\documentclass[10pt]{article}

\usepackage[english]{babel}
\usepackage[utf8]{inputenc}
\usepackage{amsmath,amssymb}
\usepackage[hidelinks]{hyperref}
\usepackage{parskip}
\usepackage{graphicx}
\usepackage{mathtools}
\usepackage{listings}
\usepackage{enumerate}
\usepackage{amsthm}
\usepackage{amsmath}
\usepackage{indentfirst} 
\newtheorem{theorem}{Theorem}
\newtheorem{lemma}{Lemma}
\newtheorem{proposition}{Proposition}

\newtheorem{remark}{Remark}
\usepackage{enumitem}
\setlist[enumerate]{leftmargin=*}

\usepackage[top=1in,  left=1in, right=1in, bottom=1in]{geometry}
\usepackage{authblk}
\usepackage{xcolor}
\usepackage{float}
\usepackage{bm}
\allowdisplaybreaks[4]

\newcommand{\RN}[1]{%
  \textup{\uppercase\expandafter{\romannumeral#1}}%
}

\title{\vspace*{-1.5cm}\huge{Derivation of an effective plate theory for parallelogram origami from bar and hinge elasticity}}
\author[1]{Hu Xu}
\author[2]{Ian Tobasco\footnote{i.tobasco@rutgers.edu}}
\author[1]{Paul Plucinsky\footnote{plucinsk@usc.edu}}
\affil[1]{\small{Aerospace and Mechanical Engineering, University of Southern California, Los Angeles, California 90089, USA}}
\affil[2]{\small{Mathematics, Rutgers University, New Brunswick, New Jersey  08854, USA}}
\date{}

\begin{document}

\maketitle

\begin{abstract} Periodic origami patterns made with repeating unit cells of creases and panels bend and twist in complex ways. In principle, such soft modes of deformation admit a simplified asymptotic description in the limit of a large number of cells. Starting from a bar and hinge model for the elastic energy of a generic four parallelogram panel origami pattern, we derive a complete set of geometric compatibility conditions identifying the pattern's soft modes in this limit. The compatibility equations form a system of partial differential equations constraining the actuation of the origami's creases (a scalar angle field) and the relative rotations of its unit cells (a pair of skew tensor fields). We show that every solution of the compatibility equations is the limit of a sequence of soft modes ---  origami deformations with finite bending energy and negligible stretching. Using these sequences, we derive a plate-like theory for parallelogram origami patterns with an explicit coarse-grained quadratic energy depending on the gradient of the crease-actuation and the relative rotations of the cells. Finally, we illustrate our theory in the context of two well-known origami designs: the Miura and Eggbox patterns. Though these patterns are distinguished in their anticlastic and synclastic bending responses, they show a universal twisting response. General soft modes captured by our theory involve a rich nonlinear interplay between actuation, bending and twisting, determined by the underlying crease geometry.

\end{abstract}

\section{Introduction}
Mechanical metamaterials are many-body elastic systems that derive their bulk elastic properties primarily from the complex geometry and arrangement of their building blocks. Origami-inspired variants, composed of patterns of stiff  panels and flexible folds,  achieve a wide range of nonlinear mechanical properties. A well-designed folding pattern can produce an origami metamaterial that deploys from flat to folded flat \cite{feng2020designs,koryo1985method,lang2018rigidly,tachi2009generalization}, is bistable \cite{faber2018bioinspired,feng2020helical,kresling2012origami,liu2019invariant}, is soft in one configuration but stiff in another \cite{filipov2015origami}, exhibits tunable stiffness \cite{silverberg2014using,zhai2018origami} and/or achieves target shapes upon folding \cite{dang2022inverse,dieleman2020jigsaw,dudte2016programming}.  These alluring properties have made origami a popular design motif in applications spanning soft robotics \cite{kim2018printing,rafsanjani2019programming,wu2021stretchable}, biomedical devices \cite{kuribayashi2006self,velvaluri2021origami} and space structures and habitats \cite{melancon2021multistable,pehrson2020self,zirbel2013accommodating}. However, they also raise many fundamental questions in mechanics: what is the interplay between geometry and elasticity in these materials? How should we characterize this interplay? How do we model it?

This paper answers these questions for parallelogram origami, a class of metamaterials composed of repeating unit cells of folds and parallelogram panels.  Fig.\;\ref{Fig:IntroFigure} shows the most well-known example --- the Miura origami \cite{koryo1985method}, whose corrugated pattern is built from a single parallelogram. Fig.\;\ref{Fig:IntroFigure}(a) illustrates a key kinematic property: the Miura can fold from flat to folded flat  through a continuous family of  rigid-body motions of the panels coordinated by the flexible folds called a mechanism \cite{pellegrino1986matrix} or floppy mode \cite{lubensky2015phonons}. Like the Miura, every parallelogram orgami pattern possesses a single degree-of-freedom (DOF) mechanism. This property makes these origami  soft  to a wide range of loads, and enables deformations that are far from intuitive. For instance, the  Miura origami in Fig.\;\ref{Fig:IntroFigure} makes a saddle shape  when bent along its centerline (Fig.\;\ref{Fig:IntroFigure}(b)); it makes a bow-tie  shape when pinched (Fig.\;\ref{Fig:IntroFigure}(c-d)). Importantly, these deformations are not mechanisms. Instead, we  call them \textit{soft modes}. In a soft mode, a mechanical metamaterial can have a small amount of strain, with energy that is sub-volumetric rather than having classical bulk scaling. 

\begin{figure}[t!]
\centering
\includegraphics[width=0.85\textwidth]{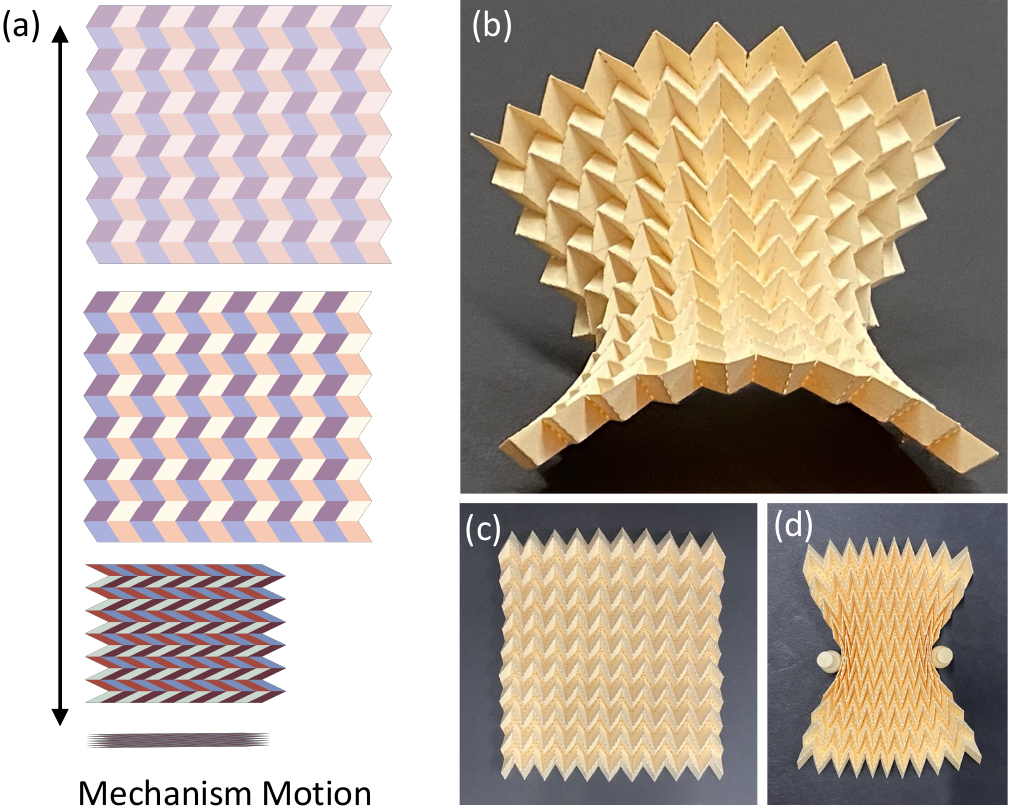} 
\caption{Miura origami  --- the canonical example in the family of parallelogram origami patterns. (a) The Miura origami can fold from flat to folded flat by a mechanism motion. (b-d) A paper model exemplifying the Miura's non-mechanistic soft modes of deformation.}
\label{Fig:IntroFigure}
\end{figure}

From an energetic point of view, it is clear why soft modes occur in parallelogram origami \textemdash\;the panels can bend when the pattern is generically deformed, and bending is energetically cheap as compared with stretching. The puzzle lies with understanding how the bending of many individual panels leads to a particular gross shape. The explanation lies in the interplay between the geometry of the crease pattern and the elasticity of the panels.  In the literature, this interplay is typically captured  through bar and hinge simulations \cite{filipov2017bar,liu2017nonlinear,schenk2011origami} or related spring methods \cite{deng2020characterization,zhou2023low}. The bar and hinge approach replaces the origami with assemblies of hinges and bars that reflect the creases and panels. Elasticity is accounted for by modeling bars/hinges as linear/torsional springs.  This approach has proven successful for simulating particular soft modes, but it leaves much to be desired theoretically. It does not, for instance, explain why the Miura origami makes a saddle shape upon bending, nor does it identify the class of all possible soft modes.  

The first theoretical discussions of soft modes in origami metamaterials are due to  Schenk and Guest \cite{schenk2013geometry} and  Wei et al.\;\cite{wei2013geometric} in their investigations of Miura origami.  By enriching the origami's rigid kinematics with  ``bending" creases along the panel diagonals, both sets of authors demonstrated  a connection between panel bending and crease folding  through a local compatibility argument involving slightly bent Miura  cells. Strikingly, this local argument shows that the Miura's in-plane and out-of-plane Poisson's ratios are equal and opposite. Since the Miura is auxetic with a negative in-plane Poisson's ratio, this identity is consistent with the anticlastic saddle-shape in Fig.\;\ref{Fig:IntroFigure}(b).

Recent lines of research have advanced this idea.  Nassar et al.\;\cite{nassar2017curvature} demonstrated the same Poisson's ratio link for Eggbox origami and postulated  that this link, when appropriately coupled to the Gauss and Codazzi--Mainardi equations from differential geometry \cite{do2016differential}, yields a set of partial differential equations (PDE) for the \emph{effective} or cell-averaged deformation of the origami. They later expanded on this idea for general parallelogram origami \cite{lebee2018fitting,nassar2022strain}, showing by a local fitting argument that a suitable interpretation of the Poisson's ratio link holds in general, and that the aforementioned coupling to classical differential geometry appears robust. McInerney et al.\;\cite{mcinerney2022discrete} also established the Poisson's ratio link for general parallelogram origami using a symmetry analysis applied to infinitesimal isometries. They too appear to be building towards an effective description of soft modes   \cite{coulais2018characteristic,czajkowski2022conformal,czajkowski2023orisometry,czajkowski2022duality}. 

What is missing from this literature is a \emph{self-consistent} coarse-graining rule --- one that uses averaging to derive a continuum field theory for the effective deformations of parallelogram origami, and then justifies that rule by producing soft modes corresponding to every possible solution of the continuum equations. 
Our previous work \cite{zheng2022continuum} produced such a rule in the simpler setting of planar kirigami. There, we linked an angle field quantifying the actuation of the kirigami's panels and slits to its overall deformation through a metric constraint.  We  justified this rule by showing that for any effective deformation and angle field satisfying this metric constraint,  there is a corresponding sequence of planar kirigami deformations converging to the given effective deformation, with elastic energy far less than bulk. We take a similar approach here. 

This paper goes beyond the metric constraint of planar kirigami to link the kinematics of slightly bent origami cells to  canonical tensor fields from surface geometry. Rather than dealing with first and second fundamental forms and applying the Gauss and Codazzi--Maindari equations, we prefer to use Cartan's elegant (and equivalent) method of moving frames \cite{cartan2001riemannian,ciarlet2008new,mardare2003fundamental,mardare2007systems}. We find Cartan's method to be ready-made for coarse graining origami. The reason is simple: when trying to fit together a  neighborhood of slightly bent origami cells, one must rotate each cell relative to the other by a small amount; these relative infinitesimal rotations are skew tensors in Cartan's theory (they are not fundamental forms). Upon coarse graining, we derive a coupled set of partial differential and algebraic equations (Eq.\;(\ref{eq:SurfaceTheory}) below) involving an angle field quantifying the cell-wise actuation of the origami's creases, along with two skew tensor fields we parameterize as vectors. When combined with a metric constraint similar to the kirgiami setting, this yields an effective surface theory for parallelogram origami. Formulas for the first and second fundamental forms follow quickly in the discussion section (see Eq.\;(\ref{eq:bendingTwistingModes})).

Our starting point is a more or less standard bar and hinge model for the elastic energy of a generic parallelogram origami pattern, which has a fully corrugated (non-planar) stress-free configuration. The total energy is the sum of panel stretching and bending energies, as well as a crease folding energy, all of which are formulated to allow for any reasonable material and geometric nonlinearity. We coarse grain this energy in the finely patterned limit where the number of cells goes to infinity while the extent of the overall pattern is fixed, assuming the following scale separation  of stiffnesses:
\begin{align*}
\text{the bulk stiffness of the panels  $\gg$  the bending stiffness of the panels $\gg$  the stiffness of the folds}. 
\end{align*}
Thus, panel stretching is severely limited, while panel bending is allowed; in comparison, the folds can actuate freely. Asymptotic analysis of the stretching and bending strains in the model leads to our effective surface theory. To justify it, we produce a globally-defined  ``recovery sequence'' associated to each of its solutions, i.e., a sequence of finite-bending and nearly-zero stretching origami deformations that converges to the given effective deformation, along with an asymptotic formula for its energy. This is the most delicate part of our work, as it requires ensuring that the stretching of our sequence is actually asymptotically negligible compared to its bending. 
In the end, we derive an effective plate theory for parallelogram origami, with a continuum bending energy that is quadratic in the gradient of the angle field and the skew tensor fields described above. 
Theorem~\ref{MainTheorem} collects the precise statements of our results; see also Eq.\;(\ref{eq:finalEnergyForm}) and the surrounding text for further discussion.

The mathematically inclined reader may recognize our terminology of recovery sequences as half of the definition of $\Gamma$-convergence \cite{braides2002gamma,dal2012introduction}, a widely used method for producing rigorous coarse-graining or homogenization results for non-convex energies. We take inspiration from prior work on small-displacement homogenization of metamaterial structures, such as the truss beam structures with pentographic sub-structures considered by Seppecher and others \cite{alibert2003truss}. Their work showed by  $\Gamma$-convergence that these structures homogenize to strain gradient and higher order elastic theories, assuming a linear  response.  There has since been a systematic effort to coarse grain the linear response of a myriad of truss analogs to mechanism-based metamaterials \cite{abdoul2018strain,dell2019pantographic,durand2022predictive,seppecher2011linear}. 
A key distinction in our work is \textit{nonlinearity} --- we address the physically relevant, fully nonlinear problem of origami deformations with finite panel rotations. While other such research is starting to emerge, for instance, on two-dimensional Kagome and checkerboard-type metmaterials \cite{dull2023variational, li2023some}, there is still much to be discovered. Rigorous analysis that embraces the nonlinearity of mechanism-based metamaterials can be a powerful tool for deriving predictive  continuum theories  with minimal ad hoc fitting parameters. We hope this paper inspires new lines of research in this direction.

We close this introduction by noting that our effective continuum theory is not of the strain gradient type; instead, and presumably as a consequence of nonlinearity, it is better seen as a  ``mechanism gradient" theory, namely, a type of \textit{generalized elastic continuum} that associates elasticity to fields parameterizing spatial variations of a locally mechanistic response. As a classical subject codified by Eringen in the 1960s--70s \cite{eringen2012microcontinuum}, generalized elasticity introduces auxiliary fields beyond deformations as a constitutive hypothesis aimed at capturing the macroscopic effect of microscale rearrangements. 
Recently, this approach has been applied to mechanical metamaterials including origami and kirigami \cite{kadic20193d,lakes2022extremal, nassar2020microtwist,saremi2020topological,sarhil2023modeling,sarhil2023size,sun2020continuum}. However, this literature has focused almost exclusively on the linear/small displacement response regime, which is not suitable for dealing with the large, macroscopic shape changes achievable by these systems. 
Our work on parallelogram origami and our related work on kirigami \cite{zheng2022continuum, zheng2023modelling} add to the discussion by deriving from first principles new generalized elastic theories that fully express the nonlinear coupling between design, deformation and actuation. We discuss generalized models further at the end.

The rest of this paper is organized as follows. Section \ref{sec:Overview} introduces the general class of parallelogram origami patterns, describes our bar and hinge energies, and states our main coarse-graining result in Theorem \ref{MainTheorem}. Sections \ref{sec:LinAlgebraSec}-\ref{sec:DerivationSec} and Appendix \ref{sec:ExistencePDE} provide the technical work needed to prove Theorem \ref{MainTheorem}. Section \ref{sec:Examples} concludes with examples showing predictions of our theory for the Miura and Eggbox origami patterns, along with a discussion of our results in the context of other continuum theories and a brief outlook for future research.

\section{Problem formulation and main results}\label{sec:Overview}

\subsection{Parallelogram origami designs and mechanism kinematics}\label{ssec:DesignKin}
We study a class of origami patterns made by tessellating a unit cell composed of four parallelograms and eight creases.  Besides the usual global Euclidean motions, generic such parallelogram origami patterns possess a continuous one parameter family of mechanism deformations, made up of rigid panel motions that fold the creases.  This section defines this class of patterns and describes their mechanism kinematics. 

Each four-panel parallelogram origami pattern has a reference unit cell  parameterized by four design vectors $\mathbf{t}_i^r\in\mathbb{R}^3,$ $i=1,2,3,4$. These vectors label the creases in a counterclockwise fashion around the central vertex of the cell as in Fig.\;\ref{Fig:idepat}(a), and they also label the boundaries of the cell. The grey panel in Fig.\;\ref{Fig:idepat}(a), for instance, has two parallel sides described by $\mathbf{t}_1^r$ and two other sides described by $\mathbf{t}_4^r$. 
Additional restrictions on the design vectors ensure the existence of a well-defined mechanism motion. We assume that the creases have a well-defined mountain-valley assignment, achieved by either partially folding a flat reference pattern or by considering a non-Euclidean pattern (one whose sector angles do not sum to $2\pi$). Specifically, we impose the nondegeneracy conditions
\begin{equation}
    \begin{aligned}\label{eq:tangentsConstraints}
      \mathbf{t}_1^r \cdot ( \mathbf{t}_2^r \times \mathbf{t}_3^r) \neq 0, \quad   \mathbf{t}_2^r \cdot (\mathbf{t}_3^r \times \mathbf{t}_4^r) \neq 0, \quad \mathbf{t}_3^r \cdot (\mathbf{t}_4^r \times \mathbf{t}_1^r) \neq 0, \quad \mathbf{t}_4^r \cdot (\mathbf{t}_1^r \times \mathbf{t}_2^r) \neq 0.
    \end{aligned}
\end{equation}
In addition, we exclude self-intersecting panels through the restrictions 
\begin{equation}
    \begin{aligned}\label{eq:selfConstraint}
    \text{conv}(\{\mathbf{0}, \mathbf{t}_1^r, \mathbf{t}_2^r\}) \cap \text{conv}(\{\mathbf{0}, \mathbf{t}_3^r, \mathbf{t}_4^r\}) = \{ \mathbf{0} \} , \quad \text{conv}(\{\mathbf{0}, \mathbf{t}_1^r, \mathbf{t}_4^r\}) \cap \text{conv}(\{\mathbf{0}, \mathbf{t}_2^r, \mathbf{t}_3^r\}) = \{ \mathbf{0} \} ,
    \end{aligned}
\end{equation}
where $\text{conv}(\cdot)$ is the convex hull of a set of points. 
Finally, to fix a reference frame, we take $\mathbf{u}_0 = \mathbf{t}_1^r - \mathbf{t}_3^r$ and $\mathbf{v}_0 = \mathbf{t}_2^r - \mathbf{t}_4^r$ to span the $\mathbf{e}_1, \mathbf{e}_2$-plane and require $\{\mathbf{u}_0, \mathbf{v}_0, \mathbf{e}_3\}$ to be a right-handed basis of $\mathbb{R}^3$: 
\begin{equation}
    \begin{aligned}\label{eq:bravaisConstraints}
    \mathbf{u}_0 \cdot \mathbf{e}_3 = 0, \quad \mathbf{v}_0 \cdot \mathbf{e}_3 = 0, \quad  \mathbf{e}_3 \cdot (\mathbf{u}_0 \times \mathbf{v}_0) > 0. 
    \end{aligned}
\end{equation}

Tessellating the resulting unit cell along the Bravais lattice vectors $\mathbf{u}_0$ and $\mathbf{v}_0$ produces a perfectly connected periodic pattern of parallelogram panels. This defines the reference origami pattern we use throughout our analysis (see Fig.\;\ref{Fig:idepat}(c)). 
Its panels and unit cell are given by
\begin{equation}
    \begin{aligned}\label{eq:OmegaCell}
      &(\text{the panels:})  && \mathcal{P}_i = \text{conv}(\{ \mathbf{0}, \mathbf{t}_i^r, \mathbf{t}_{\sigma(i)}^r ,  \mathbf{t}_i^r + \mathbf{t}_{\sigma(i)}^r \}), \quad i = 1,\ldots,4, \\
      &(\text{the unit cell:})  & &\Omega_{\text{cell}} = \mathcal{P}_1 \cup \mathcal{P}_2 \cup \mathcal{P}_3 \cup \mathcal{P}_4,
    \end{aligned}
\end{equation}
where $\sigma(\cdot)$ is a cyclic permutation of the set $\{ 1,2,3,4\}$. 
The overall pattern is
\begin{equation}
    \begin{aligned}\label{eq:tessellation}
      \mathcal{T}_{\text{ori}} = \{ \Omega_{\text{cell}} + i \mathbf{u}_0 + j \mathbf{v}_0 \colon (i, j) \in \mathbb{Z}^2\}.
    \end{aligned}
\end{equation}

\begin{figure}[t]
\centering
\includegraphics[width=0.7\textwidth]{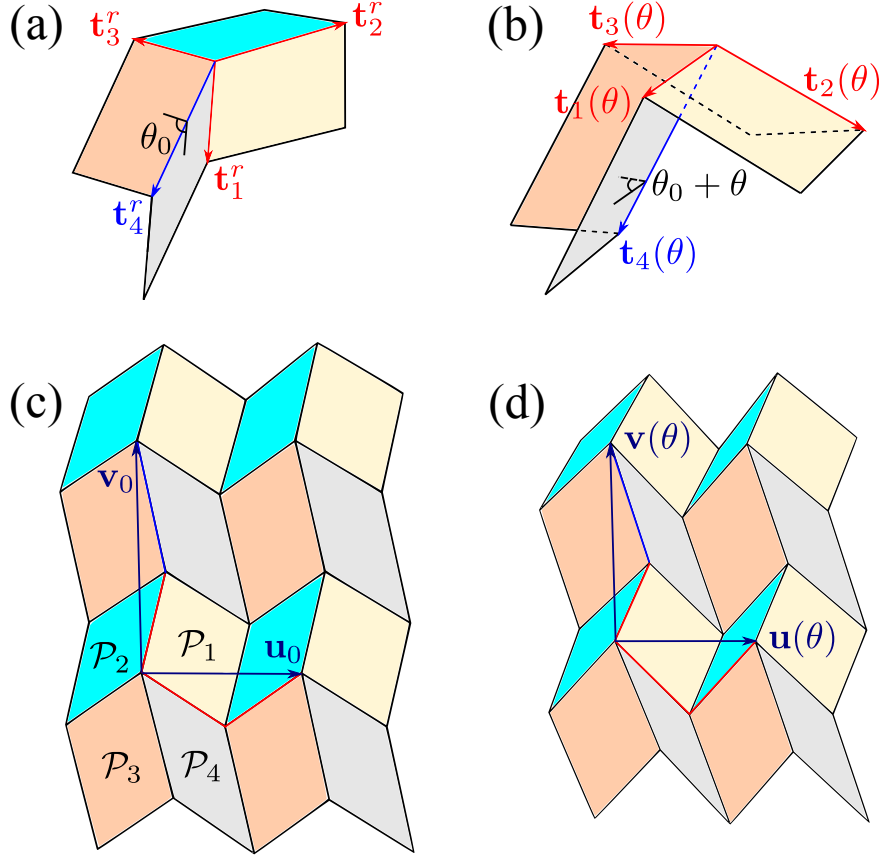} 
\caption{Labeling of the design and ideal mechanism kinematics of parallelogram origami.  (a) The unit cell in  its reference configuration  and (b)  after a mechanism deformation. (c) The pattern in  its reference configuration  and (d) after a mechanism deformation.  }
\label{Fig:idepat}
\end{figure}

We come now to the mechanism motion of parallelogram origami. Section \ref{sec:LinAlgebraSec} develops the mechanism kinematics in detail; we report the main points here as they are needed to state our coarse-graining result (Theorem \ref{MainTheorem}).
Let  $\theta_0$ denote the dihedral angle of the $\mathbf{t}_4^r$-crease in the reference unit cell, shown in Fig.\;\ref{Fig:idepat}(a). This angle belongs to either $(0,\pi)$ or $(\pi, 2 \pi)$ depending on the mountain-valley assignments, but it cannot take a value of $0$ or $\pi$ due to Eq.\;(\ref{eq:tangentsConstraints}).  In Section \ref{ssec:MechKin}, we verify that all mechanism deformations of this cell that preserve its mountain-valley assignment  are  parameterized completely up to rigid motion by folding this crease to a value $\theta_0 + \theta$ for an interval
\begin{equation}
    \begin{aligned}
      \theta \in (\theta^-, \theta^+), \quad  \theta^- = \theta^{-}(\mathbf{t}_1^r, \ldots, \mathbf{t}_4^r) < 0, \quad \theta^+ = \theta^{+}(\mathbf{t}_1^r, \ldots, \mathbf{t}_4^r) > 0.     \end{aligned}
\end{equation}
Fig.\ \ref{Fig:idepat}(b) shows the deformed tangents under this actuation. Each tangent is set by the folded $\mathbf{t}_4^r$-crease, and thus is parameterized in a mechanism by $\theta$ up to an overall rigid rotation. We denote the deformed tangents by $\mathbf{t}_i(\theta)$, $i = 1,\ldots, 4$. The interval  $(\theta^-, \theta^+)$  is  the largest open interval that preserves the reference mountain-valley assignments, per the restrictions
\begin{equation}
    \begin{aligned}\label{eq:MVPreserve}
      [\mathbf{t}_i^r \cdot ( \mathbf{t}_j^r \times \mathbf{t}_k^r)][\mathbf{t}_i(\theta) \cdot ( \mathbf{t}_j(\theta) \times \mathbf{t}_k(\theta))] >0, \quad ijk \in \{123,234,341,412\} \quad \text{for all }   \theta \in (\theta^-, \theta^+).
    \end{aligned}
\end{equation}
While it is often possible to include mechanism kinematics that change the mountain-valley assignment, we do not pursue this here, other than to point out that it would have implications for coarse-grained elastic moduli: a physical sample can easily ``pop" in and out of two different mountain-valley assignments when one of the creases is close to flat \cite{hanna2014waterbomb,silverberg2014using}. This feature is ruled out by  Eq.\;(\ref{eq:MVPreserve}).

Under a mechanism, the lattice vectors $\mathbf{u}_0 = \mathbf{t}_1^r - \mathbf{t}_3^r$ and $\mathbf{v}_0 = \mathbf{t}_2^r - \mathbf{t}_4^r$ transform to new lattice vectors through a map  $(\mathbf{u}_0, \mathbf{v}_0) \mapsto (\mathbf{u}(\theta), \mathbf{v}(\theta))$ given by
\begin{equation}
    \begin{aligned}
      \mathbf{u}(\theta) = \mathbf{t}_1(\theta) - \mathbf{t}_3(\theta), \quad \mathbf{v}(\theta) = \mathbf{t}_2(\theta) - \mathbf{t}_4(\theta).
    \end{aligned}
\end{equation}
The deformed cell remains a two-by-two array of connected and non-intersecting parallelograms, and tessellating it along $\mathbf{u}(\theta)$ and $\mathbf{v}(\theta)$  produces a new periodic pattern of parallelogram panels.  Passing through $\theta \in (\theta^{-}, \theta^+)$  yields the mechanism motion of the overall pattern (see Proposition \ref{MechProp} for a formal statement; similar observations appear in \cite{lang2017twists,mcinerney2022discrete,nassar2022strain, pratapa2019geometric}). Fig.\;\ref{Fig:idepat}(c-d) illustrates this transformation. 

For later use, observe that 
\begin{equation}
    \begin{aligned}
      \mathbf{u}(\theta) \cdot \mathbf{v}(\theta)& = \mathbf{u}_0 \cdot \mathbf{v}_0 \quad\text{for all } \theta \in (\theta^{-}, \theta^+),
    \end{aligned}
\end{equation} 
since the sector angles and lengths of each parallelogram panel are preserved by rigid deformations (i.e., $\mathbf{t}_i(\theta) \cdot \mathbf{t}_{j}(\theta) = \mathbf{t}_i^r \cdot \mathbf{t}_{j}^r$ for all $ij \in \{12,23,34,41\}$). With more effort, it is possible to find smooth and explicit parameterizations for the squared lengths  $|\mathbf{u}(\theta)|^2$ and $|\mathbf{v}(\theta)|^2$, though we will not need such parameterizations here. (For Euclidean parallelogram origami, one can apply ideas from \cite{feng2020designs, huffman1976curvature,lang2018rigidly,tachi2009generalization}; see  \cite{foschi2022explicit} for a discussion of Euclidean and non-Euclidean formulas.)  The scalar functions $\mathbf{u}(\theta) \cdot \mathbf{v}(\theta)$,  $|\mathbf{u}(\theta)|^2$ and $|\mathbf{v}(\theta)|^2$  are intrinsic quantities describing the distortion of the unit cell under a mechanism.  We are free to choose certain extrinsic quantities, like a frame of reference for the deformed tangents $\mathbf{t}_1(\theta), \ldots, \mathbf{t}_4(\theta)$. As in Eq.\;(\ref{eq:bravaisConstraints}), we require that
\begin{equation}
    \begin{aligned}
      \mathbf{u}(\theta) \cdot \mathbf{e}_3 = 0, \quad \mathbf{v}(\theta) \cdot \mathbf{e}_3 = 0, \quad \mathbf{e}_3 \cdot (\mathbf{u}(\theta) \times \mathbf{v}(\theta)) > 0  \quad  \text{for all }  \theta \in (\theta^{-}, \theta^+).
    \end{aligned}
\end{equation}
This emphasizes the effectively planar nature of the mechanism motion.

Finally, we define the shape tensor and Poisson's ratio of a given parallelogram origami design. Let $\tilde{\mathbf{u}}_0$ and $\tilde{\mathbf{v}}_0 \in \mathbb{R}^2$ denote the usual (orthogonal) projection of $\mathbf{u}_0$ and $\mathbf{v}_0$ to $\mathbb{R}^2$, with $\mathbf{u}_0=(\tilde{\mathbf{u}}_0,0)$ and $\mathbf{v}_0=(\tilde{\mathbf{v}}_0,0)$. The \emph{shape tensor} $\mathbf{A}_{\text{eff}}(\theta) \in \mathbb{R}^{3 \times 2}$ is the unique linear transformation such that
\begin{equation}
    \begin{aligned}\label{eq:shapeTensor}
      \mathbf{A}_{\text{eff}}(\theta)\tilde{\mathbf{u}}_0 = \mathbf{u}(\theta), \quad   \mathbf{A}_{\text{eff}}(\theta)\tilde{\mathbf{v}}_0 = \mathbf{v}(\theta).
    \end{aligned}
\end{equation}
For the Poisson's ratio, introduce the lattice-direction strain measures
\begin{equation}
\begin{aligned}
\varepsilon_{\mathbf{u}_0} = \frac{| \mathbf{u}(\theta+\delta \theta)| - |\mathbf{u}(\theta)|}{|\mathbf{u}(\theta)|}, \quad \varepsilon_{\mathbf{v}_0} = \frac{| \mathbf{v}(\theta+\delta \theta)| - |\mathbf{v}(\theta)|}{|\mathbf{v}(\theta)|}
\end{aligned}
\end{equation}
associated to a small perturbation $\delta \theta$, and define 
\begin{equation}
    \begin{aligned}\label{eq:PoissonsRatioDesign}
\nu(\theta) := \lim_{\delta \theta \rightarrow 0 }- \frac{\varepsilon_{\mathbf{u}_0}}{\varepsilon_{\mathbf{v}_0}} = -\Big(\frac{|\mathbf{v}(\theta)|^2}{|\mathbf{u}(\theta)|^2} \Big) \Big( \frac{\mathbf{u}'(\theta) \cdot \mathbf{u}(\theta)}{\mathbf{v}'(\theta) \cdot \mathbf{v}(\theta)}\Big).
    \end{aligned}
\end{equation}
As we verify later on in Lemma \ref{SignsLemma}, the sign of this Poisson-like ratio is  independent of $\theta$ and depends only on the design vectors of the unit cell. We use these definitions when stating our main result in Theorem \ref{MainTheorem}.

\subsection{Bar and hinge models for parallelogram origami}\label{ssec:BarHinge}


We now account for the elasticity of a general parallelogram origami pattern. 
First, to fix a reference domain, let $\Omega\subset\mathbb{R}^2$ be a simply connected planar reference domain with a smooth boundary, and assume after non-dimensionalization that $\Omega$  has characteristic length $\sim 1$ and that the undeformed configuration of the origami is a graph over this set with unit cells of length $\ell \ll 1$. Building off of the tessellation in Eq.\;(\ref{eq:tessellation}), we define the origami reference domain  as
\begin{equation}
    \begin{aligned}\label{eq:OriRef}
      \Omega_{\text{ori}}^{(\ell)} :=  \big\{ \ell \Omega_{\text{cell}} + \ell ( i \mathbf{u}_0 + j \mathbf{v}_0) \colon \ell ( i \mathbf{u}_0 + j \mathbf{v}_0) \in \Omega   , i, j \in \mathbb{Z} \big\}. 
    \end{aligned}
\end{equation}
Fig.\;\ref{Fig:oridom} shows a stylized example.

Next, we label the vertices of the origami.  Given $\Omega_{\text{ori}}^{(\ell)}$ as above, fix a connected subset $\mathcal{I}^{(\ell)}$ of $\mathbb{Z}^2$ such that $\{ \mathbf{x}_{i,j}^{(\ell)} \in \mathbb{R}^3 \colon (i,j) \in \mathcal{I}^{(\ell)} \}$ bijectively labels the vertices of the reference origami, and whose nearest neighbor relationships match those of the origami (see Fig.\;\ref{Fig:oridom}).  Such a set $\mathcal{I}^{(\ell)}$ exists given the topology of the crease pattern. We often write $\{ \mathbf{x}_{i,j}^{(\ell)}\} \equiv \{ \mathbf{x}_{i,j}^{(\ell)} \in \mathbb{R}^3 \colon (i,j) \in \mathcal{I}^{(\ell)} \}$ for short. 
With this labeling, elastic deformations of the origami are maps
\begin{equation}
    \begin{aligned}
      \mathbf{x}_{i,j}^{(\ell)} \mapsto \mathbf{y}_{i,j}^{(\ell)}
    \end{aligned}
\end{equation}
defined for $(i,j) \in \mathcal{I}^{(\ell)}$. The deformed origami pattern is called $\{\mathbf{y}_{i,j}^{(\ell)}\}$ for short.

\begin{figure}[t!]
\centering
\includegraphics[width=1\textwidth]{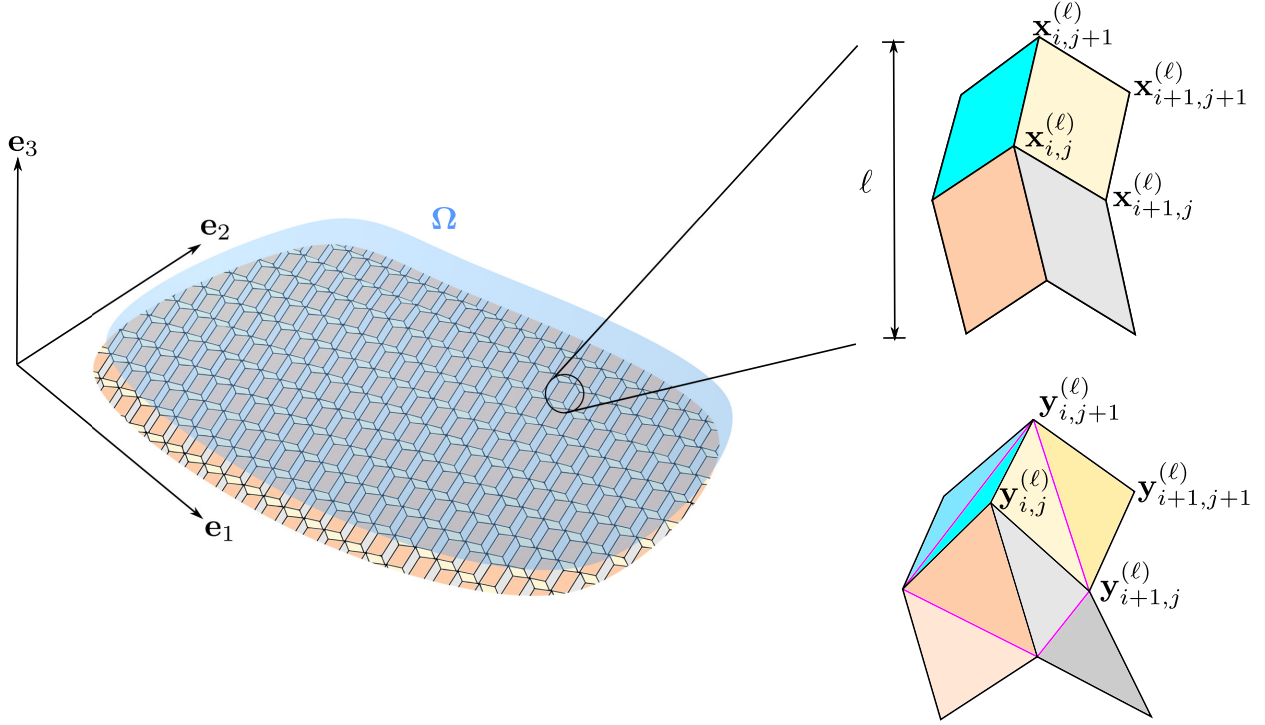} 
\caption{Reference domain and labeling for the bar and hinge model.  The  parallelogram origami  is finely patterned with unit cell of characteristic length $\sim \ell$. It covers the  effective domain $\Omega$ in the $\mathbf{e}_1$,$\mathbf{e}_2$-plane. The reference and deformed vertices are labeled by $\mathbf{x}_{i,j}^{(\ell)}$ and $\mathbf{y}_{i,j}^{(\ell)}$ with the nearest neighbor description shown.}
\label{Fig:oridom}
\end{figure}

We are now ready to define the general class of bar and hinge energies we coarse-grain. As is standard practice in the engineering origami literature \cite{filipov2017bar,liu2017nonlinear,schenk2011origami}, we consider three stored energy functions to account for the stretching, bending and folding of the origami panels and creases: 
\begin{equation}
\begin{aligned}\label{eq:TotalBarHingeEnergy}
  \mathcal{E}^{(\ell)}_{\text{tot}}(\{ \mathbf{y}_{i,j}^{(\ell)} \} ) = \mathcal{E}^{(\ell)}_{\text{str}}(\{ \mathbf{y}_{i,j}^{(\ell)} \} ) + \mathcal{E}^{(\ell)}_{\text{bend}}(\{ \mathbf{y}_{i,j}^{(\ell)} \} )+  \mathcal{E}^{(\ell)}_{\text{fold}}(\{ \mathbf{y}_{i,j}^{(\ell)} \} ).  
\end{aligned}
\end{equation}
We define each of these below.

The stretching energy $\mathcal{E}^{(\ell)}_{\text{str}}$ is generically a function of the panel strains. For the $(i,j)$-vertex in the pattern, we have four panel strains $\varepsilon_{i,j,k}^{(\ell)}$, $k=1,\dots,4$, given by
\begin{equation}
    \begin{aligned}\label{eq:PanelStrain}
      &\varepsilon^{(\ell)}_{i,j,1}  = \frac{|\mathbf{y}^{(\ell)}_{i+1,j}- \mathbf{y}^{(\ell)}_{i,j}| - |\mathbf{x}_{i+1,j}^{(\ell)} - \mathbf{x}^{(\ell)}_{i,j}|}{|\mathbf{x}_{i+1,j}^{(\ell)} - \mathbf{x}^{(\ell)}_{i,j}|}, && \varepsilon^{(\ell)}_{i,j,2} = \frac{|\mathbf{y}^{(\ell)}_{i,j+1}- \mathbf{y}^{(\ell)}_{i,j}| - |\mathbf{x}_{i,j+1}^{(\ell)} - \mathbf{x}^{(\ell)}_{i,j}|}{|\mathbf{x}_{i,j+1}^{(\ell)} - \mathbf{x}^{(\ell)}_{i,j}|},  \\
       &\varepsilon^{(\ell)}_{i,j,3}  = \frac{|\mathbf{y}^{(\ell)}_{i+1,j+1}- \mathbf{y}^{(\ell)}_{i,j}| - |\mathbf{x}_{i+1,j+1}^{(\ell)} - \mathbf{x}^{(\ell)}_{i,j}|}{|\mathbf{x}_{i+1,j+1}^{(\ell)} - \mathbf{x}^{(\ell)}_{i,j}|},  &&  \varepsilon^{(\ell)}_{i,j,4} = \frac{|\mathbf{y}^{(\ell)}_{i,j+1}- \mathbf{y}^{(\ell)}_{i+1,j}| - |\mathbf{x}_{i,j+1}^{(\ell)} - \mathbf{x}^{(\ell)}_{i+1,j}|}{|\mathbf{x}_{i,j+1}^{(\ell)} - \mathbf{x}^{(\ell)}_{i+1,j}|}.
     \end{aligned}
\end{equation}
We also use $(i,j)$ to label the panel with $\mathbf{x}_{i,j}^{(\ell)}$ as its lower left corner point. 
The stretching energy is then
\begin{equation}
    \begin{aligned}
      \mathcal{E}_{\text{str}}^{(\ell)}( \{ \mathbf{y}_{i,j}^{(\ell)}\}) =   \sum_{(i,j) \in \mathcal{I}^{(\ell)}}  A_{i,j}^{(\ell)}   \ell^{\alpha_s}   \Phi^{\text{str}}_{i,j}( \varepsilon^{(\ell)}_{i,j,1}, \varepsilon^{(\ell)}_{i,j,2},\varepsilon^{(\ell)}_{i,j,3},\varepsilon^{(\ell)}_{i,j,4})
    \end{aligned}
\end{equation}
where $\ell^{\alpha_{s}} \Phi^{\text{str}}_{i,j}(\varepsilon_1,\ldots, \varepsilon_4)$ is the stretching energy per unit reference area of the $(i,j)$-panel. Its area is $A_{i,j}^{(\ell)}$ before deformation. The functions $\Phi^{\text{str}}_{i,j}$ are smooth, non-negative and equal zero if and only if $(\varepsilon_1, \ldots, \varepsilon_4) = \mathbf{0}$. They also have the periodicity property $\Phi^{\text{str}}_{i+2,j} = \Phi^{\text{str}}_{i,j+2} = \Phi^{\text{str}}_{i,j}$, consistent with the pattern. Likewise, the reference areas satisfy $A_{i,j}^{(\ell)} = A_{i+2,j}^{(\ell)} = A^{(\ell)}_{i,j+2}$. Finally, the dimensionless ``stretching modulus" $\ell^{\alpha_s}$ quantifies the relative stiffness of stretching as compared to bending and folding.  We explain how we select the exponent $\alpha_s$ after introducing the remaining energies.

Next, we define the bending energy $\mathcal{E}^{(\ell)}_{\text{bend}}$. When the four vertices $\mathbf{x}_{i,j}^{(\ell)}$, $\mathbf{x}_{i+1,j}^{(\ell)}$, $\mathbf{x}_{i,j+1}^{(\ell)}$ and $\mathbf{x}_{i+1,j+1}^{(\ell)}$ of the $(i,j)$-panel deform to $\mathbf{y}_{i,j}^{(\ell)}$, $\mathbf{y}_{i+1,j}^{(\ell)}$, $\mathbf{y}_{i,j+1}^{(\ell)}$ and $\mathbf{y}_{i+1,j+1}^{(\ell)}$, they need not belong to a single plane. We can measure this deviation by calculating the angle between the deformed normals 
\begin{equation}
    \begin{aligned}\label{eq:getNormals}
      \mathbf{n}_{i,j,1}^{(\ell)} = \frac{(\mathbf{y}^{(\ell)}_{i+1,j} - \mathbf{y}^{(\ell)}_{i,j}) \times  (\mathbf{y}^{(\ell)}_{i,j+1} - \mathbf{y}^{(\ell)}_{i,j})}{|(\mathbf{y}^{(\ell)}_{i+1,j} - \mathbf{y}^{(\ell)}_{i,j}) \times  (\mathbf{y}^{(\ell)}_{i,j+1} - \mathbf{y}^{(\ell)}_{i,j})|} , \quad \mathbf{m}_{i,j,1}^{(\ell)} =  \frac{(\mathbf{y}^{(\ell)}_{i,j+1} - \mathbf{y}^{(\ell)}_{i+1,j+1}) \times  (\mathbf{y}^{(\ell)}_{i+1,j} - \mathbf{y}^{(\ell)}_{i+1,j+1})}{|(\mathbf{y}^{(\ell)}_{i,j+1} - \mathbf{y}^{(\ell)}_{i+1,j+1}) \times  (\mathbf{y}^{(\ell)}_{i+1,j} - \mathbf{y}^{(\ell)}_{i+1,j+1})|},
    \end{aligned}
\end{equation}
which is
\begin{equation}
    \begin{aligned}\label{eq:getBendAngles}
      \psi_{i,j,1}^{(\ell)} = \arcsin\bigg[ \frac{(\mathbf{y}^{(\ell)}_{i,j+1} - \mathbf{y}^{(\ell)}_{i+1,j})}{|\mathbf{y}^{(\ell)}_{i,j+1} - \mathbf{y}^{(\ell)}_{i+1,j}|} \cdot \Big( \mathbf{n}_{i,j,1}^{(\ell)} \times \mathbf{m}_{i,j,1}^{(\ell)}\Big) \bigg].
    \end{aligned}
\end{equation}
Alternatively, the deviation from a plane can be measured using the normals 
\begin{equation}
    \begin{aligned}
     \mathbf{n}_{i,j,2}^{(\ell)} = \frac{(\mathbf{y}^{(\ell)}_{i+1,j+1} - \mathbf{y}^{(\ell)}_{i+1,j}) \times  (\mathbf{y}^{(\ell)}_{i,j} - \mathbf{y}^{(\ell)}_{i+1,j})}{|(\mathbf{y}^{(\ell)}_{i+1,j+1} - \mathbf{y}^{(\ell)}_{i+1,j}) \times  (\mathbf{y}^{(\ell)}_{i,j} - \mathbf{y}^{(\ell)}_{i+1,j})|} , \quad \mathbf{m}_{i,j,2}^{(\ell)} =  \frac{(\mathbf{y}^{(\ell)}_{i,j} - \mathbf{y}^{(\ell)}_{i,j+1}) \times  (\mathbf{y}^{(\ell)}_{i+1,j+1} - \mathbf{y}^{(\ell)}_{i,j+1})}{|(\mathbf{y}^{(\ell)}_{i,j} - \mathbf{y}^{(\ell)}_{i,j+1}) \times  (\mathbf{y}^{(\ell)}_{i+1,j+1} - \mathbf{y}^{(\ell)}_{i,j+1})|}
    \end{aligned}
\end{equation}
and the angle 
\begin{equation}
    \begin{aligned}\label{eq:getBendAngles2}
      \psi_{i,j,2}^{(\ell)} = \arcsin\bigg[ \frac{(\mathbf{y}^{(\ell)}_{i+1,j+1} - \mathbf{y}^{(\ell)}_{i,j})}{|\mathbf{y}^{(\ell)}_{i+1,j+1} - \mathbf{y}^{(\ell)}_{i,j}|} \cdot \Big( \mathbf{n}_{i,j,2}^{(\ell)} \times \mathbf{m}_{i,j,2}^{(\ell)}\Big) \bigg].
    \end{aligned}
\end{equation}
The engineering origami literature often uses only one of the angles   $\psi_{i,j,1}^{(\ell)}$ or $\psi_{i,j,2}^{(\ell)}$ to measure panel bending, and interprets the bending kinematics of the panel in terms of an additional fold added to the panel diagonal (like in Fig.\;\ref{Fig:oridom}). For our  analysis, however, we prefer to not assume a bias towards one or another diagonal, and so we use both angles in our bending energy (a similarly unbiased calculation appears in \cite{wei2013geometric}). Our bending energy is
\begin{equation}
    \begin{aligned}\label{eq:eBend}
      \mathcal{E}_{\text{bend}}^{(\ell)}(\{ \mathbf{y}_{i,j}^{(\ell)} \} ) =  \sum_{(i,j) \in \mathcal{I}^{(\ell)}}  A_{i,j}^{(\ell)}   \ell^{\alpha_b}   \Phi^{\text{bend}}_{i,j}(\psi_{i,j,1}^{(\ell)}, \psi_{i,j,2}^{(\ell)})
    \end{aligned}
\end{equation}
where $\ell^{\alpha_b} \Phi^{\text{bend}}_{i,j}(\psi_1, \psi_2)$ is the bending energy per unit area of the $(i,j)$-panel. As with the stretching energy, the functions $\Phi_{i,j}^{\text{bend}}$ are smooth, non-negative and vanish  if and only if $(\psi_1, \psi_2)= \mathbf{0}$; they also satisfy the periodicity conditions $\Phi^{\text{bend}}_{i,j} = \Phi^{\text{bend}}_{i+2,j} = \Phi^{\text{bend}}_{i,j+2}$.  Finally, $\ell^{\alpha_b}$ quantifies the relative stiffness of bending. The exponent is specified below. 

The final term $\mathcal{E}^{(\ell)}_{\text{fold}}$ in Eq.\;(\ref{eq:TotalBarHingeEnergy}) models the energetic cost of folding the creases.  To define it, we must track how each dihedral angle changes when the pattern deforms. Again, we measure angles between deformed normals. Since there is a clear definition of dihedral angle, we may use the first set of deformed normals $\mathbf{n}_{i,j,1}^{(\ell)}$ in Eq.\;(\ref{eq:getNormals}) and their reference normals
\begin{equation}
    \begin{aligned}
      \boldsymbol{\nu}_{i,j,1}^{(\ell)} = \frac{(\mathbf{x}^{(\ell)}_{i+1,j} - \mathbf{x}^{(\ell)}_{i,j}) \times  (\mathbf{x}^{(\ell)}_{i,j+1} - \mathbf{x}^{(\ell)}_{i,j})}{|(\mathbf{x}^{(\ell)}_{i+1,j} - \mathbf{x}^{(\ell)}_{i,j}) \times  (\mathbf{x}^{(\ell)}_{i,j+1} - \mathbf{x}^{(\ell)}_{i,j})|}.
    \end{aligned}
\end{equation}
The reference and deformed horizontal folding angles at the $(i,j)$-vertex are then
\begin{equation}
    \begin{aligned}\label{eq:getBetas}
      \beta_{i,j,0} = \arcsin\bigg[\frac{(\mathbf{x}^{(\ell)}_{i+1,j} - \mathbf{x}_{i,j}^{(\ell)})}{|\mathbf{x}^{(\ell)}_{i+1,j} - \mathbf{x}_{i,j}^{(\ell)}|} \cdot  \Big(\boldsymbol{\nu}_{i,j,1}^{(\ell)} \times \boldsymbol{\nu}_{i,j-1,1}^{(\ell)} \Big)\bigg], \quad \beta_{i,j}^{(\ell)} = \arcsin\bigg[\frac{(\mathbf{y}^{(\ell)}_{i+1,j} - \mathbf{y}_{i,j}^{(\ell)})}{|\mathbf{y}^{(\ell)}_{i+1,j} - \mathbf{y}_{i,j}^{(\ell)}|} \cdot  \Big(\mathbf{n}_{i,j,1}^{(\ell)} \times \mathbf{n}_{i,j-1,1}^{(\ell)} \Big)\bigg].
    \end{aligned}
\end{equation}
Likewise, the reference and deformed vertical folding angles at the $(i,j)$-vertex are 
\begin{equation}
    \begin{aligned}\label{eq:getGammas}
      \gamma_{i,j,0}=\arcsin\bigg[\frac{(\mathbf{x}^{(\ell)}_{i,j+1} - \mathbf{x}_{i,j}^{(\ell)})}{|\mathbf{x}^{(\ell)}_{i,j+1} - \mathbf{x}_{i,j}^{(\ell)}|} \cdot  \Big(\boldsymbol{\nu}_{i,j,1}^{(\ell)} \times \boldsymbol{\nu}_{i-1,j,1}^{(\ell)} \Big)\bigg], \quad \gamma_{i,j}^{(\ell)}=\arcsin\bigg[\frac{(\mathbf{y}^{(\ell)}_{i,j+1} - \mathbf{y}_{i,j}^{(\ell)})}{|\mathbf{y}^{(\ell)}_{i,j+1} - \mathbf{y}_{i,j}^{(\ell)}|} \cdot  \Big(\mathbf{n}_{i,j,1}^{(\ell)} \times \mathbf{n}_{i-1,j,1}^{(\ell)} \Big)\bigg].
    \end{aligned}
\end{equation}
Note the reference folding angles are independent of $\ell$.
The total folding energy is 
\begin{equation}
    \begin{aligned}\label{eq:foldBarHinge}
      \mathcal{E}_{\text{fold}}^{(\ell)}( \{ \mathbf{y}_{i,j}^{(\ell)} \}) = \sum_{(i,j) \in \mathcal{I}^{(\ell)}}  A_{i,j}^{(\ell)}  \ell^{\alpha_f}  \Phi_{i,j}^{\text{fold}} \big( \beta_{i,j}^{(\ell)} - \beta_{i,j,0}, \gamma_{i,j}^{(\ell)} -\gamma_{i,j,0} \big).
    \end{aligned}
\end{equation}
The terms $\ell^{\alpha_f} \Phi^{\text{fold}}_{i,j}(\beta - \beta_0, \gamma - \gamma_0)$ give the folding energy per unit reference area of the panels.  The functions $\Phi_{i,j}^{\text{fold}}$ are smooth, non-negative and vanish if and only if $\beta = \beta_0$ and $\gamma = \gamma_0$. They satisfy the periodicity conditions $\Phi^{\text{fold}}_{i,j} = \Phi^{\text{fold}}_{i+2,j} = \Phi^{\text{fold}}_{i,j+2}$. The reference folding angles also satisfy the periodicity conditions $\beta_{i,j,0} = \beta_{i,j+2,0}  = \beta_{i+2,j,0}$ and $\gamma_{i,j,0} = \gamma_{i+2,j,0} = \gamma_{i,j+2,0}$. They encode a fully corrugated, non-planar stress-free reference configuration.
 
At this point, we have defined our general class of bar and hinge models. Since we intend to perform an asymptotic analysis of the above energies, we must now fix a scaling relationship between the stretching, bending and folding moduli $\ell^{\alpha_s}$, $\ell^{\alpha_b}$ and $\ell^{\alpha_f}$. We do so to enforce the scale separation of stiffnesses mentioned in the introduction, which manifests here via the asymptotic relationships
 \begin{equation}
\begin{aligned}\label{eq:scalingAssumptions}
  \ell^{\alpha_{s}} \gg \ell^{\alpha_b} \gg \ell^{\alpha_f} \quad \text{ when } \ell \ll 1 \quad (\text{i.e., for finely-patterned origami}).
\end{aligned}
\end{equation}
From an engineering point of view, these scalings are natural: stretching a panel should require much more energy than bending it,  which in turn should require more energy than folding a crease. 
These relationships are ensured by the requirement that
 \begin{equation}
     \begin{aligned}\label{eq:alphaAssumptions}
       \alpha_s \in (-4,-2), \quad \alpha_b = -2, \quad \alpha_f > 0.
     \end{aligned}
 \end{equation}
In fact, Eq.\;(\ref{eq:alphaAssumptions}) does more. 
For reasons that will become clear only after the statement of Theorem \ref{MainTheorem}, we shall be concerned with constructing a general family of soft origami deformations $\{ \mathbf{x}_{i,j}^{(\ell)}\}\mapsto\{ \mathbf{y}_{i,j}^{(\ell)}\}$ with
 \begin{equation}
 \begin{aligned}\label{eq:bendingDefs}
 &(\text{negligible panel strain:}) && \varepsilon_{i,j,1}^{(\ell)}, \ldots, \varepsilon_{i,j,4}^{(\ell)} \sim \ell^2, \\ 
 &(\text{slight panel bending:}) && \psi_{i,j1}^{(\ell)}, \psi_{i,j,2}^{(\ell)} \sim \ell, \\
 &(\text{large fold actuation:}) &&  \beta_{i,j}^{(\ell)} - \beta_{i,j,0}^{(\ell)},   \gamma_{i,j}^{(\ell)} -\gamma_{i,j,0}^{(\ell)} \sim 1
 \end{aligned}
 \end{equation}
independently of  $(i,j)$. The stretching, bending and folding energies of such deformations  scale as 
\begin{equation}
\begin{aligned}\label{eq:softScaling}
&\mathcal{E}_{\text{str}}^{(\ell)} (\{\mathbf{y}_{i,j}^{(\ell)} \} ) \sim \ell^{4 + \alpha_s}  , \quad \mathcal{E}_{\text{bend}}^{(\ell)} (\{\mathbf{y}_{i,j}^{(\ell)} \} ) \sim 1, \quad \mathcal{E}_{\text{fold}}^{(\ell)}(\{\mathbf{y}_{i,j}^{(\ell)} \} ) \sim \ell^{\alpha_f}.
\end{aligned}
\end{equation}
 Eq.\;(\ref{eq:alphaAssumptions}) implies that the stretching and folding energies vanish in the limit $\ell \rightarrow 0$, resulting  in an effective plate energy for parallelogram origami governed by the bending of its panels. We state this energy as part of Theorem \ref{MainTheorem}.


\subsection{The main coarse-graining result}\label{ssec:MainResult}

Our main results are as follows: first, we identify a surface theory that captures the effective origami's soft modes in the parameter regime in Eq.\;(\ref{eq:scalingAssumptions}-\ref{eq:alphaAssumptions}), leading to the first half of Theorem \ref{MainTheorem}. Then, in the second half of Theorem \ref{MainTheorem}, we produce a plate-like energy by coarse graining the bar and hinge model. Our main technical tool is an ability to construct general origami soft modes corresponding to solutions of the surface theory. We go into more detail now.

We begin with a question of local kinematics. Fig.\;\ref{Fig:pertcell} shows four slightly bent unit cells, where the panel diagonals (the purple lines in Fig.\;\ref{Fig:pertcell}(a)) are shown as stiff creases to help visualize bending of the panels. We assume that the lower-left cell in Fig.\;\ref{Fig:pertcell}(c) has been actuated by taking the dihedral angle $\theta_0$ (of the $\mathbf{t}_4^r$-crease in Fig.\;\ref{Fig:idepat}(a)) to a value $\theta_d = \theta_0 + \theta$. We also assume that its panels are bent along their diagonals through the folding angles $\ell \kappa_i$, $i =1,\ldots, 4$, shown in Fig.\;\ref{Fig:pertcell}(b). With this setup, we seek an understanding of how to fit together this deformed two-by-two set of unit cells based on the following physical heuristic: neighboring cells should not deviate too much from each other when they belong to a pattern of many such cells. In other words, like creases and like panels of neighboring cells should fold and bend in similar ways. By this heuristic, the neighboring cell in the $\mathbf{u}_0$-direction should be actuated by taking its $\theta_0$ dihedral angle to an angle $\approx \theta_d + \ell \delta \theta_{\mathbf{u}_0} $, and its panels should bend through folding angles $\approx \ell \kappa_i$. The cell in the $\mathbf{v}_0$-direction should be perturbed analogously, which introduces an additional actuation angle $\ell \delta \theta_{\mathbf{v}_0}$. 

\begin{figure}[t!]
\centering
\includegraphics[width=0.8\textwidth]{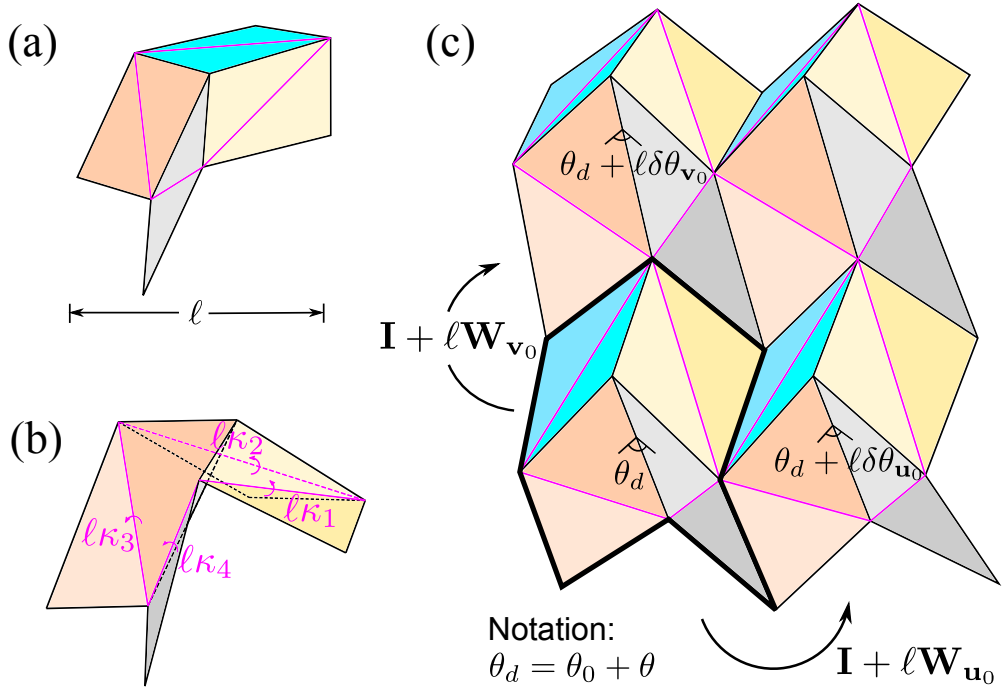} 
\caption{Local bending kinematics. (a) Stiff creases along the panel diagonals in purple provide a model for panel bending. (b) Folding angles $\ell \kappa_i$ at the stiff creases denote slight bending of the cell. (c) Notation for the local fitting problem of four neighboring cells with slightly bent kinematics. }
\label{Fig:pertcell}
\end{figure}

A basic question of compatibility emerges:  to what extent can such a two-by-two set of slightly bent origami cells fit together? The answer is given in Section \ref{sec:LinAlgebraSec}. Summarizing the results, we find that the cells need to be rotated slightly relative to each other via rotations $\approx \mathbf{I} + \ell \mathbf{W}_{\mathbf{u}_0}$ and $\approx \mathbf{I} + \ell \mathbf{W}_{\mathbf{v}_0}$ to match their deformed boundaries, as in Fig.\;\ref{Fig:pertcell}(c), and furthermore that these rotations are linked to the internal cell-based DOFs $\theta, \delta \theta_{\mathbf{u}_0} , \delta \theta_{\mathbf{v}_0} , \kappa_1, \ldots, \kappa_4$ introduced above. Note that $\mathbf{W}_{\mathbf{u}_0}$ and $\mathbf{W}_{\mathbf{v}_0}$  are $3\times3$ skew tensors, and so are completely parameterized by vectors  $\boldsymbol{\omega}_{\mathbf{u}_0}$ and $\boldsymbol{\omega}_{\mathbf{v}_0} \in \mathbb{R}^3$ via 
\begin{equation}
\begin{aligned}
\mathbf{W}_{\mathbf{u}_0} = (\boldsymbol{\omega}_{\mathbf{u}_0} \times ), \quad \mathbf{W}_{\mathbf{v}_0} = (\boldsymbol{\omega}_{\mathbf{v}_0} \times ),
\end{aligned}
\end{equation}
where $(\mathbf{a} \times)$ is defined by $(\mathbf{a} \times ) \mathbf{b} = \mathbf{a} \times \mathbf{b}$ for all $\mathbf{b} \in \mathbb{R}^3$.  Thus, there are 13 DOFs in the two-by-two set of bent cells, plus the obvious translations. Fitting the boundaries of the deformed cells together supplies eight constraints.   We show that the cells fit together with negligible stretching at leading order in $\ell$ (precisely, with panel strains $\lesssim \ell^2$) if and only if  $\theta$, $\boldsymbol{\omega}_{\mathbf{u}_0}$, $\boldsymbol{\omega}_{\mathbf{v}_0}$,  $\delta \theta_{\mathbf{u}_0} $ and $\delta \theta_{\mathbf{v}_0} $ satisfy   
\begin{equation}
\begin{aligned}\label{eq:micon}
& \bm{\omega}_{\mathbf{u}_0}\times\mathbf{v}(\theta)+\delta\theta_{\mathbf{u}_0}\mathbf{v}'(\theta)= \bm{\omega}_{\mathbf{v}_0}\times\mathbf{u}(\theta)+ \delta\theta_{\mathbf{v}_0}\mathbf{u}'(\theta), \\
& \bm{\omega}_{\mathbf{u}_0}\cdot\mathbf{v}'(\theta)= \bm{\omega}_{\mathbf{v}_0}\cdot\mathbf{u}'(\theta),
\end{aligned}
\end{equation}
and $\kappa_1, \ldots, \kappa_4$ are given linearly as functions of $\boldsymbol{\omega}_{\mathbf{u}_0}$, $\boldsymbol{\omega}_{\mathbf{v}_0}$, $\delta \theta_{\mathbf{u}_0}$ and $\delta \theta_{\mathbf{v}_0}$ in an explicit dependence with coefficients parameterized by $\theta$. Note $\mathbf{u}(\theta)$ and $\mathbf{v}(\theta)$ are the lattice vectors from Section \ref{ssec:DesignKin}. See Proposition~\ref{LocalBendProp} for the formal statement of this result. 

Eq.\;(\ref{eq:micon}) solves the local fitting problem for slightly bent parallelogram origami cells. It also plays a crucial role for coarse graining general soft modes. A basic observation is that the pattern's soft modes of deformation  are \textit{locally mechanistic}: at the scale of each unit cell they look like a mechanism, the features of which can vary from cell to cell. 
Much like in our work on  planar kirigami \cite{zheng2022continuum}, we capture this behavior by linking the first fundamental form of the effective (cell-averaged) deformation of the origami to the actuation of its unit cells, i.e., the cell-wise change in a dihedral angle $\theta$ parameterizing the mechanism motion (see Fig.\;\ref{Fig:idepat}(a-b)).  Mathematically, the effective deformation and actuation should be smooth and generally heterogeneous continuum fields $\mathbf{y}_{\text{eff}} \colon \Omega \rightarrow \mathbb{R}^3$ and $\theta \colon \Omega \rightarrow (\theta^{-}, \theta^+)$ that satisfy the metric constraint
\begin{equation}
\begin{aligned}\label{eq:firstFund}
\big(\nabla \mathbf{y}_{\text{eff}}(\mathbf{x}) \big)^T \nabla \mathbf{y}_{\text{eff}}(\mathbf{x})  = \mathbf{A}_{\text{eff}}^T(\theta(\mathbf{x})) \mathbf{A}_{\text{eff}} (\theta(\mathbf{x})) 
\end{aligned}
\end{equation}
for the shape tensor $\mathbf{A}_{\text{eff}}(\theta)$ in Eq.\;(\ref{eq:shapeTensor}). At the level of the panel strains,  Eq.\;(\ref{eq:firstFund}) enforces the requirement that these strains  vanish as $\ell\to 0$.  Even so, it is only part of the description of soft modes, since we have yet to incorporate the constraints from Eq.\;(\ref{eq:micon}) on the local kinematics of bent origami cells. For this purpose, we rewrite the metric constraint.

\begin{figure}[t!]
\centering
\includegraphics[width=1\textwidth]{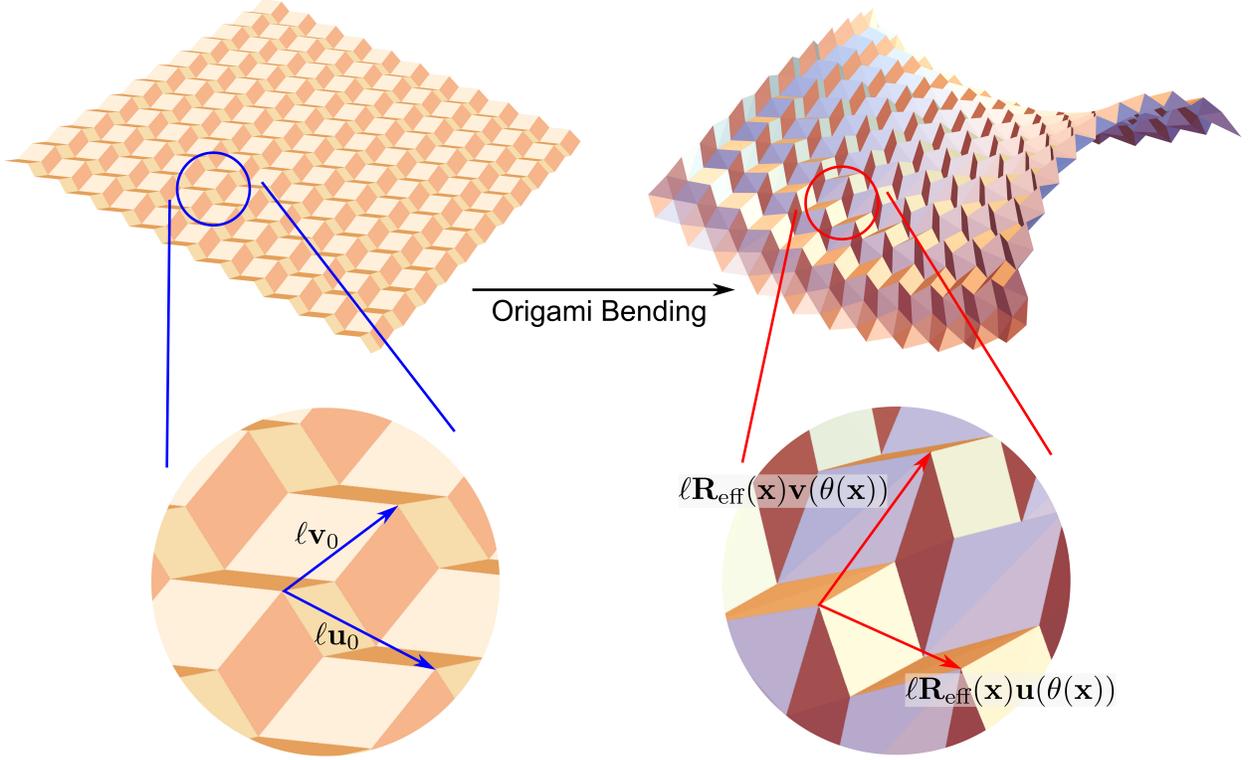} 
\caption{Illustration of a locally mechanistic origami deformation and its effective description. The reference Bravais lattice vectors $\mathbf{u}_0$ and $\mathbf{v}_0$ distort according to the local mechanism as $\mathbf{u}_0 \mapsto \mathbf{R}_{\text{eff}}(\mathbf{x}) \mathbf{u}(\theta(\mathbf{x}))$ and $\mathbf{v}_0 \mapsto \mathbf{R}_{\text{eff}}(\mathbf{x}) \mathbf{v}(\theta(\mathbf{x}))$ for a  rotation field $\mathbf{R}_{\text{eff}}(\mathbf{x})$ and spatially varying actuation field $\theta(\mathbf{x})$.  }
\label{Fig:RefDef}
\end{figure}

 By linear algebra, Eq.\;(\ref{eq:firstFund}) is equivalent to the system 
\begin{equation}
\begin{aligned}\label{eq:effectiveDecomp}
&\nabla \mathbf{y}_{\text{eff}}(\mathbf{x}) = \mathbf{R}_{\text{eff}}(\mathbf{x})  \mathbf{A}_{\text{eff}} (\theta(\mathbf{x})), \quad \text{equivalently:} \quad \begin{cases}
\partial_{\mathbf{u}_0} \mathbf{y}_{\text{eff}}(\mathbf{x}) = \mathbf{R}_{\text{eff}}(\mathbf{x}) \mathbf{u}(\theta(\mathbf{x}))  \\
 \partial_{\mathbf{v}_0} \mathbf{y}_{\text{eff}}(\mathbf{x}) = \mathbf{R}_{\text{eff}}(\mathbf{x}) \mathbf{v}(\theta(\mathbf{x}))
\end{cases}
\end{aligned}
\end{equation}
for a rotation field $\mathbf{R}_{\text{eff}}:\Omega\to SO(3)$. Note  $\partial_{\mathbf{u}_0} (\cdot) (\mathbf{x}) := \lim_{\epsilon \rightarrow 0} \epsilon^{-1} \big[(\cdot)(\mathbf{x} + \epsilon \tilde{\mathbf{u}}_0) - (\cdot)(\mathbf{x})\big]$ and $\partial_{\mathbf{v}_0} (\cdot) (\mathbf{x}) := \lim_{\epsilon \rightarrow 0}\epsilon^{-1} \big[(\cdot)(\mathbf{x} + \epsilon \tilde{\mathbf{v}}_0) - (\cdot)(\mathbf{x})\big]$ are the directional derivatives along $\mathbf{u}_0$ and $\mathbf{v}_0$, respectively (see Fig.\;\ref{Fig:RefDef}). Applying these derivatives to  $\mathbf{R}_{\text{eff}}(\mathbf{x})$ defines two skew tensor fields parameterized as vectors $\boldsymbol{\omega}_{\mathbf{u}_0},\boldsymbol{\omega}_{\mathbf{v}_0} \colon \Omega \rightarrow \mathbb{R}^3$ that measure the relative rotations of the deformed origami cells:
\begin{equation}
\begin{aligned}\label{eq:rotDerivatives}
\partial_{\mathbf{u}_0} \mathbf{R}_{\text{eff}}(\mathbf{x}) = \mathbf{R}_{\text{eff}}(\mathbf{x}) \big(\boldsymbol{\omega}_{\mathbf{u}_0}(\mathbf{x}) \times \big), \quad \partial_{\mathbf{v}_0} \mathbf{R}_{\text{eff}}(\mathbf{x}) = \mathbf{R}_{\text{eff}}(\mathbf{x}) \big(\boldsymbol{\omega}_{\mathbf{v}_0}(\mathbf{x}) \times \big).
\end{aligned}
\end{equation} 
Compatibility links the vector fields $ \boldsymbol{\omega}_{\mathbf{u}_0}(\mathbf{x})$ and 
$\boldsymbol{\omega}_{\mathbf{v}_0} (\mathbf{x})$ to the actuation angle $\theta(\mathbf{x})$. In particular,  $\partial_{\mathbf{u}_0} \partial_{\mathbf{v}_0} \mathbf{y}_{\text{eff}}(\mathbf{x}) =\partial_{\mathbf{v}_0} \partial_{\mathbf{u}_0} \mathbf{y}_{\text{eff}}(\mathbf{x})$ implies  that 
\begin{equation}
\begin{aligned}\label{eq:CompatIntro}
 \bm{\omega}_{\mathbf{u}_0}(\mathbf{x}) \times\mathbf{v}(\theta(\mathbf{x}))+\partial_{\mathbf{u}_0}  \theta(\mathbf{x}) \mathbf{v}'(\theta(\mathbf{x}))= \bm{\omega}_{\mathbf{v}_0} (\mathbf{x})\times\mathbf{u}(\theta(\mathbf{x}))+ \partial_{\mathbf{v}_0} \theta(\mathbf{x}) \mathbf{u}'(\theta(\mathbf{x})) 
\end{aligned}
\end{equation}
and $\partial_{\mathbf{u}_0} \partial_{\mathbf{v}_0} \mathbf{R}_{\text{eff}}(\mathbf{x}) =\partial_{\mathbf{v}_0} \partial_{\mathbf{u}_0} \mathbf{R}_{\text{eff}}(\mathbf{x})$ implies that 
\begin{equation}
\begin{aligned}\label{eq:CompatFinal}
  \partial_{\mathbf{u}_0} \bm{\omega}_{\mathbf{v}_0}(\mathbf{x})-\partial_{\mathbf{v}_0} \bm{\omega}_{\mathbf{u}_0}(\mathbf{x})= \bm{\omega}_{\mathbf{v}_0}(\mathbf{x})\times \bm{\omega}_{\mathbf{u}_0}(\mathbf{x}).
\end{aligned}
\end{equation}
In summary, Eqs.\;(\ref{eq:effectiveDecomp}-\ref{eq:CompatFinal}) are implied when a deformation $\mathbf{y}_\text{eff}(\mathbf{x})$ is linked to an auxiliary field (here, $\theta(\mathbf{x})$) through a constraint on its first fundamental form (Eq.\;(\ref{eq:firstFund})). In fact, these conditions are not only necessary but are also sufficient for the existence of $\mathbf{y}_\text{eff}(\mathbf{x})$ satisfying Eq.\;(\ref{eq:firstFund}), allowing us to take $\theta(\mathbf{x})$, $\boldsymbol{\omega}_{\mathbf{u}_0}(\mathbf{x})$ and $\boldsymbol{\omega}_{\mathbf{v}_0}(\mathbf{x})$ as our basic coarse-grained fields. (This result is more or less standard in Cartan's method of moving frames; for completness, we provide a proof in Proposition \ref{firstProp} of  Section \ref{ssec:EffSurfacesSection}).

We now return to the local fitting conditions in Eq.\;(\ref{eq:micon}). Eq.\;(\ref{eq:CompatIntro}) is a continuous version of the first condition in Eq.\;(\ref{eq:micon}), under the replacement $(\delta \theta_{\mathbf{u}_0},\delta \theta_{\mathbf{v}_0}) \mapsto (\partial_{\mathbf{u}_0} \theta(\mathbf{x}),\partial_{\mathbf{v}_0} \theta(\mathbf{x}))$. This is not surprising, as the statement that neighboring cells  ``fit together'' is a discrete analog of the statement that the second partial derivatives of the effective deformation commute. 
However, there is additional rigidity present due to the use of an origami-based microstructure, and indeed the second condition in Eq.\;(\ref{eq:micon}) cannot be deduced from manipulations on the metric constraint in Eq.\;(\ref{eq:firstFund}). On coarse graining, this constraint requires that
\begin{equation}
\begin{aligned}\label{eq:localBendIdent}
 \bm{\omega}_{\mathbf{u}_0}(\mathbf{x}) \cdot\mathbf{v}'(\theta(\mathbf{x}))= \bm{\omega}_{\mathbf{v}_0} (\mathbf{x}) \cdot\mathbf{u}'(\theta(\mathbf{x}) ).
\end{aligned}
\end{equation}
To be clear, we are not the first write down such a constraint; recent literature \cite{lebee2018fitting,mcinerney2022discrete, nassar2017curvature, nassar2022strain, schenk2013geometry,wei2013geometric} has linked normal curvatures of the effective surfaces described by parallelogram origami to a Poisson-like ratio of their unit cells.  However, to our knowledge we are the first to recognize the crucial link between Eq.\;(\ref{eq:localBendIdent}) and the order of magnitude of the panel strains: given Eq.\;(\ref{eq:localBendIdent}), the panel strains in a soft mode can be made $\lesssim \ell^2$, and this smallness is necessary to achieve a bending-dominated response. See Section \ref{ssec:discussSec-modes} for further discussion of the literature, and the paragraphs following Theorem \ref{MainTheorem} for more on the panel strains. 

At this point, we have produced a full set of constraints for the fields $\theta(\mathbf{x})$, $\bm{\omega}_{\mathbf{u}_0}(\mathbf{x})$ and 
$\bm{\omega}_{\mathbf{v}_0}(\mathbf{x})$ governing the macroscopic kinematics of parallelogram origami:
\begin{equation}
\begin{aligned}\label{eq:SurfaceTheory}
&\underline{\text{Effective surface theory for parallelogram origami:}}\\
& \begin{cases}
\theta(\mathbf{x}) \in (\theta^{-} , \theta^+) \\
 \bm{\omega}_{\mathbf{u}_0}(\mathbf{x}) \cdot\mathbf{v}'(\theta(\mathbf{x}))= \bm{\omega}_{\mathbf{v}_0} (\mathbf{x}) \cdot\mathbf{u}'(\theta(\mathbf{x}) ) \\ 
 \bm{\omega}_{\mathbf{u}_0}(\mathbf{x}) \times\mathbf{v}(\theta(\mathbf{x}))+\partial_{\mathbf{u}_0}  \theta(\mathbf{x}) \mathbf{v}'(\theta(\mathbf{x}))= \bm{\omega}_{\mathbf{v}_0} (\mathbf{x})\times\mathbf{u}(\theta(\mathbf{x}))+ \partial_{\mathbf{v}_0} \theta(\mathbf{x}) \mathbf{u}'(\theta(\mathbf{x}))  \\ 
 \partial_{\mathbf{v}_0} \bm{\omega}_{\mathbf{u}_0}(\mathbf{x})-\partial_{\mathbf{u}_0} \bm{\omega}_{\mathbf{v}_0}(\mathbf{x})= \bm{\omega}_{\mathbf{u}_0}(\mathbf{x})\times \bm{\omega}_{\mathbf{v}_0}(\mathbf{x}).
\end{cases}
\end{aligned}
\end{equation}
Our main result shows that for each solution of Eq.\;(\ref{eq:SurfaceTheory}), there is a unique (up to rigid body motion) parameterized surface, along with a sequence of deformed parallelogram origami patterns converging to this surface in the limit $\ell \to 0$; in other words, we construct a general family of soft modes. We also provide a formula for the asymptotic bending energy of these soft modes. 
The precise statement is as follows. Note  $\overline{\Omega}$ is the closure of the planar reference domain $\Omega$, $\langle \mathbf{x} \rangle := \frac{1}{|\Omega|} \int_{\Omega} \mathbf{x} dA$ and $\hat{\mathbf{x}}_{i,j}^{(\ell)} := (\mathbf{x}_{i,j}^{(\ell)} \cdot \mathbf{e}_1, \mathbf{x}_{i,j}^{(\ell)} \cdot \mathbf{e}_2)$.
\begin{theorem}\label{MainTheorem}
Consider any parallelogram origami design  composed of a repeating unit cell of four parallelogram panels joined at folds, and its corresponding bar and hinge energy from Sections \ref{ssec:DesignKin} and \ref{ssec:BarHinge}. Let  $\theta(\mathbf{x})$, $\bm{\omega}_{\mathbf{u}_0} (\mathbf{x})$ and $\bm{\omega}_{\mathbf{v}_0}(\mathbf{x})$ be smoothly defined on a neighborhood of $\overline{\Omega}$,  and assume that $\theta(\mathbf{x})\in(\theta^-,\theta^+)$ on this domain.
Assume also that these fields are analytic in the case that the Poisson ratio $\nu(\theta)$ in Eq.\;(\ref{eq:PoissonsRatioDesign}) is negative.  Finally, assume that $\theta(\mathbf{x})$, $\bm{\omega}_{\mathbf{u}_0} (\mathbf{x})$ and  $\bm{\omega}_{\mathbf{v}_0}(\mathbf{x})$ solve Eq.\;(\ref{eq:SurfaceTheory}) on $\Omega$. The following statements hold:

\noindent \textbf{\emph{(Surface theory.)}} There exists a unique and smooth rotation field $\mathbf{R}_{\emph{eff}} \colon \Omega \rightarrow SO(3)$ solving 
\begin{equation}
\begin{aligned}\label{eq:firstSurface}
\partial_{\mathbf{u}_0} \mathbf{R}_{\emph{eff}}(\mathbf{x}) = \mathbf{R}_{\emph{eff}}( \mathbf{x}) \big( \bm{\omega}_{\mathbf{u}_0}(\mathbf{x}) \times \big) , \quad \partial_{\mathbf{v}_0} \mathbf{R}_{\emph{eff}}(\mathbf{x}) = \mathbf{R}_{\emph{eff}}( \mathbf{x}) \big( \bm{\omega}_{\mathbf{v}_0}(\mathbf{x}) \times \big)
\end{aligned}
\end{equation}
with $\mathbf{R}_{\emph{eff}}(\langle \mathbf{x} \rangle) = \mathbf{I}$. Also, there exists a unique and smooth effective deformation $\mathbf{y}_{\emph{eff}} \colon \Omega \rightarrow \mathbb{R}^3$ satisfying
\begin{equation}
\begin{aligned}\label{eq:secondSurface}
\nabla \mathbf{y}_{\emph{eff}}(\mathbf{x}) = \mathbf{R}_{\emph{eff}}(\mathbf{x}) \mathbf{A}_{\emph{eff}}(\theta(\mathbf{x}))
\end{aligned}
\end{equation}
and $\mathbf{y}_{\emph{eff}} ( \langle \mathbf{x} \rangle) = \mathbf{0}$.

\noindent \textbf{\emph{(Origami soft modes.)}} For $\mathbf{y}_{\emph{eff}}(\mathbf{x})$ as above and any sufficiently small $ \ell$, there is a deformation $\{ \mathbf{y}_{i,j}^{(\ell)}\}$ of the design $\{ \mathbf{x}_{i,j}^{(\ell)}\}$ such that 
\begin{equation}
\begin{aligned}\label{eq:MainEsts}
&(\text{it has negligible panel strains:}) &&  |\varepsilon^{(\ell)}_{i,j,1}|, \ldots, |\varepsilon_{i,j,4}^{(\ell)}|  = O(\ell^2),  \\
&(\text{it has slight panel bending:}) && |\psi_{i,j,1}^{(\ell)}|, |\psi_{i,j,2}^{(\ell)}| = O(\ell),  \\
&(\text{it recovers the effective deformation:}) &&  | \mathbf{y}_{i,j}^{(\ell)} -\mathbf{y}_{\emph{eff}}(\hat{\mathbf{x}}_{i,j}^{(\ell)}) |  = O(\ell)  \\
\end{aligned}
\end{equation}
for all  $(i,j) \in \mathcal{I}^{(\ell)}$, i.e., for all origami vertices. 

\noindent \textbf{\emph{(Effective plate energy.)}}  The bar and hinge energy of the above origami deformations obeys
\begin{equation}
\begin{aligned}\label{eq:MainLim}
\lim_{\ell \rightarrow 0}\, \mathcal{E}_{\emph{tot}}^{(\ell)} (  \{ \mathbf{y}_{i,j}^{(\ell)} \}) =   \int_{\Omega}  \bigg\{\begin{pmatrix} \boldsymbol{\omega}_{\mathbf{u}_0}(\mathbf{x}) \\ \partial_{\mathbf{u}_0} \theta(\mathbf{x}) \end{pmatrix} \cdot  \mathbf{K}_{\mathbf{u}_{0}}(\theta(\mathbf{x}) ) \begin{pmatrix} \boldsymbol{\omega}_{\mathbf{u}_0}(\mathbf{x}) \\ \partial_{\mathbf{u}_0} \theta(\mathbf{x}) \end{pmatrix}  +  \begin{pmatrix} \boldsymbol{\omega}_{\mathbf{v}_0}(\mathbf{x}) \\ \partial_{\mathbf{v}_0} \theta(\mathbf{x}) \end{pmatrix} \cdot  \mathbf{K}_{\mathbf{v}_{0}}(\theta(\mathbf{x}) ) \begin{pmatrix} \boldsymbol{\omega}_{\mathbf{v}_0}(\mathbf{x}) \\ \partial_{\mathbf{v}_0} \theta(\mathbf{x}) \end{pmatrix} \bigg\}\, dA 
\end{aligned}
\end{equation}
where the symmetric tensors $ \mathbf{K}_{\mathbf{u}_{0}}(\theta )$ and $\mathbf{K}_{\mathbf{v}_0}(\theta) \in \mathbb{R}^{4\times4}$ are given by
\begin{equation}
\begin{aligned}\label{eq:stiffnessTensor}
 \mathbf{K}_{\mathbf{u}_{0}}(\theta ) &:= \frac{b_1}{(V_{123}(\theta))^2}  \begin{pmatrix} \mathbf{t}_2(\theta) \times \mathbf{t}_3(\theta) \\ \mathbf{t}_3(\theta) \cdot \mathbf{t}_2'(\theta) \end{pmatrix}   \otimes \begin{pmatrix} \mathbf{t}_2(\theta) \times \mathbf{t}_3(\theta) \\ \mathbf{t}_3(\theta) \cdot \mathbf{t}_2'(\theta) \end{pmatrix}  +  \frac{b_3}{(V_{341}(\theta))^2}  \begin{pmatrix} \mathbf{t}_1(\theta) \times \mathbf{t}_4(\theta) \\ -\mathbf{t}_1(\theta) \cdot \mathbf{t}_4'(\theta) \end{pmatrix}   \otimes \begin{pmatrix} \mathbf{t}_1(\theta) \times \mathbf{t}_4(\theta) \\ -\mathbf{t}_1(\theta) \cdot \mathbf{t}_4'(\theta) \end{pmatrix}, \\
 \mathbf{K}_{\mathbf{v}_{0}}(\theta ) &:= \frac{b_2}{(V_{234}(\theta))^2}  \begin{pmatrix} \mathbf{t}_3(\theta) \times \mathbf{t}_4(\theta) \\ \mathbf{t}_3'(\theta) \cdot \mathbf{t}_4(\theta) \end{pmatrix}   \otimes \begin{pmatrix} \mathbf{t}_3(\theta) \times \mathbf{t}_4(\theta) \\ \mathbf{t}_3'(\theta) \cdot \mathbf{t}_4(\theta) \end{pmatrix} +  \frac{b_4}{(V_{412}(\theta))^2}  \begin{pmatrix} \mathbf{t}_2(\theta) \times \mathbf{t}_1(\theta) \\ -\mathbf{t}_1'(\theta) \cdot \mathbf{t}_2(\theta) \end{pmatrix}   \otimes \begin{pmatrix} \mathbf{t}_2(\theta) \times \mathbf{t}_1(\theta) \\ -\mathbf{t}_1'(\theta) \cdot \mathbf{t}_2(\theta) \end{pmatrix},
\end{aligned}
\end{equation}
with $V_{ijk}(\theta) : = \mathbf{t}_i(\theta) \cdot (\mathbf{t}_j(\theta) \times \mathbf{t}_k(\theta))$. The bending moduli $b_1,\ldots,b_4$  satisfy 
\begin{equation}
\begin{aligned}\label{eq:bendingModuli} 
b_i :=  \frac{1}{2|\mathbf{u}_0 \times \mathbf{v}_0|}   \begin{pmatrix} \frac{| (\mathbf{t}_i^r   - \mathbf{t}_j^{r})  \times ( \mathbf{t}_i^r \times \mathbf{t}_j^r)|^2}{|(\mathbf{t}_i^r   - \mathbf{t}_j^{r})||( \mathbf{t}_i^r \times \mathbf{t}_j^r)|} \\ \frac{| (\mathbf{t}_i^r   + \mathbf{t}_j^{r})  \times ( \mathbf{t}_i^r \times \mathbf{t}_j^r)|^2}{|(\mathbf{t}_i^r   + \mathbf{t}_j^{r})||( \mathbf{t}_i^r \times \mathbf{t}_j^r)|} \end{pmatrix}  \cdot DD \Phi_{\mathcal{P}_i}^{\emph{bend}}( \mathbf{0} )\begin{pmatrix} \frac{| (\mathbf{t}_i^r   - \mathbf{t}_j^{r})  \times ( \mathbf{t}_i^r \times \mathbf{t}_j^r)|^2}{|(\mathbf{t}_i^r   - \mathbf{t}_j^{r})||( \mathbf{t}_i^r \times \mathbf{t}_j^r)|} \\ \frac{| (\mathbf{t}_i^r   + \mathbf{t}_j^{r})  \times ( \mathbf{t}_i^r \times \mathbf{t}_j^r)|^2}{|(\mathbf{t}_i^r   + \mathbf{t}_j^{r})||( \mathbf{t}_i^r \times \mathbf{t}_j^r)|} \end{pmatrix} , \quad ij \in \{12, 34\},  \\
b_i :=  \frac{1}{2|\mathbf{u}_0 \times \mathbf{v}_0|}   \begin{pmatrix} \frac{| (\mathbf{t}_i^r   + \mathbf{t}_j^{r})  \times ( \mathbf{t}_i^r \times \mathbf{t}_j^r)|^2}{|(\mathbf{t}_i^r   + \mathbf{t}_j^{r})||( \mathbf{t}_i^r \times \mathbf{t}_j^r)|} \\ \frac{| (\mathbf{t}_i^r   - \mathbf{t}_j^{r})  \times ( \mathbf{t}_i^r \times \mathbf{t}_j^r)|^2}{|(\mathbf{t}_i^r   - \mathbf{t}_j^{r})||( \mathbf{t}_i^r \times \mathbf{t}_j^r)|} \end{pmatrix}  \cdot DD \Phi_{\mathcal{P}_i}^{\emph{bend}}( \mathbf{0} ) \begin{pmatrix} \frac{| (\mathbf{t}_i^r   + \mathbf{t}_j^{r})  \times ( \mathbf{t}_i^r \times \mathbf{t}_j^r)|^2}{|(\mathbf{t}_i^r   + \mathbf{t}_j^{r})||( \mathbf{t}_i^r \times \mathbf{t}_j^r)|} \\ \frac{| (\mathbf{t}_i^r   - \mathbf{t}_j^{r})  \times ( \mathbf{t}_i^r \times \mathbf{t}_j^r)|^2}{|(\mathbf{t}_i^r   - \mathbf{t}_j^{r})||( \mathbf{t}_i^r \times \mathbf{t}_j^r)|} \end{pmatrix} ,  \quad ij \in \{23, 41\}.
\end{aligned}
\end{equation}
\end{theorem}
\begin{remark}
The energy densities  $\Phi_{\mathcal{P}_i}^{\emph{bend}}(\psi_1, \psi_2)$, $i=1,\dots,4$,  above are the four periodically repeating bending strain energies from Eq.\;(\ref{eq:eBend}), labeled by the panel variant in Eq.\;(\ref{eq:OmegaCell}). $DD \Phi_{\mathcal{P}_i}^{\emph{bend}}(\mathbf{0})\in \mathbb{R}^{2\times2}$ are their Hessian matrices evaluated at $(\psi_1,\psi_2)=\mathbf{0}$. 
\end{remark}
\begin{remark}\label{IntervalRemark}
Per the assumptions, there is an $\epsilon_\theta > 0$ such that $\theta(\mathbf{x}) \in [\theta^{-} + \epsilon_{\theta}, \theta^+ - \epsilon_{\theta}]$ on $\overline{\Omega}$. In other words, the actuation is bounded away from degenerate states where the origami contains unfolded or fully-folded creases. 
\end{remark}
\begin{remark}
In the theorem, the folding energy becomes negligible in the limit because we assumed a bending dominated scaling through the prescription $\alpha_f >0, \alpha_b = -2, \alpha_s \in (-4,-2)$ (recall Eqs.\;(\ref{eq:scalingAssumptions}-\ref{eq:softScaling})). If instead $\alpha_f= 0$, the limit energy becomes the sum of the bending energy in Eq.\;(\ref{eq:MainLim}) and a folding energy of the form $\int_{\Omega} W_{\emph{fold}}(\theta(\mathbf{x})) \,dA$ for an energy density $W_{\emph{fold}}(\theta) \geq 0$ that vanishes if and only if $\theta =0$. 
\end{remark}


The proof of Theorem \ref{MainTheorem} is given in Sections \ref{sec:LinAlgebraSec}-\ref{sec:DerivationSec}, with some technical details in Appendix \ref{sec:MechProofs} and \ref{sec:ExistencePDE}. Given a solution of the proposed effective surface theory in Eq.\;(\ref{eq:SurfaceTheory}) and its corresponding effective deformation, our task is to construct a locally mechanistic ansatz for the deformation of the origami, and to ensure it has negligible stretching energy and known bending energy. Writing an asymptotic expansion for the panel strains $\varepsilon_{i,j,k}^{(\ell)} = \varepsilon_{i,j,k}^{(0)} + \ell \varepsilon_{i,j,k}^{(1)} + \ell^2 \varepsilon_{i,j,k}^{(2)} +O(\ell^3)$ with $(i,j)$ indexing the panel and $k$ indexing one of four panel strains, the key is to find a way to set the first two terms in the expansion to zero: $\varepsilon_{i,j,k}^{(0)} = \varepsilon_{i,j,k}^{(1)} = 0$. To understand why, consider each panel as a plate of thickness $h \ll \ell$ and modulus $E$. The  elastic energy per panel is of the form $\mathcal{E}_{\text{plate}} \sim E   \int_{\mathcal{P}^{(\ell)}}(  h \varepsilon^2 + h^3 \kappa^2 ) dA$, where $\mathcal{P}^{(\ell)}$ is the panel and $\varepsilon$ and $\kappa$ denote stretching and curvature. Since we are after a bending-dominated limit, we take $\kappa \sim \psi_{i,j,k}^{(\ell)}/ \ell \sim 1$. If we can achieve $\varepsilon \sim \varepsilon_{i,j,k}^{(\ell)} \sim \ell^2$, it follows that  $\mathcal{E}_{\text{plate}} \sim E h \ell^6 + E h^3 \ell^2$. Bending ($Eh^3 \ell^2$) dominates over stretching ($E h \ell^6$) when $\ell^2 \ll h$, a situation that is feasible as $h \ll \ell$. (Presumably, the alternative scaling $\varepsilon \sim \ell$ leads to an effective membrane theory, but deriving such a theory is not the subject of this paper.)

Actually achieving  $\varepsilon_{i,j,k}^{(0)} = \varepsilon_{i,j,k}^{(1)} = 0$  is a delicate matter. The construction divides into three main steps. First, we deform the origami using a cell-wise sampling of the effective fields $\mathbf{y}_{\text{eff}}(\mathbf{x})$, $\mathbf{R}_{\text{eff}}(\mathbf{x})$, $\theta(\mathbf{x})$, $\boldsymbol{\omega}_{\mathbf{u}_0}(\mathbf{x})$ and $\boldsymbol{\omega}_{\mathbf{v}_0}(\mathbf{x})$. In this first step, the zeroth order strains $\varepsilon_{i,j,k}^{(0)}$ vanish  due to the metric constraint in Eq.\;(\ref{eq:secondSurface}). However, the strains at first order $\varepsilon_{i,j,k}^{(1)}$ fail to vanish, leading us to enrich the ansatz through local perturbations. We introduce two auxiliary fields --- a vector field $\boldsymbol{\omega}(\mathbf{x})$ and an angle field $\xi(\mathbf{x})$ --- corresponding to linearized rotations and linearized mechanisms of the cells. Leveraging the remaining constraints in Eqs.\;(\ref{eq:SurfaceTheory}-\ref{eq:secondSurface}), we show that the enriched construction achieves $\varepsilon_{i,j,k}^{(1)} = 0$ if and only if the auxiliary fields solve a linear PDE:
\begin{equation}
\begin{aligned}\label{eq:LPDEIntro}
\mathcal{L}(\mathbf{x}) \begin{pmatrix} \boldsymbol{\omega}(\mathbf{x}) \\ \xi(\mathbf{x})   \end{pmatrix} = \mathbf{f}(\mathbf{x}).
\end{aligned}
\end{equation}
The operator $\mathcal{L}(\mathbf{x})$  and right-hand side  $\mathbf{f}(\mathbf{x}) \in \mathbb{R}^3$ depend smoothly and explicitly on the fields $\mathbf{R}_{\text{eff}}(\mathbf{x})$, $\theta(\mathbf{x})$, $\boldsymbol{\omega}_{\mathbf{u}_0}(\mathbf{x})$ and $\boldsymbol{\omega}_{\mathbf{v}_0}(\mathbf{x})$, as well as on the design parameters of the origami unit cell. In the final step of our construction, we prove that there is always a solution to Eq.\;(\ref{eq:LPDEIntro}), and thus a  choice of  $\boldsymbol{\omega}(\mathbf{x})$ and $\xi(\mathbf{x})$ that eliminates the first order strains. This yields the desired origami soft modes, along with the plate energy in the theorem above.  


\section{Rigid and bending kinematics of parallelogram origami}\label{sec:LinAlgebraSec}

We begin our analysis by characterizing the full mechanism kinematics of an arbitrary parallelogram origami design. Then, we solve the local fitting problem for slightly bent origami unit cells.

\subsection{Mechanism kinematics}\label{ssec:MechKin}

\begin{figure}[t!]
\centering
\includegraphics[width=1\textwidth]{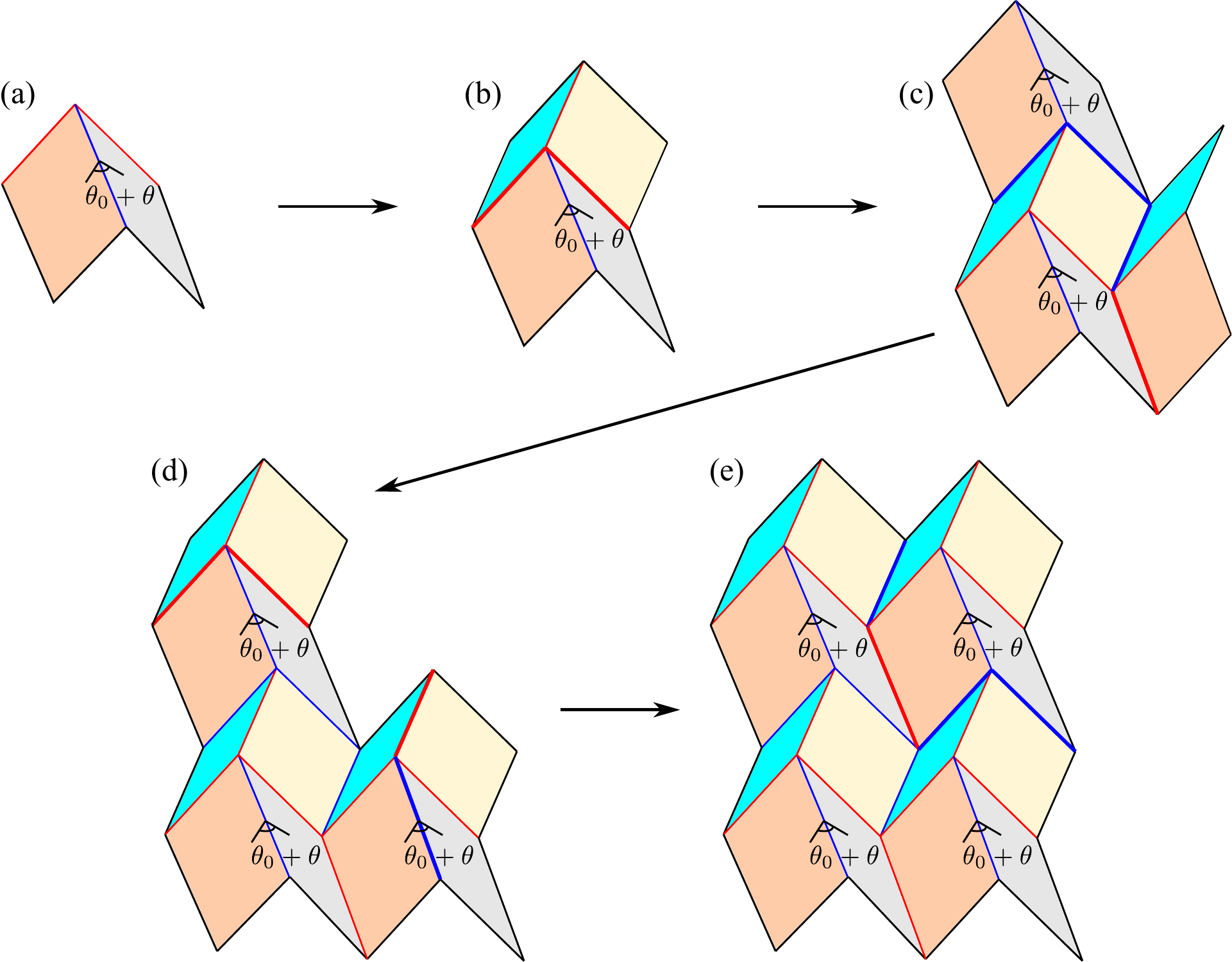} 
\caption{Uniqueness of the mechanism motion. (a-b) Actuating a single crease  determines the actuation of the opposite crease.  (c-d) This argument is repeated to construct the  motion of the adjacent cells. (e) The mechanism motion of the fourth cell is compatible with its neighbors.  }
\label{Fig:Marching}
\end{figure}

Fix a set of four design vectors $\mathbf{t}_{1}^r, \ldots, \mathbf{t}_4^{r}$ that satisfy Eq.\;(\ref{eq:tangentsConstraints}-\ref{eq:selfConstraint}) and form the creases of the reference cell, as in Fig.\;\ref{Fig:idepat}(a). The cell itself is the union of  four parallelogram panels $\Omega_{\text{cell}} = \mathcal{P}_1 \cup \ldots \cup \mathcal{P}_4$ that depend on these design vectors through Eq.\;(\ref{eq:OmegaCell}). To deform the cell rigidly, we  deform the crease vectors $\mathbf{t}_{1}^r, \ldots, \mathbf{t}_4^{r}$ to  $\mathbf{t}_{1}^d, \ldots, \mathbf{t}_4^{d}$ in such a way that
\begin{equation}
\begin{aligned}\label{eq:compatCreasesZero}
&(\text{the lengths of the creases are preserved:}) &&  |\mathbf{t}_i^d| = |\mathbf{t}_i^r|, \quad i = 1,\ldots,4 ,\\
&(\text{the angles between adjacent creases are preserved:}) &&  \mathbf{t}_i^d \cdot \mathbf{t}_{j}^d = \mathbf{t}_i^r \cdot \mathbf{t}_{j}^r, \quad ij \in \{12,23,34,41\},  \\ 
&(\text{the mountain-valley assignment is preserved:}) &&  [\mathbf{t}_i^d \cdot (\mathbf{t}_{j}^d \times \mathbf{t}_{k}^d)] [\mathbf{t}_i^r \cdot (\mathbf{t}_{j}^r \times \mathbf{t}_{k}^r)] >0, \\
& && \quad ijk \in \{123,234,341,412\} .
\end{aligned}
\end{equation}
The collection of all deformed creases 
\begin{equation}
\begin{aligned}
\mathcal{M} :=  \big\{ (\mathbf{t}_1^d, \mathbf{t}_2^d, \mathbf{t}_3^d, \mathbf{t}_4^d) \in \mathbb{R}^3 \times \mathbb{R}^3 \times \mathbb{R}^3 \times \mathbb{R}^3 \text{ subject to  Eq.\;(\ref{eq:compatCreasesZero})} \big\} 
\end{aligned}
\end{equation}
describes the mechanism kinematics of the cell consistent with the mountain-valley assignment of its reference configuration.  

A tedious but straightforward application of the implicit function theorem proves that $\mathcal{M}$ is parameterized by an analytic single DOF motion, up to an overall rigid motion. See Appendix \ref{ssec:MechProofs1} for the details.
\begin{proposition}\label{mechKinCellProp}
There is an interval $(\theta^-, \theta^+)$ containing zero and analytic functions $\mathbf{t}_i \colon (\theta^-, \theta^+) \rightarrow \mathbb{R}^3$, $i = 1,\ldots, 4$ that parameterize the mechanism set as 
\begin{equation}
\begin{aligned}\label{eq:Mmanifold}
\mathcal{M} = \big\{(  \mathbf{R} \mathbf{t}_1(\theta), \mathbf{R} \mathbf{t}_2(\theta),  \mathbf{R}\mathbf{t}_3(\theta),  \mathbf{R}\mathbf{t}_4(\theta)) \colon \mathbf{R} \in SO(3), \theta \in (\theta^{-}, \theta^+) \}.
\end{aligned}
\end{equation}
These functions also satisfy  the initial conditions $\mathbf{t}_i(0) = \mathbf{t}_i^r,$  $i = 1,\ldots, 4$, and the planarity conditions
\begin{equation}
\begin{aligned}
&\mathbf{e}_3 \cdot (\mathbf{t}_1(\theta)  - \mathbf{t}_3(\theta)) = 0, \quad  \mathbf{e}_3 \cdot (\mathbf{t}_2(\theta) - \mathbf{t}_4(\theta)) = 0, \quad \mathbf{e}_3 \cdot (  (\mathbf{t}_1(\theta)  - \mathbf{t}_3(\theta)) \times (\mathbf{t}_2(\theta)  - \mathbf{t}_4(\theta)) > 0
\end{aligned}
\end{equation}
for all $\theta \in (\theta^{-}, \theta^+)$. 
\end{proposition}

In the course of the proof, we establish an important property of the Bravais lattice vectors $\mathbf{u}(\theta) :=  \mathbf{t}_1(\theta) - \mathbf{t}_3(\theta)$ and $\mathbf{v}(\theta) := \mathbf{t}_2(\theta) -\mathbf{t}_4(\theta)$ introduced in Section \ref{ssec:DesignKin}, namely, that $|\mathbf{u}(\theta)|$ and $|\mathbf{v}(\theta)|$ depend in a strictly monotonic fashion on $\theta \in (\theta^{-}, \theta^+)$. We use this later in the coarse-graining argument.
\begin{lemma}\label{SignsLemma}
$\mathbf{u}'(\theta) \cdot \mathbf{u}(\theta)$ and $\mathbf{v}'(\theta) \cdot \mathbf{v}(\theta)$ remain either strictly negative or strictly positive on $(\theta^{-}, \theta^+)$.  Hence, $|\mathbf{u}(\theta)|$ and $|\mathbf{v}(\theta)|$ are strictly monotonic on $(\theta^{-}, \theta^+)$.
\end{lemma}
\begin{remark}
The signs of $\mathbf{u}(\theta) \cdot \mathbf{u}'(\theta)$ and $\mathbf{v}'(\theta) \cdot \mathbf{v}'(\theta)$ need not agree. 
\end{remark}
\noindent See Appendix \ref{ssec:MechProofs1} for the proof.

We now construct the mechanism motion of a given parallelogram origami pattern under the assumption that the motion contains the reference pattern and does not change the mountain-valley assignment. Start with the fact that the cell's mechanism motion is given by deforming the creases via $(\mathbf{t}_1^r, \ldots, \mathbf{t}_4^r) \mapsto (\mathbf{t}_1(\theta), \ldots, \mathbf{t}_4(\theta))$ from Proposition \ref{mechKinCellProp}.   Observe that there are unique  rotation fields $\mathbf{R}_i \colon (\theta^{-}, \theta^+) \rightarrow SO(3)$ such that 
\begin{equation}
\begin{aligned}\label{eq:MechConstruct3}
\mathbf{R}_i(\theta)  \mathbf{t}_i^r = \mathbf{t}_i(\theta), \quad \mathbf{R}_{i}(\theta) \mathbf{t}_{\sigma(i)}^r = \mathbf{t}_{\sigma(i)}(\theta), \quad \mathbf{R}_{i}(\theta) ( \mathbf{t}_i^r \times \mathbf{t}_{\sigma(i)}^r) = \mathbf{t}_i(\theta) \times \mathbf{t}_{\sigma(i)}(\theta),
\end{aligned}
\end{equation}
where $\sigma(i)$ is a cyclic permutation of $\{ 1,2,3,4\}$.   The mechanism motion of the unit cell satisfies
\begin{equation}
\begin{aligned}
\Omega^{\theta}_{\text{cell}} := \mathbf{R}_1(\theta) \mathcal{P}_1 \cup \mathbf{R}_2(\theta) \mathcal{P}_2 \cup \mathbf{R}_3(\theta) \mathcal{P}_3 \cup \mathbf{R}_4(\theta) \mathcal{P}_4, \quad \theta \in (\theta^{-}, \theta^+).
\end{aligned}
\end{equation}
In turn,  the mechanism motion of the overall pattern repeats the motion of the unit cell via 
\begin{equation}
\begin{aligned}\label{eq:MechConstructFinal}
\mathcal{T}_{\text{ori}}^{\theta}:= \big\{ \Omega^{\theta}_{\text{cell}}  + i \mathbf{u}(\theta) + j \mathbf{v}(\theta) \colon i, j \in \mathbb{Z} \big\} , \quad \theta \in (\theta^{-} , \theta^+)
\end{aligned}
\end{equation}
for the lattice vectors $\mathbf{u}(\theta)$ and $\mathbf{v}(\theta)$.
\begin{proposition}\label{MechProp}
For each $\theta \in (\theta^{-}, \theta^+)$, there is a continuous and rigid deformation $\mathbf{y}_{\theta} \colon \mathcal{T}_{\emph{ori}} \rightarrow \mathbb{R}^3$ that satisfies $\mathbf{y}_{\theta}( \mathcal{T}_{\emph{ori}} ) = \mathcal{T}_{\emph{ori}}^{\theta}$. Furthermore, up to an overall rigid motion,   $\mathbf{y}_{\theta}( \mathcal{T}_{\emph{ori}} ) $, $\theta \in (\theta^{-}, \theta^+),$ is the unique mechanism motion of $\mathcal{T}_{\emph{ori}}$ that contains $\mathcal{T}_{\emph{ori}}$ and does not change the mountain-valley assignment.  Finally, $\mathbf{y}_0(\mathcal{T}_{\emph{ori}})  = \mathcal{T}_{\emph{ori}}$ and the parameterization is analytic in $\theta$. 
\end{proposition}
\noindent It is not hard to show the existence of such a mechanism motion. Uniqueness is more subtle --- its proof involves repeated applications of the implicit function theorem under the hypothesis that the mountain-valley assignment is unchanged. We sketch  the proof using Fig.\;\ref{Fig:Marching} below; again, see Appendix \ref{ssec:MechProofs3} for the details.

Start by folding the crease between two adjacent panels of a cell in a manner that takes the dihedral angle  from $\theta_0$ to $\theta_0 + \theta$ as  in Fig.\;\ref{Fig:Marching}(a). We claim that the pattern's mechanism motion is uniquely determined by this folding. Indeed, the deformed wedge in red in Fig.\;\ref{Fig:Marching}(a) must be matched by a corresponding folding of the neighboring pair of panels. A basic geometry argument  dictates that there is exactly one way to do this folding while preserving the mountain-valley assignment (see Appendix \ref{ssec:MechProofs2}). This solution is sketched in Fig.\;\ref{Fig:Marching}(b). In other words, the folding of a single crease determines the folding of the opposite crease about a vertex. This argument can be repeated to fit  the adjacent pair of panels uniquely and compatibly to the right of the deformed unit cell  and above it (Fig.\;\ref{Fig:Marching}(c)), and then again to the right and above (Fig.\;\ref{Fig:Marching}(d)).  Moving on to the fourth cell, there are now two deformed wedges to contend with. It is either the case that the two ways to march the construction yield the same result, or an incompatibility arises and the pattern fails to have a mechanism motion that contains the reference configuration.  The latter is impossible, since the periodically repeating cells in Fig.\;\ref{Fig:Marching}(e) define a compatible mechanism motion. Uniqueness of the mechanism motion follows.

\subsection{Local bending kinematics}\label{ssec:LocalBend}
Having characterized the mechanism motions of generic parallelogram origami patterns, we now study their local bending kinematics. We start by relaxing the mechanism kinematics of a single unit cell to allow for perturbations that preserve the distances between neighboring vertices to leading order.   We  then solve the fitting problem for such ``slightly bent''  neighboring unit cells, as discussed in Section \ref{ssec:MainResult}.

\begin{figure}[t!]
\centering
\includegraphics[width=1\textwidth]{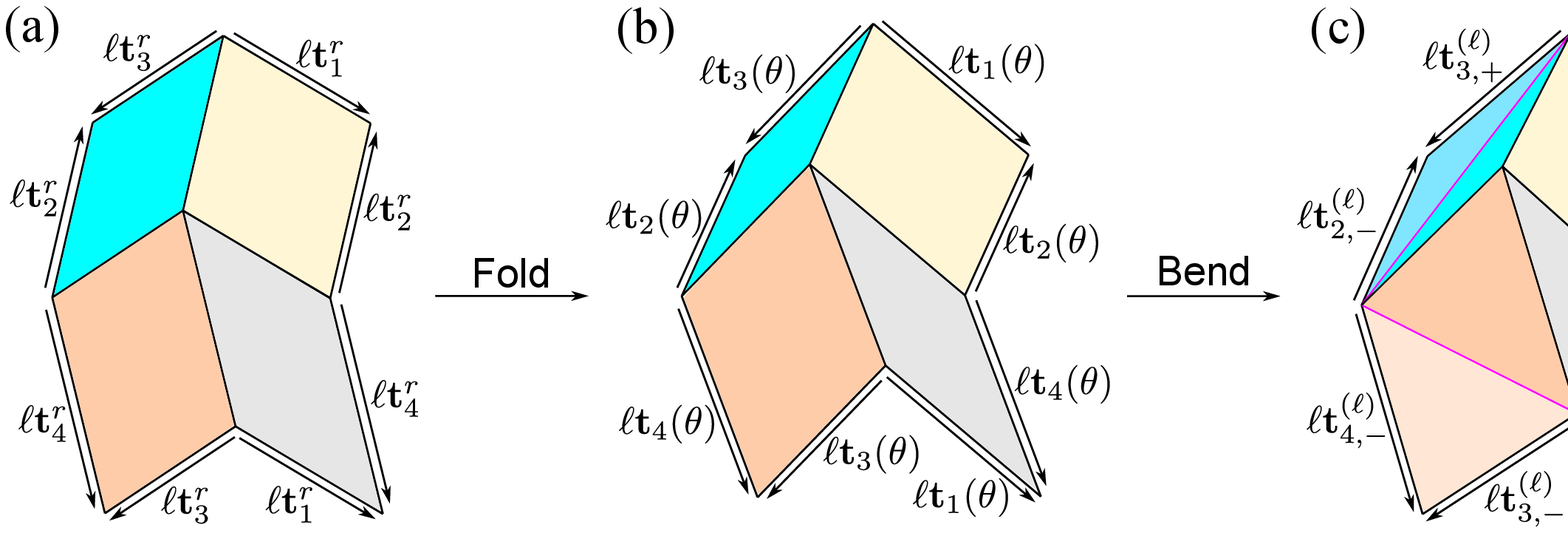} 
\caption{Notation for the local bending kinematics of a single cell of characteristic length $\ell$. Boundary of (a) the reference configurations, (b) after folding and (c) after folding and bending.}
\label{Fig:idepercell}
\end{figure}

We first fix a single parallelogram unit cell $\Omega_{\text{cell}}^{(\ell)} = \ell \Omega_{\text{cell}}$ that has been scaled by $0< \ell \ll 1$, so that the crease vectors are $\ell \mathbf{t}_1^r, \ldots, \ell \mathbf{t}_4^r$.   Per the previous section, the mechanism deformation is parameterized up to rigid motion by deforming the creases as  $\ell \mathbf{t}_i^r  \mapsto \ell \mathbf{t}_i(\theta)$, $i = 1,\ldots,4$, $\theta \in (\theta^{-}, \theta^+)$, where the deformed $\mathbf{t}_1(\theta), \ldots, \mathbf{t}_4(\theta)$ creases satisfy Eq.\;(\ref{eq:MechConstruct1}-\ref{eq:MechConstruct2}). Fig.\;\ref{Fig:idepercell}(a-b) sketches this deformation. 


We seek a small perturbation to the mechanism that does not change lengths at leading order. We first argue that we may fix the interior crease vectors to be $\mathbf{t}_i(\theta)$, $i = 1,\ldots, 4$, and distort only the cell boundaries. Indeed, suppose we distort the interior crease vectors through small perturbations 
\begin{equation}
\begin{aligned}\label{eq:infinitePert}
\mathbf{t}_i(\theta) \mapsto \mathbf{t}_i(\theta) + \ell \boldsymbol{\delta} \mathbf{t}_i, \quad i = 1,\ldots, 4
\end{aligned}
\end{equation}
 that preserve lengths to leading order in $\ell$. For this to be the case, it is necessary and sufficient that the perturbations $\boldsymbol{\delta} \mathbf{t}_i \in \mathbb{R}^3$, $i =1,\ldots,4,$ satisfy the orthogonality conditions
\begin{equation}
\begin{aligned}\label{eq:compatBendFirst}
&\boldsymbol{\delta} \mathbf{t}_i \cdot \mathbf{t}_i(\theta) = 0 \quad i = 1,\ldots,4, \quad  \text{ and } \quad (\boldsymbol{\delta} \mathbf{t}_i - \boldsymbol{\delta} \mathbf{t}_j) \cdot ( \mathbf{t}_i(\theta) - \mathbf{t}_j(\theta)) = 0 \quad ij \in \{ 12,23,34,41\}. 
\end{aligned}
\end{equation}
Such perturbations turn out to be a combination of an overall infinitesimal rotation of the cell and a perturbation of the mechanism angle $\theta \in (\theta^{-}, \theta^+)$.
\begin{lemma}
Eq.\;(\ref{eq:compatBendFirst}) holds if and only if 
\begin{equation}
\begin{aligned}\label{eq:compatBendSecond}
\boldsymbol{\delta} \mathbf{t}_i = \boldsymbol{\omega} \times \mathbf{t}_i(\theta) + \xi \mathbf{t}_i'(\theta), \quad i = 1,\ldots, 4
\end{aligned}
\end{equation}
for some $\boldsymbol{\omega} \in \mathbb{R}^3$ and $\xi \in \mathbb{R}$. The perturbation satisfies 
\begin{equation}
\begin{aligned}\label{eq:compatBendThird}
\mathbf{t}_i(\theta) + \ell \boldsymbol{\delta}  \mathbf{t}_i = \big[\mathbf{I} + \ell( \boldsymbol{\omega} \times)\big] \mathbf{t}_i(\theta + \ell \xi) + O(\ell^2), \quad i = 1,\ldots, 4.
\end{aligned}
\end{equation}
\end{lemma}
\begin{proof}
Proposition \ref{mechKinCellProp} shows that the set of mechanism deformations of a cell $\mathcal{M}$ is a  $4$-dimensional manifold, since each $\mathbf{t}_i(\theta)$ is smooth.  The key observation is that Eq.\;(\ref{eq:infinitePert}-\ref{eq:compatBendFirst}) describe algebraic conditions that are necessary and sufficient for an infinitesimal perturbation of the mechanism. As a result, each $(\mathbf{R}\boldsymbol{\delta} \mathbf{t}_1,\ldots, \mathbf{R}\boldsymbol{\delta} \mathbf{t}_4)$ must belong to the tangent space of the manifold in Eq.\;(\ref{eq:Mmanifold}) at the point $(\mathbf{R}\mathbf{t}_1(\theta), \ldots, \mathbf{R} \mathbf{t}_4(\theta))$. Since the dimension of the tangent space of a manifold is the same as that of the manifold, we conclude that the solution space to Eq.\;(\ref{eq:compatBendFirst}) is four dimensional.

Next, observe that the parameterization in Eq.\;(\ref{eq:compatBendSecond}) solves all the conditions in Eq.\;(\ref{eq:compatBendFirst}) for any choice of $\boldsymbol{\omega}$ and $\xi$. We thus complete the proof by showing  that $\boldsymbol{\omega}$ and $\xi$ are four linearly independent DOFs in this parameterization. Assume  
\begin{equation}
    \begin{aligned}\label{eq:TrivialLemmaBend1}
      \boldsymbol{\omega} \times \mathbf{t}_i(\theta) + \xi \mathbf{t}_i'(\theta) =  \mathbf{0} \quad \text{for all } i = 1,\ldots,4.  
    \end{aligned}
\end{equation}
By subtracting the $i = 1,3$ equations and the $i = 2,4$ equations, Eq.\;(\ref{eq:TrivialLemmaBend1}) implies that 
\begin{equation}
\begin{aligned}\label{eq:TrivialLemmaBend2}
\boldsymbol{\omega} \times \mathbf{u}(\theta) + \xi \mathbf{u}'(\theta) = \mathbf{0} , \quad 
\boldsymbol{\omega} \times \mathbf{v}(\theta) + \xi \mathbf{v}'(\theta) = \mathbf{0} ,
\end{aligned}
\end{equation}
for $\mathbf{u}(\theta) = \mathbf{t}_1(\theta) - \mathbf{t}_3(\theta)$ and $\mathbf{v}(\theta) = \mathbf{t}_2(\theta) - \mathbf{t}_4(\theta)$, which span the plane normal to $\mathbf{e}_3$ by construction.  To characterize Eq.\;(\ref{eq:TrivialLemmaBend2}), we write $\boldsymbol{\omega}$ in the $\{ \mathbf{u}(\theta), \mathbf{v}(\theta), \mathbf{e}_3\}$ basis as   $\boldsymbol{\omega} = \omega_u \mathbf{u}(\theta) + \omega_v \mathbf{v}(\theta) + \omega_3 \mathbf{e}_3$ and obtain the equations  $\omega_v (\mathbf{e}_3 \cdot (\mathbf{v}(\theta) \times \mathbf{u}(\theta)) \mathbf{e}_3 + \omega_3 \mathbf{e}_3 \times \mathbf{u}(\theta) +  \xi \mathbf{u}'(\theta) = \mathbf{0}$  and $\omega_u (\mathbf{e}_3 \cdot (\mathbf{u}(\theta) \times \mathbf{v}(\theta)) \mathbf{e}_3 + \omega_3 \mathbf{e}_3 \times \mathbf{v}(\theta) +  \xi \mathbf{v}'(\theta) = \mathbf{0}$.  Note that $|\mathbf{u}(\theta)|, |\mathbf{v}(\theta)|$ are strictly monotonic functions of $\theta$ by Lemma \ref{SignsLemma}. As a result, $\mathbf{u}'(\theta)$ and $\mathbf{v}'(\theta)$ have non-zero components in the $\mathbf{u}(\theta)$ and $\mathbf{v}(\theta)$ directions. It follows that  $\{ \mathbf{e}_3, \mathbf{e}_3 \times \mathbf{u}(\theta), \mathbf{u}'(\theta)\}$   and $\{ \mathbf{e}_3, \mathbf{e}_3 \times \mathbf{v}(\theta), \mathbf{v}'(\theta)\}$ both span $\mathbb{R}^3$. Thus, Eq.\;(\ref{eq:TrivialLemmaBend2}) implies $\boldsymbol{\omega} = \mathbf{0}$ and $\xi = 0$. This proves the desired linear independence and  shows that Eq.\;(\ref{eq:compatBendSecond}) fully characterizes  the solutions to Eq.\;(\ref{eq:compatBendFirst}).

 The perturbation in Eq.\;(\ref{eq:compatBendThird}) follows by Taylor expansion. \end{proof}

Due to the structure of the perturbations in Eq.\;(\ref{eq:compatBendThird}), we are free to fix the interior creases as $\mathbf{t}_i(\theta)$, $i =1,\ldots,4$, and distort only the corner nodes of the cell to describe its bending kinematics.  The point is that a bent cell with interior crease vectors $(\mathbf{I} + \ell (\boldsymbol{\omega} \times)) \mathbf{t}_i(\tilde{\theta} + \ell \xi)$ can be  transformed to leading order to a bent  cell with fixed interior crease vectors $\mathbf{t}_i(\theta)$ through an  infinitesimal rotation of the  cell $(\mathbf{I} - \ell (\boldsymbol{\omega} \times ) )$ and by replacing $\tilde{\theta} + \ell \xi$ with $\theta$. With this, the  bending kinematics of a unit cell are described as follows. Observe that the four corner points of the mechanism deformation are the sum of the the deformed crease vectors $\mathbf{s}_1(\theta) := \mathbf{t}_1(\theta) + \mathbf{t}_2(\theta), \ldots, \mathbf{s}_4(\theta) = \mathbf{t}_4(\theta) + \mathbf{t}_1(\theta)$. Since three of the four vertices of each panel are now fixed, the corner points must be perturbed along the normal of the panel. So, $\mathbf{s}_i(\theta) \mapsto  \mathbf{s}_i(\theta) + \ell \boldsymbol{\delta} \mathbf{s}_i$ satisfies 
\begin{equation}
\begin{aligned}
&\boldsymbol{\delta} \mathbf{s}_1 =  \kappa_1  ( \mathbf{t}_1(\theta) \times \mathbf{t}_2(\theta)), \;\;\; \boldsymbol{\delta} \mathbf{s}_2 =  \kappa_2  ( \mathbf{t}_2(\theta) \times \mathbf{t}_3(\theta)), \;\;\; \boldsymbol{\delta} \mathbf{s}_3 =  \kappa_3  ( \mathbf{t}_3(\theta) \times \mathbf{t}_4(\theta)), \;\;\; \boldsymbol{\delta} \mathbf{s}_4 =  \kappa_4 ( \mathbf{t}_4(\theta) \times \mathbf{t}_1(\theta))
\end{aligned}
\end{equation}
for panel curvatures $\kappa_1, \ldots, \kappa_4 \in \mathbb{R}$. For future reference, the boundaries of the slightly bent cell are $\ell \mathbf{t}_{i,\pm}^{(\ell)}(\theta, \boldsymbol{\kappa})$ as indicated in Fig.\;\ref{Fig:idepercell}(c), where 
\begin{equation}
\begin{aligned}\label{eq:bentBoundaries1}
\mathbf{t}_{1,-}^{(\ell)}(\theta,\boldsymbol{\kappa}) :=   \mathbf{t}_{1}(\theta) + \ell \kappa_4 ( \mathbf{t}_4(\theta) \times \mathbf{t}_1(\theta)) , \quad   \mathbf{t}_{3,-}^{(\ell)} (\theta,\boldsymbol{\kappa}) := \mathbf{t}_3(\theta) + \ell \kappa_3 ( \mathbf{t}_3(\theta) \times \mathbf{t}_4(\theta)), \\
\mathbf{t}_{1,+}^{(\ell)}(\theta,\boldsymbol{\kappa}) :=   \mathbf{t}_{1}(\theta) + \ell \kappa_1 ( \mathbf{t}_1(\theta) \times \mathbf{t}_2(\theta)) , \quad   \mathbf{t}_{3,+}^{(\ell)} (\theta,\boldsymbol{\kappa}) := \mathbf{t}_3(\theta) + \ell \kappa_2 ( \mathbf{t}_2(\theta) \times \mathbf{t}_3(\theta)), \\
\mathbf{t}_{2,-}^{(\ell)}(\theta,\boldsymbol{\kappa}) :=   \mathbf{t}_{2}(\theta) + \ell \kappa_2 ( \mathbf{t}_2(\theta) \times \mathbf{t}_3(\theta)) , \quad   \mathbf{t}_{4,-}^{(\ell)} (\theta,\boldsymbol{\kappa}) := \mathbf{t}_4(\theta) + \ell \kappa_3 ( \mathbf{t}_3(\theta) \times \mathbf{t}_4(\theta)), \\
\mathbf{t}_{2,+}^{(\ell)}(\theta,\boldsymbol{\kappa}) :=   \mathbf{t}_{2}(\theta) + \ell \kappa_1 ( \mathbf{t}_1(\theta) \times \mathbf{t}_2(\theta)) , \quad   \mathbf{t}_{4,+}^{(\ell)} (\theta,\boldsymbol{\kappa}) := \mathbf{t}_4(\theta) + \ell \kappa_4 ( \mathbf{t}_4(\theta) \times \mathbf{t}_1(\theta)).
\end{aligned}
\end{equation}
This completes our parameterization of a single slightly bent cell. 

\begin{figure}[t!]
\centering
\includegraphics[width=0.9\textwidth]{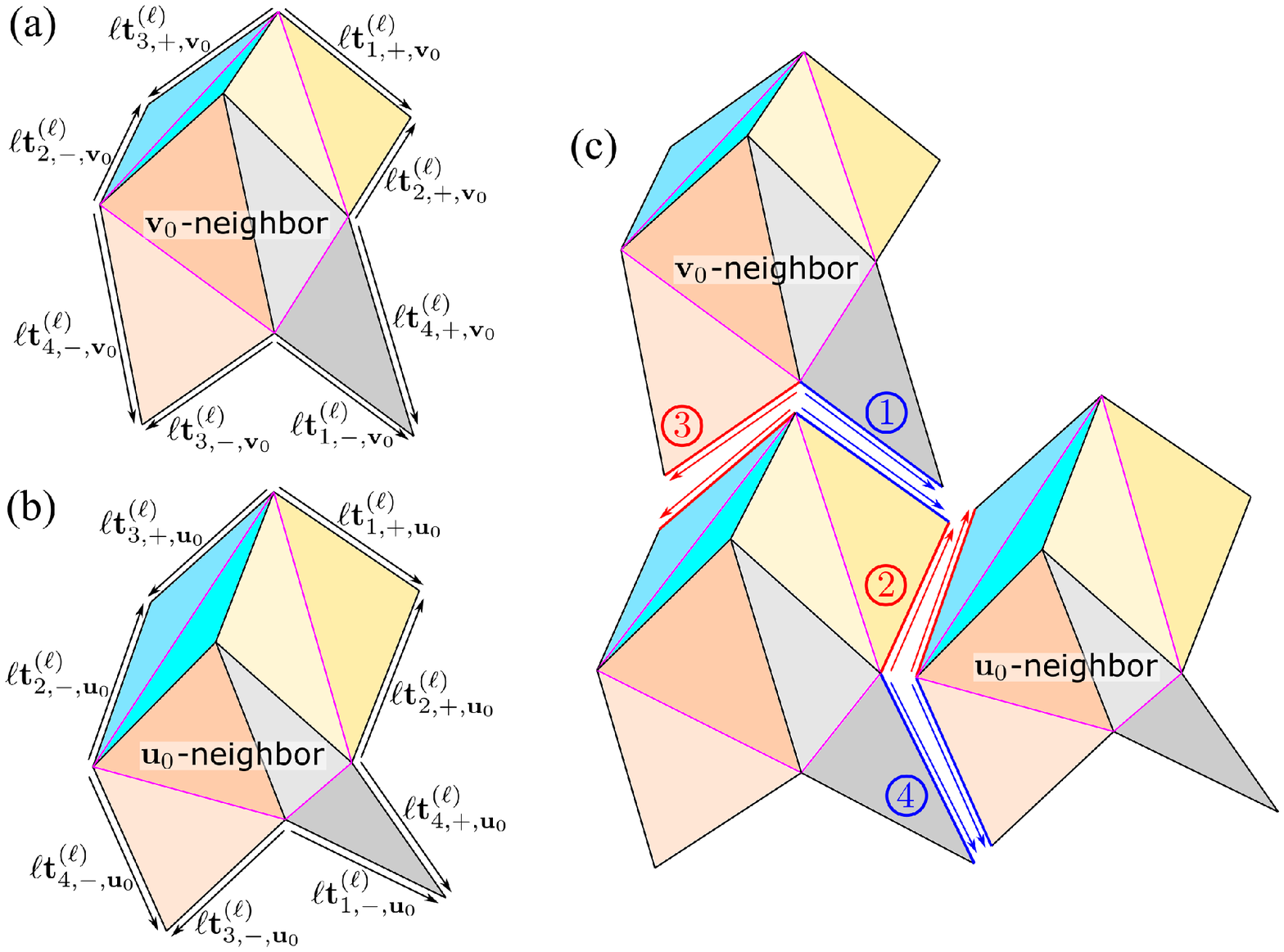} 
\caption{Neighbor compatibility for cells of characteristic length $\ell$. (a) Notation for the $\mathbf{v}_0$-neighbor and (b) $\mathbf{u}_0$-neighbor. (c) Description for fitting together three neighboring cells. The four vector compatibility conditions  indicated appear in Eq.\;(\ref{eq:neighborBendFinal}).}
\label{Fig:FitCells}
\end{figure}

We now show how to fit together neighboring slightly bent unit cells.  Let Eq.\;(\ref{eq:bentBoundaries1}) describe the boundaries of an actuated and slightly bent $\Omega_{\text{cell}}^{(\ell)}$, as in Fig.\;\ref{Fig:idepercell}(c). Then, following Fig.\;\ref{Fig:FitCells}(a), assume that the neighboring cell $\Omega_{\text{cell}}^{(\ell)} + \ell \mathbf{v}_0$ is actuated to a slightly bent configuration with cell boundary tangents 
\begin{equation}
\begin{aligned}\label{eq:bentBoundaries2}
\mathbf{t}_{i,\pm, \mathbf{v}_0}^{(\ell)}(\theta,\boldsymbol{\kappa}, \delta \theta_{\mathbf{v}_0} , \boldsymbol{\omega}_{\mathbf{v}_0})  :=   (\mathbf{I}+ \ell (\boldsymbol{\omega}_{\mathbf{v}_0} \times)) \mathbf{t}_{i,\pm}^{(\ell)}(\theta + \ell  \delta \theta_{\mathbf{v}_0} ,\boldsymbol{\kappa}), \quad i = 1,\ldots,4.
\end{aligned}
\end{equation}
Finally, following Fig.\;\ref{Fig:FitCells}(b), assume that the neighboring cell $\Omega_{\text{cell}}^{(\ell)} + \ell \mathbf{u}_0$ is actuated to a slightly bent configuration with  cell boundary tangents 
\begin{equation}
\begin{aligned}\label{eq:bentBoundaries3}
\mathbf{t}_{i,\pm, \mathbf{u}_0}^{(\ell)}(\theta,\boldsymbol{\kappa}, \delta \theta_{\mathbf{u}_0} , \boldsymbol{\omega}_{\mathbf{u}_0})  :=   (\mathbf{I}+ \ell (\boldsymbol{\omega}_{\mathbf{u}_0} \times)) \mathbf{t}_{i,\pm}^{(\ell)}(\theta + \ell  \delta \theta_{\mathbf{u}_0} ,\boldsymbol{\kappa}), \quad i = 1,\ldots,4.
\end{aligned}
\end{equation}
The assumptions here are motivated by the physical observation that two slightly bent neighboring cells in an overall pattern do not deviate much from each other in their kinematics. Compared to the original cell, the neighbors pick up an actuation  angle  $\ell\delta \theta_{(\cdot)}$ to account for small changes in the folds between the cells. They also pick up an  infinitesimal rotation  $\mathbf{I} + \ell ( \boldsymbol{\omega}_{(\cdot)}\times )$, which is needed to fit the cells together compatibly at leading order. (Small perturbations of the curvatures $\boldsymbol{\kappa} \mapsto \boldsymbol{\kappa} + \ell \boldsymbol{\delta} \boldsymbol{\kappa}_{(\cdot)}$ are also possible. However, they enter the kinematics at a higher order in $\ell$.) 

Fig.\;\ref{Fig:FitCells}(c) illustrates the local fitting problem constraining the parameters $\boldsymbol{\kappa}$, $\delta \theta_{\mathbf{u}_0}$, $\delta \theta_{\mathbf{v}_0}$, $\boldsymbol{\omega}_{\mathbf{u}_0}$ and  $\boldsymbol{\omega}_{\mathbf{v}_0}$ in Eq.\;(\ref{eq:bentBoundaries1}-\ref{eq:bentBoundaries3}). There are four compatibility conditions to solve. They concern an approximate fitting of neighboring tangents of the form 
\begin{equation}
\begin{aligned}\label{eq:neighborBendFinal}
\mathbf{t}^{(\ell)}_{1,-,\mathbf{v}_0}(\theta,\boldsymbol{\kappa}, \delta \theta_{\mathbf{v}_0} , \boldsymbol{\omega}_{\mathbf{v}_0}) - \mathbf{t}^{(\ell)}_{1,+}(\theta,\boldsymbol{\kappa}) = O(\ell^2),  \quad \mathbf{t}^{(\ell)}_{3,-,\mathbf{u}_0}(\theta,\boldsymbol{\kappa}, \delta \theta_{\mathbf{v}_0} , \boldsymbol{\omega}_{\mathbf{v}_0}) - \mathbf{t}^{(\ell)}_{3,+}(\theta,\boldsymbol{\kappa}) = O(\ell^2), \\ 
\mathbf{t}^{(\ell)}_{2,-,\mathbf{u}_0}(\theta,\boldsymbol{\kappa}, \delta \theta_{\mathbf{u}_0} , \boldsymbol{\omega}_{\mathbf{u}_0}) - \mathbf{t}^{(\ell)}_{2,+}(\theta,\boldsymbol{\kappa}) = O(\ell^2), \quad \mathbf{t}^{(\ell)}_{4,-,\mathbf{u}_0}(\theta,\boldsymbol{\kappa}, \delta \theta_{\mathbf{u}_0} , \boldsymbol{\omega}_{\mathbf{u}_0}) - \mathbf{t}^{(\ell)}_{4,+}(\theta,\boldsymbol{\kappa}) = O(\ell^2).
\end{aligned}
\end{equation}
Once these conditions are solved, cell translations can be chosen resulting in gaps $\sim \ell^3$ between the neighboring cells and strains $\sim \ell^2$. (Doing so globally, however, requires additional enrichment of the ansatz by local perturbations, as will be described in Section \ref{ssec:GlobalOrigami}.) 
\begin{proposition}\label{LocalBendProp} 
Eq.\;(\ref{eq:neighborBendFinal}) holds if and only if  $\boldsymbol{\omega}_{\mathbf{u}_0}$, $\boldsymbol{\omega}_{\mathbf{v}_0}$  $\delta \theta_{\mathbf{u}_0} $ and $\delta \theta_{\mathbf{v}_0} $ satisfy  
\begin{equation}
\begin{aligned}\label{eq:micon1}
& \bm{\omega}_{\mathbf{u}_0}\times\mathbf{v}(\theta)+\delta\theta_{\mathbf{u}_0}\mathbf{v}'(\theta)= \bm{\omega}_{\mathbf{v}_0}\times\mathbf{u}(\theta)+ \delta\theta_{\mathbf{v}_0}\mathbf{u}'(\theta), \\
& \bm{\omega}_{\mathbf{u}_0}\cdot\mathbf{v}'(\theta)= \bm{\omega}_{\mathbf{v}_0}\cdot\mathbf{u}'(\theta)
\end{aligned}
\end{equation}
and $\boldsymbol{\kappa} = (\kappa_1, \ldots, \kappa_4)$ satisfies
\begin{equation}
\begin{aligned}\label{eq:theKappas1}
&\begin{pmatrix}
\kappa_1 \\ \kappa_2 \\ \kappa_3 \\ \kappa_4 
\end{pmatrix} = \underbrace{\begin{pmatrix} \frac{(\mathbf{t}_2(\theta)\times\mathbf{t}_3(\theta))^T}{V_{123}(\theta)}  & \frac{\mathbf{t}_3(\theta)\cdot\mathbf{t}'_2(\theta)}{V_{123}(\theta)} & 0 & 0  \\ 0 & 0  &   \frac{(\mathbf{t}_3(\theta)\times\mathbf{t}_4(\theta))^T}{V_{234}(\theta) }& \frac{\mathbf{t}'_3(\theta)\cdot\mathbf{t}_4(\theta)}{V_{234}(\theta)} \\
\frac{(\mathbf{t}_4(\theta) \times \mathbf{t}_1(\theta))^T}{V_{314}(\theta)}  & \frac{\mathbf{t}_1(\theta) \cdot \mathbf{t}_4'(\theta)}{V_{314}(\theta)}  & 0 & 0   \\ 0 & 0 & \frac{(\mathbf{t}_1(\theta) \times \mathbf{t}_2(\theta))^T}{V_{421}(\theta)} & \frac{\mathbf{t}_1'(\theta) \cdot \mathbf{t}_2(\theta)}{V_{421}(\theta)}  \end{pmatrix}}_{=:\mathbf{B}(\theta)}  \begin{pmatrix} \boldsymbol{\omega}_{\mathbf{u}_0} \\ \delta \theta_{\mathbf{u}_0} \\ \boldsymbol{\omega}_{\mathbf{v}_0} \\ \delta \theta_{\mathbf{v}_0}  \end{pmatrix} 
\end{aligned}
\end{equation}
where $V_{ijk}(\theta) = \mathbf{t}_i(\theta) \cdot (\mathbf{t}_j(\theta) \times \mathbf{t}_k(\theta))$. 
\end{proposition}
\begin{remark}\label{EquivRemark} In the course of the proof, we also establish that Eq.\;(\ref{eq:micon1}) holds if and only if 
\begin{equation}
\begin{aligned}\label{eq:genParamOmega}
\boldsymbol{\omega}_{\mathbf{u}_0} = \tau \mathbf{u}(\theta) + \kappa (\mathbf{u}'(\theta) \cdot \mathbf{u}(\theta)) \mathbf{v}(\theta) +  \Big(\frac{\delta\theta_{\mathbf{u}_0} \mathbf{v}'(\theta) \cdot \mathbf{u}(\theta) - \delta \theta_{\mathbf{v}_0} \mathbf{u}'(\theta) \cdot \mathbf{u}(\theta) }{\mathbf{e}_3 \cdot ( \mathbf{u}(\theta) \times \mathbf{v}(\theta))}  \Big) \mathbf{e}_3, \\
\boldsymbol{\omega}_{\mathbf{v}_0} = \kappa (\mathbf{v}'(\theta) \cdot \mathbf{v}(\theta))  \mathbf{u}(\theta) - \tau  \mathbf{v}(\theta) +  \Big(\frac{\delta\theta_{\mathbf{u}_0} \mathbf{v}'(\theta) \cdot \mathbf{v}(\theta) - \delta \theta_{\mathbf{v}_0} \mathbf{u}'(\theta) \cdot \mathbf{v}(\theta)}{\mathbf{e}_3 \cdot ( \mathbf{u}(\theta) \times \mathbf{v}(\theta))} \Big) \mathbf{e}_3,
\end{aligned}
\end{equation}
for some $\kappa$ and $\tau  \in \mathbb{R}$. These parameters can be interpreted as ``bend" ($\kappa$) and ``twist" ($\tau$) of the unit cell; see Section \ref{ssec:discussSec-modes} for further discussion and Section \ref{ssec:Examples} for examples of these modes.
\end{remark}
\begin{proof}
By expanding out Eq.\;(\ref{eq:neighborBendFinal}), we conclude from the $O(\ell)$ terms that it is equivalent to 
\begin{equation}
\begin{aligned}\label{eq:compBendEquiv}
&\big[ \boldsymbol{\omega}_{\mathbf{v}_0}  + \kappa_4 \mathbf{t}_4(\theta)  + \kappa_1 \mathbf{t}_2(\theta) \big] \times \mathbf{t}_1(\theta) + \delta \theta_{\mathbf{v}_0}  \mathbf{t}_1'(\theta) = \mathbf{0}, 
 &&\big[\boldsymbol{\omega}_{\mathbf{v}_0} - \kappa_3 \mathbf{t}_4(\theta)  - \kappa_2 \mathbf{t}_2(\theta) \big] \times \mathbf{t}_3(\theta) + \delta \theta_{\mathbf{v}_0}  \mathbf{t}'_3(\theta)  = \mathbf{0},    \\
&\big[\boldsymbol{\omega}_{\mathbf{u}_0} - \kappa_2 \mathbf{t}_3(\theta) - \kappa_1 \mathbf{t}_1(\theta)  \big] \times \mathbf{t}_2(\theta) + \delta \theta_{\mathbf{u}_0} \mathbf{t}_2'(\theta) = \mathbf{0},  
 &&\big[\boldsymbol{\omega}_{\mathbf{u}_0} + \kappa_3 \mathbf{t}_3(\theta) + \kappa_4 \mathbf{t}_1(\theta)  \big] \times  \mathbf{t}_4(\theta) + \delta \theta_{\mathbf{u}_0} \mathbf{t}_4'(\theta) = \mathbf{0}.
\end{aligned}
\end{equation}
We focus on necessary conditions to Eq.\;(\ref{eq:compBendEquiv}) and then turn to sufficiency.  

To find necessary conditions, we eliminate the panel curvatures $\kappa_1,\dots,\kappa_4$ from Eq.\;(\ref{eq:compBendEquiv}) while simultaneously bringing in the Bravais lattice vectors $\mathbf{u}(\theta)=\mathbf{t}_1(\theta) - \mathbf{t}_3(\theta)$ and $\mathbf{v}(\theta)=\mathbf{t}_2(\theta) - \mathbf{t}_4(\theta)$. Taking the difference of the first and second equation, and similarly the third and fourth, in Eq.\;(\ref{eq:compBendEquiv}) gives that  
\begin{equation}
 \begin{aligned}\label{eq:comp4}
 \bm{\omega}_{\mathbf{v}_0}\times\mathbf{u}(\theta)+\big[\kappa_1\mathbf{t}_2(\theta)+ \kappa_4\mathbf{t}_4(\theta)\big]\times\mathbf{t}_1(\theta) -\big[-\kappa_2\mathbf{t}_2(\theta)-\kappa_3\mathbf{t}_4(\theta)\big]\times\mathbf{t}_3(\theta)+ \delta\theta_{\mathbf{v}_0}\mathbf{u}'(\theta) &=\mathbf{0}, \\ 
 \bm{\omega}_{\mathbf{u}_0}\times\mathbf{v}(\theta)+\big[-\kappa_2\mathbf{t}_3(\theta)-\kappa_1\mathbf{t}_1(\theta)\big]\times\mathbf{t}_2(\theta) -\big[\kappa_3\mathbf{t}_3(\theta)+\kappa_4\mathbf{t}_1(\theta)\big]\times\mathbf{t}_4(\theta)+ \delta\theta_{\mathbf{u}_0}\mathbf{v}'(\theta) &=\mathbf{0}.
\end{aligned}
\end{equation}
Taking the difference of these equations produces  the first identity in Eq.\;(\ref{eq:micon1}); the $\kappa_i$'s cancel. 

Next, we dot  the four compatibility conditions in Eq.\;(\ref{eq:compBendEquiv}) by $\mathbf{t}'_1(\theta) \times \mathbf{t}_1(\theta), \mathbf{t}'_3(\theta) \times \mathbf{t}_3(\theta)$, $\mathbf{t}_2'(\theta) \times \mathbf{t}_2(\theta)$, $\mathbf{t}_4'(\theta) \times \mathbf{t}_4(\theta)$, respectively, to obtain the necessary conditions
\begin{equation}
\begin{aligned}\label{eq:compDiff}
 \bm{\omega}_{\mathbf{v}_0}\cdot\mathbf{t}'_1(\theta)+\kappa_1\mathbf{t}_2(\theta)\cdot\mathbf{t}'_1(\theta)+\kappa_4\mathbf{t}_4(\theta)\cdot\mathbf{t}'_1(\theta) =0, \quad \bm{\omega}_{\mathbf{v}_0}\cdot\mathbf{t}'_3(\theta)-\kappa_2\mathbf{t}_2(\theta)\cdot\mathbf{t}'_3(\theta)-\kappa_3\mathbf{t}_4(\theta)\cdot\mathbf{t}'_3(\theta) =0, \\
 \bm{\omega}_{\mathbf{u}_0}\cdot\mathbf{t}'_2(\theta)-\kappa_1\mathbf{t}_1(\theta)\cdot\mathbf{t}'_2(\theta)-\kappa_2\mathbf{t}_3(\theta)\cdot\mathbf{t}'_2(\theta)=0, \quad \bm{\omega}_{\mathbf{u}_0}\cdot\mathbf{t}'_4(\theta)+\kappa_4\mathbf{t}_1(\theta)\cdot\mathbf{t}'_4(\theta)+\kappa_3\mathbf{t}_3(\theta)\cdot\mathbf{t}'_4(\theta) =0,
\end{aligned}
\end{equation}
after manipulations involving $\mathbf{t}_i'(\theta)\cdot \mathbf{t}_i(\theta) = 0$ and  $(\mathbf{t}_i'(\theta) \times \mathbf{t}_i(\theta)) \cdot (\mathbf{b} \times \mathbf{t}_i(\theta)) = (\mathbf{t}_i'(\theta) \cdot \mathbf{b}) |\mathbf{t}_i^r|^2$, $\mathbf{b} \in \mathbb{R}^3$. Taking differences on the rows in Eq.\;(\ref{eq:compDiff}) gives that
\begin{equation}
 \begin{aligned}\label{eq:compDiffDiff}
 \bm{\omega}_{\mathbf{v}_0}\cdot\mathbf{u}'(\theta)+\kappa_1\mathbf{t}_2(\theta)\cdot\mathbf{t}'_1(\theta) + \kappa_4\mathbf{t}_4(\theta)\cdot\mathbf{t}'_1(\theta) + \kappa_2\mathbf{t}_2(\theta)\cdot\mathbf{t}'_3(\theta)+\kappa_3\mathbf{t}_4(\theta)\cdot\mathbf{t}'_3(\theta) =0, \\
 \bm{\omega}_{\mathbf{u}_0}\cdot\mathbf{v}'(\theta)-\kappa_1\mathbf{t}_1(\theta)\cdot\mathbf{t}'_2(\theta)- \kappa_2\mathbf{t}_3(\theta)\cdot\mathbf{t}'_2(\theta)- \kappa_4\mathbf{t}_1(\theta)\cdot\mathbf{t}'_4(\theta)-\kappa_3\mathbf{t}_3(\theta)\cdot\mathbf{t}'_4(\theta)=0.
  \end{aligned}
\end{equation}
Taking the difference of these two equations and using the product rule of differentiation gives that 
\begin{equation}
 \begin{aligned}\label{eq:compAlmostFinal}
& \bm{\omega}_{\mathbf{u}_0}\cdot\mathbf{v}'(\theta)  -\bm{\omega}_{\mathbf{v}_0}\cdot\mathbf{u}'(\theta)  =\kappa_1\big(\mathbf{t}_1(\theta)\cdot\mathbf{t}_2(\theta)\big)'+ \kappa_2\big(\mathbf{t}_2(\theta)\cdot\mathbf{t}_3(\theta)\big)'+\kappa_3\big(\mathbf{t}_3(\theta)\cdot\mathbf{t}_4(\theta)\big)' + \kappa_4\big(\mathbf{t}_1(\theta)\cdot\mathbf{t}_4(\theta)\big)'.
  \end{aligned}
\end{equation}
Since $\mathbf{t}_1(\theta) \cdot \mathbf{t}_2(\theta) = \mathbf{t}_1^r \cdot \mathbf{t}_2^r, \ldots, \mathbf{t}_1(\theta)\cdot\mathbf{t}_4(\theta) = \mathbf{t}_1^r \cdot \mathbf{t}_4^r$ are constant, the righthand side of this equation is zero. The second identity in  Eq.\;(\ref{eq:micon1}) follows. 

As a final set of necessary conditions, we dot the first condition  in Eq.\;(\ref{eq:compBendEquiv}) by $\mathbf{t}_2(\theta)$, the second by $\mathbf{t}_4(\theta)$, the third by $\mathbf{t}_3(\theta)$, and the fourth by $\mathbf{t}_1(\theta)$. After rearranging the terms in these dot products and using that $\mathbf{t}_i(\theta) \cdot (\mathbf{t}_j(\theta) \times \mathbf{t}_k(\theta)) \neq 0$ for all $i \neq j \neq k \in \{ 1,2,3,4\}$, we obtain the formula for the $\kappa_i's$ in Eq.\;(\ref{eq:theKappas1}).  In summary, we derive that Eq.\;(\ref{eq:micon1}) and (\ref{eq:theKappas1}) are necessary to solve Eq.\;(\ref{eq:compBendEquiv}). 

We turn to sufficiency. Observe that Eq.\;(\ref{eq:compBendEquiv}) comprises at most eight linearly independent constraints since the first, second, third, and fourth conditions are orthogonal to $\mathbf{t}_1(\theta)$, $\mathbf{t}_3(\theta)$, $\mathbf{t}_2(\theta)$ and $\mathbf{t}_4(\theta)$, respectively. To complete the proof, it is enough to show that Eq.\;(\ref{eq:micon1}) and Eq.\;(\ref{eq:theKappas1}) are, in fact, eight linearly independent constraints. Clearly Eq.\;(\ref{eq:theKappas1}) are four constraints.  For Eq.\;(\ref{eq:micon1}), we write $\boldsymbol{\omega}_{\mathbf{u}_0}$ and $\boldsymbol{\omega}_{\mathbf{v}_0}$ in the $\{ \mathbf{u}(\theta), \mathbf{v}(\theta), \mathbf{e}_3\}$ basis and obtain, through standard algebraic manipulations, that its general parameterization is Eq.\;(\ref{eq:genParamOmega}) where $\kappa$ and $\tau$ are two DOFs. Thus, Eq.\;(\ref{eq:micon1}) constrains four of the six DOFs in $\boldsymbol{\omega}_{\mathbf{u}_0}$ and $\boldsymbol{\omega}_{\mathbf{v}_0}$, so it is four total constraints. This completes the proof. 
\end{proof}

\section{Derivation of the effective plate theory}\label{sec:DerivationSec}

Here, we derive the effective plate theory from  Theorem~\ref{MainTheorem}. First, we establish the surface theory part of the result using the kinematics of slightly bent parallelogram origami cells obtained in the previous section. Then, we construct the desired global origami soft modes. Section \ref{sec:CompleteTheProof} ends by proving Theorem~\ref{MainTheorem}.

\subsection{Effective surface theory}\label{ssec:EffSurfacesSection}
The previous section showed that the slightly bent kinematics of neighboring parallelogram origami cells are characterized by parameters $\theta \in (\theta^{-}, \theta^+)$, $\delta \theta_{\mathbf{u}_0} $, $\delta \theta_{\mathbf{v}_0} $,  $\boldsymbol{\omega}_{\mathbf{u}_0}$ and $\boldsymbol{\omega}_{\mathbf{v}_0}$  solving Eq.\;(\ref{eq:micon1}). We now extend these parameters  to smooth fields on the reference domain $\Omega$ through the substitutions 
\begin{equation}
\begin{aligned}
(\theta,  \delta \theta_{\mathbf{u}_0} , \delta \theta_{\mathbf{v}_0} ,  \boldsymbol{\omega}_{\mathbf{v}_0}, \boldsymbol{\omega}_{\mathbf{u}_0}) \rightarrow  (\theta(\mathbf{x}) ,  \partial_{\mathbf{u}_0} \theta(\mathbf{x}) ,\partial_{\mathbf{v}_0} \theta(\mathbf{x}) , \boldsymbol{\omega}_{\mathbf{v}_0}(\mathbf{x}) , \boldsymbol{\omega}_{\mathbf{u}_0}(\mathbf{x}) ).
\end{aligned}
\end{equation}
As such, the constraints for slightly bent cells become global pointwise constraints of the form
\begin{equation}
\begin{aligned}\label{eq:effPDE1}
&\theta(\mathbf{x}) \in (\theta^{-}, \theta^+), \\ 
&\boldsymbol{\omega}_{\mathbf{u}_0}(\mathbf{x}) \times \mathbf{v}(\theta(\mathbf{x}))  + \partial_{\mathbf{u}_0} \theta(\mathbf{x}) \mathbf{v}'(\theta(\mathbf{x}))   =  \boldsymbol{\omega}_{\mathbf{v}_0}(\mathbf{x}) \times \mathbf{u}(\theta(\mathbf{x}))  + \partial_{\mathbf{v}_0} \theta(\mathbf{x}) \mathbf{u}'(\theta(\mathbf{x})), \\
&\boldsymbol{\omega}_{\mathbf{u}_0}(\mathbf{x}) \cdot \mathbf{v}'(\theta(\mathbf{x})) = \boldsymbol{\omega}_{\mathbf{v}_0} (\mathbf{x}) \cdot \mathbf{u}'(\theta(\mathbf{x})).
\end{aligned}
\end{equation} 
We claim that these fields describe a parameterized surface when supplemented with the PDE constraint
\begin{equation}
\begin{aligned}\label{eq:effPDE2}
 \partial_{\mathbf{v}_0} \bm{\omega}_{\mathbf{u}_0}(\mathbf{x})-\partial_{\mathbf{u}_0} \bm{\omega}_{\mathbf{v}_0}(\mathbf{x})= \bm{\omega}_{\mathbf{u}_0}(\mathbf{x})\times \bm{\omega}_{\mathbf{v}_0}(\mathbf{x}).
\end{aligned}
\end{equation}
The surface, in particular, is given by a deformation $\mathbf{y}_{\text{eff}} \colon \Omega \rightarrow \mathbb{R}^3$ with 
\begin{equation}
\begin{aligned}\label{eq:effPDE3}
&\partial_{\mathbf{u}_0}  \mathbf{y}_{\text{eff}}(\mathbf{x}) = \mathbf{R}_{\text{eff}}(\mathbf{x}) \mathbf{u}(\theta(\mathbf{x})) , &&  \partial_{\mathbf{v}_0}  \mathbf{y}_{\text{eff}}(\mathbf{x}) = \mathbf{R}_{\text{eff}}(\mathbf{x}) \mathbf{v}(\theta(\mathbf{x})) ,\\
&\partial_{\mathbf{u}_0} \mathbf{R}_{\text{eff}}(\mathbf{x}) = \mathbf{R}_{\text{eff}}(\mathbf{x}) (\boldsymbol{\omega}_{\mathbf{u}_0} (\mathbf{x})  \times ), && \partial_{\mathbf{v}_0} \mathbf{R}_{\text{eff}}(\mathbf{x}) = \mathbf{R}_{\text{eff}}(\mathbf{x}) (\boldsymbol{\omega}_{\mathbf{v}_0} (\mathbf{x})  \times ). 
\end{aligned}
\end{equation} 
In fact, this is the core content of Cartan's method of moving frames  \cite{cartan2001riemannian}, as it applies to our problem of coarse graining parallelogram origami. For completeness, we provide a proof. Note $\langle \mathbf{x} \rangle = \frac{1}{|\Omega|} \int_{\Omega} \mathbf{x} dA$.

\begin{proposition}\label{firstProp}
Let $\theta(\mathbf{x})$, $\bm{\omega}_{\mathbf{u}_0} (\mathbf{x})$ and $\bm{\omega}_{\mathbf{v}_0}(\mathbf{x})$  be smooth fields on $\Omega$ solving Eqs.\;(\ref{eq:effPDE1}-\ref{eq:effPDE2}). There exists a unique and smooth rotation field $\mathbf{R}_{\emph{eff}} \colon \Omega \rightarrow SO(3)$  with $\mathbf{R}_{\emph{eff}}(\langle \mathbf{x} \rangle) = \mathbf{I}$ solving the Pfaff system
\begin{equation}
\begin{aligned}\label{eq:firstSurface1}
\partial_{\mathbf{u}_0} \mathbf{R}_{\emph{eff}}(\mathbf{x}) = \mathbf{R}_{\emph{eff}}( \mathbf{x}) \big( \bm{\omega}_{\mathbf{u}_0}(\mathbf{x}) \times \big) , \quad \partial_{\mathbf{v}_0} \mathbf{R}_{\emph{eff}}(\mathbf{x}) = \mathbf{R}_{\emph{eff}}( \mathbf{x}) \big( \bm{\omega}_{\mathbf{v}_0}(\mathbf{x}) \times \big).
\end{aligned}
\end{equation}
In addition, there exists a unique and smooth deformation $\mathbf{y}_{\emph{eff}} \colon \Omega \rightarrow \mathbb{R}^3$  such that $\mathbf{y}_{\emph{eff}} ( \langle \mathbf{x} \rangle) = \mathbf{0}$ and 
\begin{equation}
\begin{aligned}\label{eq:secondSurface1}
\nabla \mathbf{y}_{\emph{eff}}(\mathbf{x}) = \mathbf{R}_{\emph{eff}}(\mathbf{x}) \mathbf{A}_{\emph{eff}}(\theta(\mathbf{x})).
\end{aligned}
\end{equation}
\end{proposition}
\noindent  The proof is based on well-known mathematical facts, which we collect in Lemmas \ref{firstLemmaEff}-\ref{secondLemmaEff} below. A convenient modern reference is \cite{ciarlet2008new}, Theorems 2-6.
\begin{lemma}\label{firstLemmaEff}
Let $U\subset\mathbb{R}^2$ be a simply connected domain with a smooth boundary, let $\mathbf{A}_{\alpha} \colon U \rightarrow \mathbb{R}^{3\times3}$, $\alpha = 1,2$ be smooth matrix fields that satisfy
\begin{equation}
\begin{aligned}\label{eq:TheAiPDE}
\partial_1 \mathbf{A}_2(\eta_1, \eta_2) - \partial_2 \mathbf{A}_1(\eta_1, \eta_2)  + \mathbf{A}_1(\eta_1, \eta_2) \mathbf{A}_2(\eta_1, \eta_2) - \mathbf{A}_2(\eta_1, \eta_2) \mathbf{A}_1(\eta_1, \eta_2) = \mathbf{0}
\end{aligned}
\end{equation}
on $U$ and let  $(\bar{\eta}_{1}, \bar{\eta}_{2}) \in U$ and $\bar{\mathbf{F}}\in \mathbb{R}^{3\times3}$ be given. There is one and only one  smooth matrix field $\mathbf{F} \colon U \rightarrow \mathbb{R}^{3\times3}$ solving the Pfaff system
\begin{equation}
\begin{aligned}
\partial_1 \mathbf{F}(\eta_1, \eta_2) = \mathbf{F}(\eta_1, \eta_2) \mathbf{A}_{1}(\eta_1, \eta_2), \quad \partial_2 \mathbf{F}(\eta_1, \eta_2) = \mathbf{F}(\eta_1, \eta_2) \mathbf{A}_{2}(\eta_1, \eta_2)
\end{aligned}
\end{equation}
along with the condition $\mathbf{F}(\bar{\eta}_1, \bar{\eta}_2) = \bar{\mathbf{F}}$.
\end{lemma}
\begin{lemma}\label{secondLemmaEff}
Let $U\subset\mathbb{R}^2$ be a simply connected domain with a  smooth boundary, and let $\mathbf{h}_{\alpha} \colon U \rightarrow \mathbb{R}^3$, $\alpha =1,2$, be smooth vector fields that satisfy 
\begin{equation}
\begin{aligned}
\partial_1 \mathbf{h}_2(\eta_1, \eta_2) = \partial_2 \mathbf{h}_1(\eta_1, \eta_2)
\end{aligned}
\end{equation} 
on $U$. There exists a smooth vector field $\boldsymbol{\varphi} \colon U \rightarrow \mathbb{R}^3$, unique up to an additive constant, solving
\begin{equation}
\partial_1 \boldsymbol{\varphi}(\eta_1, \eta_2) = \mathbf{h}_1(\eta_1, \eta_2), \quad \partial_2 \boldsymbol{\varphi}(\eta_1, \eta_2) = \mathbf{h}_2(\eta_1,\eta_2).
\end{equation}
\end{lemma}
\begin{proof}[Proof of Proposition \ref{firstProp}.] Assume that Eqs.\;(\ref{eq:effPDE1}-\ref{eq:effPDE2}) hold for smooth fields $\theta(\mathbf{x})$, $\boldsymbol{\omega}_{\mathbf{u}_0}(\mathbf{x})$ and $\boldsymbol{\omega}_{\mathbf{v}_0}(\mathbf{x})$ on $\Omega$, a simply connected domain with a  smooth boundary. 

Recall that the  Bravais lattice vectors $\tilde{\mathbf{u}}_0$ and $\tilde{\mathbf{v}}_0$ span $\mathbb{R}^2$. Thus, there exists reciprocal vectors $\tilde{\mathbf{u}}_0^r$ and $\tilde{\mathbf{v}}_0^r$ such that $\tilde{\mathbf{u}}_0^r \cdot \tilde{\mathbf{u}}_0 = 1$, $\tilde{\mathbf{u}}_0^r \cdot \tilde{\mathbf{v}}_0 = 0$, $\tilde{\mathbf{v}}_0^r \cdot \tilde{\mathbf{v}}_0 = 1$ and $\tilde{\mathbf{v}}_0^r \cdot \tilde{\mathbf{u}}_0 = 0$. Hence $\mathbf{x} \in \Omega$ satisfies $\mathbf{x}  = (\mathbf{x} \cdot \tilde{\mathbf{u}}_0^r) \tilde{\mathbf{u}}_0 +  (\mathbf{x} \cdot \tilde{\mathbf{v}}_0^r) \tilde{\mathbf{v}}_0$.  Let $\eta_1:= \mathbf{x} \cdot \tilde{\mathbf{u}}_0^r$ and $\eta_2 :=\mathbf{x} \cdot \tilde{\mathbf{v}}_0^r$, so that $\mathbf{x}$ is parameterized as $\mathbf{x}(\eta_1, \eta_2) = \eta_1 \tilde{\mathbf{u}}_0 + \eta_2 \tilde{\mathbf{v}}_0$.  Finally, let $U := \{ (\eta_1, \eta_2) \colon \mathbf{x}(\eta_1, \eta_2) \in \Omega\}$. By construction, $U$ is a simply connected domain with a  smooth boundary. Also, there is a unique $(\bar{\eta}_1, \bar{\eta}_2) \in U$ such that $\mathbf{x}(\bar{\eta}_1, \bar{\eta}_2) = \langle \mathbf{x} \rangle$. We write $\boldsymbol{\eta} = (\eta_1, \eta_2)$ for short below. 

Building on these definitions, let $\hat{\theta} (\boldsymbol{\eta}) := \theta(\mathbf{x}(\boldsymbol{\eta}))$, $\boldsymbol{\omega}_1(\boldsymbol{\eta}) := \boldsymbol{\omega}_{\mathbf{u}_0}(\mathbf{x}(\boldsymbol{\eta}))$ and  $\boldsymbol{\omega}_2(\boldsymbol{\eta}):= \boldsymbol{\omega}_{\mathbf{v}_0}(\mathbf{x}(\boldsymbol{\eta}))$.  We claim that Eq.\;(\ref{eq:effPDE2}) implies  Eq.\;(\ref{eq:TheAiPDE}) for the tensors $\mathbf{A}_1(\boldsymbol{\eta}) : =  [\boldsymbol{\omega}_1(\boldsymbol{\eta}) \times ]$ and $\mathbf{A}_2(\boldsymbol{\eta}) : =  [\boldsymbol{\omega}_2(\boldsymbol{\eta}) \times ]$. This claim follows from a few basic observations. First, the definitions of $\eta_1, \eta_2$  give  that 
\begin{equation}
\begin{aligned}\label{eq:partialDerivativesRelate}
\partial_1 \boldsymbol{\omega}_2(\boldsymbol{\eta})  &= \lim_{\epsilon \rightarrow 0} \epsilon^{-1} \big( \boldsymbol{\omega}_2(\eta_1 + \epsilon, \eta_2) - \boldsymbol{\omega}_2(\eta_1, \eta_2)    \big) \\
&=  \lim_{\epsilon \rightarrow 0} \epsilon^{-1} \big(\boldsymbol{\omega}_{\mathbf{v}_0}(\mathbf{x}(\eta_1, \eta_2) +  \epsilon \tilde{\mathbf{u}}_0 ) - \boldsymbol{\omega}_{\mathbf{v}_0}(\mathbf{x}(\eta_1, \eta_2))    \big) =  \partial_{\mathbf{u}_0} \boldsymbol{\omega}_{\mathbf{v}_0} (\mathbf{x}(\boldsymbol{\eta})), \\
\partial_2 \boldsymbol{\omega}_1(\boldsymbol{\eta})  &= \lim_{\epsilon \rightarrow 0} \epsilon^{-1} \big( \boldsymbol{\omega}_1(\eta_1, \eta_2 + \epsilon) - \boldsymbol{\omega}_1(\eta_1, \eta_2)    \big) \\ 
& = \lim_{\epsilon \rightarrow 0} \epsilon^{-1} \big(\boldsymbol{\omega}_{\mathbf{u}_0}(\mathbf{x}(\eta_1, \eta_2) +  \epsilon \tilde{\mathbf{v}}_0 ) - \boldsymbol{\omega}_{\mathbf{u}_0}(\mathbf{x}(\eta_1, \eta_2))    \big)  = \partial_{\mathbf{v}_0} \boldsymbol{\omega}_{\mathbf{u}_0} (\mathbf{x}(\boldsymbol{\eta})). 
\end{aligned}
\end{equation}
As a result, Eq.\;(\ref{eq:effPDE2}) implies that
\begin{equation}
\begin{aligned}\label{eq:firstObProp1}
 \partial_{2}\boldsymbol{\omega}_1(\boldsymbol{\eta})-\partial_1\boldsymbol{\omega}_2(\boldsymbol{\eta})=\boldsymbol{\omega}_1(\boldsymbol{\eta})\times \boldsymbol{\omega}_2(\boldsymbol{\eta})
\end{aligned}
\end{equation}
on $U$. Next, notice that 
\begin{equation}
\begin{aligned}\label{eq:secondObProp1}
\mathbf{a}_1 - \mathbf{a}_2   -  \mathbf{b}_1 \times \mathbf{b}_2 = \mathbf{0} \quad \iff \quad (\mathbf{a}_1 \times) - (\mathbf{a}_2 \times) - (\mathbf{b}_1 \times ) ( \mathbf{b}_2 \times) + (\mathbf{b}_2 \times) ( \mathbf{b}_1 \times) = \mathbf{0} 
\end{aligned}
\end{equation}
for any $\mathbf{a}_1, \mathbf{a}_2, \mathbf{b}_1, \mathbf{b}_2 \in \mathbb{R}^3$. Indeed, standard properties of the cross-product yield $\mathbf{v} \times ( \mathbf{a}_1 - \mathbf{a}_2   -  \mathbf{b}_1 \times \mathbf{b}_2) = \big[ (\mathbf{a}_1 \times) - (\mathbf{a}_2 \times) - (\mathbf{b}_1 \times ) ( \mathbf{b}_2 \times) + (\mathbf{b}_2 \times) ( \mathbf{b}_1 \times)\big] \mathbf{v}$. Thus, Eqs.\;(\ref{eq:firstObProp1}-\ref{eq:secondObProp1}) furnish the PDE
\begin{equation}
\begin{aligned}
&(\partial_{1}\boldsymbol{\omega}_2(\boldsymbol{\eta}) \times )  -  (\partial_2\boldsymbol{\omega}_1(\boldsymbol{\eta}) \times )  + (\boldsymbol{\omega}_1(\boldsymbol{\eta}) \times )  (\boldsymbol{\omega}_2(\boldsymbol{\eta}) \times )  - (\boldsymbol{\omega}_2(\boldsymbol{\eta}) \times )(\boldsymbol{\omega}_1(\boldsymbol{\eta}) \times )  = \mathbf{0}
\end{aligned}
\end{equation}
on $U$.  The claim follows since $(\partial_1 \boldsymbol{\omega}_2(\boldsymbol{\eta}) \times ) = \partial_1 ( \boldsymbol{\omega}_2(\boldsymbol{\eta}) \times )$ and $(\partial_2 \boldsymbol{\omega}_1(\boldsymbol{\eta}) \times ) = \partial_2 ( \boldsymbol{\omega}_1(\boldsymbol{\eta}) \times )$.

By Lemma \ref{firstLemmaEff}, there is a smooth and unique $\mathbf{R} \colon U \rightarrow \mathbb{R}^{3\times3}$ solving the Pfaff system
\begin{equation}
\begin{aligned}\label{eq:Pfaff1}
\partial_1 \mathbf{R}(\boldsymbol{\eta})  = \mathbf{R}(\boldsymbol{\eta})  \big(\boldsymbol{\omega}_1(\boldsymbol{\eta}) \times \big), \quad \partial_2 \mathbf{R}(\boldsymbol{\eta})  = \mathbf{R}(\boldsymbol{\eta})  \big(\boldsymbol{\omega}_2(\boldsymbol{\eta}) \times \big)
\end{aligned}
\end{equation}
with $\mathbf{R}(\bar{\boldsymbol{\eta}}) = \mathbf{I}$. To show that $\mathbf{R}(\boldsymbol{\eta})$ is a rotation for all $\boldsymbol{\eta} \in U$, observe that
\begin{equation}
    \begin{aligned}\label{eq:constDerivative}
        \partial_i [\det \mathbf{R}(\boldsymbol{\eta})] = \det \mathbf{R}(\boldsymbol{\eta})  \text{Tr}\big(\big[\partial_i \mathbf{R}(\boldsymbol{\eta}) \big] \mathbf{R}^{-1}(\boldsymbol{\eta}) \big)  = \det \mathbf{R}(\boldsymbol{\eta}) \text{Tr} \big([\boldsymbol{\omega}_i(\boldsymbol{\eta}) \times] \big) = 0, \quad i = 1,2
    \end{aligned}
\end{equation}
where $\mathbf{R}(\boldsymbol{\eta})$ is invertible. Since $\det \mathbf{R}(\boldsymbol{\eta})$ is smooth and $\det \mathbf{R}(\bar{\boldsymbol{\eta}}) = 1$, a continuation argument using Eq.\;(\ref{eq:constDerivative}) gives that $\det \mathbf{R}(\boldsymbol{\eta}) = 1$ everywhere on $U$. As such, $\mathbf{R}^{-1}(\boldsymbol{\eta})$ is well-defined and smooth. It also satisfies $\partial_i \mathbf{R}^{-1}(\boldsymbol{\eta}) = (\boldsymbol{\omega}_i(\boldsymbol{\eta}) \times) \mathbf{R}^{-1}(\boldsymbol{\eta})$, $i = 1,2,$ by differentiation of the inverse. Consequently, 
\begin{equation}
    \begin{aligned}
        \partial_i \big( \mathbf{R}^{-T}(\boldsymbol{\eta}) \mathbf{R}^{-1}(\boldsymbol{\eta}) \big) = \big[(\boldsymbol{\omega}_i(\boldsymbol{\eta}) \times ) \mathbf{R}^{-1}(\boldsymbol{\eta})\big]^{T} \mathbf{R}^{-1}(\boldsymbol{\eta}) + \mathbf{R}^{-T}(\boldsymbol{\eta}) (\boldsymbol{\omega}_i(\boldsymbol{\eta}) \times) \mathbf{R}^{-1}(\boldsymbol{\eta}) = \mathbf{0}, \quad i =1,2
    \end{aligned}
\end{equation}
and hence  $\mathbf{R}^{-T}(\boldsymbol{\eta}) \mathbf{R}^{-1}(\boldsymbol{\eta}) = \mathbf{I}$  since $\mathbf{R}(\bar{\boldsymbol{\eta}}) = \mathbf{I} = \mathbf{R}^{-1}(\bar{\boldsymbol{\eta}})$.  Finally, $\mathbf{R}(\boldsymbol{\eta}) \in SO(3)$ because $\det \mathbf{R}(\boldsymbol{\eta}) = 1$.

Having shown that $\mathbf{R}(\boldsymbol{\eta})$ is a rotation field, we define $\mathbf{R}_{\text{eff}} \colon \Omega \rightarrow SO(3)$  as $\mathbf{R}_{\text{eff}}(\mathbf{x}(\boldsymbol{\eta})) =  \mathbf{R}(\boldsymbol{\eta})$ for  $\boldsymbol{\eta} \in U$. Similar to Eq.\;(\ref{eq:partialDerivativesRelate}), the derivatives of these fields are related via 
\begin{equation}
\begin{aligned}
\partial_1 \mathbf{R}(\boldsymbol{\eta}) &= \mathbf{R}(\boldsymbol{\eta}) \big[ \boldsymbol{\omega}_1(\boldsymbol{\eta}) \times \big] = \mathbf{R}_{\text{eff}}(\mathbf{x}(\boldsymbol{\eta}))\big[ \boldsymbol{\omega}_{\mathbf{u}_0} (\mathbf{x}(\boldsymbol{\eta})) \times \big]= \partial_{\mathbf{u}_0} \mathbf{R}_{\text{eff}}( \mathbf{x}(\boldsymbol{\eta})),  \\ 
\partial_2 \mathbf{R}(\boldsymbol{\eta}) &= \mathbf{R}(\boldsymbol{\eta}) \big[ \boldsymbol{\omega}_2(\boldsymbol{\eta}) \times \big] = \mathbf{R}_{\text{eff}}(\mathbf{x}(\boldsymbol{\eta}))\big[ \boldsymbol{\omega}_{\mathbf{v}_0} (\mathbf{x}(\boldsymbol{\eta})) \times \big] = \partial_{\mathbf{v}_0} \mathbf{R}_{\text{eff}}( \mathbf{x}(\boldsymbol{\eta})).
\end{aligned}
\end{equation}
Thus,  $\mathbf{R}_{\text{eff}}(\mathbf{x})$ is the unique and smooth rotation field satisfying $\mathbf{R}_{\text{eff}}(\langle \mathbf{x} \rangle) = \mathbf{I}$ and the second set of identities in Eq.\;(\ref{eq:effPDE3}). (Concerning uniqueness, if there were a second such rotation field $\mathbf{Q}_{\text{eff}}(\mathbf{x}) \neq \mathbf{R}_{\text{eff}}(\mathbf{x})$ satisfying these conditions, then there is also a rotation field $\mathbf{Q}(\boldsymbol{\eta}) := \mathbf{Q}_{\text{eff}}(\mathbf{x}(\boldsymbol{\eta}))$  that is $\neq \mathbf{R}(\boldsymbol{\eta})$ and yet solves Eq.\;(\ref{eq:Pfaff1}) with $\mathbf{Q}(\boldsymbol{\eta}) = \mathbf{I}$. This would contradict the uniqueness of solutions to Pfaff systems.)  

Next, recall that $\hat{\theta} (\boldsymbol{\eta}) = \theta(\mathbf{x}(\boldsymbol{\eta}))$. So,
\begin{equation}
\begin{aligned}
&\partial_1 \big[ \mathbf{R}(\boldsymbol{\eta}) \mathbf{v}(\hat{\theta} (\boldsymbol{\eta})) \big] =    \mathbf{R}_{\text{eff}}(\mathbf{x}(\boldsymbol{\eta})) \Big[ \boldsymbol{\omega}_{\mathbf{u}_0}(\mathbf{x}(\boldsymbol{\eta})) \times \mathbf{v}(\theta(\mathbf{x}(\boldsymbol{\eta}))) + \partial_{\mathbf{u}_0} \theta( \mathbf{x}(\boldsymbol{\eta})) \mathbf{v}'(\theta(\mathbf{x}(\boldsymbol{\eta}))) \Big], \\
&\partial_2 \big[ \mathbf{R}(\boldsymbol{\eta}) \mathbf{u}(\hat{\theta} (\boldsymbol{\eta})) \big] =    \mathbf{R}_{\text{eff}}(\mathbf{x}(\boldsymbol{\eta})) \Big[ \boldsymbol{\omega}_{\mathbf{v}_0}(\mathbf{x}(\boldsymbol{\eta})) \times \mathbf{u}(\theta(\mathbf{x}(\boldsymbol{\eta}))) + \partial_{\mathbf{v}_0} \theta( \mathbf{x}(\boldsymbol{\eta})) \mathbf{u}'(\theta(\mathbf{x}(\boldsymbol{\eta}))) \Big].
\end{aligned}
\end{equation}
Since Eq.\;(\ref{eq:effPDE1}) holds, $\partial_1 \big[ \mathbf{R}(\boldsymbol{\eta}) \mathbf{v}(\hat{\theta} (\boldsymbol{\eta})) \big] = \partial_2 \big[ \mathbf{R}(\boldsymbol{\eta}) \mathbf{u}(\hat{\theta} (\boldsymbol{\eta})) \big]$. It follows from Lemma \ref{secondLemmaEff} that there is a unique and smooth  deformation $\boldsymbol{\varphi} \colon U \rightarrow \mathbb{R}^3$ such that  $\boldsymbol{\varphi}(\bar{\boldsymbol{\eta}}) = \mathbf{0}$ and 
\begin{equation}
\begin{aligned}
\partial_1 \boldsymbol{\varphi}(\boldsymbol{\eta}) =  \mathbf{R}(\boldsymbol{\eta}) \mathbf{u}(\hat{\theta} (\boldsymbol{\eta})), \quad \partial_2 \boldsymbol{\varphi}(\boldsymbol{\eta}) =  \mathbf{R}(\boldsymbol{\eta}) \mathbf{v}(\hat{\theta} (\boldsymbol{\eta})).
\end{aligned}
\end{equation}
Similar to the argument for the rotation field, we conclude that $\mathbf{y}_{\text{eff}} \colon \Omega \rightarrow \mathbb{R}^3$ defined by $\mathbf{y}_{\text{eff}}(\mathbf{x}(\boldsymbol{\eta})) =  \boldsymbol{\varphi}(\boldsymbol{\eta})$, $\boldsymbol{\eta} \in U$ is the unique and smooth deformation with $\mathbf{y}_{\text{eff}}(\langle \mathbf{x} \rangle) = \mathbf{0}$ and such that the first set of identities in Eq.\;(\ref{eq:effPDE3}) holds. This completes the proof of the proposition. 
\end{proof}

\subsection{Construction of soft origami modes}\label{ssec:GlobalOrigami}

We now construct general soft origami modes. Our plan is illustrated in Fig.\;\ref{Fig:EggboxGap}. We sample the PDE in Eqs.\;(\ref{eq:effPDE1}-\ref{eq:effPDE2}) in a cell-wise manner to deform the origami cells. We impose the local cell-wise bending kinematics from Section \ref{ssec:LocalBend} on the construction, but suffer incompatibilities across the cell boundaries which must be resolved. Of course, such incompatibility is inevitable when sampling a smooth PDE to construct a discrete structure. Here, the intercell gaps turn out to be coupled at leading order to an underdetermined linear PDE system with three equations and four unknowns. We recognize how to solve this PDE, which we do in Appendix \ref{sec:ExistencePDE}. This leads to the desired soft modes.

\begin{figure}[t!]
\centering
\includegraphics[width=1\textwidth]{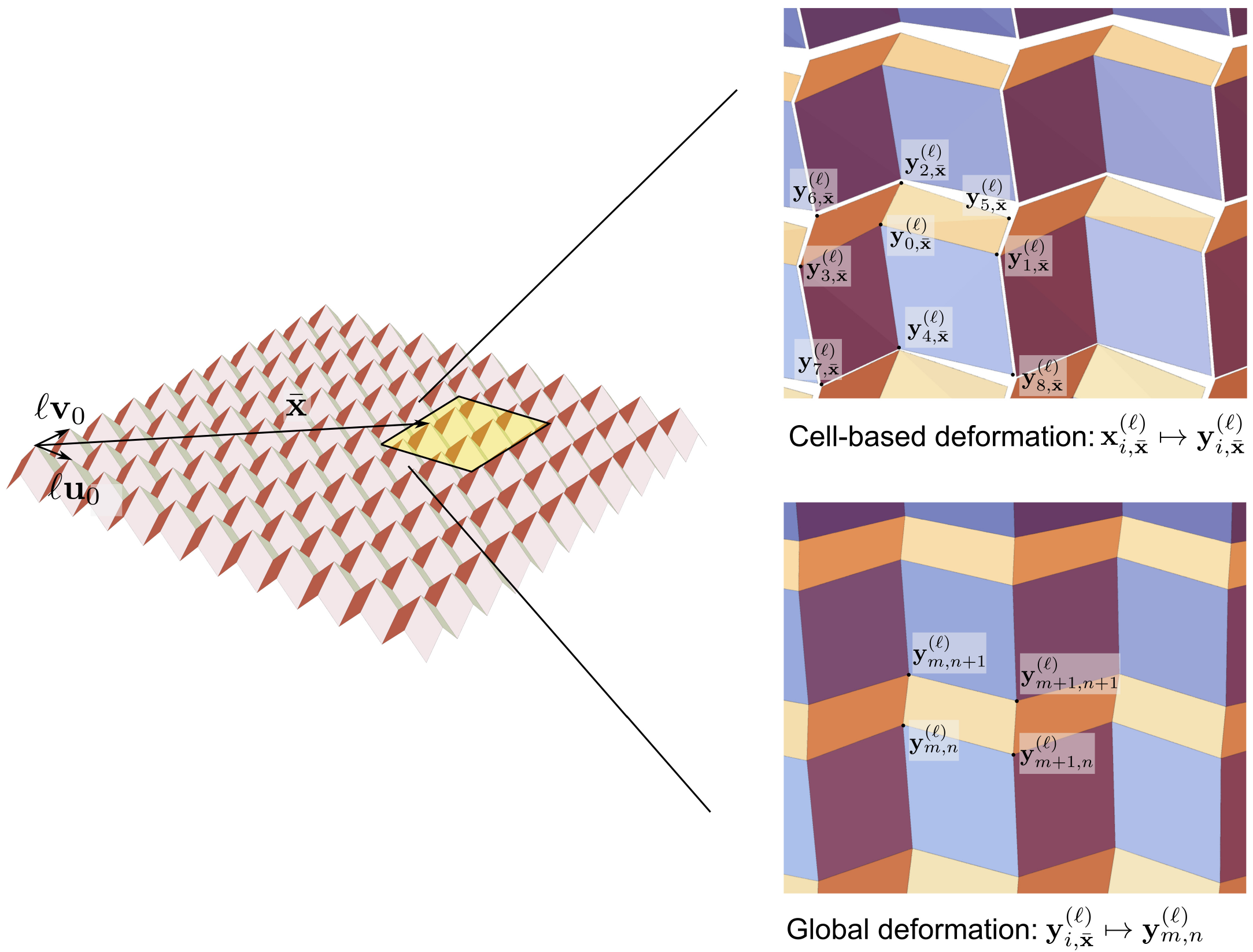} 
\caption{Illustration of the cell-based construction of origami soft modes. Each cell is made up of nine points. These points are deformed according to the local bending kinematics in Section \ref{ssec:LocalBend}, with parameters that vary from cell to cell according to a solution of the effective surface theory in Section \ref{ssec:EffSurfacesSection}. This construction leads to small gaps between adjacent cells,  which are then closed by an averaging procedure as shown.}
\label{Fig:EggboxGap}
\end{figure}

\paragraph{Preliminaries.} Let us quickly recall the setup of the coarse-graining problem for the reader's convenience. We consider any parallelogram origami design built from crease vectors $\mathbf{t}_1^r, \ldots, \mathbf{t}_4^r$ per Section \ref{ssec:DesignKin}. The design has a unique mechanism motion given by deforming the crease vectors as $\mathbf{t}_i^r \mapsto \mathbf{t}_i(\theta)$ for $\theta \in (\theta^{-}, \theta^+)$. The Bravais lattice vectors of the reference configuration are $\mathbf{u}_0 = \mathbf{t}_1^r - \mathbf{t}_3^r$ and $\mathbf{v}_0 = \mathbf{t}_2^r - \mathbf{t}_4^r$. These span the plane with normal $\mathbf{e}_3$ and satisfy $\mathbf{e}_3 \cdot (\mathbf{u}_0 \times \mathbf{v}_0) > 0$. Analogously, their deformed counterparts are $\mathbf{u}(\theta) = \mathbf{t}_1(\theta) - \mathbf{t}_3(\theta)$ and $\mathbf{v}(\theta) = \mathbf{t}_2(\theta) - \mathbf{t}_4(\theta)$. They too are taken to span the plane with normal $\mathbf{e}_3$ and to have $\mathbf{e}_3 \cdot (\mathbf{u}(\theta) \times \mathbf{v}(\theta)) > 0$.  Finally, $\tilde{\mathbf{u}}_0$ and $\tilde{\mathbf{v}}_0$ denote the orthogonal projections of $\mathbf{u}_0$ and $\mathbf{v}_0$ onto $\mathbb{R}^2$.  

\paragraph{The PDE description.} Let  $\Omega$  be the planar reference domain from Section \ref{ssec:BarHinge}. Let $\theta(\mathbf{x})$, $\boldsymbol{\omega}_{\mathbf{u}_0}(\mathbf{x})$ and $\boldsymbol{\omega}_{\mathbf{v}_0}(\mathbf{x})$ be smoothly defined on a neighborhood of  $\overline{\Omega}$. In the case that the design’s Poisson’s ratio $\nu(\theta)$ in Eq.\;(\ref{eq:PoissonsRatioDesign}) is negative, assume further that these fields are analytic. Finally, assume that the fields  solve the PDE in Eqs.\;(\ref{eq:effPDE1}-\ref{eq:effPDE2}) on $\Omega$. 
Extract from this PDE solution the unique and smooth rotation field $\mathbf{R}_{\text{eff}} \colon \Omega \rightarrow SO(3)$ and  effective deformation $\mathbf{y}_{\text{eff}} \colon \Omega \rightarrow \mathbb{R}^3$ satisfying (\ref{eq:effPDE3}) along with the constraints $\mathbf{R}_{\text{eff}} (\langle \mathbf{x} \rangle) = \mathbf{I}$ and $\mathbf{y}_{\text{eff}}(\langle \mathbf{x} \rangle ) = \mathbf{0}$. 
Next,  define a smooth analog of the bending parameters in Eq.\;(\ref{eq:theKappas1}) through a vector field $\boldsymbol{\kappa} \colon \Omega \rightarrow \mathbb{R}^4$, $\boldsymbol{\kappa}(\mathbf{x}) := ( \kappa_1(\mathbf{x}) ,\ldots, \kappa_4(\mathbf{x}))$  that replaces $\theta, \delta \theta_{\mathbf{u}_0}$, $\delta \theta_{\mathbf{v}_0}$, $\boldsymbol{\omega}_{\mathbf{u}_0}$ and $\boldsymbol{\omega}_{\mathbf{v}_0}$ in that description with $\theta(\mathbf{x})$, $\partial_{\mathbf{u}_0} \theta (\mathbf{x})$, $\partial_{\mathbf{v}_0} \theta(\mathbf{x})$, $\boldsymbol{\omega}_{\mathbf{v}_0}(\mathbf{x})$  and $\boldsymbol{\omega}_{\mathbf{v}_0}(\mathbf{x})$. That is,  
\begin{equation}
\begin{aligned}\label{eq:theKappaFields}
\begin{pmatrix}
\kappa_1(\mathbf{x}) \\ \kappa_2(\mathbf{x}) \\ \kappa_3(\mathbf{x}) \\ \kappa_4(\mathbf{x}) 
\end{pmatrix} = \mathbf{B}(\theta(\mathbf{x}))  \begin{pmatrix}  \boldsymbol{\omega}_{\mathbf{u}_0}(\mathbf{x}) \\ \partial_{\mathbf{u}_0} \theta(\mathbf{x}) \\
\boldsymbol{\omega}_{\mathbf{v}_0}(\mathbf{x}) \\ \partial_{\mathbf{v}_0} \theta(\mathbf{x})
 \end{pmatrix}
\end{aligned}
\end{equation}
for $\mathbf{B}(\theta)$  in Eq.\;(\ref{eq:theKappas1}). Our origami construction is based on sampling the fields $\theta(\mathbf{x})$, $\boldsymbol{\omega}_{\mathbf{u}_0}(\mathbf{x})$, $\boldsymbol{\omega}_{\mathbf{v}_0}(\mathbf{x})$,  $\mathbf{R}_{\text{eff}}(\mathbf{x})$, $\mathbf{y}_{\text{eff}}(\mathbf{x})$ and $\boldsymbol{\kappa}(\mathbf{x})$.

\paragraph{A cell-based construction.} Let $\Omega_{\text{ori}}^{(\ell)}$ be the origami reference domain from Section \ref{ssec:BarHinge}. Its panels have characteristic length $\sim \ell \ll 1$.  Let $\mathcal{I}_{\text{cell}}^{(\ell)} := \{ \ell( i \tilde{\mathbf{u}}_0 + j \tilde{\mathbf{v}}_0 ) \subset \Omega \colon i, j \in \mathbb{Z}   \}$ consist of the central vertices of the unit cells in $\Omega_{\text{ori}}^{(\ell)}$.  For each $\bar{\mathbf{x}} \in \mathcal{I}_{\text{cell}}^{(\ell)}$, observe that a single undeformed cell is fully specified by nine vertices: the interior vertex $\mathbf{x}_{0,\bar{\mathbf{x}}}^{(\ell)} :=  (\bar{\mathbf{x}} , 0 )$  and the boundary vertices 
\begin{equation}
\begin{aligned}\label{eq:cellbyCell}
& \mathbf{x}_{1, \bar{\mathbf{x}}}^{(\ell)} := \mathbf{x}_{0,\bar{\mathbf{x}}}^{(\ell)}  + \ell \mathbf{t}_1^{r}, \quad  \mathbf{x}_{2, \bar{\mathbf{x}}}^{(\ell)} := \mathbf{x}_{0,\bar{\mathbf{x}}}^{(\ell)}  + \ell \mathbf{t}_2^{r} , \quad \mathbf{x}_{3, \bar{\mathbf{x}}}^{(\ell)} := \mathbf{x}_{0,\bar{\mathbf{x}}}^{(\ell)}  + \ell \mathbf{t}_3^{r}, \quad \mathbf{x}_{4, \bar{\mathbf{x}}}^{(\ell)} := \mathbf{x}_{0,\bar{\mathbf{x}}}^{(\ell)}  + \ell \mathbf{t}_4^{r},\\
&\mathbf{x}_{5,\bar{\mathbf{x}}}^{(\ell)} := \mathbf{x}_{1,\bar{\mathbf{x}}}^{(\ell)} + \ell \mathbf{t}_2^r, \quad \mathbf{x}_{6,\bar{\mathbf{x}}}^{(\ell)} := \mathbf{x}_{2,\bar{\mathbf{x}}}^{(\ell)} + \ell \mathbf{t}_3^r, \quad \mathbf{x}_{7,\bar{\mathbf{x}}}^{(\ell)} := \mathbf{x}_{3,\bar{\mathbf{x}}}^{(\ell)} + \ell \mathbf{t}_4^r, \quad \mathbf{x}_{8,\bar{\mathbf{x}}}^{(\ell)} := \mathbf{x}_{4,\bar{\mathbf{x}}}^{(\ell)} + \ell \mathbf{t}_1^r.
\end{aligned}
\end{equation}
To deform each such cell we must deform the points $\mathbf{x}_{i,\bar{\mathbf{x}}}^{(\ell)} \mapsto \mathbf{y}_{i,\bar{\mathbf{x}}}^{(\ell)}$, $i = 0, 1,\ldots,8$. We proceed by constructing a preferred deformation for each cell, and then by closing the intercell gaps. The labeling here is as in Fig.\;\ref{Fig:EggboxGap}.

For each $\bar{\mathbf{x}} \in \mathcal{I}_{\text{cell}}^{(\ell)}$, we deform the interior vertex of the origami cell $\mathbf{x}_{0,\bar{\mathbf{x}}}^{(\ell)}$ to 
\begin{equation}
\begin{aligned}\label{eq:cellwise1}
\mathbf{y}_{0,\bar{\mathbf{x}}}^{(\ell)} := \mathbf{y}_{\text{eff}}(\bar{\mathbf{x}}) +\ell \mathbf{d}(\bar{\mathbf{x}})
\end{aligned}
\end{equation} 
for some smooth vector field $\mathbf{d}(\mathbf{x})$ that will be chosen later. We then deform the adjacent vertices $\mathbf{x}_{i,\bar{\mathbf{x}}}^{(\ell)}$, $i =1,\ldots,4$, to
\begin{equation}
\begin{aligned}\label{eq:cellwise2}
\mathbf{y}_{i,\bar{\mathbf{x}}}^{(\ell)} :=  \mathbf{y}_{0,\bar{\mathbf{x}}}^{(\ell)}  + \ell \big[  \mathbf{R}_{\text{eff}}(\bar{\mathbf{x}}) \big( \mathbf{I} + \ell (\boldsymbol{\omega}(\bar{\mathbf{x}}) \times ) \big)\big]  \mathbf{t}_i(\theta(\bar{\mathbf{x}}) + \ell \xi(\bar{\mathbf{x}}) ), \quad i = 1,\ldots, 4
\end{aligned}
\end{equation}
for some smooth vector and scalar fields $\boldsymbol{\omega}(\mathbf{x})$ and $\xi(\mathbf{x})$ that will also be chosen later. Finally, we deform the corner vertices $\mathbf{x}_{i,\bar{\mathbf{x}}}^{(\ell)}$, $i =5, \ldots,8$, to 
\begin{equation}
\begin{aligned}\label{eq:cellwise3}
\mathbf{y}_{5,\bar{\mathbf{x}}}^{(\ell)} :=  \mathbf{y}_{1,\bar{\mathbf{x}}}^{(\ell)}  + \ell \big[  \mathbf{R}_{\text{eff}}(\bar{\mathbf{x}}) \big( \mathbf{I} + \ell (\boldsymbol{\omega}(\bar{\mathbf{x}}) \times ) \big)\big]  \mathbf{t}^{(\ell)}_{2,+}(\theta(\bar{\mathbf{x}}) + \ell \xi(\bar{\mathbf{x}}), \boldsymbol{\kappa}(\bar{\mathbf{x}})  ),  \\
\mathbf{y}_{6,\bar{\mathbf{x}}}^{(\ell)} :=  \mathbf{y}_{2,\bar{\mathbf{x}}}^{(\ell)}  + \ell \big[  \mathbf{R}_{\text{eff}}(\bar{\mathbf{x}}) \big( \mathbf{I} + \ell (\boldsymbol{\omega}(\bar{\mathbf{x}}) \times ) \big)\big]  \mathbf{t}^{(\ell)}_{3,+}(\theta(\bar{\mathbf{x}}) + \ell \xi(\bar{\mathbf{x}}), \boldsymbol{\kappa}(\bar{\mathbf{x}})  ), \\
\mathbf{y}_{7,\bar{\mathbf{x}}}^{(\ell)} :=  \mathbf{y}_{3,\bar{\mathbf{x}}}^{(\ell)}  + \ell \big[  \mathbf{R}_{\text{eff}}(\bar{\mathbf{x}}) \big( \mathbf{I} + \ell (\boldsymbol{\omega}(\bar{\mathbf{x}}) \times ) \big)\big]  \mathbf{t}^{(\ell)}_{4,-}(\theta(\bar{\mathbf{x}}) + \ell \xi(\bar{\mathbf{x}}), \boldsymbol{\kappa}(\bar{\mathbf{x}})  ),  \\
\mathbf{y}_{8,\bar{\mathbf{x}}}^{(\ell)} :=  \mathbf{y}_{4,\bar{\mathbf{x}}}^{(\ell)}  + \ell \big[  \mathbf{R}_{\text{eff}}(\bar{\mathbf{x}}) \big( \mathbf{I} + \ell (\boldsymbol{\omega}(\bar{\mathbf{x}}) \times ) \big)\big]  \mathbf{t}^{(\ell)}_{1,-}(\theta(\bar{\mathbf{x}}) + \ell \xi(\bar{\mathbf{x}}), \boldsymbol{\kappa}(\bar{\mathbf{x}})  ),
\end{aligned}
\end{equation}
using the definitions for the slightly bent boundary tangents in Eq.\;(\ref{eq:bentBoundaries1}) (see also Fig.\;\ref{Fig:FitCells}(c)). 
Note in the analysis of gaps to come, we assume that the sampled fields are (or can be) smoothly defined on a neighborhood of $\overline{\Omega}$. This assumption will be verified for $\mathbf{d}(\mathbf{x})$, $\boldsymbol{\omega}(\mathbf{x})$ and $\xi(\mathbf{x})$ once we choose these fields; it holds for $\mathbf{y}_{\text{eff}}(\mathbf{x})$, $\mathbf{R}_{\text{eff}}(\mathbf{x})$ and $\boldsymbol{\kappa}(\mathbf{x})$ given their definitions using $\theta(\mathbf{x})$, $\boldsymbol{\omega}_{\mathbf{u}_0}(\mathbf{x})$ and $\boldsymbol{\omega}_{\mathbf{v}_0}(\mathbf{x})$.

The additional fields $\mathbf{d}(\mathbf{x})$, $\boldsymbol{\omega}(\mathbf{x})$ and $\xi(\mathbf{x})$ enrich the ansatz without changing the lengths of the slightly bent cells at leading order. In particular, when sampled at the cell-level, $\mathbf{d}(\mathbf{x})$ and $\boldsymbol{\omega}(\mathbf{x})$ correspond to an infinitesimal  rigid body motion, and $\xi(\mathbf{x})$ is a perturbation of the cell's mechanism. This degeneracy is crucial to producing an overall origami construction with panel strain $\sim \ell^2$, as opposed to $\sim \ell$.

\paragraph{Intercell incompatibilities.} All quantities  in  the cell-wise ansatz in Eqs.\;(\ref{eq:cellwise1}-\ref{eq:cellwise3}) are specified by the effective surface theory in Eq.\;(\ref{eq:SurfaceTheory}), except for the fields $\mathbf{d}(\mathbf{x})$, $\boldsymbol{\omega}(\mathbf{x})$ and $\xi(\mathbf{x})$. We now derive an auxiliary PDE involving these fields, governing the gaps between neighboring cells  at $O(\ell^2)$. The gaps are defined by 
\begin{equation}
\begin{aligned}\label{eq:theGapsDef}
&\mathbf{g}_{31, \bar{\mathbf{x}}}^{(\ell)}  := \mathbf{y}_{3, \bar{\mathbf{x}}+ \ell \tilde{\mathbf{u}}_0}^{(\ell)} -  \mathbf{y}_{1, \bar{\mathbf{x}}}^{(\ell)}, \quad \mathbf{g}_{65, \bar{\mathbf{x}}}^{(\ell)}  := \mathbf{y}_{6, \bar{\mathbf{x}}+ \ell \tilde{\mathbf{u}}_0}^{(\ell)} -  \mathbf{y}_{5, \bar{\mathbf{x}}}^{(\ell)}, \quad \mathbf{g}_{78,\bar{\mathbf{x}}}^{(\ell)} := \mathbf{y}_{7, \bar{\mathbf{x}}+ \ell \tilde{\mathbf{u}}_0}^{(\ell)} -  \mathbf{y}_{8, \bar{\mathbf{x}}}^{(\ell)},  \\
&\mathbf{g}_{42, \bar{\mathbf{x}}}^{(\ell)} := \mathbf{y}_{4, \bar{\mathbf{x}}+ \ell \tilde{\mathbf{v}}_0}^{(\ell)} -  \mathbf{y}_{2, \bar{\mathbf{x}}}^{(\ell)}, \quad \mathbf{g}_{85, \bar{\mathbf{x}}}^{(\ell)} := \mathbf{y}_{8, \bar{\mathbf{x}}+ \ell \tilde{\mathbf{v}}_0}^{(\ell)} -  \mathbf{y}_{5, \bar{\mathbf{x}}}^{(\ell)}, \quad \mathbf{g}_{76, \bar{\mathbf{x}}}^{(\ell)} := \mathbf{y}_{7, \bar{\mathbf{x}}+ \ell \tilde{\mathbf{v}}_0}^{(\ell)} -  \mathbf{y}_{6, \bar{\mathbf{x}}}^{(\ell)}
\end{aligned}
\end{equation}
for $\bar{\mathbf{x}} \in \mathcal{I}^{(\ell)}_{\text{cell}}$. First, we relate all the gaps to two distinguished ones $\mathbf{g}_{31,\bar{\mathbf{x}}}^{(\ell)}$ and $\mathbf{g}_{42,\bar{\mathbf{x}}}^{(\ell)}$.
\begin{lemma}\label{LemmaCell1}
The ansatz satisfies 
\begin{equation}
\begin{aligned}
\mathbf{g}_{65,\bar{\mathbf{x}}}^{(\ell)} - \mathbf{g}_{31,\bar{\mathbf{x}}}^{(\ell)} = O(\ell^3), \quad \mathbf{g}_{78,\bar{\mathbf{x}}}^{(\ell)} - \mathbf{g}_{31,\bar{\mathbf{x}}}^{(\ell)} = O(\ell^3), \quad \mathbf{g}_{85,\bar{\mathbf{x}}}^{(\ell)} - \mathbf{g}_{42,\bar{\mathbf{x}}}^{(\ell)} = O(\ell^3), \quad \mathbf{g}_{76,\bar{\mathbf{x}}}^{(\ell)} - \mathbf{g}_{42,\bar{\mathbf{x}}}^{(\ell)} = O(\ell^3).
\end{aligned}
\end{equation}
\end{lemma}
\begin{proof}
The proof builds on the analysis in Proposition \ref{LocalBendProp}. Essentially,   after  a Taylor expansion of the gaps in $\ell$, we are able to match the conditions for approximate compatibility in that proposition.  

Consider any $\bar{\mathbf{x}} \in \mathcal{I}_{\text{cell}}^{(\ell)}.$ We show explicitly that $\mathbf{g}_{65,\bar{\mathbf{x}}}^{(\ell)} - \mathbf{g}_{31,\bar{\mathbf{x}}}^{(\ell)} = O(\ell^3)$. The other three statements follow by a  similar argument. Observe first that $\mathbf{y}_{6,\bar{\mathbf{x}}}^{(\ell)}$  in Eq.\;(\ref{eq:cellwise3}) can also be written as 
\begin{equation}
\begin{aligned}
\mathbf{y}_{6,\bar{\mathbf{x}}}^{(\ell)}  =  \mathbf{y}_{3,\bar{\mathbf{x}}}^{(\ell)} +  \ell \big[  \mathbf{R}_{\text{eff}}(\bar{\mathbf{x}}) \big( \mathbf{I} + \ell (\boldsymbol{\omega}(\bar{\mathbf{x}}) \times ) \big)\big]  \mathbf{t}^{(\ell)}_{2,-}(\theta(\bar{\mathbf{x}}) + \ell \xi(\bar{\mathbf{x}}), \boldsymbol{\kappa}(\bar{\mathbf{x}})  )
\end{aligned}
\end{equation}
since $ \mathbf{t}_3(\theta) + \mathbf{t}_{2,-}(\theta, \boldsymbol{\kappa}) = \mathbf{t}_2(\theta) + \mathbf{t}_{3,+}(\theta, \boldsymbol{\kappa})$ by Eq.\;(\ref{eq:bentBoundaries1}). It follows that 
\begin{equation}
\begin{aligned}
\mathbf{g}_{65,\bar{\mathbf{x}}}^{(\ell)} - \mathbf{g}_{31,\bar{\mathbf{x}}}^{(\ell)} &=  \ell \big[  \mathbf{R}_{\text{eff}}(\bar{\mathbf{x}} + \ell \tilde{\mathbf{u}}_0 ) \big( \mathbf{I} + \ell (\boldsymbol{\omega}(\bar{\mathbf{x}} + \ell \tilde{\mathbf{u}}_0) \times ) \big)\big]  \mathbf{t}^{(\ell)}_{2,-}(\theta(\bar{\mathbf{x}}+ \ell \tilde{\mathbf{u}}_0) + \ell \xi(\bar{\mathbf{x}}+ \ell \tilde{\mathbf{u}}_0), \boldsymbol{\kappa}(\bar{\mathbf{x}} + \ell \tilde{\mathbf{u}}_0)  )  \\
&\qquad  - \ell \big[  \mathbf{R}_{\text{eff}}(\bar{\mathbf{x}}) \big( \mathbf{I} + \ell (\boldsymbol{\omega}(\bar{\mathbf{x}}) \times ) \big)\big]  \mathbf{t}^{(\ell)}_{2,+}(\theta(\bar{\mathbf{x}}) + \ell \xi(\bar{\mathbf{x}}), \boldsymbol{\kappa}(\bar{\mathbf{x}})  ) \\
&= \ell  \mathbf{R}_{\text{eff}}(\bar{\mathbf{x}}) \big[ \mathbf{I} + \ell ( \boldsymbol{\omega}_{\mathbf{u}_0}(\bar{\mathbf{x}}) \times ) \big] \big[ \mathbf{I} + \ell (\boldsymbol{\omega}(\bar{\mathbf{x}}) \times ) \big]  \mathbf{t}^{(\ell)}_{2,-}(\theta(\bar{\mathbf{x}}) + \ell  \xi(\bar{\mathbf{x}})  + \ell \partial_{\mathbf{u}_0} \theta( \bar{\mathbf{x}}) ,\boldsymbol{\kappa} (\bar{\mathbf{x}}) ) + O(\ell^3) \\
&\qquad  - \ell \big[  \mathbf{R}_{\text{eff}}(\bar{\mathbf{x}}) \big( \mathbf{I} + \ell (\boldsymbol{\omega}(\bar{\mathbf{x}}) \times ) \big)\big]  \mathbf{t}^{(\ell)}_{2,+}(\theta(\bar{\mathbf{x}}) + \ell \xi(\bar{\mathbf{x}}), \boldsymbol{\kappa}(\bar{\mathbf{x}})  )  \\ 
&= \ell \big[  \mathbf{R}_{\text{eff}}(\bar{\mathbf{x}}) \big( \mathbf{I} + \ell (\boldsymbol{\omega}(\bar{\mathbf{x}}) \times ) \big)\big]  \big[ \mathbf{I} + \ell ( \boldsymbol{\omega}_{\mathbf{u}_0}(\bar{\mathbf{x}}) \times ) \big] \mathbf{t}^{(\ell)}_{2,-}(\theta(\bar{\mathbf{x}})  + \ell \partial_{\mathbf{u}_0} \theta( \bar{\mathbf{x}}) ,\boldsymbol{\kappa} (\bar{\mathbf{x}}) ) + O(\ell^3) \\
&\qquad  - \ell \big[  \mathbf{R}_{\text{eff}}(\bar{\mathbf{x}}) \big( \mathbf{I} + \ell (\boldsymbol{\omega}(\bar{\mathbf{x}}) \times ) \big)\big]  \mathbf{t}^{(\ell)}_{2,+}(\theta(\bar{\mathbf{x}}) , \boldsymbol{\kappa}(\bar{\mathbf{x}})  )  ,
\end{aligned} 
\end{equation}
The last equality uses that $\mathbf{t}^{(\ell)}_{2,-}(\theta + \ell \xi, \boldsymbol{\kappa}) =  \mathbf{t}^{(\ell)}_{2,-}(\theta , \boldsymbol{\kappa}) + \ell \xi \mathbf{t}_2'(\theta)  +O(\ell^2)$ and $\mathbf{t}^{(\ell)}_{2,+}(\theta + \ell \xi, \boldsymbol{\kappa}) =  \mathbf{t}^{(\ell)}_{2,+}(\theta , \boldsymbol{\kappa}) + \ell \xi \mathbf{t}_2'(\theta)  +O(\ell^2)$. Next, set $\mathbf{R}^{(\ell)}(\bar{\mathbf{x}}) := \big[  \mathbf{R}_{\text{eff}}(\bar{\mathbf{x}}) \big( \mathbf{I} + \ell (\boldsymbol{\omega}(\bar{\mathbf{x}}) \times ) \big)\big] $ and observe  that 
\begin{equation}
\begin{aligned}
\mathbf{g}_{65,\bar{\mathbf{x}}}^{(\ell)} - \mathbf{g}_{31,\bar{\mathbf{x}}}^{(\ell)} &=  \ell \mathbf{R}^{(\ell)}(\bar{\mathbf{x}})\big[  \mathbf{t}^{(\ell)}_{2,-,\tilde{\mathbf{u}}_0}(\theta(\bar{\mathbf{x}}) , \boldsymbol{\kappa} (\bar{\mathbf{x}})  , \partial_{\mathbf{u}_0} \theta( \bar{\mathbf{x}}) , \boldsymbol{\omega}_{\mathbf{u}_0} (\bar{\mathbf{x}}) )   -   \mathbf{t}^{(\ell)}_{2,+}(\theta(\bar{\mathbf{x}}), \boldsymbol{\kappa}(\bar{\mathbf{x}})  )\big] + O(\ell^3) 
\end{aligned}
\end{equation}
 using the definition in Eq.\;(\ref{eq:bentBoundaries2}). Proposition \ref{LocalBendProp}  furnishes that 
 \begin{equation}
 \begin{aligned}
  \mathbf{t}^{(\ell)}_{2,-,\tilde{\mathbf{u}}_0}(\theta(\bar{\mathbf{x}}) , \boldsymbol{\kappa} (\bar{\mathbf{x}})  , \partial_{\mathbf{u}_0} \theta( \bar{\mathbf{x}}) , \boldsymbol{\omega}_{\mathbf{u}_0} (\bar{\mathbf{x}}) )   -   \mathbf{t}^{(\ell)}_{2,+}(\theta(\bar{\mathbf{x}}), \boldsymbol{\kappa}(\bar{\mathbf{x}})  ) = O(\ell^2)
 \end{aligned}
 \end{equation} 
since Eq.\;(\ref{eq:effPDE1}) and Eq.\;(\ref{eq:theKappaFields})  match  Eq.\;(\ref{eq:micon1}) and Eq.\;(\ref{eq:theKappas1}) under the replacement of their arguments. We conclude that  $\mathbf{g}_{65,\bar{\mathbf{x}}}^{(\ell)} - \mathbf{g}_{31,\bar{\mathbf{x}}}^{(\ell)} = O(\ell^3)$ as desired.  
\end{proof}
We now compute the leftover gaps $\mathbf{g}_{31,\bar{\mathbf{x}}}^{(\ell)}$ and $\mathbf{g}_{42,\bar{\mathbf{x}}}^{(\ell)}$ at leading order in $\ell$.
\begin{lemma}\label{LemmaCell2}
The ansatz satisfies 
\begin{equation}
\begin{aligned}\label{eq:gGapsIdent}
\mathbf{g}_{31,\bar{\mathbf{x}}}^{(\ell)} &= \ell^2 \Big\{ \partial_{\mathbf{u}_0} \mathbf{d}(\bar{\mathbf{x}}) + \frac{1}{2} \partial_{\mathbf{u}_0} \big[ \mathbf{R}_{\emph{eff}}(\bar{\mathbf{x}})\big(  \mathbf{t}_1(\theta(\bar{\mathbf{x}})) + \mathbf{t}_3(\theta(\bar{\mathbf{x}}))\big) \big]    - \mathbf{R}_{\emph{eff}}(\bar{\mathbf{x}}) \big[ \boldsymbol{\omega}(\bar{\mathbf{x}}) \times \mathbf{u}(\theta(\bar{\mathbf{x}}))  + \xi(\bar{\mathbf{x}}) \mathbf{u}'(\theta(\bar{\mathbf{x}}) )\big]\Big\}+ O(\ell^3)  \\
\mathbf{g}_{42,\bar{\mathbf{x}}}^{(\ell)}  &=  \ell^2 \Big\{ \partial_{\mathbf{v}_0} \mathbf{d}(\bar{\mathbf{x}}) + \frac{1}{2} \partial_{\mathbf{v}_0} \big[ \mathbf{R}_{\emph{eff}}(\bar{\mathbf{x}})\big(  \mathbf{t}_2(\theta(\bar{\mathbf{x}})) + \mathbf{t}_4(\theta(\bar{\mathbf{x}}))\big) \big]   - \mathbf{R}_{\emph{eff}}(\bar{\mathbf{x}}) \big[ \boldsymbol{\omega}(\bar{\mathbf{x}}) \times \mathbf{v}(\theta(\bar{\mathbf{x}}))  + \xi(\bar{\mathbf{x}}) \mathbf{v}'(\theta(\bar{\mathbf{x}}) )\big]\Big\}+ O(\ell^3) 
\end{aligned}
\end{equation}
for all $\bar{\mathbf{x}} \in \mathcal{I}_{\emph{cell}}^{(\ell)}$.
\end{lemma}
\begin{proof}
The proof  follows by Taylor expansion and algebraic manipulation involving the identities in Eq.\;(\ref{eq:effPDE3}). Let $\bar{\mathbf{x}} \in \mathcal{I}_{\text{cell}}^{(\ell)}$. Observe that $\mathbf{y}^{(\ell)}_{3, \bar{\mathbf{x}} + \ell \tilde{\mathbf{u}}_0}$ can be expanded in powers of $\ell$ as
\begin{equation}
\begin{aligned}
\mathbf{y}^{(\ell)}_{3, \bar{\mathbf{x}} + \ell \tilde{\mathbf{u}}_0} &=   \mathbf{y}_{0,\bar{\mathbf{x}} + \ell \tilde{\mathbf{u}}_0}^{(\ell)} + \ell \mathbf{R}_{\text{eff}}(\bar{\mathbf{x}} + \ell \tilde{\mathbf{u}}_0) \big[ \mathbf{I} + \ell ( \boldsymbol{\omega}(\bar{\mathbf{x}} + \ell \tilde{\mathbf{u}}_0) \times ) \big] \mathbf{t}_3(\theta(\bar{\mathbf{x}} + \ell \tilde{\mathbf{u}}_0) + \ell \xi(\bar{\mathbf{x}} + \ell  \tilde{\mathbf{u}}_0))  \\
&=  \mathbf{y}_{\text{eff}}(\bar{\mathbf{x}}) + \ell \Big ( \mathbf{d}(\bar{\mathbf{x}}) +  \partial_{\mathbf{u}_0} \mathbf{y}_{\text{eff}}(\bar{\mathbf{x}})  + \mathbf{R}_{\text{eff}}(\bar{\mathbf{x}}) \mathbf{t}_3(\theta(\bar{\mathbf{x}})) \Big)  + \ell^2 \Big\{ \frac{1}{2} \partial_{\mathbf{u}_0} \partial_{\mathbf{u}_0} \mathbf{y}_{\text{eff}}(\bar{\mathbf{x}})  +  \partial_{\mathbf{u}_0} \mathbf{d}(\bar{\mathbf{x}})  \\
&\qquad + \mathbf{R}_{\text{eff}}(\bar{\mathbf{x}}) \big[  \boldsymbol{\omega}(\bar{\mathbf{x}}) \times \mathbf{t}_3(\theta(\bar{\mathbf{x}}))  + \xi(\bar{\mathbf{x}}) \mathbf{t}_3'(\theta(\bar{\mathbf{x}}) \big]  + \partial_{\mathbf{u}_0} \big[ \mathbf{R}_{\text{eff}}(\bar{\mathbf{x}}) \mathbf{t}_3(\theta(\bar{\mathbf{x}})) \big] \Big\} + O(\ell^3).
\end{aligned}
\end{equation}
Likewise,
\begin{equation}
\begin{aligned}
\mathbf{y}^{(\ell)}_{1,\bar{\mathbf{x}}} &= \mathbf{y}_{\text{eff}}(\bar{\mathbf{x}}) + \ell \Big ( \mathbf{d}(\bar{\mathbf{x}})  +  \mathbf{R}_{\text{eff}}(\bar{\mathbf{x}}) \mathbf{t}_1(\theta(\bar{\mathbf{x}}))) \Big) +  \ell^2\Big( \mathbf{R}_{\text{eff}}(\bar{\mathbf{x}}) \big[  \boldsymbol{\omega}(\bar{\mathbf{x}}) \times \mathbf{t}_1(\theta(\bar{\mathbf{x}}))  + \xi(\bar{\mathbf{x}}) \mathbf{t}_1'(\theta(\bar{\mathbf{x}}) \big] \Big). 
\end{aligned}
\end{equation}
Their difference  is  
\begin{equation}
\begin{aligned}
\mathbf{g}_{31,\bar{\mathbf{x}}}^{(\ell)} &= \ell \Big( \partial_{\mathbf{u}_0} \mathbf{y}_{\text{eff}}(\bar{\mathbf{x}})  - \mathbf{R}_{\text{eff}}(\bar{\mathbf{x}}) \mathbf{u}(\theta(\bar{\mathbf{x}})) \Big)  + \ell^2 \Big\{   \frac{1}{2} \partial_{\mathbf{u}_0} \partial_{\mathbf{u}_0} \mathbf{y}_{\text{eff}}(\bar{\mathbf{x}}) + \partial_{\mathbf{u}_0} \mathbf{d}(\bar{\mathbf{x}})   \\
 &\qquad - \mathbf{R}_{\text{eff}}(\bar{\mathbf{x}}) \big[  \boldsymbol{\omega}(\bar{\mathbf{x}}) \times \mathbf{u}(\theta(\bar{\mathbf{x}}))  + \xi(\bar{\mathbf{x}}) \mathbf{u}'(\theta(\bar{\mathbf{x}}) \big]  + \partial_{\mathbf{u}_0} \big[ \mathbf{R}_{\text{eff}}(\bar{\mathbf{x}}) \mathbf{t}_3(\theta(\bar{\mathbf{x}})) \big] \Big\} + O(\ell^3)
\end{aligned}
\end{equation}
since $\mathbf{u}(\theta) = \mathbf{t}_1(\theta) - \mathbf{t}_3(\theta)$. By Eq.\;(\ref{eq:effPDE3}), the term at $O(\ell)$ vanishes and
\begin{equation}
\begin{aligned}
 \partial_{\mathbf{u}_0} \big[ \mathbf{R}_{\text{eff}}(\bar{\mathbf{x}}) \mathbf{t}_3(\theta(\bar{\mathbf{x}})) \big] &= \frac{1}{2}  \partial_{\mathbf{u}_0} \big[ \mathbf{R}_{\text{eff}}(\bar{\mathbf{x}}) \big(\mathbf{t}_1(\theta(\bar{\mathbf{x}})) + \mathbf{t}_3(\theta(\bar{\mathbf{x}}))\big) \big] - \frac{1}{2} \partial_{\mathbf{u}_0} \big[ \mathbf{R}_{\text{eff}}(\bar{\mathbf{x}}) \big(\mathbf{t}_1(\theta(\bar{\mathbf{x}})) - \mathbf{t}_3(\theta(\bar{\mathbf{x}}))\big) \big]  \\
 &=  \frac{1}{2}  \partial_{\mathbf{u}_0} \big[ \mathbf{R}_{\text{eff}}(\bar{\mathbf{x}}) \big(\mathbf{t}_1(\theta(\bar{\mathbf{x}})) + \mathbf{t}_3(\theta(\bar{\mathbf{x}}))\big) \big] - \frac{1}{2} \partial_{\mathbf{u}_0} \partial_{\mathbf{u}_0} \mathbf{y}_{\text{eff}}(\bar{\mathbf{x}}).
 \end{aligned}
 \end{equation}
 This establishes the first part of Eq.\;(\ref{eq:gGapsIdent}). The result for $\mathbf{g}_{42,\bar{\mathbf{x}}}^{(\ell)}$  follows by a similar set of manipulations. 
\end{proof}

Finally, we state conditions under which the leftover gaps in Eq.\;(\ref{eq:gGapsIdent}) vanish at $O(\ell^2)$.
\begin{lemma}\label{LemmaCell3}
There is a vector field $\mathbf{d}(\mathbf{x})$ solving
\begin{equation}
\begin{aligned}\label{eq:dSolve}
\begin{cases}\partial_{\mathbf{u}_0} \mathbf{d}(\mathbf{x}) =  \mathbf{R}_{\emph{eff}}(\mathbf{x}) \big[ \boldsymbol{\omega}(\mathbf{x}) \times \mathbf{u}(\theta(\mathbf{x}))  + \xi(\mathbf{x}) \mathbf{u}'(\theta(\mathbf{x}) )\big] - \frac{1}{2} \partial_{\mathbf{u}_0} \big[ \mathbf{R}_{\emph{eff}}(\mathbf{x})\big(  \mathbf{t}_1(\theta(\mathbf{x})) + \mathbf{t}_3(\theta(\mathbf{x}))\big) \big] \\
\partial_{\mathbf{v}_0} \mathbf{d}(\mathbf{x}) =    \mathbf{R}_{\emph{eff}}(\mathbf{x}) \big[ \boldsymbol{\omega}(\mathbf{x}) \times \mathbf{v}(\theta(\mathbf{x}))  + \xi(\mathbf{x}) \mathbf{v}'(\theta(\mathbf{x}) )\big] - \frac{1}{2} \partial_{\mathbf{v}_0} \big[ \mathbf{R}_{\emph{eff}}(\mathbf{x})\big(  \mathbf{t}_2(\theta(\mathbf{x})) + \mathbf{t}_4(\theta(\mathbf{x}))\big) \big]
\end{cases}
\end{aligned}
\end{equation}
on $\Omega$ if and only if $\boldsymbol{\omega}(\mathbf{x})$ and $\xi(\mathbf{x})$ solve 
\begin{equation}
\begin{aligned}\label{eq:theAuxPDE0}
&\partial_{\mathbf{u}_0} \Big(\mathbf{R}_{\emph{eff}}(\mathbf{x}) \big[ \boldsymbol{\omega}(\mathbf{x}) \times \mathbf{v}(\theta(\mathbf{x})) + \xi(\mathbf{x}) \mathbf{v}'(\theta(\mathbf{x})) \big] \Big)  + \frac{1}{2} \partial_{\mathbf{u}_0} \partial_{\mathbf{v}_0} \Big( \mathbf{R}_{\emph{eff}}(\mathbf{x}) \big[ \mathbf{t}_1(\theta(\mathbf{x})) + \mathbf{t}_3(\theta(\mathbf{x})) \big]\Big)  \\
&\quad = \partial_{\mathbf{v}_0} \Big(\mathbf{R}_{\emph{eff}}(\mathbf{x}) \big[ \boldsymbol{\omega}(\mathbf{x}) \times \mathbf{u}(\theta(\mathbf{x})) + \xi(\mathbf{x}) \mathbf{u}'(\theta(\mathbf{x})) \big] \Big)  + \frac{1}{2} \partial_{\mathbf{u}_0} \partial_{\mathbf{v}_0} \Big( \mathbf{R}_{\emph{eff}}(\mathbf{x}) \big[ \mathbf{t}_2(\theta(\mathbf{x})) + \mathbf{t}_4(\theta(\mathbf{x})) \big]\Big)
\end{aligned}
\end{equation}
on $\Omega$. The vector field $\mathbf{d}(\mathbf{x})$ is smoothly extendable to a neighborhood of $\overline{\Omega}$ as long as $\boldsymbol{\omega}(\mathbf{x})$ and $\boldsymbol{\xi}(\mathbf{x})$ are.
\end{lemma}
\begin{proof}
Since partial derivatives commute, we derive the PDE in Eq.\;(\ref{eq:theAuxPDE0}) by taking the $\partial_{\mathbf{v}_0}$ derivative of the first equation and the $\partial_{\mathbf{u}_0}$ derivative of the second and setting them equal to each other. Since $\Omega$ is simply connected, it follows by standard arguments (Lemma \ref{LemmaPDE1}) that finding a solution $(\boldsymbol{\omega}(\mathbf{x}), \xi(\mathbf{x}))$ to Eq.\;(\ref{eq:theAuxPDE0})  is sufficient for the existence of $\mathbf{d}(\mathbf{x})$ solving Eq.\;(\ref{eq:dSolve}). The fact that $\mathbf{d}(\mathbf{x})$ is smoothly extendable follows from its definition using a smooth extension theorem (e.g., the ones in \cite{evans2022partial}, Chapter 5 will do).
\end{proof}

\paragraph{The global origami construction.} To finish, we need to know that there is a smooth solution $(\boldsymbol{\omega}(\mathbf{x}), \xi(\mathbf{x}))$ of the  linear system of PDEs in Eq.\;(\ref{eq:theAuxPDE0}).
The proof of this is given in  Appendix \ref{sec:ExistencePDE}. After some inspired manipulations, the proof amounts to a lengthy application of technical but standard results from PDE theory.
Mechanics and physics oriented readers can safely skip this proof as the techniques are not required in the rest of the paper, only the result. 
Accepting the existence of a smooth solution $(\boldsymbol{\omega}(\mathbf{x}), \xi(\mathbf{x}))$ to Eq.\;(\ref{eq:theAuxPDE0}), we apply it to our ansatz. We choose $\mathbf{d}(\mathbf{x})$ as in Lemma \ref{LemmaCell3}. It follows from Lemmas \ref{LemmaCell1}-\ref{LemmaCell3} that all of the gaps in the cell-wise ansatz are small:
\begin{equation}
\begin{aligned}\label{eq:cellwiseGapsFinal}
\big( \mathbf{g}_{31,\bar{\mathbf{x}}}^{(\ell)} , \mathbf{g}_{65, \bar{\mathbf{x}}}^{(\ell)}  ,  \mathbf{g}_{78,\bar{\mathbf{x}}}^{(\ell)},   \mathbf{g}_{42, \bar{\mathbf{x}}}^{(\ell)} ,  \mathbf{g}_{85, \bar{\mathbf{x}}}^{(\ell)} , \mathbf{g}_{76, \bar{\mathbf{x}}}^{(\ell)}  \big) = O(\ell^3)
\end{aligned}
\end{equation}
for all $\bar{\mathbf{x}} \in \mathcal{I}_{\text{cell}}^{(\ell)}$.  

Recall from Section \ref{ssec:BarHinge}  that  $\{ \mathbf{x}_{i,j}^{(\ell)} \in \mathbb{R}^3 \colon (i,j) \in \mathcal{I}^{(\ell)} \}$ bijectively labels the vertices of the origami in $\Omega_{\text{cell}}^{(\ell)}$, in a way that is consistent with the nearest neighbor description in Fig.\;\ref{Fig:oridom}. A valid origami deformation is given by mapping all these points as $\mathbf{x}_{i,j}^{(\ell)} \mapsto \mathbf{y}_{i,j}^{(\ell)}$. We now specify  $\{ \mathbf{y}_{i,j}^{(\ell)}  \}$ by closing the gaps generated in our above cell-wise ansatz.  In particular, for each $\mathbf{x}_{i,j}^{(\ell)}$ and $(i,j) \in \mathcal{I}^{(\ell)}$, we define 
\begin{equation}
\begin{aligned}\label{eq:TheOriConstruct}
\mathbf{y}_{i,j}^{(\ell)} := \begin{cases}
\mathbf{y}_{0,\bar{\mathbf{x}}}^{(\ell)} & \text{ if } \mathbf{x}_{i,j}^{(\ell)} = \mathbf{x}_{0,\bar{\mathbf{x}}}^{(\ell)} \text{ for some } \bar{\mathbf{x}} \in \mathcal{I}_{\text{cell}}^{(\ell)} \\
 \frac{1}{2} (\mathbf{y}_{1,\bar{\mathbf{x}}}^{(\ell)} + \mathbf{y}_{3,\bar{\mathbf{x}}+ \ell \tilde{\mathbf{u}}_0}^{(\ell)}) &  \text{ if } \mathbf{x}_{i,j}^{(\ell)} = \mathbf{x}_{1,\bar{\mathbf{x}}}^{(\ell)} \text{ for some } \bar{\mathbf{x}} \in \mathcal{I}_{\text{cell}}^{(\ell)} \\
 \frac{1}{2} (\mathbf{y}_{2,\bar{\mathbf{x}}}^{(\ell)} + \mathbf{y}_{4,\bar{\mathbf{x}}+ \ell \tilde{\mathbf{v}}_0}^{(\ell)}) &  \text{ if } \mathbf{x}_{i,j}^{(\ell)} = \mathbf{x}_{2,\bar{\mathbf{x}}}^{(\ell)} \text{ for some } \bar{\mathbf{x}} \in \mathcal{I}_{\text{cell}}^{(\ell)}  \\
 \frac{1}{4} (\mathbf{y}_{5,\bar{\mathbf{x}}}^{(\ell)} + \mathbf{y}_{6,\bar{\mathbf{x}} + \ell \tilde{\mathbf{u}}_0}^{(\ell)} + \mathbf{y}_{8,\bar{\mathbf{x}} + \ell \tilde{\mathbf{v}}_0}^{(\ell)} + \mathbf{y}_{7,\bar{\mathbf{x}} + \ell \tilde{\mathbf{u}}_0+ \ell \tilde{\mathbf{v}}_0}^{(\ell)}) & \text{ if } \mathbf{x}_{i,j}^{(\ell)} = \mathbf{x}_{5,\bar{\mathbf{x}}}^{(\ell)} \text{ for some } \bar{\mathbf{x}} \in \mathcal{I}_{\text{cell}}^{(\ell)}.
\end{cases}
\end{aligned}
\end{equation}
Fig.\;\ref{Fig:EggboxGap} is helpful for visualizing this definition. As an example,  $\mathbf{y}_{m+1,n+1}^{(\ell)}$ in the figure is equal to $\frac{1}{4} (\mathbf{y}_{5,\bar{\mathbf{x}}}^{(\ell)} + \mathbf{y}_{6,\bar{\mathbf{x}} + \ell \tilde{\mathbf{u}}_0}^{(\ell)} + \mathbf{y}_{8,\bar{\mathbf{x}} + \ell \tilde{\mathbf{v}}_0}^{(\ell)} + \mathbf{y}_{7,\bar{\mathbf{x}} + \ell \tilde{\mathbf{u}}_0+ \ell \tilde{\mathbf{v}}_0}^{(\ell)})$. 

Eq.\;(\ref{eq:TheOriConstruct}) completes the construction: the full ansatz is specified by deforming the vertices $\mathbf{x}_{i,j}^{(\ell)}$ to $\mathbf{y}_{i,j}^{(\ell)}$ through the formulas in Eqs.\;(\ref{eq:cellwise1}), (\ref{eq:cellwise2}), (\ref{eq:cellwise3}) and (\ref{eq:TheOriConstruct}) for any choice of fields $\mathbf{y}_{\text{eff}}(\mathbf{x})$, $\mathbf{R}_{\text{eff}}(\mathbf{x})$, $\theta(\mathbf{x})$, $\boldsymbol{\omega}_{\mathbf{u}_0}(\mathbf{x})$, $\boldsymbol{\omega}_{\mathbf{v}_0}(\mathbf{x})$, $\boldsymbol{\kappa}(\mathbf{x})$, $\boldsymbol{\omega}(\mathbf{x})$, $\xi(\mathbf{x})$ and $\mathbf{d}(\mathbf{x})$ that satisfy the assumptions and constraints outlined in this section.

We claim that these origami deformations have panel strains $\sim \ell^2$, and that they approximate the given effective deformation $\mathbf{y}_\text{eff}(\mathbf{x})$. To prove this, recall that the panel strains are defined by $\varepsilon_{i,j,1}^{(\ell)}, \ldots,\varepsilon_{i,j,4}^{(\ell)}$ in Eq.\;(\ref{eq:PanelStrain}), and that  $\hat{\mathbf{x}}_{i,j}^{(\ell)}$ denotes the orthogonal projection of $\mathbf{x}_{i,j}^{(\ell)}$ onto $\mathbb{R}^2$.
\begin{proposition}\label{GlobalOrigamiProp}
The origami deformation $\{\mathbf{y}_{i,j}^{(\ell)}\}$ satisfies (i) $|\varepsilon^{(\ell)}_{i,j,1}|, \ldots, |\varepsilon_{i,j,4}^{(\ell)}|  = O(\ell^2)$ for all $(i,j) \in \mathcal{I}^{(\ell)}$, and (ii) $|\mathbf{y}_{i,j}^{(\ell)} - \mathbf{y}_{\emph{eff}}(\hat{\mathbf{x}}_{i,j}^{(\ell)})| = O(\ell)$ for all $(i,j) \in \mathcal{I}^{(\ell)}$.
\end{proposition}
\begin{proof}
For part (i), we estimate the lengths of neighboring vertices in the origami construction $\{ \mathbf{y}_{i,j}^{(\ell)} \}$ by those in the cell-wise  construction $\{ \mathbf{y}_{i,\bar{\mathbf{x}}}^{(\ell)} \}$.  We show in particular that the deviation between these lengths is $O(\ell^3)$. The desired estimate then follows because each cell in the latter is deformed by an approximate mechanism deformation, resulting in sufficiently small strain.  We demonstrate the result explicitly for one representative example. The other cases follow analogously. 

Consider the panel strain $\varepsilon_{i,j,1}^{(\ell)}$ in Eq.\;(\ref{eq:PanelStrain}) for $i$ and $j$ in Eq.\;(\ref{eq:TheOriConstruct}) such that  $\mathbf{y}_{i+1,j}^{(\ell)} =  \frac{1}{4} (\mathbf{y}_{5,\bar{\mathbf{x}}}^{(\ell)} + \mathbf{y}_{6,\bar{\mathbf{x}} + \ell \tilde{\mathbf{u}}_0}^{(\ell)} + \mathbf{y}_{8,\bar{\mathbf{x}} + \ell \tilde{\mathbf{v}}_0}^{(\ell)} + \mathbf{y}_{7,\bar{\mathbf{x}} + \ell \tilde{\mathbf{u}}_0+ \ell \tilde{\mathbf{v}}_0}^{(\ell)})$ and $\mathbf{y}_{i,j}^{(\ell)}  =  \frac{1}{2} (\mathbf{y}_{2,\bar{\mathbf{x}}}^{(\ell)} + \mathbf{y}_{4,\bar{\mathbf{x}}+ \ell \tilde{\mathbf{v}}_0}^{(\ell)})$, $\bar{\mathbf{x}} \in \mathcal{I}^{(\ell)}_{\text{cell}}$. It follows that 
\begin{equation}
\begin{aligned}\label{eq:case52Part0}
\mathbf{y}_{i+1,j}^{(\ell)}  - \mathbf{y}_{i,j}^{(\ell)} &=  \mathbf{y}_{5, \bar{\mathbf{x}}}^{(\ell)} + \frac{1}{4} \big( 2 \mathbf{g}_{65,\bar{\mathbf{x}}}^{(\ell)} +  \mathbf{g}_{85,\bar{\mathbf{x}}}^{(\ell)}  + \mathbf{g}_{76,\bar{\mathbf{x}}+ \ell \tilde{\mathbf{u}}_0}^{(\ell)} \big)  - \mathbf{y}_{2,\bar{\mathbf{x}}}^{(\ell)} + \frac{1}{2} \mathbf{g}_{42,\bar{\mathbf{x}}}^{(\ell)} = \mathbf{y}_{5, \bar{\mathbf{x}}}^{(\ell)}  - \mathbf{y}_{2, \bar{\mathbf{x}}}^{(\ell)} + O(\ell^3)
\end{aligned}
\end{equation}
using Eq\;(\ref{eq:theGapsDef}) and Eq.\;(\ref{eq:cellwiseGapsFinal}). The definitions in Eq.\;(\ref{eq:cellwise2}-\ref{eq:cellwise3}) then give that
\begin{equation}
\begin{aligned}\label{eq:case52}
| \mathbf{y}_{5,\bar{\mathbf{x}}}^{(\ell)} - \mathbf{y}_{2,\bar{\mathbf{x}}}^{(\ell)}|  &=  |\ell \big[  \mathbf{R}_{\text{eff}}(\bar{\mathbf{x}}) \big( \mathbf{I} + \ell (\boldsymbol{\omega}(\bar{\mathbf{x}}) \times ) \big)\big]  \mathbf{t}^{(\ell)}_{1,+}(\theta(\bar{\mathbf{x}}) + \ell \xi(\bar{\mathbf{x}}), \boldsymbol{\kappa}(\bar{\mathbf{x}})  )|  \\
&= \ell |  \mathbf{t}^{(\ell)}_{1,+}(\theta(\bar{\mathbf{x}}) + \ell \xi(\bar{\mathbf{x}}), \boldsymbol{\kappa}(\bar{\mathbf{x}})  )| + O(\ell^3)
\end{aligned}
\end{equation}
since $\mathbf{t}_1(\tilde{\theta}) + \mathbf{t}_{2,+}(\tilde{\theta}, \boldsymbol{\kappa}) = \mathbf{t}_2(\tilde{\theta}) + \mathbf{t}_{1,+}(\tilde{\theta}, \boldsymbol{\kappa})$ and since $|\mathbf{R}(\mathbf{I} + \ell (\boldsymbol{\omega} \times )) \mathbf{v}| = \sqrt{|\mathbf{v}|^2 +O(\ell^2)} = |\mathbf{v}| + O(\ell^2)$ for any $\mathbf{R} \in SO(3)$ and $\boldsymbol{\omega}, \mathbf{v} \in \mathbb{R}^3$. Similarly,  $|\mathbf{t}^{(\ell)}_{1,+}(\tilde{\theta}, \boldsymbol{\kappa} )| = \sqrt{ |\mathbf{t}_1(\tilde{\theta})|^2 + O(\ell^2)} = |\mathbf{t}_1(\tilde{\theta})| + O(\ell^2)$ using the definition in Eq.\;(\ref{eq:bentBoundaries1}) and Taylor expansion, and $|\mathbf{t}_1(\tilde{\theta})| = |\mathbf{t}_1^r|$. It follows that 
\begin{equation}
\begin{aligned}\label{eq:case52Final}
\ell |  \mathbf{t}^{(\ell)}_{1,+}(\theta(\bar{\mathbf{x}}) + \ell \xi(\bar{\mathbf{x}}), \boldsymbol{\kappa}(\bar{\mathbf{x}})  )| = \ell |  \mathbf{t}_{1}(\theta(\bar{\mathbf{x}}) + \ell \xi(\bar{\mathbf{x}}))|  + O(\ell^3)=  \ell |\mathbf{t}_1^r|.
\end{aligned}
\end{equation}
Since $\ell |\mathbf{t}_1^r| = |  \mathbf{x}_{5,\bar{\mathbf{x}}}^{(\ell)} - \mathbf{x}_{2,\bar{\mathbf{x}}}^{(\ell)}| = |\mathbf{x}_{i+1,j}^{(\ell)} - \mathbf{x}_{i,j}^{(\ell)}|$ by assumption, we conclude from Eqs.\;(\ref{eq:case52Part0}-\ref{eq:case52Final}) that 
\begin{equation}
\begin{aligned}
|\mathbf{y}_{i+1,j}^{(\ell)} - \mathbf{y}_{i,j}^{(\ell)}| = | \mathbf{y}_{5,\bar{\mathbf{x}}}^{(\ell)} - \mathbf{y}_{2,\bar{\mathbf{x}}}^{(\ell)}| + O(\ell^3) = \ell |\mathbf{t}_1^r| + O(\ell^3) = |\mathbf{x}_{i+1,j}^{(\ell)} - \mathbf{x}_{i,j}^{(\ell)}| + O(\ell^3). 
\end{aligned}
\end{equation}
Hence $|\varepsilon_{i,j,1}^{(\ell)}| = O(\ell^2)$.

For part (ii), consider any $\mathbf{y}_{i,j}^{(\ell)}$  and $\mathbf{x}_{i,j}^{(\ell)}$, $(i,j) \in \mathcal{I}^{(\ell)}$. Note that $\hat{\mathbf{x}}_{i,j}^{(\ell)} = (\mathbf{x}_{i,j}^{(\ell)} \cdot \mathbf{e}_1, \mathbf{x}_{i,j}^{(\ell)} \cdot \mathbf{e}_2)$. Since each $\mathbf{y}_{i,j}^{(\ell)}$ is defined by averaging points in the cell-wise construction and since the gaps between neighboring cells are $O(\ell^3)$, we have that $\mathbf{y}_{i,j}^{(\ell)} = \mathbf{y}_{k,\bar{\mathbf{x}}}^{(\ell)} + O(\ell^3)$ for some $k \in \{ 0,1,\ldots,8\}$ and $\bar{\mathbf{x}} \in \mathcal{I}^{(\ell)}_{\text{cell}}$. Since the vertices inside the reference and deformed cell are at most a distance $O(\ell)$ apart,  $\mathbf{y}_{k,\bar{\mathbf{x}}}^{(\ell)}   = \mathbf{y}_{0,\bar{\mathbf{x}}}^{(\ell)}+ O(\ell)$ and $\hat{\mathbf{x}}_{i,j}^{(\ell)} = \bar{\mathbf{x}} + O(\ell)$.  Using the definition $\mathbf{y}_{0,\bar{\mathbf{x}}}^{(\ell)} := \mathbf{y}_{\text{eff}}(\bar{\mathbf{x}}) + \ell \mathbf{d}(\bar{\mathbf{x}})$ from Eq.\;(\ref{eq:cellwise1}),
\begin{equation}
\begin{aligned}
|\mathbf{y}_{i,j}^{(\ell)} - \mathbf{y}_{\text{eff}}(\hat{\mathbf{x}}_{i,j}^{(\ell)})  | &= |\mathbf{y}_{0,\bar{\mathbf{x}}}^{(\ell)}+ O(\ell)-  \mathbf{y}_{\text{eff}}(\hat{\mathbf{x}}_{i,j}^{(\ell)})|  =| \mathbf{y}_{\text{eff}}(\bar{\mathbf{x}})  + O(\ell) - \mathbf{y}_{\text{eff}}(\bar{\mathbf{x}} + O(\ell))| = O(\ell).
\end{aligned}
\end{equation}
Here we used that $\mathbf{y}_{\text{eff}}(\mathbf{x})$ is smooth. 
\end{proof}

\subsection{Proof of Theorem~\ref{MainTheorem}}\label{sec:CompleteTheProof}

We are ready to prove Theorem~\ref{MainTheorem}, which we do in a series of short steps. Consider any parallelogram origami design from Section \ref{ssec:DesignKin}, and  the bar and hinge energy of this origami design in  Section~\ref{ssec:BarHinge}. Let $\theta(\mathbf{x})$,  $\bm{\omega}_{\mathbf{u}_0} (\mathbf{x})$ and $\bm{\omega}_{\mathbf{v}_0}(\mathbf{x})$ solve the surface theory and have the regularity properties from the statement of the theorem, so that we can guarantee a smooth solution to Eq.\;(\ref{eq:theAuxPDE0}).  The first few results of the theorem have already been proven, but we reiterate them for clarity. 

\noindent \textbf{Step 1}. \textit{There exists a unique and smooth rotation field $\mathbf{R}_{\emph{eff}} \colon \Omega \rightarrow SO(3)$ with $\mathbf{R}_{\emph{eff}}(\langle \mathbf{x} \rangle) = \mathbf{I}$  such that $\partial_{\mathbf{u}_0} \mathbf{R}_{\emph{eff}}(\mathbf{x}) = \mathbf{R}_{\emph{eff}}(\mathbf{x}) (\boldsymbol{\omega}_{\mathbf{u}_0}(\mathbf{x}) \times)$ and $\partial_{\mathbf{v}_0} \mathbf{R}_{\emph{eff}}(\mathbf{x}) = \mathbf{R}_{\emph{eff}}(\mathbf{x}) (\boldsymbol{\omega}_{\mathbf{v}_0}(\mathbf{x}) \times)$.  In addition, there exists a unique and smooth  effective deformation $\mathbf{y}_{\emph{eff}} \colon \Omega \rightarrow \mathbb{R}^3$  such that $\mathbf{y}_{\emph{eff}} ( \langle \mathbf{x} \rangle) = \mathbf{0}$ and $\nabla \mathbf{y}_{\emph{eff}}(\mathbf{x}) = \mathbf{R}_{\emph{eff}}(\mathbf{x}) \mathbf{A}_{\emph{eff}}(\theta(\mathbf{x}))$.}

\begin{proof} This follows directly from Proposition \ref{firstProp}. \end{proof}

\noindent Fix the unique $\mathbf{R}_{\text{eff}}(\mathbf{x})$ and $\mathbf{y}_{\text{eff}}(\mathbf{x})$ as above. Choose a smooth solution $(\boldsymbol{\omega}(\mathbf{x}), \xi(\mathbf{x}))$ to Eq.\;(\ref{eq:theAuxPDE0}), the existence of which is guaranteed by Proposition \ref{AuxPDEProp} of Appendix \ref{sec:ExistencePDE}, and construct a global origami deformation $\{\mathbf{y}_{i,j}^{(\ell)}\}$ from the fields  $\theta(\mathbf{x})$, $\bm{\omega}_{\mathbf{u}_0} (\mathbf{x})$, $\bm{\omega}_{\mathbf{v}_0}(\mathbf{x})$, $\mathbf{R}_{\text{eff}}(\mathbf{x})$, $\mathbf{y}_{\text{eff}}(\mathbf{x})$, $\boldsymbol{\omega}(\mathbf{x})$ and $\xi(\mathbf{x})$, as done in Section \ref{ssec:GlobalOrigami}. Recall that the panels strains are defined by $\varepsilon_{i,j,1}^{(\ell)}, \ldots, \varepsilon_{i,j,4}^{(\ell)}$ in Eq.\;(\ref{eq:PanelStrain}).

\noindent \textbf{Step 2}. \textit{The origami deformation $\{ \mathbf{y}_{i,j}^{(\ell)} \}$ satisfies $|\varepsilon^{(\ell)}_{i,j,1}|, \ldots, |\varepsilon_{i,j,4}^{(\ell)}|  = O(\ell^2)$ for all $(i,j) \in \mathcal{I}^{(\ell)}$.}

\begin{proof} This follows directly from the first part of Proposition \ref{GlobalOrigamiProp}. \end{proof}

\noindent Next, recall that $\hat{\mathbf{x}}_{i,j}^{(\ell)}$ denotes the orthogonal projection of $\mathbf{x}_{i,j}^{(\ell)}$ onto $\mathbb{R}^2$.

\noindent \textbf{Step 3}.\;\textit{The origami deformation $\{ \mathbf{y}_{i,j}^{(\ell)} \}$ approximates the effective deformation $\mathbf{y}_{\emph{eff}}(\mathbf{x})$ in that $|\mathbf{y}_{i,j}^{(\ell)} - \mathbf{y}_{\emph{eff}}(\hat{\mathbf{x}}_{i,j}^{(\ell)})| = O(\ell)$ for all $(i,j) \in \mathcal{I}^{(\ell)}$.}

\begin{proof} This follows directly from the second part of Proposition \ref{GlobalOrigamiProp}. \end{proof}

We now calculate the bending energy of each individual panel. This involves a Taylor expansion of the angles $\psi_{i,j,1}^{(\ell)}$ and $\psi_{i,j,2}^{(\ell)}$ defined in Eqs.\;(\ref{eq:getBendAngles}) and (\ref{eq:getBendAngles2}). As these angles are cumbersome nonlinear functions of the panel vertices, we supply several useful preliminary identities based on the schematic in Fig.\;\ref{Fig:CalcBending}.  Let $\mathbf{R}^{(\ell)} = \mathbf{R} + O(\ell)$ for some $\mathbf{R} \in SO(3)$ and $\mathbf{t}_i^{(\ell)} = \mathbf{t}_i(\theta) + O(\ell)$, $i = 1,\ldots,4$, for the deformed tangents $\mathbf{t}_i(\theta)$ from Section \ref{ssec:DesignKin}. Note we do not assume that $\mathbf{R}^{(\ell)}$ is a rotation.

\begin{figure}[t!]
\centering
\includegraphics[width=1\textwidth]{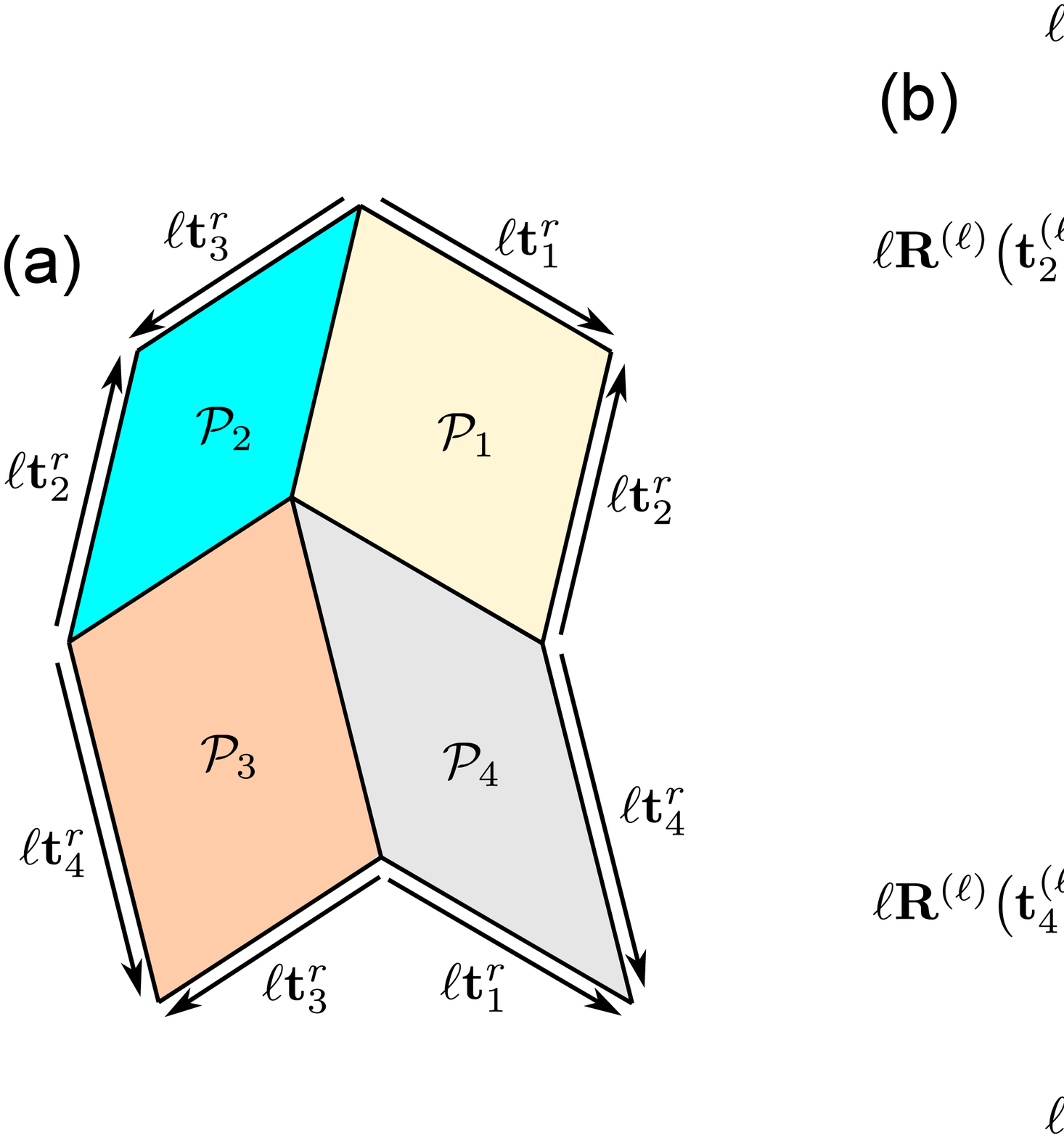} 
\caption{Schematic for panel bending. The points and vectors drawn here are used to calculate the bending strains of the panels.  Note $\mathbf{R}^{(\ell)} = \mathbf{R} + O(\ell)$ for some $\mathbf{R} \in SO(3)$ and $\mathbf{t}_i^{(\ell)} = \mathbf{t}_i(\theta) + O(\ell)$, $i = 1,\ldots,4$, for the deformed tangents $\mathbf{t}_i(\theta)$ from Section \ref{ssec:DesignKin}.}
\label{Fig:CalcBending}
\end{figure}

\noindent \textbf{Step 3}.\;\textit{The vertices $\mathbf{y}_{1}^{(\ell)},\ldots, \mathbf{y}_4^{(\ell)}$ in Fig.\;\ref{Fig:CalcBending}  satisfy
\begin{align}\label{eq:ugly}
&\underline{\text{Panel 1:}} \nonumber \\
&\begin{cases}
\frac{\mathbf{y}_1^{(\ell)} - \mathbf{y}_2^{(\ell)}}{|\mathbf{y}_1^{(\ell)} - \mathbf{y}_2^{(\ell)}|}  \cdot \Big( \Big[\frac{(\mathbf{y}_1^{(\ell)} - \mathbf{y}_0^{(\ell)}  )\times ( \mathbf{y}_2^{(\ell)} - \mathbf{y}_0^{(\ell)})}{|(\mathbf{y}_1^{(\ell)} - \mathbf{y}_0^{(\ell)}  )\times ( \mathbf{y}_2^{(\ell)} - \mathbf{y}_0^{(\ell)})|} \Big] \times \Big[\frac{(\mathbf{y}_2^{(\ell)} - \mathbf{y}_5^{(\ell)}  )\times ( \mathbf{y}_1^{(\ell)} - \mathbf{y}_5^{(\ell)})}{|(\mathbf{y}_2^{(\ell)} - \mathbf{y}_5^{(\ell)}  )\times ( \mathbf{y}_1^{(\ell)} - \mathbf{y}_5^{(\ell)})|} \Big] \Big) = \ell \kappa_1 \frac{|(\mathbf{t}_1(\theta) - \mathbf{t}_2(\theta)) \times \mathbf{n}_{12}(\theta)|^2}{|\mathbf{t}_1(\theta) - \mathbf{t}_2(\theta)| |\mathbf{n}_{12}(\theta)|^2 } + O(\ell^2),    \\
\frac{\mathbf{y}_5^{(\ell)} - \mathbf{y}_0^{(\ell)}}{|\mathbf{y}_5^{(\ell)} - \mathbf{y}_0^{(\ell)}|}  \cdot  \Big( \Big[\frac{(\mathbf{y}_5^{(\ell)} - \mathbf{y}_1^{(\ell)}  )\times ( \mathbf{y}_0^{(\ell)} - \mathbf{y}_1^{(\ell)})}{|(\mathbf{y}_5^{(\ell)} - \mathbf{y}_1^{(\ell)}  )\times ( \mathbf{y}_0^{(\ell)} - \mathbf{y}_1^{(\ell)})|} \Big] \times \Big[\frac{(\mathbf{y}_0^{(\ell)} - \mathbf{y}_2^{(\ell)}  )\times ( \mathbf{y}_5^{(\ell)} - \mathbf{y}_2^{(\ell)})}{|(\mathbf{y}_0^{(\ell)} - \mathbf{y}_2^{(\ell)}  )\times ( \mathbf{y}_5^{(\ell)} - \mathbf{y}_2^{(\ell)})|} \Big] \Big)  = \ell \kappa_1 \frac{|(\mathbf{t}_1(\theta) + \mathbf{t}_2(\theta)) \times \mathbf{n}_{12}(\theta)|^2}{|\mathbf{t}_1(\theta) + \mathbf{t}_2(\theta)| |\mathbf{n}_{12}(\theta)|^2 } + O(\ell^2).
\end{cases} \nonumber \\ 
&\underline{\text{Panel 2:}} \nonumber \\
&\begin{cases}
\frac{\mathbf{y}_0^{(\ell)} - \mathbf{y}_6^{(\ell)}}{|\mathbf{y}_0^{(\ell)} - \mathbf{y}_6^{(\ell)}|} \cdot   \Big(\Big[\frac{(\mathbf{y}_0^{(\ell)} - \mathbf{y}_3^{(\ell)}  )\times ( \mathbf{y}_6^{(\ell)} - \mathbf{y}_3^{(\ell)})}{|(\mathbf{y}_0^{(\ell)} - \mathbf{y}_3^{(\ell)}  )\times ( \mathbf{y}_6^{(\ell)} - \mathbf{y}_3^{(\ell)})|} \Big] \times \Big[\frac{( \mathbf{y}_6^{(\ell)} - \mathbf{y}_2^{(\ell)}) \times (\mathbf{y}_0^{(\ell)} - \mathbf{y}_2^{(\ell)}  )}{|(\mathbf{y}_0^{(\ell)} - \mathbf{y}_2^{(\ell)}  )\times ( \mathbf{y}_6^{(\ell)} - \mathbf{y}_2^{(\ell)})|} \Big]  \Big)   =  \ell \kappa_2 \frac{|(\mathbf{t}_2(\theta) + \mathbf{t}_3(\theta)) \times \mathbf{n}_{23}(\theta)|^2}{|\mathbf{t}_2(\theta) + \mathbf{t}_3(\theta)| |\mathbf{n}_{23}(\theta)|^2 } + O(\ell^2),    \\
\frac{\mathbf{y}_2^{(\ell)} - \mathbf{y}_3^{(\ell)}}{|\mathbf{y}_2^{(\ell)} - \mathbf{y}_3^{(\ell)}|} \cdot \Big( \Big[\frac{ ( \mathbf{y}_2^{(\ell)} - \mathbf{y}_0^{(\ell)}) \times (\mathbf{y}_3^{(\ell)} - \mathbf{y}_0^{(\ell)}  )}{|  ( \mathbf{y}_2^{(\ell)} - \mathbf{y}_0^{(\ell)}) \times (\mathbf{y}_3^{(\ell)} - \mathbf{y}_0^{(\ell)}  )|} \Big] \times \Big[\frac{(\mathbf{y}_3^{(\ell)} - \mathbf{y}_6^{(\ell)}  )\times ( \mathbf{y}_2^{(\ell)} - \mathbf{y}_6^{(\ell)})}{|(\mathbf{y}_3^{(\ell)} - \mathbf{y}_6^{(\ell)}  )\times ( \mathbf{y}_2^{(\ell)} - \mathbf{y}_6^{(\ell)})|} \Big]\Big) =  \ell \kappa_2 \frac{|(\mathbf{t}_2(\theta) - \mathbf{t}_3(\theta)) \times \mathbf{n}_{23}(\theta)|^2}{|\mathbf{t}_2(\theta) - \mathbf{t}_3(\theta)| |\mathbf{n}_{23}(\theta)|^2 } + O(\ell^2).
\end{cases} \nonumber  \\ 
&\underline{\text{Panel 3:}} \nonumber  \\
&\begin{cases}
\frac{\mathbf{y}_4^{(\ell)} - \mathbf{y}_3^{(\ell)}}{|\mathbf{y}_4^{(\ell)} - \mathbf{y}_3^{(\ell)}|} \cdot \Big( \Big[\frac{(\mathbf{y}_4^{(\ell)} - \mathbf{y}_7^{(\ell)}  )\times ( \mathbf{y}_3^{(\ell)} - \mathbf{y}_7^{(\ell)})}{|(\mathbf{y}_4^{(\ell)} - \mathbf{y}_7^{(\ell)}  )\times ( \mathbf{y}_3^{(\ell)} - \mathbf{y}_7^{(\ell)})|} \Big] \times \Big[\frac{( \mathbf{y}_3^{(\ell)} - \mathbf{y}_0^{(\ell)}) \times (\mathbf{y}_4^{(\ell)} - \mathbf{y}_0^{(\ell)}  )}{|(\mathbf{y}_3^{(\ell)} - \mathbf{y}_0^{(\ell)}  )\times ( \mathbf{y}_4^{(\ell)} - \mathbf{y}_0^{(\ell)})|} \Big] \Big)    =  \ell \kappa_3 \frac{|(\mathbf{t}_3(\theta) - \mathbf{t}_4(\theta)) \times \mathbf{n}_{34}(\theta)|^2}{|\mathbf{t}_3(\theta) - \mathbf{t}_4(\theta)| |\mathbf{n}_{34}(\theta)|^2 } + O(\ell^2),    \\
\frac{\mathbf{y}_0^{(\ell)} - \mathbf{y}_7^{(\ell)}}{|\mathbf{y}_0^{(\ell)} - \mathbf{y}_7^{(\ell)}|} \cdot \Big( \Big[\frac{ ( \mathbf{y}_0^{(\ell)} - \mathbf{y}_4^{(\ell)}) \times (\mathbf{y}_7^{(\ell)} - \mathbf{y}_4^{(\ell)}  )}{|  ( \mathbf{y}_0^{(\ell)} - \mathbf{y}_4^{(\ell)}) \times (\mathbf{y}_7^{(\ell)} - \mathbf{y}_4^{(\ell)}  )|} \Big] \times \Big[\frac{(\mathbf{y}_7^{(\ell)} - \mathbf{y}_3^{(\ell)}  )\times ( \mathbf{y}_0^{(\ell)} - \mathbf{y}_3^{(\ell)})}{|(\mathbf{y}_7^{(\ell)} - \mathbf{y}_3^{(\ell)}  )\times ( \mathbf{y}_0^{(\ell)} - \mathbf{y}_3^{(\ell)})|} \Big] \Big)    =  \ell \kappa_3 \frac{|(\mathbf{t}_3(\theta) + \mathbf{t}_4(\theta)) \times \mathbf{n}_{34}(\theta)|^2}{|\mathbf{t}_3(\theta) + \mathbf{t}_4(\theta)| |\mathbf{n}_{34}(\theta)|^2 } + O(\ell^2). 
\end{cases} \nonumber  \\ 
&\underline{\text{Panel 4:}} \nonumber  \\
&\begin{cases}
\frac{\mathbf{y}_8^{(\ell)} - \mathbf{y}_0^{(\ell)}}{|\mathbf{y}_8^{(\ell)} - \mathbf{y}_0^{(\ell)}|} \cdot \Big( \Big[\frac{(\mathbf{y}_8^{(\ell)} - \mathbf{y}_4^{(\ell)}  )\times ( \mathbf{y}_0^{(\ell)} - \mathbf{y}_4^{(\ell)})}{|(\mathbf{y}_8^{(\ell)} - \mathbf{y}_4^{(\ell)}  )\times ( \mathbf{y}_0^{(\ell)} - \mathbf{y}_4^{(\ell)})|} \Big] \times \Big[\frac{( \mathbf{y}_0^{(\ell)} - \mathbf{y}_1^{(\ell)}) \times (\mathbf{y}_8^{(\ell)} - \mathbf{y}_1^{(\ell)}  )}{|(\mathbf{y}_0^{(\ell)} - \mathbf{y}_1^{(\ell)}  )\times ( \mathbf{y}_8^{(\ell)} - \mathbf{y}_1^{(\ell)})|} \Big]  \Big)   =  \ell \kappa_4 \frac{|(\mathbf{t}_4(\theta) + \mathbf{t}_1(\theta)) \times \mathbf{n}_{41}(\theta)|^2}{|\mathbf{t}_4(\theta) + \mathbf{t}_1(\theta)| |\mathbf{n}_{41}(\theta)|^2 } + O(\ell^2),    \\
\frac{\mathbf{y}_1^{(\ell)} - \mathbf{y}_4^{(\ell)}}{|\mathbf{y}_1^{(\ell)} - \mathbf{y}_4^{(\ell)}|}  \cdot \Big( \Big[\frac{ ( \mathbf{y}_1^{(\ell)} - \mathbf{y}_8^{(\ell)}) \times (\mathbf{y}_4^{(\ell)} - \mathbf{y}_8^{(\ell)}  )}{|  ( \mathbf{y}_1^{(\ell)} - \mathbf{y}_8^{(\ell)}) \times (\mathbf{y}_4^{(\ell)} - \mathbf{y}_8^{(\ell)}  )|} \Big] \times \Big[\frac{(\mathbf{y}_4^{(\ell)} - \mathbf{y}_0^{(\ell)}  )\times ( \mathbf{y}_1^{(\ell)} - \mathbf{y}_0^{(\ell)})}{|(\mathbf{y}_4^{(\ell)} - \mathbf{y}_0^{(\ell)}  )\times ( \mathbf{y}_1^{(\ell)} - \mathbf{y}_0^{(\ell)})|} \Big] \Big)  =  \ell \kappa_4 \frac{|(\mathbf{t}_4(\theta) - \mathbf{t}_1(\theta)) \times \mathbf{n}_{41}(\theta)|^2}{|\mathbf{t}_4(\theta) - \mathbf{t}_1(\theta)| |\mathbf{n}_{41}(\theta)|^2 } + O(\ell^2),
\end{cases}
\end{align}
where $\mathbf{n}_{ij}(\theta) := \mathbf{t}_i(\theta) \times \mathbf{t}_j(\theta)$.}
\begin{proof}
Observe from the diagram of Panel 1 that 
\begin{equation}
\begin{aligned}\label{eq:AnnoyingBendCalc1}
&(\mathbf{y}_1^{(\ell)} - \mathbf{y}_0^{(\ell)}  )\times  (\mathbf{y}_2^{(\ell)} - \mathbf{y}_0^{(\ell)}  ) = \ell^2 \big[\text{cof }\mathbf{R}^{(\ell)} \big]  (\mathbf{t}_1^{(\ell)} \times \mathbf{t}_2^{(\ell)}), \\
&  (\mathbf{y}_2^{(\ell)} - \mathbf{y}_5^{(\ell)}  ) \times (\mathbf{y}_1^{(\ell)} - \mathbf{y}_5^{(\ell)} )   = -\ell^2 \big[\text{cof }\mathbf{R}^{(\ell)} \big] \big\{  (\mathbf{t}_2^{(\ell)} \times \mathbf{t}_1^{(\ell)}) + \ell \kappa_1 ( \mathbf{t}_2^{(\ell)} - \mathbf{t}_1^{(\ell)}) \times ( \mathbf{t}_1^{(\ell)} \times \mathbf{t}_2^{(\ell)})  \big\} ,  \\
&(\mathbf{y}_5^{(\ell)} - \mathbf{y}_1^{(\ell)}  )\times ( \mathbf{y}_0^{(\ell)} - \mathbf{y}_1^{(\ell)}) = - \ell^2 \big[\text{cof }\mathbf{R}^{(\ell)} \big] \big\{  (\mathbf{t}_2^{(\ell)} \times \mathbf{t}_1^{(\ell)}) - \ell \kappa_1 \mathbf{t}_1^{(\ell)}  \times (\mathbf{t}_1^{(\ell)} \times \mathbf{t}_2^{(\ell)}) \big\},  \\ 
&(\mathbf{y}_0^{(\ell)} - \mathbf{y}_2^{(\ell)}  )\times ( \mathbf{y}_5^{(\ell)} - \mathbf{y}_2^{(\ell)}) =  - \ell^2 \big[\text{cof }\mathbf{R}^{(\ell)} \big] \big\{  (\mathbf{t}_2^{(\ell)} \times \mathbf{t}_1^{(\ell)}) + \ell \kappa_1 \mathbf{t}_2^{(\ell)}  \times (\mathbf{t}_1^{(\ell)} \times \mathbf{t}_2^{(\ell)}) \big\} 
\end{aligned}
\end{equation}
by standard properties of the cross product and cofactor,  namely that $(\text{cof } \mathbf{A})(\mathbf{a} \times \mathbf{b}) = \mathbf{A} \mathbf{a} \times \mathbf{A} \mathbf{b}$ for $\mathbf{A} \in \mathbb{R}^{3\times3}$ and $\mathbf{a}, \mathbf{b} \in \mathbb {R}^3$. Taking the cross product of the first two and second two equations above, respectively, yields 
\begin{equation}
\begin{aligned}\label{eq:AnnoyingBendCalc2}
&\big[(\mathbf{y}_1^{(\ell)} - \mathbf{y}_0^{(\ell)}  )\times  (\mathbf{y}_2^{(\ell)} - \mathbf{y}_0^{(\ell)}  )  \big] \times \big[(\mathbf{y}_2^{(\ell)} - \mathbf{y}_5^{(\ell)}  ) \times (\mathbf{y}_1^{(\ell)} - \mathbf{y}_5^{(\ell)}  ) \big]  \\
&\qquad \qquad  = - \ell^5 \kappa_1  \big[ \text{cof} \big[ \text{cof } \mathbf{R}^{(\ell)} \big]  \big] (\mathbf{t}_1^{(\ell)} \times \mathbf{t}_2^{(\ell)}) \times \big\{ ( \mathbf{t}_2^{(\ell)} - \mathbf{t}_1^{(\ell)}) \times  (\mathbf{t}_1^{(\ell)} \times \mathbf{t}_2^{(\ell)})  \big \}  \\
& \big[ (\mathbf{y}_5^{(\ell)} - \mathbf{y}_1^{(\ell)}  )\times ( \mathbf{y}_0^{(\ell)} - \mathbf{y}_1^{(\ell)}) \big]  \times \big[(\mathbf{y}_0^{(\ell)} - \mathbf{y}_2^{(\ell)}  )\times ( \mathbf{y}_5^{(\ell)} - \mathbf{y}_2^{(\ell)}) \big]  \\
&\qquad \qquad  =  \ell^5 \kappa_1  \big[ \text{cof} \big[ \text{cof } \mathbf{R}^{(\ell)} \big]  \big] (\mathbf{t}_1^{(\ell)} \times \mathbf{t}_2^{(\ell)}) \times \big\{ ( \mathbf{t}_2^{(\ell)} + \mathbf{t}_1^{(\ell)}) \times  (\mathbf{t}_1^{(\ell)} \times \mathbf{t}_2^{(\ell)})  \big \}  + O(\ell^6)
\end{aligned}
\end{equation}
That $\mathbf{R}^{(\ell)} = \mathbf{R} + O(\ell)$ for a rotation $\mathbf{R}$ furnishes the identities $\big[\text{cof }\mathbf{R}^{(\ell)} \big]  = \mathbf{R} + O(\ell)$ and  $\big[ \text{cof} \big[ \text{cof } \mathbf{R}^{(\ell)} \big]  \big] = \mathbf{R} + O(\ell)$. Thus, 
\begin{equation}
\begin{aligned}
&\Big[\tfrac{(\mathbf{y}_1^{(\ell)} - \mathbf{y}_0^{(\ell)}  )\times ( \mathbf{y}_2^{(\ell)} - \mathbf{y}_0^{(\ell)})}{|(\mathbf{y}_1^{(\ell)} - \mathbf{y}_0^{(\ell)}  )\times ( \mathbf{y}_2^{(\ell)} - \mathbf{y}_0^{(\ell)})|} \Big] \times \Big[\tfrac{(\mathbf{y}_2^{(\ell)} - \mathbf{y}_5^{(\ell)}  )\times ( \mathbf{y}_1^{(\ell)} - \mathbf{y}_5^{(\ell)})}{|(\mathbf{y}_2^{(\ell)} - \mathbf{y}_5^{(\ell)}  )\times ( \mathbf{y}_1^{(\ell)} - \mathbf{y}_5^{(\ell)})|} \Big]    = -  \ell \kappa_1 \tfrac{\mathbf{R} [ \mathbf{n}_{12}(\theta) \times \{ (\mathbf{t}_2(\theta) - \mathbf{t}_1(\theta)) \times  \mathbf{n}_{12}(\theta)  \} ]}{|\mathbf{n}_{12}(\theta)|^2} + O(\ell^2) , \\
& \Big[\tfrac{(\mathbf{y}_5^{(\ell)} - \mathbf{y}_1^{(\ell)}  )\times ( \mathbf{y}_0^{(\ell)} - \mathbf{y}_1^{(\ell)})}{|(\mathbf{y}_5^{(\ell)} - \mathbf{y}_1^{(\ell)}  )\times ( \mathbf{y}_0^{(\ell)} - \mathbf{y}_1^{(\ell)})|} \Big] \times \Big[\tfrac{(\mathbf{y}_0^{(\ell)} - \mathbf{y}_2^{(\ell)}  )\times ( \mathbf{y}_5^{(\ell)} - \mathbf{y}_2^{(\ell)})}{|(\mathbf{y}_0^{(\ell)} - \mathbf{y}_2^{(\ell)}  )\times ( \mathbf{y}_5^{(\ell)} - \mathbf{y}_2^{(\ell)})|} \Big]  =   \ell \kappa_1 \tfrac{\mathbf{R}[ \mathbf{n}_{12}(\theta) \times \{ (\mathbf{t}_2(\theta) + \mathbf{t}_1(\theta)) \times  \mathbf{n}_{12}(\theta)  \} ]}{|\mathbf{n}_{12}(\theta)|^2} + O(\ell^2) , 
\end{aligned}
\end{equation}
since each $\mathbf{t}_i^{(\ell)} = \mathbf{t}_i(\theta) + O(\ell)$. The desired result  for Panel 1 follows after noting that $\mathbf{y}_1^{(\ell)} - \mathbf{y}_2^{(\ell)} =  \mathbf{R}( \mathbf{t}_1(\theta) - \mathbf{t}_2(\theta)) + O(\ell)$ and  $\mathbf{y}_5^{(\ell)} - \mathbf{y}_0^{(\ell)} = \mathbf{R}( \mathbf{t}_1(\theta) + \mathbf{t}_2(\theta))  + O(\ell)$. The calculations for Panels 2-4 are similar and not done here for brevity. 
\end{proof}

\noindent Now recall from Section \ref{ssec:MechKin} that each deformed tangent $\mathbf{t}_i(\theta)$ is related to the reference tangent  $\mathbf{t}_i^r$ through $\theta$-dependent panel rotations. We therefore establish the following.

\noindent \textbf{Step 4}.\;\textit{In each formula in Eq.\;(\ref{eq:ugly}), 
\begin{equation}
\begin{aligned}\label{eq:bijpm}
\frac{|(\mathbf{t}_i(\theta) \pm \mathbf{t}_i(\theta)) \times \mathbf{n}_{ij}(\theta)|^2}{|\mathbf{t}_i(\theta) \pm \mathbf{t}_j(\theta)| |\mathbf{n}_{ij}(\theta)|^2 } =\frac{| (\mathbf{t}_i^r   \pm \mathbf{t}_j^{r})  \times ( \mathbf{t}_i^r \times \mathbf{t}_j^r)|^2}{|(\mathbf{t}_i^r   \pm \mathbf{t}_j^{r})||( \mathbf{t}_i^r \times \mathbf{t}_j^r)|^2} =: b_{ij}^{\pm} , \quad ij \in \{12,23,34,41\}, 
\end{aligned}
\end{equation}
i.e., these terms do not depend on $\theta$.} 
\begin{proof}
Since we are dealing with adjacent crease vectors, $\mathbf{t}_i(\theta) = \mathbf{R}_i(\theta) \mathbf{t}_i^r$ and $\mathbf{t}_j(\theta) = \mathbf{R}_i(\theta) \mathbf{t}_j^r$ for some rotation $\mathbf{R}_i(\theta)$. As a result, $\mathbf{n}_{ij}(\theta) = ( \mathbf{R}_i(\theta) \mathbf{t}_i^r \times  \mathbf{R}_i(\theta) \mathbf{t}_j^r) = \mathbf{R}_i(\theta) ( \mathbf{t}_i^r \times \mathbf{t}_j^r)$ and thus $(\mathbf{t}_i(\theta) \pm \mathbf{t}_i(\theta)) \times \mathbf{n}_{ij}(\theta) = [\mathbf{R}_i(\theta) ( \mathbf{t}_i^r \pm \mathbf{t}_j^r)]  \times[ \mathbf{R}_i(\theta) (\mathbf{t}_i^r \times \mathbf{t}_j^r)] = \mathbf{R}_i(\theta)  [( \mathbf{t}_i^r \pm \mathbf{t}_j^r) \times ( \mathbf{t}_i^r \times \mathbf{t}_j^r)]$. 
\end{proof}

Next, we derive asymptotic formulas for the angles $\psi_{i,j,1}^{(\ell)}$ and $\psi_{i,j,1}^{(\ell)}$. The formulas are organized using the notation for the reference vertices  $\mathbf{x}_{i,\bar{\mathbf{x}}}^{(\ell)}$ introduced in Section \ref{ssec:GlobalOrigami}.  They also involve the panel curvature fields $\kappa_1(\mathbf{x}), \ldots, \kappa_4(\mathbf{x})$  in Eq.\;(\ref{eq:theKappaFields}) and the moduli $b_{ij}^{\pm}$ above. 

\noindent \textbf{Step 5}. \textit{The origami deformation $\{ \mathbf{y}_{i,j}^{(\ell)} \}$ satisfies $|\psi^{(\ell)}_{i,j,1}|, |\psi_{i,j,2}^{(\ell)}|  = O(\ell)$ for all $(i,j) \in \mathcal{I}^{(\ell)}$.  Specifically, 
\begin{equation}
\begin{aligned}\label{eq:psiResult1}
\psi^{(\ell)}_{i,j,1} = \begin{cases}
\ell b_{12}^- \kappa_{1}(\bar{\mathbf{x}})  + O(\ell^2) &  \text{ if } \mathbf{x}_{i,j}^{(\ell)} = \mathbf{x}_{0,\bar{\mathbf{x}}}^{(\ell)} \text { for some } \bar{\mathbf{x}} \in \mathcal{I}_{\emph{cell}}^{(\ell)},  \\
\ell b_{23}^+ \kappa_{2}(\bar{\mathbf{x}})  + O(\ell^2)  &  \text{ if } \mathbf{x}_{i,j}^{(\ell)} = \mathbf{x}_{1,\bar{\mathbf{x}}}^{(\ell)} \text { for some }  \bar{\mathbf{x}} \in \mathcal{I}_{\emph{cell}}^{(\ell)}, \\
\ell b_{34}^- \kappa_{3}(\bar{\mathbf{x}})  + O(\ell^2) &  \text{ if } \mathbf{x}_{i,j}^{(\ell)} = \mathbf{x}_{5,\bar{\mathbf{x}}}^{(\ell)} \text { for some }  \bar{\mathbf{x}} \in \mathcal{I}_{\emph{cell}}^{(\ell)},  \\
\ell b_{41}^+ \kappa_{4}(\bar{\mathbf{x}}) + O(\ell^2)  &  \text{ if } \mathbf{x}_{i,j}^{(\ell)} = \mathbf{x}_{2,\bar{\mathbf{x}}}^{(\ell)} \text{ for some } \bar{\mathbf{x}} \in \mathcal{I}_{\emph{cell}}^{(\ell)},
\end{cases}
\end{aligned}
\end{equation}
and 
\begin{equation}
\begin{aligned}\label{eq:psiResult2}
\psi^{(\ell)}_{i,j,2} = \begin{cases}
\ell b_{12}^+ \kappa_{1}(\bar{\mathbf{x}}) + O(\ell^2)  &  \text{ if } \mathbf{x}_{i,j}^{(\ell)} = \mathbf{x}_{0,\bar{\mathbf{x}}}^{(\ell)} \text { for some } \bar{\mathbf{x}} \in \mathcal{I}_{\emph{cell}}^{(\ell)},  \\
\ell b_{23}^- \kappa_{2}(\bar{\mathbf{x}}) + O(\ell^2)  &  \text{ if } \mathbf{x}_{i,j}^{(\ell)} = \mathbf{x}_{1,\bar{\mathbf{x}}}^{(\ell)} \text { for some }  \bar{\mathbf{x}} \in \mathcal{I}_{\emph{cell}}^{(\ell)}, \\
\ell b_{34}^+ \kappa_{3}(\bar{\mathbf{x}})  + O(\ell^2) &  \text{ if } \mathbf{x}_{i,j}^{(\ell)} = \mathbf{x}_{5,\bar{\mathbf{x}}}^{(\ell)} \text { for some }  \bar{\mathbf{x}} \in \mathcal{I}_{\emph{cell}}^{(\ell)},  \\
\ell b_{41}^- \kappa_{4}(\bar{\mathbf{x}})  + O(\ell^2) &  \text{ if } \mathbf{x}_{i,j}^{(\ell)} = \mathbf{x}_{2,\bar{\mathbf{x}}}^{(\ell)} \text{ for some } \bar{\mathbf{x}} \in \mathcal{I}_{\emph{cell}}^{(\ell)}.
\end{cases}
\end{aligned}
\end{equation}}
\begin{proof}
Again, each $\mathbf{y}_{i,j}^{(\ell)}$ is defined by averaging points in the cell-wise construction where the gaps between neighboring cells are $O(\ell^3)$ small. Thus, we have $\mathbf{y}_{i,j}^{(\ell)} = \mathbf{y}_{k,\bar{\mathbf{x}}}^{(\ell)} + O(\ell^3)$ for appropriate $k \in \{ 0,1,\ldots,8\}$ for some $\bar{\mathbf{x}} \in \mathcal{I}^{(\ell)}_{\text{cell}}$. We use this to estimate functions which depend on the origami vertices  $\mathbf{y}_{i,j}^{(\ell)}$ by the cell-wise vertices $\mathbf{y}_{k,\bar{\mathbf{x}}}^{(\ell)}$. Eventually, we  obtain an expression for the angles  $\psi^{(\ell)}_{i,j,1}$ and $\psi_{i,j,2}^{(\ell)}$ that is very similar to Eq.\;(\ref{eq:ugly}). The proof then  follows from the previous two steps.   

  Let us focus on the case that $\mathbf{y}_{i,j}^{(\ell)} = \mathbf{y}_{0,\bar{\mathbf{x}}}^{(\ell)}$ for some $\bar{\mathbf{x}}$ in Eq.\;(\ref{eq:TheOriConstruct}).  Recalling the definitions of $\mathbf{n}_{i,j,1}^{(\ell)}$ and $\mathbf{m}_{i,j,1}^{(\ell)}$ in Eq.\;(\ref{eq:getNormals}), we have 
 \begin{equation}
 \begin{aligned}
&\mathbf{n}_{i,j,1}^{(\ell)}= \frac{(\mathbf{y}_{1,\bar{\mathbf{x}}}^{(\ell)} - \mathbf{y}_{0,\bar{\mathbf{x}}}^{(\ell)} + O(\ell^3)) \times (\mathbf{y}_{2,\bar{\mathbf{x}}}^{(\ell)} - \mathbf{y}_{0,\bar{\mathbf{x}}}^{(\ell)} + O(\ell^3))}{|(\mathbf{y}_{1,\bar{\mathbf{x}}}^{(\ell)} - \mathbf{y}_{0,\bar{\mathbf{x}}}^{(\ell)} + O(\ell^3)) \times (\mathbf{y}_{2,\bar{\mathbf{x}}}^{(\ell)} - \mathbf{y}_{0,\bar{\mathbf{x}}}^{(\ell)} + O(\ell^3))|} = \frac{(\mathbf{y}_{1,\bar{\mathbf{x}}}^{(\ell)} - \mathbf{y}_{0,\bar{\mathbf{x}}}^{(\ell)}) \times (\mathbf{y}_{2,\bar{\mathbf{x}}}^{(\ell)} - \mathbf{y}_{0,\bar{\mathbf{x}}}^{(\ell)})}{|(\mathbf{y}_{1,\bar{\mathbf{x}}}^{(\ell)} - \mathbf{y}_{0,\bar{\mathbf{x}}}^{(\ell)}) \times (\mathbf{y}_{2,\bar{\mathbf{x}}}^{(\ell)} - \mathbf{y}_{0,\bar{\mathbf{x}}}^{(\ell)})|}  + O(\ell^2) ,\\
&\mathbf{m}_{i,j,1}^{(\ell)}= \frac{(\mathbf{y}_{2,\bar{\mathbf{x}}}^{(\ell)} - \mathbf{y}_{5,\bar{\mathbf{x}}}^{(\ell)} + O(\ell^3)) \times (\mathbf{y}_{1,\bar{\mathbf{x}}}^{(\ell)} - \mathbf{y}_{5,\bar{\mathbf{x}}}^{(\ell)} + O(\ell^3))}{|(\mathbf{y}_{2,\bar{\mathbf{x}}}^{(\ell)} - \mathbf{y}_{5,\bar{\mathbf{x}}}^{(\ell)} + O(\ell^3)) \times (\mathbf{y}_{1,\bar{\mathbf{x}}}^{(\ell)} - \mathbf{y}_{5,\bar{\mathbf{x}}}^{(\ell)} + O(\ell^3))|} = \frac{(\mathbf{y}_{2,\bar{\mathbf{x}}}^{(\ell)} - \mathbf{y}_{5,\bar{\mathbf{x}}}^{(\ell)}) \times (\mathbf{y}_{1,\bar{\mathbf{x}}}^{(\ell)} - \mathbf{y}_{5,\bar{\mathbf{x}}}^{(\ell)})}{|(\mathbf{y}_{2,\bar{\mathbf{x}}}^{(\ell)} - \mathbf{y}_{5,\bar{\mathbf{x}}}^{(\ell)}) \times (\mathbf{y}_{1,\bar{\mathbf{x}}}^{(\ell)} - \mathbf{y}_{5,\bar{\mathbf{x}}}^{(\ell)})|}  + O(\ell^2), \\
&\frac{\mathbf{y}_{i+1,j}^{(\ell)} - \mathbf{y}_{i,j+1}^{(\ell)}}{|\mathbf{y}_{i+1,j}^{(\ell)} - \mathbf{y}_{i,j+1}^{(\ell)}|}  =  \frac{\mathbf{y}_{1,\bar{\mathbf{x}}}^{(\ell)} - \mathbf{y}_{2,\bar{\mathbf{x}}}^{(\ell)} + O(\ell^3)}{|  \mathbf{y}_{1,\bar{\mathbf{x}}}^{(\ell)} - \mathbf{y}_{2,\bar{\mathbf{x}}}^{(\ell)} + O(\ell^3) |}  =  \frac{\mathbf{y}_{1,\bar{\mathbf{x}}}^{(\ell)} - \mathbf{y}_{2,\bar{\mathbf{x}}}^{(\ell)}}{|  \mathbf{y}_{1,\bar{\mathbf{x}}}^{(\ell)} - \mathbf{y}_{2,\bar{\mathbf{x}}}^{(\ell)}|}  + O(\ell^2).
 \end{aligned}
 \end{equation}
Thus, the angle $\psi_{i,j,1}^{(\ell)}$ defined in Eq.\;(\ref{eq:getBendAngles}) satisfies 
 \begin{equation}
\begin{aligned}\label{eq:firstCalcPsi}
\psi_{i,j,1}^{(\ell)} = \arcsin \Big( \Big[ \tfrac{\mathbf{y}_{1,\bar{\mathbf{x}}}^{(\ell)} - \mathbf{y}_{2,\bar{\mathbf{x}}}^{(\ell)}}{|  \mathbf{y}_{1,\bar{\mathbf{x}}}^{(\ell)} - \mathbf{y}_{2,\bar{\mathbf{x}}}^{(\ell)}|}  \Big] \cdot \Big\{   \Big[  \tfrac{(\mathbf{y}_{1,\bar{\mathbf{x}}}^{(\ell)} - \mathbf{y}_{0,\bar{\mathbf{x}}}^{(\ell)}) \times (\mathbf{y}_{2,\bar{\mathbf{x}}}^{(\ell)} - \mathbf{y}_{0,\bar{\mathbf{x}}}^{(\ell)})}{|(\mathbf{y}_{1,\bar{\mathbf{x}}}^{(\ell)} - \mathbf{y}_{0,\bar{\mathbf{x}}}^{(\ell)}) \times (\mathbf{y}_{2,\bar{\mathbf{x}}}^{(\ell)} - \mathbf{y}_{0,\bar{\mathbf{x}}}^{(\ell)})|} \Big] \times  \Big[\tfrac{(\mathbf{y}_{2,\bar{\mathbf{x}}}^{(\ell)} - \mathbf{y}_{5,\bar{\mathbf{x}}}^{(\ell)}) \times (\mathbf{y}_{1,\bar{\mathbf{x}}}^{(\ell)} - \mathbf{y}_{5,\bar{\mathbf{x}}}^{(\ell)})}{|(\mathbf{y}_{2,\bar{\mathbf{x}}}^{(\ell)} - \mathbf{y}_{5,\bar{\mathbf{x}}}^{(\ell)}) \times (\mathbf{y}_{1,\bar{\mathbf{x}}}^{(\ell)} - \mathbf{y}_{5,\bar{\mathbf{x}}}^{(\ell)})|}  \Big]  \Big\} + O(\ell^2)  \Big) .
\end{aligned}
 \end{equation}
 In this formula, each $\mathbf{y}_{i,\bar{\mathbf{x}}}^{(\ell)}$ is akin to $\mathbf{y}_i^{(\ell)}$ in Fig.\;\ref{Fig:CalcBending} and Eq.\;(\ref{eq:ugly}) with $\mathbf{R}^{(\ell)} = \mathbf{R}_{\text{eff}}(\bar{\mathbf{x}}) \big[ \mathbf{I} + \ell ( \boldsymbol{\omega}(\bar{\mathbf{x}}) \times ) \big]$, $\mathbf{t}_i^{(\ell)} = \mathbf{t}_i(\theta(\bar{\mathbf{x}}) + \ell \xi(\bar{\mathbf{x}}))$ and $\kappa_i = \kappa_i(\bar{\mathbf{x}})$. Consequently,  Eq.\;(\ref{eq:firstCalcPsi}) simplifies to 
 \begin{equation}
 \begin{aligned}
 \psi_{i,j,1}^{(\ell)} &=\arcsin \Big( \ell \kappa_1(\bar{\mathbf{x}})  \frac{|(\mathbf{t}_1(\theta(\bar{\mathbf{x}})) - \mathbf{t}_2(\theta (\bar{\mathbf{x}}))) \times \mathbf{n}_{12}(\theta(\bar{\mathbf{x}}))|^2}{|\mathbf{t}_1(\theta(\bar{\mathbf{x}})) - \mathbf{t}_2(\theta(\bar{\mathbf{x}}))| |\mathbf{n}_{12}(\theta(\bar{\mathbf{x}}))|^2 } + O(\ell^2)  \Big)  \\
 &=  \ell \kappa_1(\bar{\mathbf{x}})  \frac{|(\mathbf{t}_1(\theta(\bar{\mathbf{x}})) - \mathbf{t}_2(\theta (\bar{\mathbf{x}}))) \times \mathbf{n}_{12}(\theta(\bar{\mathbf{x}}))|^2}{|\mathbf{t}_1(\theta(\bar{\mathbf{x}})) - \mathbf{t}_2(\theta(\bar{\mathbf{x}}))| |\mathbf{n}_{12}(\theta(\bar{\mathbf{x}}))|^2 } + O(\ell^2) = \ell b_{12}^{-} \kappa_1(\bar{\mathbf{x}}) + O(\ell^2)  
 \end{aligned}
 \end{equation}
 due to the results of the previous two steps and since $\sin (\beta) \approx \beta$ to leading order. The other statements in Eq.\;(\ref{eq:psiResult1}-\ref{eq:psiResult2}) follow in an analogous fashion. 
\end{proof}

At this point, we have proved everything in Theorem~\ref{MainTheorem} up to Eq.\;(\ref{eq:MainLim}). The rest of the proof deals with the bar and hinge energy   $\mathcal{E}^{(\ell)}_{\text{tot}} ( \{ \mathbf{y}_{i,j}^{(\ell)} \}) =  \mathcal{E}_{\text{str}}^{(\ell)} ( \{ \mathbf{y}_{i,j}^{(\ell)} \}) + \mathcal{E}_{\text{bend}}^{(\ell)} ( \{ \mathbf{y}_{i,j}^{(\ell)} \}) + \mathcal{E}_{\text{fold}}^{(\ell)} ( \{ \mathbf{y}_{i,j}^{(\ell)} \})$ from Section \ref{ssec:BarHinge}.

\noindent  \textbf{Step 6.}\;\textit{The origami deformations  $\{ \mathbf{y}_{i,j}^{(\ell)} \}$ satisfy $\lim_{\ell \rightarrow 0} \mathcal{E}_{\emph{str}}^{(\ell)} ( \{ \mathbf{y}_{i,j}^{(\ell)} \}) = 0$.}
\begin{proof}
Step 1 proved $|\varepsilon^{(\ell)}_{i,j,1}|, \ldots, |\varepsilon_{i,j,4}^{(\ell)}|  = O(\ell^2)$ for all $(i,j) \in \mathcal{I}^{(\ell)}$. Thus, 
\begin{equation}
\begin{aligned}\label{eq:preLimStrEnergy}
\mathcal{E}_{\text{str}}^{(\ell)}( \{ \mathbf{y}_{i,j}^{(\ell)}\}) &=   \sum_{(i,j) \in \mathcal{I}^{(\ell)}}  A_{i,j}^{(\ell)}   \ell^{\alpha_s}   \Phi^{\text{str}}_{i,j}( \varepsilon^{(\ell)}_{i,j,1}, \varepsilon^{(\ell)}_{i,j,2},\varepsilon^{(\ell)}_{i,j,3},\varepsilon^{(\ell)}_{i,j,4}) =  \sum_{(i,j) \in \mathcal{I}^{(\ell)}}  O(\ell^2) \ell^{\alpha_s}  \Phi^{\text{str}}_{i,j}( O(\ell^2) ) \\
&=   \sum_{(i,j) \in \mathcal{I}^{(\ell)}}  O(\ell^{2 + \alpha_s} ) \big[  O(\ell^2) \cdot  DD \Phi^{\text{str}}_{i,j} ( \mathbf{0}) O(\ell^2)  \big] =   \sum_{(i,j) \in \mathcal{I}^{(\ell)}} O(\ell^{6 + \alpha_s} ) ,
\end{aligned}
\end{equation}
where the strain energy $ \Phi^{\text{str}}_{i,j}(\cdot)$ has been Taylor expanded using that it is  smooth and minimized on  $(\varepsilon^{(\ell)}_{i,j,1}, \varepsilon^{(\ell)}_{i,j,2},\varepsilon^{(\ell)}_{i,j,3},\varepsilon^{(\ell)}_{i,j,4}) = \mathbf{0}$. Next, the set $\mathcal{I}^{(\ell)}$ contains $\sim 1/\ell^2$ points, since the characteristic length of $\Omega$ is $\sim 1$. It follows that  $\mathcal{E}_{\text{str}}^{(\ell)}( \{ \mathbf{y}_{i,j}^{(\ell)}\}) =  O(\ell^{4 + \alpha_s})$.  The stretching energy vanishes, since $\alpha_s \in (-4, -2)$ by the hypothesis  in Eq.\;(\ref{eq:alphaAssumptions}).
\end{proof}

\noindent  \textbf{Step 7.}\;\textit{The origami deformations  $\{ \mathbf{y}_{i,j}^{(\ell)} \}$ satisfy  $\lim_{\ell \rightarrow 0} \mathcal{E}_{\emph{fold}}^{(\ell)} ( \{ \mathbf{y}_{i,j}^{(\ell)} \}) = 0$.}
\begin{proof}
By their definitions in Eqs.(\ref{eq:getBetas}-\ref{eq:getGammas}), $\beta_{i,j,0}^{(\ell)}, \beta_{i,j}^{(\ell)}, \gamma_{i,j,0}^{(\ell)}, \gamma_{i,j}^{(\ell)}$ take values in the compact set $[-\pi, \pi]$.  As such, $\Phi_{i,j}^{\text{fold}}( \beta_{i,j}^{(\ell)} - \beta_{i,j,0}^{(\ell)},   \gamma_{i,j}^{(\ell)} - \gamma_{i,j,0}^{(\ell)}) \leq \max_{\beta, \gamma \in [-2\pi, 2\pi]}  \Phi_{i,j}^{\text{fold}} ( \beta, \gamma) = : M_{i,j}$. Also, $\Phi_{i,j}^{\text{fold}} = \Phi_{i+2,j}^{\text{fold}} = \Phi_{i, j+2}^{\text{fold}}$  for all $(i,j) \in \mathcal{I}^{(\ell)}$, leading to $M_{i,j} = M_{i+2, j} = M_{i,j+2}$ for all  $(i,j) \in \mathcal{I}^{(\ell)}$. Thus,  there is an absolute constant $M$ such that $M_{i,j} \leq M$ for all $(i,j) \in \mathcal{I}^{(\ell)}$. It follows that 
\begin{equation}
\begin{aligned}
\mathcal{E}_{\text{fold}}^{(\ell)} ( \{ \mathbf{y}_{i,j}^{(\ell)} \}) =  \sum_{(i,j) \in \mathcal{I}^{(\ell)}}  A_{i,j}^{(\ell)}  \ell^{\alpha_f}  \Phi_{i,j}^{\text{fold}} \big( \beta_{i,j}^{(\ell)} - \beta_{i,j,0}^{(\ell)}, \gamma_{i,j}^{(\ell)} -\gamma_{i,j,0}^{(\ell)} \big)  \leq M  \ell^{\alpha_f}   \sum_{(i,j) \in \mathcal{I}^{(\ell)}}  A_{i,j}^{(\ell)} \leq C\ell^{\alpha_f}  
\end{aligned}
\end{equation}
since  each $A_{i,j}^{(\ell)} \sim \ell^2$ and $ \mathcal{I}^{(\ell)}$ contains $\sim 1/\ell^2$ points. The result follows as $\alpha_f >0$ is assumed in Eq.\;(\ref{eq:alphaAssumptions}).
\end{proof}

\noindent  \textbf{Step 8.}\;\textit{The origami deformations  $\{ \mathbf{y}_{i,j}^{(\ell)} \}$ satisfy  
\begin{equation}
\begin{aligned}\label{eq:LimBendEnergy}
\lim_{\ell \rightarrow 0} \mathcal{E}_{\emph{bend}}^{(\ell)} ( \{ \mathbf{y}_{i,j}^{(\ell)} \}) = \int_{\Omega}  \bigg\{\begin{pmatrix} \boldsymbol{\omega}_{\mathbf{u}_0}(\mathbf{x}) \\ \partial_{\mathbf{u}_0} \theta(\mathbf{x}) \end{pmatrix} \cdot  \mathbf{K}_{\mathbf{u}_{0}}(\theta(\mathbf{x}) ) \begin{pmatrix} \boldsymbol{\omega}_{\mathbf{u}_0}(\mathbf{x}) \\ \partial_{\mathbf{u}_0} \theta(\mathbf{x}) \end{pmatrix}  +  \begin{pmatrix} \boldsymbol{\omega}_{\mathbf{v}_0}(\mathbf{x}) \\ \partial_{\mathbf{v}_0} \theta(\mathbf{x}) \end{pmatrix} \cdot  \mathbf{K}_{\mathbf{v}_{0}}(\theta(\mathbf{x}) ) \begin{pmatrix} \boldsymbol{\omega}_{\mathbf{v}_0}(\mathbf{x}) \\ \partial_{\mathbf{v}_0} \theta(\mathbf{x}) \end{pmatrix} \bigg\} dA 
\end{aligned}
\end{equation}
for $\mathbf{K}_{\mathbf{u}_0}(\theta)$ and $\mathbf{K}_{\mathbf{v}_0}(\theta)$ in Eq.\;(\ref{eq:stiffnessTensor}).}
\begin{proof}
Since there are four panels per unit cell, $ \mathcal{E}_{\text{bend}}^{(\ell)} ( \{ \mathbf{y}_{i,j}^{(\ell)} \})$ defined panel-by-panel in Eq.\;(\ref{eq:eBend}) can be written in a cell-by-cell sum as
\begin{equation}
\begin{aligned}
\mathcal{E}_{\text{bend}}^{(\ell)}(\{ \mathbf{y}_{i,j}^{(\ell)} \} )= \sum_{\bar{\mathbf{x}} \equiv  \ell (i \tilde{\mathbf{u}}_0 + j \tilde{\mathbf{v}}_0) \in \mathcal{I}_{\text{cell}}^{(\ell)}}    & \Big\{  A_{i,j}^{(\ell)}   \ell^{\alpha_b}   \Phi^{\text{bend}}_{i,j}(\psi_{i,j,1}^{(\ell)}, \psi_{i,j,2}^{(\ell)})  + A_{i-1,j-1}^{(\ell)}   \ell^{\alpha_b}   \Phi^{\text{bend}}_{i-1,j-1}(\psi_{i-1,j-1,1}^{(\ell)}, \psi_{i-1,j-1,2}^{(\ell)})  \\
&+ A_{i,j-1}^{(\ell)}   \ell^{\alpha_b}   \Phi^{\text{bend}}_{i,j-1}(\psi_{i,j-1,1}^{(\ell)}, \psi_{i,j-1,2}^{(\ell)})   +  A_{i-1,j}^{(\ell)}   \ell^{\alpha_b}   \Phi^{\text{bend}}_{i-1,j}(\psi_{i-1,j,1}^{(\ell)}, \psi_{i-1,j,2}^{(\ell)}) \Big\}. 
\end{aligned}
\end{equation}
This uses the definition of  $\mathcal{I}^{(\ell)}_{\text{cell}}$ at the start of the cell-based construction in Section \ref{ssec:GlobalOrigami}. Note that  the four panel areas $A_{i,j}^{(\ell)}$  are  $\ell^2 |\mathbf{t}_k^r \times \mathbf{t}_{l}^r|$ for  $kl \in \{12,23,34,41\}$. By Step 5,
\begin{equation}
\begin{aligned}\label{eq:EbendUgly}
\mathcal{E}_{\text{bend}}^{(\ell)}(\{ \mathbf{y}_{i,j}^{(\ell)} \} )  = \sum_{\bar{\mathbf{x}} \in  \mathcal{I}_{\text{cell}}^{(\ell)}} &\bigg\{ \ell^{2+\alpha_b} |\mathbf{t}_1^r \times \mathbf{t}_2^r| \Phi^{\text{bend}}_{\mathcal{P}_1}(\ell b_{12}^{-} \kappa_{1}(\bar{\mathbf{x}}) + O(\ell^2), \ell b_{12}^{+} \kappa_{1}(\bar{\mathbf{x}}) + O(\ell^2))  \\
&\quad + \ell^{2+\alpha_b} |\mathbf{t}_2^r \times \mathbf{t}_3^r| \Phi^{\text{bend}}_{\mathcal{P}_2}(\ell b_{23}^{+} \kappa_{2}(\bar{\mathbf{x}}) + O(\ell^2), \ell b_{23}^{-} \kappa_{2}(\bar{\mathbf{x}}) + O(\ell^2)) \\
&\quad + \ell^{2+\alpha_b} |\mathbf{t}_3^r \times \mathbf{t}_4^r| \Phi^{\text{bend}}_{\mathcal{P}_3}(\ell b_{34}^{-} \kappa_{3}(\bar{\mathbf{x}}) + O(\ell^2), \ell b_{34}^{+} \kappa_{3}(\bar{\mathbf{x}}) + O(\ell^2)) \\
&\quad +  \ell^{2+\alpha_b} |\mathbf{t}_4^r \times \mathbf{t}_1^r| \Phi^{\text{bend}}_{\mathcal{P}_4}(\ell b_{41}^{+} \kappa_{4}(\bar{\mathbf{x}}) + O(\ell^2), \ell b_{41}^{-} \kappa_{4}(\bar{\mathbf{x}}) + O(\ell^2)) \bigg\}  
\end{aligned}
\end{equation}
where we have used the periodicity condition $\Phi^{\text{bend}}_{i,j} = \Phi^{\text{bend}}_{i,j+2} =  \Phi^{\text{bend}}_{i+2,j}$ for all $(i,j) \in \mathcal{I}^{(\ell)}$ to label the energy densities according to the four types of panels being bent in Eq.\;(\ref{eq:OmegaCell}). Focusing on the first term in Eq.\;(\ref{eq:EbendUgly}), a Taylor expansion gives that 
\begin{equation}
    \begin{aligned}
         \Phi^{\text{bend}}_{\mathcal{P}_1}(\ell b_{12}^{-} \kappa_{1}(\bar{\mathbf{x}}) + O(\ell^2), \ell b_{12}^{+} \kappa_{1}(\bar{\mathbf{x}}) + O(\ell^2)) &= \frac{\ell^2}{2}   \begin{pmatrix} b_{12}^{-} \\ b_{12}^+ \end{pmatrix}  \cdot DD \Phi_{\mathcal{P}_1}^{\text{bend}}( \mathbf{0} ) \begin{pmatrix} b_{12}^{-} \\ b_{12}^+ \end{pmatrix} \big(\kappa_1(\bar{\mathbf{x}})\big)^2 + O(\ell^3) \\
         &= \ell^2 \frac{|\mathbf{u}_0 \times \mathbf{v}_0|}{|\mathbf{t}_1^r \times \mathbf{t}_2^r|}   b_1  \big(\kappa_1(\bar{\mathbf{x}})\big)^2 + O(\ell^3)
    \end{aligned}
\end{equation}
since $\Phi^{\text{bend}}_{\mathcal{P}_1}$ is smooth with minimum $\Phi^{\text{bend}}_{\mathcal{P}_1}(\mathbf{0})  =\mathbf{0}$. The second equality follows by matching the definitions for $b_{12}^{\pm}$ in Eq.\;(\ref{eq:bijpm}) to $b_1$ in Eq.\;(\ref{eq:bendingModuli}). Analogous Taylor expansions of each term in Eq.\;(\ref{eq:EbendUgly}) yield
\begin{equation}
\begin{aligned}\label{eq:eBendLessUgly}
\mathcal{E}_{\text{bend}}^{(\ell)}(\{ \mathbf{y}_{i,j}^{(\ell)} \} )  &= \sum_{\bar{\mathbf{x}} \in  \mathcal{I}_{\text{cell}}^{(\ell)}} \ell^{4+ \alpha_b} |\mathbf{u}_0 \times \mathbf{v}_0|  \Big\{ \sum_{k =1,\ldots,4}   b_k (\kappa_k(\bar{\mathbf{x}}))^2 + O(\ell)\Big\}.
\end{aligned}
\end{equation}

We now take $\ell \rightarrow 0$  in Eq.\;(\ref{eq:eBendLessUgly}). Note  $\alpha_b = -2$ by assumption, each    $\kappa_1(\mathbf{x}), \ldots, \kappa_4(\mathbf{x})$  is smooth on $\Omega$ and  $\lim_{\ell \rightarrow 0} \sum_{\bar{\mathbf{x}} \in  \mathcal{I}_{\text{cell}}^{(\ell)}} \ell^{2} |\mathbf{u}_0 \times \mathbf{v}_0|  = \Omega$. Thus, 
\begin{equation}
\begin{aligned}
\lim_{\ell \rightarrow 0} \mathcal{E}_{\text{bend}}^{(\ell)}(\{ \mathbf{y}_{i,j}^{(\ell)} \} )  =\int_{\Omega} \Big\{  b_1 (\kappa_1(\mathbf{x}))^2 +  b_2 (\kappa_2(\mathbf{x}))^2 +  b_3 (\kappa_3(\mathbf{x}))^2 +  b_4 (\kappa_4(\mathbf{x}))^2 \Big\}  dA .
\end{aligned}
\end{equation} 
Substituting the formulas in Eq.\;(\ref{eq:theKappaFields}) for the panel curvature fields $\kappa_1(\mathbf{x}), \ldots, \kappa_4(\mathbf{x})$ and recalling the definition of $\mathbf{B}(\theta)$ proves the result.
\end{proof}

This completes the proof of Theorem~\ref{MainTheorem}.

\section{Discussion, examples and outlook} \label{sec:Examples}

We end by discussing some of the most intriguing aspects of our effective plate theory for origami, and by illustrating these aspects in examples. Section \ref{ssec:discussSec-modes} starts by rewriting the surface theory for $\theta(\mathbf{x})$, $\boldsymbol{\omega}_{\mathbf{u}_0}(\mathbf{x})$ and $\boldsymbol{\omega}_{\mathbf{v}_0}(\mathbf{x})$ into one involving more standard quantities such as the first and second fundamental forms of the effective deformation $\mathbf{y}_\text{eff}(\mathbf{x})$. Along the way, we identify the three basic soft modes of deformation available to parallelogram origami (actuation, bending and twisting modes) and provide variables to parameterize general deformations in terms of these modes. 
Section \ref{ssec:Discussion2} goes on to discuss our effective theory in the context of more familiar continuum theories, namely Kirchhoff's plate theory \cite{friesecke2002theorem} and models of generalized elastic continua \cite{eringen2012microcontinuum}. Section \ref{ssec:Examples} applies our theory to the examples of Miura and Eggbox origami. Finally, Section \ref{ssec:Outlook} ends  with a brief commentary on future directions.

\subsection{Reformulation of the effective surface theory}\label{ssec:discussSec-modes}

So far, we have described the effective deformations of parallelogram origami through field variables  $\theta(\mathbf{x})$, $\boldsymbol{\omega}_{\mathbf{u}_0}(\mathbf{x})$ and  $\boldsymbol{\omega}_{\mathbf{v}_0}(\mathbf{x})$ that solve Eq.\;(\ref{eq:SurfaceTheory}). As a quick recap, this equation is composed of the angle restriction $\theta(\mathbf{x}) \in (\theta^-, \theta^+)$, the algebraic constraints 
\begin{equation}
    \begin{aligned}\label{eq:algConstraints}
        &\boldsymbol{\omega}_{\mathbf{u}_0}(\mathbf{x}) \times \mathbf{v}(\theta(\mathbf{x}))  + \partial_{\mathbf{u}_0} \theta(\mathbf{x}) \mathbf{v}'(\theta(\mathbf{x}))   =  \boldsymbol{\omega}_{\mathbf{v}_0}(\mathbf{x}) \times \mathbf{u}(\theta(\mathbf{x}))  + \partial_{\mathbf{v}_0} \theta(\mathbf{x}) \mathbf{u}'(\theta(\mathbf{x})), \\
&\boldsymbol{\omega}_{\mathbf{u}_0}(\mathbf{x}) \cdot \mathbf{v}'(\theta(\mathbf{x})) = \boldsymbol{\omega}_{\mathbf{v}_0} (\mathbf{x}) \cdot \mathbf{u}'(\theta(\mathbf{x}))
    \end{aligned}
\end{equation}
and the compatibility condition  $\partial_{\mathbf{v}_0} \bm{\omega}_{\mathbf{u}_0}(\mathbf{x})-\partial_{\mathbf{u}_0} \bm{\omega}_{\mathbf{v}_0}(\mathbf{x})= \bm{\omega}_{\mathbf{u}_0}(\mathbf{x})\times \bm{\omega}_{\mathbf{v}_0}(\mathbf{x})$.

\paragraph{Actuation, bend and twist.} We now show how to write the effective theory in a way that emphasizes the three natural modes of deformation for parallelogram origami: actuation, bend, and twist. First, observe that  Remark~\ref{EquivRemark} allows us to explicitly solve for the algebraic constraints in Eq.\;(\ref{eq:algConstraints}). To do so, we expand the vector fields $\boldsymbol{\omega}_{\mathbf{u}_0}(\mathbf{x})$ and $\boldsymbol{\omega}_{\mathbf{v}_0}(\mathbf{x})$ in the basis $\{ \mathbf{u}(\theta(\mathbf{x})), \mathbf{v}(\theta(\mathbf{x})), \mathbf{e}_3\}$ and choose the components of this expansion using the remark. This gives
\begin{equation}
\begin{aligned}\label{eq:omegaParamsSolveAlg}
&\boldsymbol{\omega}_{\mathbf{u}_0}(\mathbf{x}) = \tau(\mathbf{x})  \mathbf{u}(\theta(\mathbf{x})) + \kappa(\mathbf{x}) [\mathbf{u}'(\theta(\mathbf{x})) \cdot \mathbf{u}(\theta(\mathbf{x}))] \mathbf{v}(\theta(\mathbf{x})) + \big[  \boldsymbol{\Gamma}_{\mathbf{u}_0}(\theta(\mathbf{x})) \cdot \nabla_0 \theta(\mathbf{x})  \big]\mathbf{e}_3, \\
&\boldsymbol{\omega}_{\mathbf{v}_0}(\mathbf{x}) = \kappa(\mathbf{x}) [\mathbf{v}'(\theta(\mathbf{x})) \cdot \mathbf{v}(\theta(\mathbf{x}))]  \mathbf{u}(\theta(\mathbf{x})) - \tau(\mathbf{x})   \mathbf{v}(\theta(\mathbf{x})) +\big[\boldsymbol{\Gamma}_{\mathbf{v}_0}(\theta(\mathbf{x})) \cdot \nabla_0 \theta(\mathbf{x})  \big] \mathbf{e}_3
\end{aligned}
\end{equation}
for some $\kappa, \tau \colon \Omega \rightarrow \mathbb{R}$, and where
\begin{equation}
\begin{aligned}
\boldsymbol{\Gamma}_{\mathbf{u}_0}(\theta) := \frac{[\mathbf{u}(\theta) \cdot \mathbf{v}'(\theta) ]\tilde{\mathbf{e}}_1 - [\mathbf{u}(\theta) \cdot \mathbf{u}'(\theta)] \tilde{\mathbf{e}}_2}{\mathbf{e}_3 \cdot (\mathbf{u}(\theta) \times \mathbf{v}(\theta))},  \quad 
\boldsymbol{\Gamma}_{\mathbf{v}_0}(\theta) := \frac{[\mathbf{v}(\theta) \cdot \mathbf{v}'(\theta) ]\tilde{\mathbf{e}}_1 - [\mathbf{v}(\theta) \cdot \mathbf{u}'(\theta)] \tilde{\mathbf{e}}_2}{\mathbf{e}_3 \cdot (\mathbf{u}(\theta) \times \mathbf{v}(\theta))}
\end{aligned}
\end{equation}
with $\{\tilde{\mathbf{e}}_1, \tilde{\mathbf{e}}_2\}$ as the standard basis of $\mathbb{R}^2$. The gradient $\nabla_0$ is defined as $\nabla_0 \theta(\mathbf{x}) := \partial_{\mathbf{u}_0} \theta(\mathbf{x}) \tilde{\mathbf{e}}_1 + \partial_{\mathbf{v}_0} \theta (\mathbf{x}) \tilde{\mathbf{e}}_2$. 

Substituting Eq.\;(\ref{eq:omegaParamsSolveAlg}) into the remaining conditions of the effective surface theory leads to a system of three scalar PDEs in $\theta(\mathbf{x})$, $\kappa(\mathbf{x})$ and $\tau(\mathbf{x})$ equivalent to the original ones:
\begin{equation}
\begin{aligned}\label{eq:GenPDESurf}
&\underline{\text{Equivalent formulation of the effective surface theory for parallelogram origami:}}\\
&\begin{cases}
\theta(\mathbf{x}) \in (\theta^{-}, \theta^+), \\
\nabla_0^{\perp} \cdot \big[ \boldsymbol{\Gamma}( \theta(\mathbf{x})) \nabla_0 \theta(\mathbf{x}) \big]  = \begin{pmatrix} \kappa(\mathbf{x})  \\ \tau(\mathbf{x}) \end{pmatrix} \cdot \boldsymbol{\Lambda}(\theta(\mathbf{x}))   \begin{pmatrix} \kappa(\mathbf{x})  \\ \tau(\mathbf{x}) \end{pmatrix}, \\
\Big[ \mathbf{M}_{\mathbf{u}_0}(\theta(\mathbf{x}) ) \partial_{\mathbf{u}_0} - \mathbf{M}_{\mathbf{v}_0}(\theta(\mathbf{x}) ) \partial_{\mathbf{v}_0} + \mathbf{M}(\theta(\mathbf{x}), \nabla_0\theta(\mathbf{x}))   \Big]  \begin{pmatrix} \kappa(\mathbf{x})  \\ \tau(\mathbf{x}) \end{pmatrix}  = \mathbf{0}.
\end{cases}
\end{aligned}
\end{equation}  
In the standard Cartesian basis, the purely $\theta$-dependent  $\mathbb{R}^{2\times2}$ tensors in these expressions are 
\begin{equation}
\begin{aligned}
&\boldsymbol{\Gamma}(\theta) := \frac{1}{A(\theta)} \begin{pmatrix} \mathbf{u}(\theta) \cdot \mathbf{v}'(\theta) & - \mathbf{u}(\theta) \cdot \mathbf{u}'(\theta) \\ \mathbf{v}(\theta) \cdot \mathbf{v}'(\theta) & -\mathbf{v}(\theta) \cdot \mathbf{u}'(\theta) \end{pmatrix},  &&  \boldsymbol{\Lambda}(\theta)  := -A(\theta) \begin{pmatrix} \big[\mathbf{u}'(\theta) \cdot \mathbf{u}(\theta) \big] \big[ \mathbf{v}'(\theta) \cdot \mathbf{v}(\theta) \big] & 0 \\ 0 & 1 \end{pmatrix}, \\
&\mathbf{M}_{\mathbf{u}_0}(\theta) := \begin{pmatrix}
[\mathbf{v}'(\theta) \cdot \mathbf{v}(\theta) ] [ \mathbf{u}(\theta) \cdot \mathbf{e}_1] & -\mathbf{v}(\theta) \cdot \mathbf{e}_1 \\ [\mathbf{v}'(\theta) \cdot \mathbf{v}(\theta) ] [\mathbf{u}(\theta) \cdot \mathbf{e}_2]  & -\mathbf{v}(\theta) \cdot \mathbf{e}_2
\end{pmatrix},  && \mathbf{M}_{\mathbf{v}_0}(\theta)  := \begin{pmatrix}
[\mathbf{u}'(\theta) \cdot \mathbf{u}(\theta) ] [ \mathbf{v}(\theta) \cdot \mathbf{e}_1] & \mathbf{u}(\theta) \cdot \mathbf{e}_1 \\ [\mathbf{u}'(\theta) \cdot \mathbf{u}(\theta) ] [\mathbf{v}(\theta) \cdot \mathbf{e}_2]  & \mathbf{u}(\theta) \cdot \mathbf{e}_2
\end{pmatrix},
\end{aligned}
\end{equation}
where $A(\theta) = \mathbf{e}_3 \cdot (\mathbf{u}(\theta) \times \mathbf{v}(\theta))$ is the area of the parallelogram formed by $\mathbf{u}(\theta)$ and $\mathbf{v}(\theta)$. 
The last tensor $\mathbf{M}(\theta, \tilde{\mathbf{g}}) \in \mathbb{R}^{2\times2}$ in Eq.\;(\ref{eq:GenPDESurf}) has a lengthy expression, whose coefficients in the standard basis are
\begin{equation}
\begin{aligned}
M_{11}(\theta, \tilde{\mathbf{g}}) & := -[\tilde{\mathbf{e}}_1 \cdot \boldsymbol{\Gamma}(\theta) \tilde{\mathbf{g}}]  [ \mathbf{v}'(\theta)  \cdot \mathbf{v}(\theta)] [\mathbf{u}(\theta) \cdot \mathbf{e}_2 ] + [\tilde{\mathbf{e}}_2 \cdot \boldsymbol{\Gamma}(\theta)\tilde{\mathbf{g}}]  [ \mathbf{u}'(\theta)  \cdot \mathbf{u}(\theta)] [\mathbf{v}(\theta) \cdot \mathbf{e}_2 ] \\
&\qquad  + (\tilde{\mathbf{g}} \cdot \tilde{\mathbf{e}}_1) \big[\tilde{\mathbf{e}}_1 \cdot \mathbf{M}_{\mathbf{u}_0}'(\theta) \tilde{\mathbf{e}}_1 \big]  - (\tilde{\mathbf{g}} \cdot \tilde{\mathbf{e}}_2) \big[\tilde{\mathbf{e}}_1 \cdot \mathbf{M}_{\mathbf{v}_0}'(\theta) \tilde{\mathbf{e}}_1 \big],  \\
M_{21}(\theta, \tilde{\mathbf{g}}) &:=  [\tilde{\mathbf{e}}_1 \cdot \boldsymbol{\Gamma}(\theta) \tilde{\mathbf{g}}]  [ \mathbf{v}'(\theta)  \cdot \mathbf{v}(\theta)] [\mathbf{u}(\theta) \cdot \mathbf{e}_1 ] - [\tilde{\mathbf{e}}_2 \cdot \boldsymbol{\Gamma}(\theta) \tilde{\mathbf{g}}]  [ \mathbf{u}'(\theta)  \cdot \mathbf{u}(\theta)] [\mathbf{v}(\theta) \cdot \mathbf{e}_1 ] \\
&\qquad  + (\tilde{\mathbf{g}} \cdot \tilde{\mathbf{e}}_1) \big[\tilde{\mathbf{e}}_2 \cdot \mathbf{M}_{\mathbf{u}_0}'(\theta) \tilde{\mathbf{e}}_1 \big]  - (\tilde{\mathbf{g}} \cdot \tilde{\mathbf{e}}_2) \big[\tilde{\mathbf{e}}_2 \cdot \mathbf{M}_{\mathbf{v}_0}'(\theta) \tilde{\mathbf{e}}_1 \big],  \\
M_{12}(\theta, \tilde{\mathbf{g}}) &:=  [\tilde{\mathbf{e}}_1 \cdot \boldsymbol{\Gamma}(\theta) \tilde{\mathbf{g}}]  [\mathbf{v}(\theta) \cdot \mathbf{e}_2 ]  + [\tilde{\mathbf{e}}_2 \cdot \boldsymbol{\Gamma}(\theta)\tilde{\mathbf{g}}]  [\mathbf{u}(\theta) \cdot \mathbf{e}_2 ] ,  \\
&\qquad  + (\tilde{\mathbf{g}} \cdot \tilde{\mathbf{e}}_1) \big[\tilde{\mathbf{e}}_1 \cdot \mathbf{M}_{\mathbf{u}_0}'(\theta) \tilde{\mathbf{e}}_2 \big]  - (\tilde{\mathbf{g}} \cdot \tilde{\mathbf{e}}_2) \big[\tilde{\mathbf{e}}_1 \cdot \mathbf{M}_{\mathbf{v}_0}'(\theta) \tilde{\mathbf{e}}_2 \big],  \\
M_{22}(\theta, \tilde{\mathbf{g}}) &:= -[\tilde{\mathbf{e}}_1 \cdot \boldsymbol{\Gamma}(\theta) \tilde{\mathbf{g}}]  [\mathbf{v}(\theta) \cdot \mathbf{e}_1 ]  -  [\tilde{\mathbf{e}}_2 \cdot \boldsymbol{\Gamma}(\theta) \tilde{\mathbf{g}}]  [\mathbf{u}(\theta) \cdot \mathbf{e}_1 ] \\ 
&\qquad  + (\tilde{\mathbf{g}} \cdot \tilde{\mathbf{e}}_1) \big[\tilde{\mathbf{e}}_2 \cdot \mathbf{M}_{\mathbf{u}_0}'(\theta) \tilde{\mathbf{e}}_2 \big]  - (\tilde{\mathbf{g}} \cdot \tilde{\mathbf{e}}_2) \big[\tilde{\mathbf{e}}_2 \cdot \mathbf{M}_{\mathbf{v}_0}'(\theta) \tilde{\mathbf{e}}_2 \big].
\end{aligned}
\end{equation}
Note $\mathbf{M}(\theta, \tilde{\mathbf{g}})$ is linear in $\tilde{\mathbf{g}}$ for each fixed $\theta$. The operator $\nabla_0^{\perp} \cdot $ in Eq.\;(\ref{eq:GenPDESurf})  satisfies  $\nabla_0^{\perp} \cdot \big[ \boldsymbol{\Gamma}(\theta(\mathbf{x})) \nabla_0 \theta(\mathbf{x}) \big]$ $:= \partial_{\mathbf{v}_0} \big[\boldsymbol{\Gamma}(\theta(\mathbf{x})) \nabla_0 \theta(\mathbf{x}) \big]\cdot \tilde{\mathbf{e}}_1 -  \partial_{\mathbf{u}_0} \big[\boldsymbol{\Gamma}(\theta(\mathbf{x})) \nabla_0 \theta(\mathbf{x})\big] \cdot \tilde{\mathbf{e}}_2$. 

Given a solution $(\theta(\mathbf{x}), \kappa(\mathbf{x}), \tau(\mathbf{x}))$ of Eq.\;(\ref{eq:GenPDESurf}), the effective deformation  of the origami can be recovered by defining $\boldsymbol{\omega}_{\mathbf{u}_0}(\mathbf{x})$ and $\boldsymbol{\omega}_{\mathbf{v}_0}(\mathbf{x})$ via Eq.\;(\ref{eq:omegaParamsSolveAlg}), and solving for $\mathbf{R}_{\text{eff}}(\mathbf{x})$ and $\mathbf{y}_{\text{eff}}(\mathbf{x})$ 
via Proposition \ref{firstProp}.

\paragraph{The fundamental forms of parallelogram origami.}  As in the original formulation of the surface theory, the field $\theta(\mathbf{x})$ gives the cell-wise actuation of the mechanism. The new fields $\kappa(\mathbf{x})$ and $\tau(\mathbf{x})$ reflect two additional basic modes of deformation of the origami in question: these can be understood as bending and twisting modes. To see why, first recall from the differential geometry of surfaces \cite{do2016differential} that a deformation $\mathbf{y} \colon \Omega\subset\mathbb{R}^2 \rightarrow \mathbb{R}^3$ has first and second fundamental forms defined by
\begin{equation}
\begin{aligned}\label{eq:theFunds}
\mathbf{I}_{\mathbf{y}}(\mathbf{x}) := \big(\nabla \mathbf{y}(\mathbf{x})\big)^T  \nabla \mathbf{y}(\mathbf{x}), \quad \mathbf{II}_{\mathbf{y}}(\mathbf{x}) := -\big(\nabla \mathbf{y}(\mathbf{x})\big)^T  \nabla \mathbf{n}_{\mathbf{y}}(\mathbf{x})
\end{aligned}
\end{equation}
where $\mathbf{n}_{\mathbf{y}}(\mathbf{x}) = \tfrac{\partial_1 \mathbf{y}(\mathbf{x}) \times \partial_2 \mathbf{y} (\mathbf{x})}{| \partial_1 \mathbf{y}(\mathbf{x}) \times \partial_2 \mathbf{y} (\mathbf{x})|}$ is a unit normal to the surface parameterized by $\mathbf{y}(\mathbf{x})$. 
These forms are frame-indifferent descriptors of  the intrinsic and extrinsic geometry of surfaces, and are used in typical plate theories as measures of strain corresponding to stretching and bending of the plate. Here, for parallelogram origami, the fundamental forms obey a set of algebraic constraints involving $\theta(\mathbf{x}), \kappa(\mathbf{x})$ and $\tau(\mathbf{x})$. In particular, by Eqs.\;(\ref{eq:secondSurface}), (\ref{eq:omegaParamsSolveAlg}) and (\ref{eq:theFunds}), every effective deformation in our theory must satisfy 
\begin{equation}
\begin{aligned}\label{eq:bendingTwistingModes}
&\underline{\text{origami actuation modes:}}\\
&\mathbf{I}_{\mathbf{y}_{\text{eff}}}(\mathbf{x}) = \mathbf{A}_{\text{eff}}^T(\theta(\mathbf{x})) \mathbf{A}_{\text{eff}}(\theta(\mathbf{x})), \\
&\underline{\text{origami bending modes:}}\\
&\tilde{\mathbf{u}}_0 \cdot \mathbf{II}_{\mathbf{y}_{\text{eff}}}   (\mathbf{x})  \tilde{\mathbf{u}}_0 = \mathbf{u}(\theta(\mathbf{x})) \cdot ( \boldsymbol{\omega}_{\mathbf{u}_0}(\mathbf{x}) \times \mathbf{e}_3 ) = -A(\theta(\mathbf{x})) \big[ \mathbf{u}'(\theta(\mathbf{x})) \cdot \mathbf{u}(\theta(\mathbf{x})) \big] \kappa(\mathbf{x}), \\ 
&\tilde{\mathbf{v}}_0 \cdot \mathbf{II}_{\mathbf{y}_{\text{eff}}}   (\mathbf{x})  \tilde{\mathbf{v}}_0 = \mathbf{v}(\theta(\mathbf{x})) \cdot ( \boldsymbol{\omega}_{\mathbf{v}_0}(\mathbf{x}) \times \mathbf{e}_3 ) = A(\theta(\mathbf{x})) \big[ \mathbf{v}'(\theta(\mathbf{x})) \cdot \mathbf{v}(\theta(\mathbf{x})) \big] \kappa(\mathbf{x}),
 \\
&\underline{\text{origami twisting modes:}} \\
&\tilde{\mathbf{u}}_0 \cdot \mathbf{II}_{\mathbf{y}_{\text{eff}}}   (\mathbf{x})  \tilde{\mathbf{v}}_0 = \mathbf{u}(\theta(\mathbf{x})) \cdot ( \boldsymbol{\omega}_{\mathbf{v}_0}(\mathbf{x}) \times \mathbf{e}_3 ) = A(\theta(\mathbf{x}))  \tau(\mathbf{x}), \\ 
&\tilde{\mathbf{v}}_0 \cdot \mathbf{II}_{\mathbf{y}_{\text{eff}}}   (\mathbf{x})  \tilde{\mathbf{u}}_0 = \mathbf{v}(\theta(\mathbf{x})) \cdot ( \boldsymbol{\omega}_{\mathbf{u}_0}(\mathbf{x}) \times \mathbf{e}_3 ) = A(\theta(\mathbf{x}))  \tau(\mathbf{x}).
\end{aligned}
\end{equation}
The first constraint links the first fundamental form to the actuation and reflects the fact that the pattern should deform in a manner consistent with a local mechanism, as discussed in Section \ref{ssec:MainResult}. The other two sets of constraints link the second fundamental form to $\kappa(\mathbf{x})$ and $\tau(\mathbf{x})$. Observe that $\kappa(\mathbf{x})$ is proportional to the normal curvatures in the $\mathbf{u}_0$ and $\mathbf{v}_0$ directions, while  $\tau(\mathbf{x})$ is  proportional to the cross terms. This dichotomy reflects our terminology --- bending for $\kappa$ and twisting for $\tau$ --- and is illustrated in the forthcoming  Fig.\;\ref{Fig:MiuraEggboxExample}.

\paragraph{In-plane and out-of-plane Poisson's ratios.}  Many notable features of the soft modes of parallelogram origami can be derived from the constraints in Eq.\;(\ref{eq:bendingTwistingModes}). To demonstrate this, 
we now re-derive and discuss an important identity from the literature on origami metamaterials, popularly stated as \textit{the origami's in-plane and out-of-plane Poisson's ratio are equal and opposite} \cite{lebee2018fitting,mcinerney2022discrete, nassar2017curvature, nassar2022strain, schenk2013geometry,wei2013geometric}. 

First, eliminate $\kappa(\mathbf{x})$ from the middle constraint in   Eq.\;(\ref{eq:bendingTwistingModes})  to obtain  
\begin{equation}
\begin{aligned}\label{eq:importantImplication}
&[ \mathbf{v}'(\theta(\mathbf{x})) \cdot \mathbf{v}(\theta(\mathbf{x}))] [\tilde{\mathbf{u}}_0 \cdot \mathbf{II}_{\mathbf{y}_{\text{eff}}}(\mathbf{x})\tilde{\mathbf{u}}_0 ] = - [ \mathbf{u}'(\theta(\mathbf{x})) \cdot \mathbf{u}(\theta(\mathbf{x}))] [\tilde{\mathbf{v}}_0 \cdot \mathbf{II}_{\mathbf{y}_{\text{eff}}}(\mathbf{x})\tilde{\mathbf{v}}_0 ] .
\end{aligned}
\end{equation}
Next, introduce the  normal curvatures of $\mathbf{y}_{\text{eff}}(\mathbf{x})$  in the $\mathbf{u}_0$- and $\mathbf{v}_0$-directions,
\begin{equation}
\begin{aligned}
\kappa_{\mathbf{u}_0}(\mathbf{x}) :=  \frac{\tilde{\mathbf{u}}_0 \cdot \mathbf{II}_{\mathbf{y}_{\text{eff}}}(\mathbf{x}) \tilde{\mathbf{u}}_0}{\tilde{\mathbf{u}}_0 \cdot \mathbf{I}_{\mathbf{y}_{\text{eff}}}(\mathbf{x}) \tilde{\mathbf{u}}_0} , \quad  \kappa_{\mathbf{v}_0}(\mathbf{x}) :=  \frac{\tilde{\mathbf{v}}_0 \cdot \mathbf{II}_{\mathbf{y}_{\text{eff}}}(\mathbf{x}) \tilde{\mathbf{v}}_0}{ \tilde{\mathbf{v}}_0 \cdot \mathbf{I}_{\mathbf{y}_{\text{eff}}}(\mathbf{x}) \tilde{\mathbf{v}}_0},
\end{aligned}
\end{equation}
and recall the Poisson's ratio $\nu(\theta)$ defined in Eq.\;(\ref{eq:PoissonsRatioDesign}). Combining the first identity in Eq.\;(\ref{eq:bendingTwistingModes}) with Eq.\;(\ref{eq:importantImplication}) gives the desired result: if $\kappa_{\mathbf{v}_0}(\mathbf{x})$ is non-zero, then 
\begin{equation}
\begin{aligned}\label{eq:curvaturesToStrains}
\frac{\kappa_{\mathbf{u}_0}(\mathbf{x})}{\kappa_{\mathbf{v}_0}(\mathbf{x})} = \nu(\theta(\mathbf{x})).
\end{aligned}
\end{equation}

Eq.\;(\ref{eq:curvaturesToStrains}) quantifies the coupling between  actuation and bending of parallelogram origami.  As mentioned in the introduction, it was discovered independently by Schenk and Guest \cite{schenk2013geometry} and Wei et al.\;\cite{wei2013geometric}  in the case of the Miura origami and later expanded upon and generalized to  parallelogram origami by  Nassar et al.\;\cite{lebee2018fitting,nassar2017curvature,nassar2022strain} and McInerney et al.\;\cite{mcinerney2022discrete}. We point out that the terminology ``in-plane Poisson's ratio" for $\nu(\theta)$ and ``out-of-plane Poisson's ratio'' for $-\kappa_{\mathbf{u}_0}/\kappa_{\mathbf{v}_0}$ is really a vestige of earlier works on Miura and Eggbox origami. In these cases, the strains and normal curvatures are appropriately orthogonal --- $\mathbf{u}(\theta)\cdot\mathbf{v}(\theta)=0$ for Miura and Eggobx, as explained in Section \ref{ssec:Examples} below --- making their ratios proper Poisson's ratios. However, such orthogonality is not present for generic parallelogram origami patterns, and  interpreting  Eq.\;(\ref{eq:curvaturesToStrains}) as a statement on Poisson's ratios is a bit of an abuse of terminology, although one that we too find convenient.

\paragraph{Effectively planar origami.} There is a case in which the identity in Eq.\;(\ref{eq:curvaturesToStrains}) fails: effectively planar origami. This case is obtained by setting   $\kappa(\mathbf{x})=\tau(\mathbf{x}) = 0$, and thus $\mathbf{II}_{\mathbf{y}_{\text{eff}}}(\mathbf{x}) = \mathbf{0}$.  Our surface theory then reduces to a single governing equation in the actuation angle:
\begin{equation}
\begin{aligned}
(\text{planer origami soft modes:}) \quad \nabla_0^{\perp} \cdot \big[ \boldsymbol{\Gamma}( \theta(\mathbf{x})) \nabla_0 \theta(\mathbf{x}) \big]  = 0. 
\end{aligned}
\end{equation}
In fact, we derived essentially the same field theory in our previous work on planar kirigami metamaterials \cite{zheng2022continuum,zheng2023modelling}, albeit with a different formula for the tensor $\boldsymbol{\Gamma}(\theta)$, and with $\theta(\mathbf{x})$ encoding the cell-wise actuation of the kirigami's slits (the kirigami analog of creases). This observation suggests a universal character to the soft modes in locally mechanistic systems; for instance, it raises the question of whether a version of the full surface theory in Eq.\;(\ref{eq:SurfaceTheory}) could hold for non-planar kirigami deformations. We leave this to future work.

\subsection{Discussion of the effective plate theory for origami}\label{ssec:Discussion2}

We are now ready to compare our effective plate theory for origami with other, more familiar continuum theories.
The plate theory in Theorem \ref{MainTheorem} involves both a quadratic stored energy function in the fields $\nabla\theta(\mathbf{x})$, $\boldsymbol{\omega}_{\mathbf{u}_0}(\mathbf{x})$ and $\boldsymbol{\omega}_{\mathbf{v}_0}(\mathbf{x})$, and a hard constraint on these fields expressed in the form of the surface theory discussed above. That a hard constraint should arise when coarse-graining the soft modes of origami is not  surprising after all, especially given the microscopic origins of our effective energy. Nonlinear plate theories, when derived asymptotically for thin bodies in a bending-dominated limit, usually involve a constraint on their first fundamental forms. In the classical Kirchhoff plate theory the first fundamental form is required to be the identity matrix (see \cite{friesecke2002theorem} for a modern, variational derivation); similar constraints arise in plate models of ``active sheets" made of hydrogels or liquid crystal elastomers \cite{lewicka2011scaling,plucinsky2018actuation}. These constraints are also well-understood in the physical literature in the language of non-Euclidean plate theory \cite{aharoni2014geometry, efrati2009elastic,klein2007shaping}, and have provided key design guidance for engineers looking to ``program'' desired shapes in applications \cite{aharoni2018universal,mostajeran2016encoding,plucinsky2018patterning}. 

Origami is different. It too admits a plate-like theory, as we have shown, but not one that involves only a constraint on the first fundamental form. Rather, there is an algebraic coupling between the first and second fundamental forms implied by the discrete nature of the origami's unit cells --- again, see Eq.\;(\ref{eq:localBendIdent}) or Eq.\;(\ref{eq:importantImplication}).  In this way, the task of coarse graining origami, and mechanical metamaterials more broadly, adds to our understanding of continuum mechanics by forcing us to go beyond the paradigm of classical (or even non-Euclidean) plate theories. 
This merits further discussion. 

\paragraph{Effective plate theory for origami.}
To aid in the discussion, we first rewrite the plate theory from Theorem \ref{MainTheorem} in a more ``classical'' form, using the first and second fundamental forms derived in the previous section. Recall the effective energy density from the main result: 
\begin{equation}
\begin{aligned}\label{eq:Weff}
W_{\text{eff}}(\mathbf{x}) := \begin{pmatrix} \boldsymbol{\omega}_{\mathbf{u}_0}(\mathbf{x}) \\ \partial_{\mathbf{u}_0}\theta(\mathbf{x})  \end{pmatrix} \cdot  \mathbf{K}_{\mathbf{u}_{0}}(\theta(\mathbf{x})) \begin{pmatrix} \boldsymbol{\omega}_{\mathbf{u}_0}(\mathbf{x}) \\  \partial_{\mathbf{u}_0} \theta(\mathbf{x}) \end{pmatrix}  +  \begin{pmatrix} \boldsymbol{\omega}_{\mathbf{v}_0}(\mathbf{x}) \\ \partial_{\mathbf{v}_0} \theta(\mathbf{x}) \end{pmatrix} \cdot  \mathbf{K}_{\mathbf{v}_{0}}(\theta(\mathbf{x}) ) \begin{pmatrix} \boldsymbol{\omega}_{\mathbf{v}_0}(\mathbf{x}) \\ \partial_{\mathbf{v}_0} \theta(\mathbf{x})\end{pmatrix}
\end{aligned}.
\end{equation}
 Eliminating $\kappa(\mathbf{x})$ and $\tau(\mathbf{x})$  in Eq.\;(\ref{eq:omegaParamsSolveAlg})  via Eq.\;(\ref{eq:bendingTwistingModes}) gives the fields $\boldsymbol{\omega}_{\mathbf{u}_0}(\mathbf{x})$ and $\boldsymbol{\omega}_{\mathbf{u}_0}(\mathbf{x})$  as
\begin{equation}
\begin{aligned}\label{eq:finalOmegaParam}
&\boldsymbol{\omega}_{\mathbf{u}_0}(\mathbf{x}) =\Big(\tfrac{\tilde{\mathbf{v}}_{0} \cdot \mathbf{II}_{\mathbf{y}_{\text{eff}}}(\mathbf{x}) \tilde{\mathbf{u}}_0}{ A(\theta(\mathbf{x})) }\Big)  \mathbf{u}(\theta(\mathbf{x})) - \Big(\tfrac{\tilde{\mathbf{u}}_{0} \cdot \mathbf{II}_{\mathbf{y}_{\text{eff}}}(\mathbf{x}) \tilde{\mathbf{u}}_0}{ A(\theta(\mathbf{x})) }\Big) \mathbf{v}(\theta(\mathbf{x})) +  \big[  \boldsymbol{\Gamma}_{\mathbf{u}_0}(\theta(\mathbf{x})) \cdot \nabla_0 \theta(\mathbf{x})  \big] \mathbf{e}_3, \\
&\boldsymbol{\omega}_{\mathbf{v}_0}(\mathbf{x}) = \Big(\tfrac{\tilde{\mathbf{v}}_{0} \cdot \mathbf{II}_{\mathbf{y}_{\text{eff}}}(\mathbf{x}) \tilde{\mathbf{v}}_0}{ A(\theta(\mathbf{x})) }\Big) \mathbf{u}(\theta(\mathbf{x})) - \Big(\tfrac{\tilde{\mathbf{u}}_{0} \cdot \mathbf{II}_{\mathbf{y}_{\text{eff}}}(\mathbf{x}) \tilde{\mathbf{v}}_0}{ A(\theta(\mathbf{x})) }\Big)   \mathbf{v}(\theta(\mathbf{x})) +  \big[  \boldsymbol{\Gamma}_{\mathbf{v}_0}(\theta(\mathbf{x})) \cdot \nabla_0 \theta(\mathbf{x})  \big] \mathbf{e}_3.
\end{aligned}
\end{equation}
 Due to the linearity of this parameterization in $\mathbf{II}_{\mathbf{y}_{\text{eff}}}(\mathbf{x})$ and $\nabla \theta(\mathbf{x})$, $W_{\text{eff}}(\mathbf{x})$ can also be written as a quadratic form in these variables:
\begin{equation}
\begin{aligned}
W_{\text{eff}}(\mathbf{x})=  \begin{pmatrix} \mathbf{II}_{\mathbf{y}_{\text{eff}}}(\mathbf{x}) \\ \nabla \theta(\mathbf{x})^T \end{pmatrix} \colon \mathbb{K}(\theta(\mathbf{x})) \colon \begin{pmatrix} \mathbf{II}_{\mathbf{y}_{\text{eff}}}(\mathbf{x}) \\ \nabla \theta(\mathbf{x})^T \end{pmatrix}
\end{aligned}
\end{equation}
for a fourth order tensor $\mathbb{K}(\theta) \in \mathbb{R}^{3 \times 2 \times 3 \times 2}$ with the major symmetry  $\mathbf{A} \colon \mathbb{K}(\theta) \colon \mathbf{B} = \mathbf{B} \colon \mathbb{K}(\theta) \colon \mathbf{A}$ for all $\mathbf{A}, \mathbf{B} \in \mathbb{R}^{3\times2}$. This tensor can be made explicit using Eqs.\;(\ref{eq:stiffnessTensor}-\ref{eq:bendingModuli}) and Eqs.\;(\ref{eq:Weff}-\ref{eq:finalOmegaParam}).

To write the plate theory in full, recognize that the  constraints linking the first and second fundamental forms of the effective deformation $\mathbf{y}_{\text{eff}}(\mathbf{x})$ to the actuation $\theta(\mathbf{x})$ are the first constraint in Eq.\;(\ref{eq:bendingTwistingModes}) and Eq.\;(\ref{eq:importantImplication}). Correspondingly, we define the set of  admissible angles and  symmetric tensors as 
\begin{equation}
\begin{aligned}\label{eq:preFinalFormEnergy}
\mathcal{A}_{\text{ori}} := \Big\{ &( \theta, \mathbf{G} , \mathbf{H})  \in  ( \theta^{-}, \theta^+) \times \mathbb{R}^{2 \times 2}_{\text{sym}} \times \mathbb{R}^{2\times2}_{\text{sym}}   \quad \text {subject to } \\
& \quad  \mathbf{G} = \mathbf{A}_{\text{eff}}^T(\theta) \mathbf{A}_{\text{eff}}(\theta) \text{ and }  [\mathbf{v}'(\theta) \cdot \mathbf{v}(\theta)][ \tilde{\mathbf{u}}_0 \cdot \mathbf{H} \tilde{\mathbf{u}}_0 ]= -[\mathbf{u}'(\theta) \cdot \mathbf{u}(\theta)][ \tilde{\mathbf{v}}_0 \cdot \mathbf{H} \tilde{\mathbf{v}}_0 ] \Big\} .
\end{aligned}
\end{equation}
The complete limiting plate theory for parallelogram origami is thus 
\begin{equation}
\begin{aligned}\label{eq:finalEnergyForm}
\mathcal{E}_{\text{eff}}( \mathbf{y}_{\text{eff}}, \theta) := \begin{cases}
 \int_{\Omega}  \begin{pmatrix} \mathbf{II}_{\mathbf{y}_{\text{eff}}}(\mathbf{x}) \\ \nabla \theta(\mathbf{x})^T \end{pmatrix} \colon \mathbb{K}(\theta(\mathbf{x})) \colon \begin{pmatrix} \mathbf{II}_{\mathbf{y}_{\text{eff}}}(\mathbf{x}) \\ \nabla \theta(\mathbf{x})^T \end{pmatrix} dA   & \text{ if }  (\theta(\mathbf{x}), \mathbf{I}_{\mathbf{y}_{\text{eff}}}(\mathbf{x}), \mathbf{II}_{\mathbf{y}_{\text{eff}}}(\mathbf{x}) ) \in \mathcal{A}_{\text{ori}}  \text{ on $\Omega$} \\
 +\infty & \text{ otherwise.}
 \end{cases} 
\end{aligned}
\end{equation}

\paragraph{Comparison with Kirchhoff's plate theory.} 
We now compare Eq.\;(\ref{eq:finalEnergyForm}) against Kirchhoff's plate theory, which is based on the constrained bending energy
\begin{equation}
\begin{aligned}\label{eq:classicalKirchhoff}
\mathcal{E}_{\text{iso}}(\mathbf{y}) := \begin{cases}
\frac{1}{24} \int_{\Omega} \Big\{  2 \mu | \mathbf{II}_{\mathbf{y}}(\mathbf{x})|^2 + \frac{\lambda \mu}{\mu + \lambda/2} \big[\text{Tr}\big(  \mathbf{II}_{\mathbf{y}}(\mathbf{x}) \big) \big]^2 \Big\}  dA & \text{ if }  \mathbf{I}_{\mathbf{y}}(\mathbf{x}) = \mathbf{I} \text{ on $\Omega$} \\
+\infty & \text{ otherwise}
\end{cases}
\end{aligned}
\end{equation}
where $\lambda$ and $\mu$ are the  Lam\'{e} constants. The constraint  $\mathbf{I}_{\mathbf{y}}(\mathbf{x}) = \mathbf{I}$ in Kirchhoff's theory selects isometric deformations, i.e., ones that preserve Euclidean lengths and angles. To each isometric deformation the theory assigns an energy that is quadratic in its second fundamental form, the bending strain measure in the theory. However, Kirchhoff's theory does not have any analog of the second constraint in $\mathcal{A}_{\text{ori}}$ on the origami's second fundamental form. 

Of course, there is a pointwise constraint on the second fundamental forms in Kirchhoff's plate theory, although it is not usually made explicit in its formulation. According to Gauss's theorem egregium, any sufficiently smooth deformation $\mathbf{y}(\mathbf{x})$ satisfying the isometry constraint $\mathbf{I}_{\mathbf{y}}(\mathbf{x}) =\mathbf{I}$ must have a second fundamental form whose determinant is zero: $\det\mathbf{II}_{\mathbf{y}}(\mathbf{x})=0$.
More generally, when applied to shell theory or to non-Euclidean plates, Gauss's theorem constrains the ratio of the determinants of the first and second fundamental forms to be equal to the Gauss curvature of the shell, a quantity that is completely determined by the first fundamental form \cite{do2016differential}. Gauss's theorem also applies to the effective deformations of origami, but in a way that is implied by the coupling $\mathbf{I}_{\mathbf{y}_{\text{eff}}}(\mathbf{x}) = \mathbf{A}_{\text{eff}}^T(\theta(\mathbf{x})) \mathbf{A}_{\text{eff}}(\theta(\mathbf{x}))$ in $\mathcal{A}_{\text{ori}}$; the last part of $\mathcal{A}_{\text{ori}}$ is an additional coupling of microscopic origin, not a manifestation of Gauss's theorem. 


\paragraph{Origami as a generalized elastic continua.} 

The second key difference between our effective theory for origami and standard continuum theories is the actuation field $\theta(\mathbf{x})$, which reflects origami's local degrees of freedom, i.e., the existence of a mechanism. 
This field is felt throughout our theory. It enters nonlinearly into the moduli $\mathbb{K}(\theta(\mathbf{x}))$, its gradient enters as an additional quadratic strain measure, and it couples to the first and second fundamental forms via constraints that differ greatly from those of a Euclidean isometry (see the examples in Fig.\;\ref{Fig:MiuraEggboxExample}). The presence of an actuation field makes our theory not only a generalization of Kirchhoff's plate theory, but also an example of a \textit{generalized elastic continuum}, which we discuss now.

As mentioned in the introduction, generalized elasticity refers to a class of models that incorporate auxiliary fields --- beyond the deformation --- as additional degrees of freedom leading to an enriched set of governing equations. The Cosserat brothers first proposed the concept via ``director models" in the early 1900s \cite{cosserat1909theorie}. Eringen later codified the subject in much detail \cite{eringen2012microcontinuum} by developing a wide range of microcontinuum models, including the popular micropolar, microstretch, and micromorphic theories. These theories introduce the auxiliary  fields to capture mechanical rearrangements at the microscale with consequences for elasticity at larger scales. Perhaps, as a result, a few lines of research beyond our own \cite{zheng2023modelling} have connected mechanical metamaterials  to this  subject. 
R.\;S.\;Lakes  recognized that extremal hinged lattices cannot be modeled using  conventional elasticity and advocated for a Cosserat continuum \cite{lakes2022extremal}; Nassar et al.\;\cite{nassar2020microtwist} proposed a ``microtwist" continuum theory to explain topological polarization in the kagome lattice; Sarhil et al.\;\cite{sarhil2023modeling,sarhil2023size} modeled the size effects in metamaterial beams using relaxed micromorphic theories; a review article has noted the link \cite{kadic20193d}; others \cite{saremi2020topological,sun2020continuum} have made analogies to these theories to arrive at strain gradient models. 

 Like our theory, the generalized elastic models referenced above introduce elastic energy terms that are  quadratic  in the gradient of the auxiliary field(s). Unlike our theory, these models have mostly focused on the linear response regime, making them small displacement theories; particularly, their  stored energies are quadratic in a linear measure of strain. This assumption is not consistent with the large deformation response of parallelogram origami undergoing a soft mode. Our generalized elastic theory in Eq.\;(\ref{eq:finalEnergyForm}) captures these large soft modes. It couples the first and second fundamental forms of their effective deformations to the origami's local actuation. This coupling is made explicit by the constraints in $\mathcal{A}_{\text{ori}}$. Only after getting this coupling right does the theory resemble a more-or-less familiar generalized elastic continua,  with an energy density that is quadratic in the appropriate nonlinear strain measure $\mathbf{II}_{\mathbf{y}_{\text{eff}}}(\mathbf{x})$ and the gradient of the auxiliary field $\nabla \theta(\mathbf{x})$.



\subsection{Examples: the Miura and Eggbox patterns}\label{ssec:Examples}
\begin{figure}[t!]
\centering
\includegraphics[width=1\textwidth]{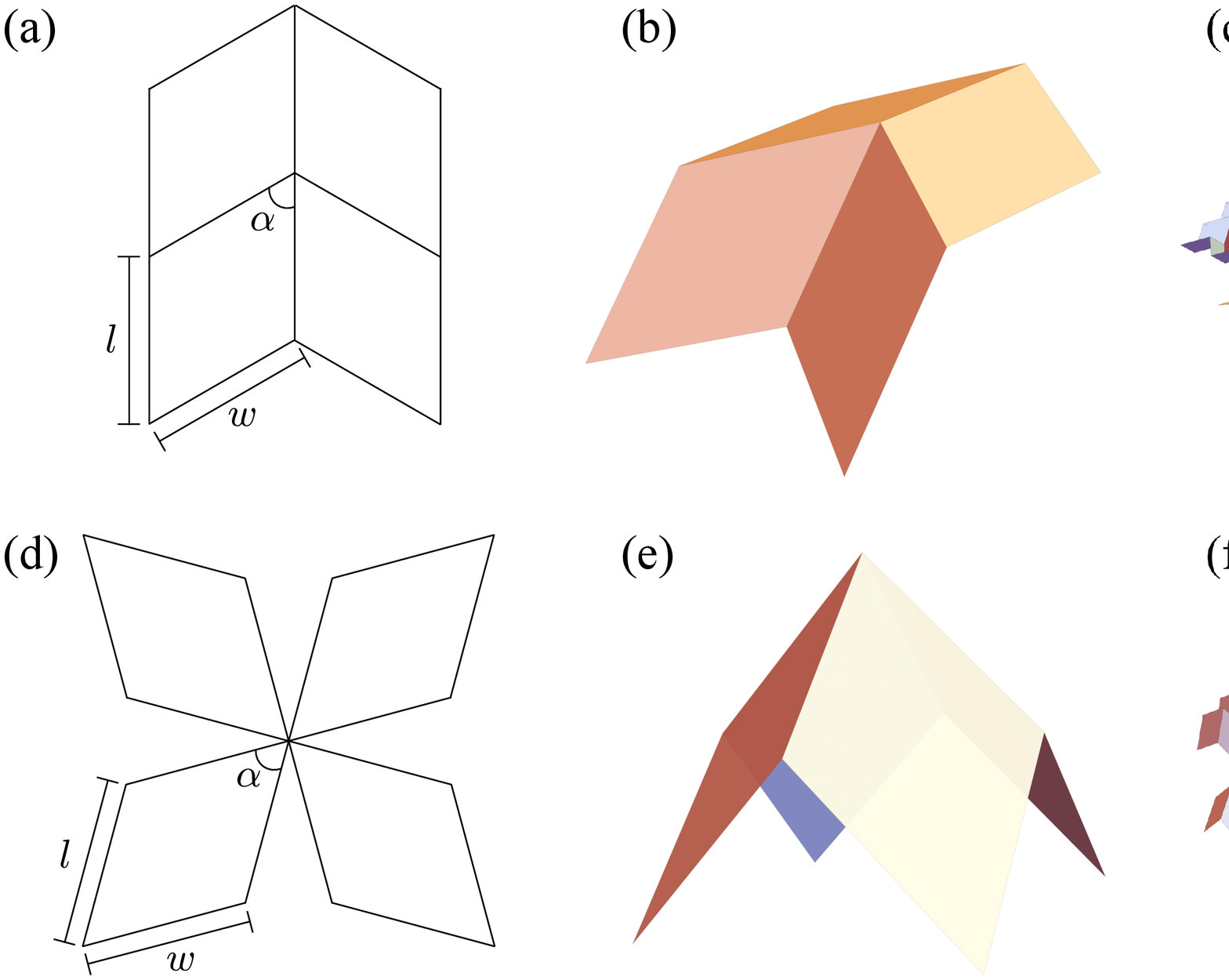} 
\caption{Miura and Eggbox origami. The top half of the figure displays the Miura pattern: (a) the flattened unit cell and its design parameters; (b) the partly folded reference cell; (c) the pattern and a folded state that highlights its auxeticity. The bottom half of the figure displays the Eggbox pattern: (d) the flattened unit cell and its design parameters; (e) the partly folded reference cell; (f) the pattern and a folded state that highlights its conventional Poisson's ratio.  }
\label{Fig:MiuraEggboxAuxeticity}
\end{figure}
We close by demonstrating a few solutions to the surface theory in Eq.\;(\ref{eq:GenPDESurf}) for the Miura and Eggbox origami. Both  are canonical examples of parallelogram origami, obtained by tessellating a cell composed of a single parallelogram of length $l>0$, width $w>0$ and angle $\alpha \in (0, \pi/2)$. Fig.\;\ref{Fig:MiuraEggboxAuxeticity}(a) and (d) show flattened versions of each unit cell. The Miura pattern in Fig.\;\ref{Fig:MiuraEggboxAuxeticity}(a-c) is Euclidean. Its sector angles sum to $2 \pi$  around each vertex, yielding a flat crease pattern. The Eggbox pattern in Fig.\;\ref{Fig:MiuraEggboxAuxeticity}(d-f) is \textit{not} Euclidean.  Its sector angles sum to   $4 \alpha$  and $4( \pi -\alpha)$ in an alternating fashion around each vertex. It cannot  flatten out. 

Regarding kinematics,  both patterns have  a single degree-of-freedom mechanism described by an orthogonal parameterization $\mathbf{u}(\theta) = \lambda_u(\theta) \mathbf{e}_1$ and $\mathbf{v}(\theta) = \lambda_v(\theta) \mathbf{e}_2$ for $\lambda_{u}(\theta), \lambda_v(\theta) >0$. However, the Miura is auxetic since both its sides contract/expand on folding/unfolding, while  the Eggbox is not auxetic since one side contracts and the other expands on folding (see Fig.\;\ref{Fig:MiuraEggboxAuxeticity}(c) and (f)). This difference furnishes
\begin{equation}
    \begin{aligned}\label{eq:dichotomy}
        \text{Miura:} \quad  \lambda_u'(\theta) \lambda_v'(\theta) > 0, \quad \text{Eggbox:} \quad \lambda_u'(\theta) \lambda_v'(\theta) < 0.
    \end{aligned}
\end{equation}
These inequalities have significant implications for the nature of bending and twisting in the Miura and Eggbox patterns. By Eqs.\;(\ref{eq:bendingTwistingModes}), the Gauss curvature of a general soft mode satisfies 
\begin{equation}
\begin{aligned}\label{eq:signGC}
\text{sign of the Gauss curvature}   = \text{sign} \Big( -\lambda_u'(\theta) \lambda_v'(\theta) \lambda_u(\theta) \lambda_v(\theta) \kappa^2 - \tau^2   \Big).
\end{aligned}
\end{equation}
Thus, the Miura always deforms with effectively negative Gauss curvature, while the Eggbox can deform with positive or negative Gauss curvature  depending on the amount of bending ($\kappa$) versus twisting ($\tau$) in the soft mode. We illustrate this in examples below.

For the examples, we focus on cases where the parallelogram building block of the pattern is fixed at $w = l = 1$ and $\alpha = \tfrac{2}{3} \pi$. This yields orthogonal Bravais lattice vectors given by
\begin{equation}
\begin{aligned}\label{eq:ExamplesBravParam}
&\underline{\text{Miura:}} \\
&\lambda_u(\theta) =\sqrt{3} \cos \Big( \frac{\theta + \pi/6}{2} \Big) , \qquad \lambda_v(\theta) = 2\sqrt{2}\big(5 - 3 \cos( \theta+ \pi/6) \big)^{-1/2} , \qquad   \theta \in (- \pi/6, 5\pi/6),  \\
&\underline{\text{Eggbox:}}\\ 
&\lambda_u(\theta) = 2 \sin\Big[\frac{1}{2} \big( \theta + \arccos(1- \cos \theta)\big) \Big],\qquad    \lambda_v(\theta) = \lambda_u(-\theta), \qquad \theta \in (-\pi/3, \pi/3),
\end{aligned}
\end{equation}
after a lengthy calculation involving the cell's rigid kinematics (which we omit). 
Fig.\;\ref{Fig:MiuraEggboxAuxeticity}(b) and (e) shows each cell's partly folded reference configuration, obtained by setting $\theta= 0$. In contrast to the depictions in Figs.\;\ref{Fig:idepat} and \ref{Fig:pertcell}, the $\theta$-parameterization of the Eggbox's mechanism motion does not strictly measure the change in the dihedral angle of any one of its creases. This distinction does not impact the theory.  
 
 We construct solutions of the surface theory for these Miura and Eggbox patterns using a one-dimensional ansatz for the fields $\kappa(\mathbf{x}) \equiv \kappa(\mathbf{x}\cdot \tilde{\mathbf{v}}_0^r)$, $\tau(\mathbf{x}) \equiv \tau(\mathbf{x}\cdot \tilde{\mathbf{v}}_0^r)$ and $\theta(\mathbf{x}) \equiv \theta(\mathbf{x}\cdot \tilde{\mathbf{v}}_0^r)$, explained in Appendix \ref{sec:1DSolve}. 
 This ansatz turns the PDEs in Eq.\;(\ref{eq:GenPDESurf}) into easily solvable ordinary differential equations (ODE); Fig.\;\ref{Fig:MiuraEggboxExample} shows deformations for both patterns obtained using the procedure in the appendix. Fig.\;\ref{Fig:MiuraEggboxExample}(a-c) displays bending and twisting modes for Miura origami, while Fig.\;\ref{Fig:MiuraEggboxExample}(d-f) displays analogous Eggbox modes. Pure bending is shown in the top pane, pure twist on the bottom and a mixture of bending and twist in the middle. We discuss these examples in the context of their Gauss curvature, using the illuminating formulas in Eqs.\;(\ref{eq:dichotomy}-\ref{eq:signGC}).

\begin{figure}[t!]
\centering
\includegraphics[width = 1\textwidth]{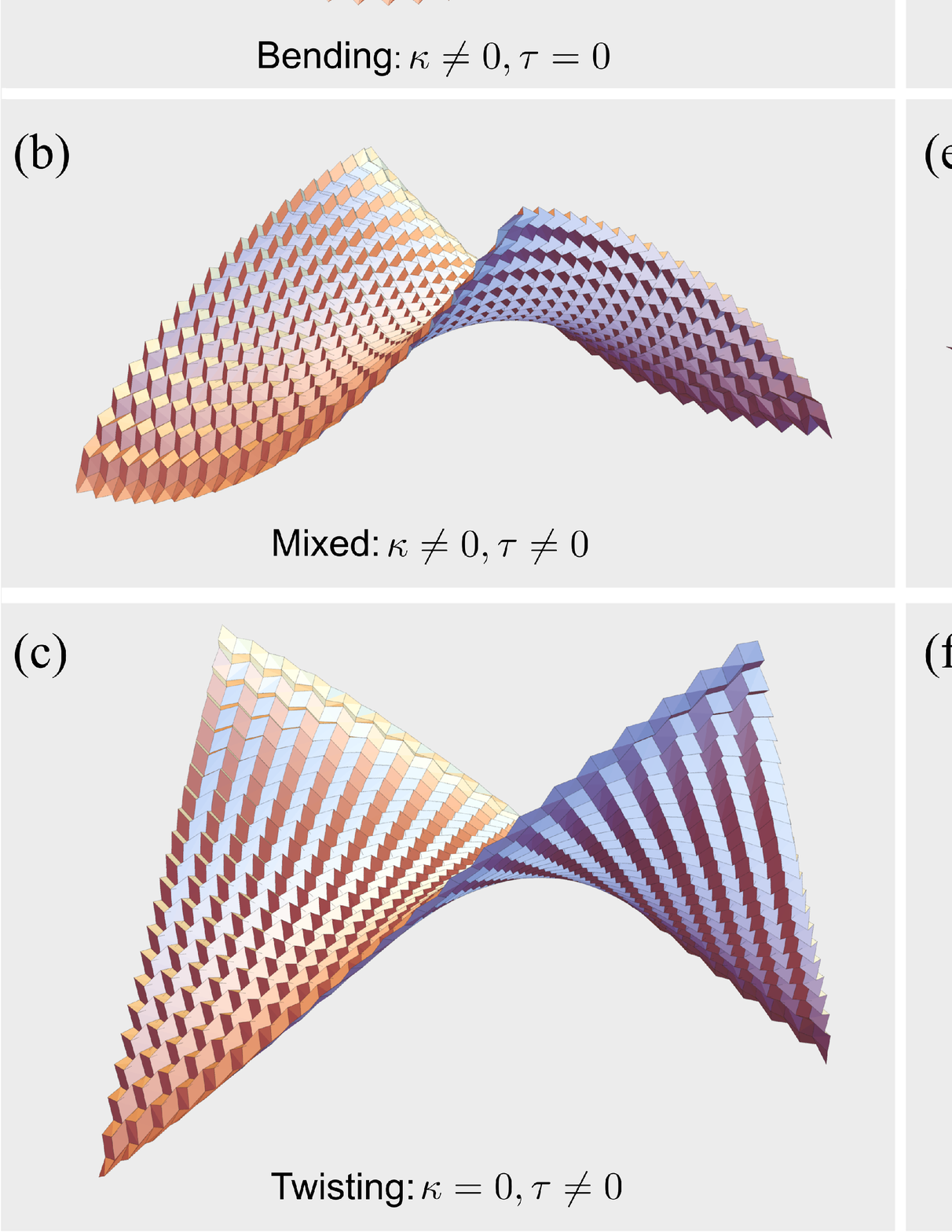} 
\caption{Examples of bending ($\kappa$) and twisting ($\tau$) deformations in Miura and Eggbox origami. (a-c) Miura examples: (a) pure bending, (b) combination of bend and twist and  (c) pure twist.  (d-f) Eggbox examples: (d) pure bending, (e) combination of bend and twist and (f) pure twist.}
\label{Fig:MiuraEggboxExample}
\end{figure}

The pure bending case for the Miura origami in Fig.\;\ref{Fig:MiuraEggboxExample}(a)  reproduces the  symmetric saddle reported throughout the literature \cite{callens2018flat,lebee2015folds,schenk2013geometry, wei2013geometric}. That of the Eggbox in  Fig.\;\ref{Fig:MiuraEggboxExample}(d)  displays  a symmetric cap, the bending motif of the opposite extreme familiar to those who have studied this pattern \cite{nassar2017curvature,pratapa2019geometric,schenk2011origami}. These deformations are obtained by imposing the ``no-twist'' constraint $\tau(s) = 0$ in the governing equations in Appendix \ref{sec:1DSolve} (Eq.\;(\ref{eq:tauKappaAnalytical}-\ref{eq:finalODE})), and  choosing the other free parameters in these equations to produce an actuation field $\theta(s)$ that is symmetric about the midline of the domain. They also illustrate the basic fact that the Gauss curvature of pure bending is strictly negative for Miura origami, strictly positive for Eggbox and generally depends only on whether the pattern's mechanism is auxetic or not (per Eq.\;(\ref{eq:signGC})). 

The twisting modes of parallelogram origami offer no such dichotomy. Pure twist has negative Gauss curvature regardless of whether the pattern is Miura, Eggbox or \textit{any other} type of parallelogram origami --- set $\kappa=0$ and $\tau\neq 0$ in Eq. (\ref{eq:signGC}). We illustrate this  fact by producing essentially the same type of pure twist mode for the two patterns in Fig.\;\ref{Fig:MiuraEggboxExample}(c) and (f).  These were made by imposing $\kappa(s) = 0$ and  choosing the remaining free parameters to produce  another symmetric ODE solution for $\theta(s)$. While many authors have emphasized a link  between the in-plane and out-of-plane Poisson's ratio in these types of  origami patterns \cite{lebee2018fitting,mcinerney2022discrete, nassar2017curvature, nassar2022strain, schenk2013geometry,wei2013geometric} (see Section \ref{ssec:discussSec-modes}), our results  demonstrate a subtlety in the physical interpretation of this link: the pattern's auxeticity is only directly linked to certain normal curvatures associated with bending modes.  It has no such  link to the principal curvatures of generic deformations. This is consistent with the fact that non-auxetic patterns like the Eggbox can accommodate both positive and negative Gauss curvature,  depending on the amount of twist. 

As a final demonstration, we highlight deformation modes that simultaneously bend and twist. Fig.\;\ref{Fig:MiuraEggboxExample}(b) displays an example for the Miura. Bend and twist produce the same Gauss curvature in this case and are seemingly cooperative --- the deformation looks like a twisted saddle. However, for Eggbox origami, bend and twist have opposing influences on the Gauss curvature. Fig.\;\ref{Fig:MiuraEggboxExample}(e) displays an Eggbox mode that is bending dominated on the left of the sample and twist dominated on the right, leading to pronounced regions of positive and negative Gauss curvature in a single sample. 

\subsection{Outlook}\label{ssec:Outlook}

This paper derived an effective plate theory for parallelogram origami patterns from bar and hinge elasticity, and demonstrated the accompanying surface theory for the examples of Miura and Eggbox origami (Fig.\;\ref{Fig:MiuraEggboxExample}). These two examples are a ``drop in the bucket" as far as the general theory is concerned. Future work is needed on the numerical front to address generic patterns and nontrivial PDE solutions to the surface theory in Eq.\;(\ref{eq:GenPDESurf}). Some ideas in this direction can be found in \cite{marazzato2022mixed,marazzato2023computation,marazzato2023H2}, inspired by the results of Leb\'{e}e et al.\;\cite{lebee2018fitting} on Miura origami and the results of the last two authors on planar kirigami. Another idea is to relax the PDE by introducing an associated elastic energy penalty; we developed this approach for planar kirigami where it proved successful at capturing general soft modes  \cite{zheng2023modelling}. Beyond solving the surface theory, an accompanying effort is needed to develop a simulation framework for the full effective plate theory in Eq.\;(\ref{eq:finalEnergyForm}) subject to physically relevant boundary conditions. Importantly, this would allow the theory developed in this paper to be  compared  with experiments of parallelogram origami patterns \textit{under loads} --- a true measure of the success and utility of a mechanical theory. Purely mathematical questions remain as well, such as the technical matter of the analyticity hypothesis in the case of a negative Poisson's ratio, and the deeper question of a fully ansatz-free derivation using the technique of $\Gamma$-convergence. 
Open questions aside, our examples and discussion have highlighted the richness of our theory and the insights it provides for parallelogram origami.

\section*{Acknowledgments}
\noindent  H.X. and P.P. acknowledge support from the National Science Foundation (CMMI-CAREER-2237243).  I.T. acknowledges support from the National Science Foundation (DMS-CAREER-2350161). All authors acknowledge support from the Army Research Office (ARO-W911NF2310137).

\appendix

\section{Analysis of mechanism kinematics}\label{sec:MechProofs}

This section proves Propositions \ref{mechKinCellProp} and \ref{MechProp} and  Lemma \ref{SignsLemma}, thus characterizing the mechanism kinematics of the cell and the pattern. 

\subsection{Mechanism kinematics of a single cell}\label{ssec:MechProofs1}
Recall that the mechanism set of a unit cell is $\mathcal{M} := \big\{ ( \mathbf{t}_1^d, \ldots, \mathbf{t}_4^d) \colon \mathbb{R}^3 \times \ldots \times \mathbb{R}^3 \colon \text{subject to Eq.\;(\ref{eq:compatCreasesZero})}  \big\}$. We now parameterize this set.

Without loss of generality, fix two adjacent deformed vectors by setting 
\begin{equation}
\begin{aligned}\label{eq:fixVectors}
\mathbf{t}_1^d = \mathbf{t}_1^r, \quad \mathbf{t}_4^d = \mathbf{t}_4^r
\end{aligned}
\end{equation}
to eliminate an overall rigid rotation of the unit cell. Then, the other two deformed vectors satisfy 
\begin{equation}
\begin{aligned}\label{eq:getDefCreases}
\mathbf{t}_2^d(\gamma) = \mathbf{R}_{\mathbf{t}_1^r}(\gamma) \mathbf{t}_2^r, \quad \mathbf{t}_3^d(\theta) = \mathbf{R}_{\mathbf{t}_4^r}(\theta) \mathbf{t}_3^r \quad \text{ for  angles $\gamma, \theta \in (-\pi, \pi)$,}
\end{aligned}
\end{equation}
where $\mathbf{R}_{\mathbf{t}}(\varphi):=|\mathbf{t}|^{-2} \mathbf{t}\otimes\mathbf{t} +\cos\varphi\big(\mathbf{I}-|\mathbf{t}|^{-2}\mathbf{t}\otimes\mathbf{t}\big)+|\mathbf{t}|^{-1}\sin\varphi \big(\mathbf{t}\times \big)$ denotes a counterclockwise rotation about a  non-zero axis vector $\mathbf{t}$ through an angle $\varphi$. The angle between the two vectors $\mathbf{t}_2^d(\gamma)$ and $\mathbf{t}_3^d(\theta)$ must be the same as that of their reference vectors, 
\begin{equation}
\begin{aligned}\label{eq:fcompCompat}
f_{\text{comp}}(\gamma, \theta) :=  \mathbf{t}_2^d(\gamma) \cdot \mathbf{t}_3^d(\theta)  - \mathbf{t}_2^r \cdot \mathbf{t}_3^r = 0. 
\end{aligned}
\end{equation}
Finally, the requirement that the mountain-valley assignment of the deformed cell be the same as the reference cell is an inequality of the form
\begin{equation}
\begin{aligned}\label{eq:MVAssign}
\mathbf{f}_{\text{mv}}(\gamma, \theta) := \begin{pmatrix}  [\mathbf{t}_1^d \cdot (\mathbf{t}_2^d(\gamma) \times \mathbf{t}_3^d(\theta))][\mathbf{t}_1^r \cdot (\mathbf{t}_2^r \times \mathbf{t}_3^r)]  \\ [\mathbf{t}_2^d(\gamma) \cdot (\mathbf{t}_3^d(\theta) \times \mathbf{t}_4^d) ][\mathbf{t}_2^r \cdot (\mathbf{t}_3^r \times \mathbf{t}_4^r)]  \\  [\mathbf{t}_3^d(\theta) \cdot (\mathbf{t}_4^d \times \mathbf{t}_1^d) ][\mathbf{t}_3^r \cdot (\mathbf{t}_4^r \times \mathbf{t}_1^r)] \\  [\mathbf{t}_4^d \cdot (\mathbf{t}_1^d \times \mathbf{t}_2^d(\gamma)) ][\mathbf{t}_4^r \cdot (\mathbf{t}_1^r \times \mathbf{t}_2^r)] \end{pmatrix}  > \mathbf{0},
\end{aligned}
\end{equation}
enforced element-wise. Eqs.\;(\ref{eq:fixVectors}-\ref{eq:MVAssign}) parameterize the rigid kinematics of the unit cell in Eq.\;(\ref{eq:compatCreasesZero}) up to an overall rigid rotation.

Next, we show that there exists a unique mechanism motion of the unit cell that contains the reference configuration and does not change the mountain-valley assignment. 
\begin{lemma}\label{ImplicitLemma}
There exists an open interval $(\theta^-, \theta^+)$ containing $0$ and a unique continuously differentiable function $\gamma \colon (\theta^{-}, \theta^+) \rightarrow \mathbb{R}$ such that $\gamma(0) = 0$ and 
\begin{equation}
\begin{aligned}\label{eq:firstCompResult}
f_{\emph{comp}}(\gamma(\theta), \theta) = 0, \quad \mathbf{f}_{\emph{mv}}(\gamma(\theta), \theta)  >  \mathbf{0} \quad \text{for all } \theta \in (\theta^{-}, \theta^+).
\end{aligned}
\end{equation}
In addition, $\gamma(\theta)$ is  analytic and its derivative satisfies
\begin{equation}
    \begin{aligned}\label{eq:gammaPrime}
        \gamma'(\theta) = \frac{|\mathbf{t}_1^r|}{|\mathbf{t}_4^r|}\frac{\mathbf{t}_2^d(\gamma(\theta)) \cdot (\mathbf{t}_3^d(\theta) \times \mathbf{t}_4^d)}{\mathbf{t}^d_1 \cdot ( \mathbf{t}_2^d(\gamma(\theta)) \times \mathbf{t}_3^d(\theta))}.
    \end{aligned}
\end{equation}
\end{lemma}
\begin{proof}
The proof is a standard application of the implicit function theorem. Observe that $f_{\text{comp}}(0,0) = 0$ and $f_{\text{comp}}(\gamma, \theta)$ is analytic. Also,
\begin{equation}
    \begin{aligned}\label{eq:partialIdentities}
        &\partial_{\gamma} f_{\text{comp}}(\gamma, \theta) = |\mathbf{t}_1^r|^{-1} \mathbf{R}_{\mathbf{t}_1^r}(\gamma)(\mathbf{t}_1^r \times \mathbf{t}_2^r) \cdot \mathbf{t}_3^d(\theta) = |\mathbf{t}_1^r|^{-1} (\mathbf{t}_1^d \times \mathbf{t}_2^d(\gamma) ) \cdot \mathbf{t}_3^d(\theta)  \\
        &\partial_{\theta} f_{\text{comp}}(\gamma, \theta) = |\mathbf{t}_4^r|^{-1} \mathbf{R}_{\mathbf{t}_4^r}(\theta) (\mathbf{t}_4^r \times \mathbf{t}_3^r) \cdot \mathbf{t}_2^d(\gamma)= | \mathbf{t}_4^r|^{-1} ( \mathbf{t}_4^d \times \mathbf{t}_3^d(\theta)) \cdot \mathbf{t}_2^d(\gamma)
    \end{aligned}
\end{equation}
by differentiating Eq.\;(\ref{eq:fcompCompat}) and using the identities in Eqs.\;(\ref{eq:fixVectors}-\ref{eq:getDefCreases}). It follows that 
$\partial_{\gamma} f_{\text{comp}}(0,0) = |\mathbf{t}_1^r|^{-1}(\mathbf{t}_1^r \times \mathbf{t}_2^r) \cdot \mathbf{t}_3^r >0$ by the mountain-valley assumptions on  $\mathbf{t}_1^r, \ldots, \mathbf{t}_4^r$. Thus, by the implicit function theorem, there is an open interval $(\tilde{\theta}^{-}, \tilde{\theta}^+)$ containing $0$ and a unique analytic function $\gamma \colon (\tilde{\theta}^{-}, \tilde{\theta}^+) \rightarrow \mathbb{R}$ with $\gamma(0) = 0$ and $f_{\text{comp}}(\gamma(\theta), \theta) = 0$ for all $\theta \in (\tilde{\theta}^{-}, \tilde{\theta}^+)$.  Since $\mathbf{f}_{\text{mv}}(\gamma(0), 0) =  \mathbf{f}_{\text{mv}}(0, 0) > \mathbf{0}$
trivially  and  $\mathbf{f}_{\text{mv}}(\gamma(\theta), \theta)$ is continuous in $\theta$, there  exists an open interval $(\theta^{-}, \theta^+) \subset (\tilde{\theta}^{-}, \tilde{\theta}^+)$ such that  $0 \in (\theta^{-}, \theta^+) $ and  $\mathbf{f}_{\text{mv}}(\gamma(\theta), \theta) > \mathbf{0}$ for all  $\theta \in (\theta^{-}, \theta^+)$. This proves Eq.\;(\ref{eq:firstCompResult}). For Eq.\;(\ref{eq:gammaPrime}),  the chain rule gives that $\partial_{\gamma} f_{\text{comp}}(\gamma(\theta), \theta)\gamma'(\theta) = - \partial_{\theta}  f_{\text{comp}}(\gamma(\theta), \theta)$ on $(\theta^{-}, \theta^+)$. The first identity in Eq.\;(\ref{eq:partialIdentities}) and $\mathbf{f}_{\text{mv}}(\gamma(\theta), \theta) > \mathbf{0}$ gives that $\partial_{\gamma} f_{\text{comp}}(\gamma(\theta), \theta)$ is non-zero. Dividing by it and using both identities in Eq.\;(\ref{eq:partialIdentities}) yields Eq.\;(\ref{eq:gammaPrime}).
\end{proof}

\noindent Henceforth, $(\theta^{-}, \theta^+)$ denotes the maximal open interval for which the statement of Lemma \ref{ImplicitLemma}  holds.

Proposition \ref{mechKinCellProp} is a direct consequence of Lemma \ref{ImplicitLemma}.  Consider the creases defined implicitly in Lemma \ref{ImplicitLemma} by deforming $( \mathbf{t}_1^r, \mathbf{t}_2^r, \mathbf{t}_3^r, \mathbf{t}_4^r)$ to $(\mathbf{t}_1^r, \mathbf{t}^d_2(\gamma(\theta)), \mathbf{t}_3^d(\theta), \mathbf{t}_4^r)$ for $\theta \in (\theta^{-}, \theta^+)$. It follows that 
\begin{equation}
\begin{aligned}\label{eq:MechConstruct0}
\mathcal{M} = \{ (\mathbf{Q} \mathbf{t}_1^r, \mathbf{Q} \mathbf{t}^d_2(\gamma(\theta)), \mathbf{Q} \mathbf{t}_3^d(\theta), \mathbf{Q}\mathbf{t}_4^r) \colon \mathbf{Q} \in SO(3), \theta \in (\theta^{-}, \theta^+) \big\},
\end{aligned}
\end{equation}
since the lemma parameterizes the mechanism set up to an overall rigid rotation of the vectors.  This parameterization also satisfies $(\mathbf{t}_1^r, \mathbf{t}^d_2(\gamma(0)), \mathbf{t}_3^d(0), \mathbf{t}_4^r) = (\mathbf{t}_1^r, \mathbf{t}_2^r , \mathbf{t}_3^r, \mathbf{t}_4^r).$  

Next, observe that there is a rotation field $\mathbf{R}_0 \colon (\theta^{-}, \theta^+) \rightarrow SO(3)$ with $\mathbf{R}_0(0) = \mathbf{I}$ such that the vector fields  $\mathbf{t}_i \colon (\theta^{-},\theta^+) \rightarrow \mathbb{R}^3$ defined by
\begin{equation}
\begin{aligned}\label{eq:MechConstruct1}
\mathbf{t}_1(\theta) := \mathbf{R}_0(\theta) \mathbf{t}_1^r, \quad \mathbf{t}_2(\theta) := \mathbf{R}_0(\theta) \mathbf{t}_2^d(\gamma(\theta)), \quad \mathbf{t}_3(\theta) :=  \mathbf{R}_0(\theta) \mathbf{t}_3^d(\theta), \quad \mathbf{t}_4(\theta) := \mathbf{R}_0(\theta) \mathbf{t}_4^r
\end{aligned}
\end{equation}\
satisfy the planarity conditions
\begin{equation}
\begin{aligned}\label{eq:MechConstruct2}
\mathbf{e}_3 \cdot (\mathbf{t}_1(\theta) - \mathbf{t}_3(\theta))  = 0, \quad \mathbf{e}_3 \cdot (\mathbf{t}_2(\theta) - \mathbf{t}_4(\theta))  = 0, \quad \mathbf{e}_3 \cdot \big[ (\mathbf{t}_1(\theta) - \mathbf{t}_3(\theta)) \times (\mathbf{t}_2(\theta) - \mathbf{t}_4(\theta)) \big] > 0
\end{aligned}
\end{equation}
for all $\theta \in (\theta^{-}, \theta^+)$ and the initial condition $\mathbf{t}_i(0) = \mathbf{t}_i^r$ for all $i = 1,\ldots,4.$ Furthermore, since $ \mathbf{t}_2^d(\gamma(\theta))$ and $\mathbf{t}_3^d(\theta)$ are analytic by Lemma \ref{ImplicitLemma}, $\mathbf{R}_0(\theta)$ can be chosen to be analytic without loss of generality. The vector fields $\mathbf{t}_i(\theta)$, $i = 1,\ldots, 4$, are the deformed crease vectors introduced in the proposition.
\begin{proof}[Proof of Proposition \ref{mechKinCellProp}.] $\mathbf{t}_1(\theta), \ldots, \mathbf{t}_4(\theta)$ are analytic and possess the desired initial conditions and planarity conditions by construction. To complete the proof, we simply observe from  Eqs.\;(\ref{eq:MechConstruct0}) and (\ref{eq:MechConstruct1}) that $\mathcal{M} = \{ (\mathbf{Q}\mathbf{R}_0^T(\theta) \mathbf{t}_1(\theta), \ldots, \mathbf{Q}\mathbf{R}_0^T(\theta) \mathbf{t}_4(\theta)) \colon  \mathbf{Q} \in SO(3), \theta \in (\theta^{-}, \theta^+)\} = \{ (\mathbf{R} \mathbf{t}_1(\theta), \ldots, \mathbf{R} \mathbf{t}_4(\theta)) \colon  \mathbf{R} \in SO(3), \theta \in (\theta^{-}, \theta^+)\}$, as desired.  
\end{proof}

We now turn to the proof of Lemma \ref{SignsLemma}, which establishes the monotonicity of the functions $|\mathbf{u}(\theta)|$ and $|\mathbf{v}(\theta)|$ on $(\theta^{-}, \theta^+)$ for $\mathbf{u}(\theta) := \mathbf{t}_1(\theta) - \mathbf{t}_3(\theta)$, $\mathbf{v}(\theta) := \mathbf{t}_2(\theta) - \mathbf{t}_4(\theta)$ and $\mathbf{t}_1(\theta), \ldots, \mathbf{t}_4(\theta)$ from Eq.\;(\ref{eq:MechConstruct1}).  
\begin{proof}[Proof of Lemma \ref{SignsLemma}] 
The desired monotonicity follows because we only consider mechanism deformations that preserve the mountain-valley assignment. 
In particular, the definitions above give that $\mathbf{u}(\theta)  = \mathbf{R}_0(\theta)( \mathbf{t}_1^r -  \mathbf{R}_{\mathbf{t}^r_4}(\theta) \mathbf{t}_3^r)$. Thus,  
\begin{equation}
\begin{aligned}
\mathbf{u}'(\theta) \cdot \mathbf{u}(\theta) &= \frac{1}{2}\frac{d}{d\theta}\big(  | \mathbf{u}(\theta)|^2 \big)= \frac{1}{2}\frac{d}{d\theta} \Big[ \big( \mathbf{t}_1^r -  \mathbf{R}_{\mathbf{t}^r_4}(\theta) \mathbf{t}_3^r\big) \cdot \big(\mathbf{t}_1^r -  \mathbf{R}_{\mathbf{t}^r_4}(\theta) \mathbf{t}_3^r\big)\Big] =  - \frac{d}{d \theta}\Big[ \mathbf{t}_1^r \cdot  \mathbf{R}_{\mathbf{t}^r_4}(\theta) \mathbf{t}_3^r \Big]
\\ &= -\Big[|\mathbf{t}_4^r|^{-1} \mathbf{t}_1^r \cdot  \mathbf{R}_{\mathbf{t}^r_4}(\theta) (\mathbf{t}_4^r \times  \mathbf{t}_3^r) \Big]  = -|\mathbf{t}_4^r|^{-1} \mathbf{t}_1^d \cdot ( \mathbf{t}_4^d \times \mathbf{t}_3^d(\theta)) = -|\mathbf{t}_4^r|^{-1} \mathbf{t}_1(\theta) \cdot ( \mathbf{t}_4(\theta) \times \mathbf{t}_3(\theta)) \neq 0
\end{aligned}
\end{equation}
 on $(\theta^-, \theta^+)$ due to the mountain-valley inequalities in Eq.\;(\ref{eq:MVPreserve}). For the other case, a similar set of identities furnishes $\mathbf{v}'(\theta) \cdot \mathbf{v}(\theta) = - \gamma'(\theta) |\mathbf{t}_1^r|^{-1}[ \mathbf{t}_4(\theta) \cdot ( \mathbf{t}_1(\theta) \times \mathbf{t}_2(\theta))] \neq 0$ on $(\theta^{-}, \theta^+)$, where the non-vanishing assertion follows from Eq.\;(\ref{eq:gammaPrime}) and  Eq.\;(\ref{eq:MVPreserve}).  Since $\mathbf{u}'(\theta) \cdot \mathbf{u}(\theta)$ and $\mathbf{v}'(\theta) \cdot \mathbf{v}(\theta)$ are continuous functions that do not vanish on $(\theta^{-}, \theta^+)$, they are either strictly positive or strictly negative.
\end{proof}

\subsection{Mechanism kinematics of neighboring cells}\label{ssec:MechProofs2}
We now construct a mechanism motion of two neighboring unit cells, where the motion is assumed to contain the reference state and to preserve the mountain-valley assignment. 

Clearly, each cell must   deform along its  respective mechanism motion. Thus, as a necessary condition, each cell's kinematics are described by a parameterization of the deformed creases  given by Lemma \ref{ImplicitLemma}, up to an overall rotation. Specifically, we can assume that the first cell deforms via
\begin{equation}
\begin{aligned}
(\mathbf{t}_1^r, \mathbf{t}_2^r, \mathbf{t}_3^r, \mathbf{t}_4^r) \mapsto (\mathbf{t}_1^r, \mathbf{t}_2^d(\gamma(\theta)), \mathbf{t}_3^d(\theta), \mathbf{t}_4^r)  \quad \text{ for some } \theta \in (\theta^{-}, \theta^+).
\end{aligned}
\end{equation}
Then, the second cell must deform as 
\begin{equation}
\begin{aligned}
(\mathbf{t}_1^r, \mathbf{t}_2^r, \mathbf{t}_3^r, \mathbf{t}_4^r) \mapsto \mathbf{R} (\mathbf{t}_1^r, \mathbf{t}_2^d(\gamma(\tilde{\theta})), \mathbf{t}_3^d(\tilde{\theta}), \mathbf{t}_4^r)  \quad \text{ for some } \tilde{\theta} \in (\theta^{-}, \theta^+), \mathbf{R} \in SO(3).
\end{aligned}
\end{equation}
The two cells are subject to the compatibility conditions for joining them together at their boundaries:
\begin{equation}
\begin{aligned}\label{eq:CompatNeighbor}
&(\text{$\mathbf{u}_0$-neighbors:}) &&  \mathbf{R} \mathbf{t}_2^d(\gamma(\tilde{\theta})) =  \mathbf{t}^d_2(\gamma(\theta)) \quad \text{ and } \quad  \mathbf{R} \mathbf{t}_4^r =  \mathbf{t}_4^r,  \\
&(\text{$\mathbf{v}_0$-neighbors:})&&  \mathbf{R} \mathbf{t}_3^d(\tilde{\theta}) =  \mathbf{t}_3^d(\theta) \quad \text{ and } \quad \mathbf{R} \mathbf{t}_1^r = \mathbf{t}_1^r , 
\end{aligned}
\end{equation}
where the notion of $\mathbf{u}_0$ and  $\mathbf{v}_0$-neighbors is as in Fig.\;\ref{Fig:idepat}. A necessary condition for such compatibility is obtained by dotting the two equations together, which gives that
\begin{equation}
\begin{aligned}
&(\text{$\mathbf{u}_0$-neighbors:}) &&  f_{\mathbf{u}_0}(\tilde{\theta}, \theta) :=  \mathbf{t}_2^d(\gamma(\tilde{\theta})) \cdot \mathbf{t}_4^r -  \mathbf{t}^d_2(\gamma(\theta)) \cdot \mathbf{t}_4^r = 0, \\
&(\text{$\mathbf{v}_0$-neighbors:}) &&  f_{\mathbf{v}_0}(\tilde{\theta}, \theta) :=  \mathbf{t}_3^d(\theta) \cdot \mathbf{t}_1^r -  \mathbf{t}^d_3(\theta) \cdot \mathbf{t}_1^r = 0.
\end{aligned}
\end{equation}  
\begin{lemma}\label{ImplicitLemma2}
$\tilde{\theta}(\theta) = \theta$ is the only continuously differentiable function satisfying $\tilde{\theta}(0)  = 0$ and  $f_{\mathbf{u}_0}(\tilde{\theta}(\theta), \theta) = 0$ on $(\theta^{-}, \theta^+)$.  The same result holds for  $f_{\mathbf{v}_0}(\tilde{\theta}, \theta)$.
\end{lemma}
\begin{proof}
We focus on $f_{\mathbf{u}_0}(\tilde{\theta}, \theta)$, since the result for $f_{\mathbf{v}_0}(\tilde{\theta}, \theta)$ follows by the same argument. Again we invoke the implicit function theorem. Observe that $f_{\mathbf{v}_0}(0,0) =0$ and  that 
\begin{equation}
\begin{aligned}
\partial_{\tilde{\theta}} f_{\mathbf{u}_0}(\tilde{\theta},\theta) = \gamma'(\tilde{\theta}) |\mathbf{t}_1^r|^{-1}\mathbf{R}_{\mathbf{t}_1^r}(\gamma(\tilde{\theta})) (\mathbf{t}_1^r \times \mathbf{t}_2^r) \cdot \mathbf{t}_4^r =  \gamma'(\tilde{\theta}) |\mathbf{t}_1^r|^{-1} \big(\mathbf{t}_1^d \times \mathbf{t}_2^d(\gamma(\tilde{\theta}))\big) \cdot \mathbf{t}_4^r.
\end{aligned}
\end{equation}
Since $\mathbf{f}_{\text{mv}}(\gamma(\tilde{\theta}),\tilde{\theta}) > \mathbf{0}$ for all $\tilde{\theta} \in (\theta^{-}, \theta^+)$ and since $\gamma'(\tilde{\theta})$ satisfies Eq.\;(\ref{eq:gammaPrime}),  $\partial_{\tilde{\theta}} f_{\mathbf{u}_0}(\tilde{\theta},\theta) \neq 0$ for all $\tilde{\theta} \in (\theta^{-}, \theta^+)$ and all $\theta$.  Applying the implicit function theorem at the origin $(\tilde{\theta}, \theta) = \mathbf{0}$ gives a unique continuously differentiable $\tilde{\theta}(\theta)$ that solves $f_{\mathbf{u}_0}(\tilde{\theta}(\theta), \theta) = 0$ and $\tilde{\theta}(0) =0$ on some open interval containing the origin. The interval evidently contains $(\theta^{-}, \theta^+)$ because $\partial_{\tilde{\theta}} f_{\mathbf{u}_0}(\tilde{\theta},\theta)$ does not vanish on the latter. Since $f_{\mathbf{u}_0}(\theta, \theta) = 0$ for all $(\theta^{-}, \theta^+)$, $\tilde{\theta}(\theta) = \theta$ is the unique solution, as desired. 
\end{proof}

This lemma shows that the crease-actuation of two neighboring unit cells is the same for any mechanism motion containing the reference configuration.  Setting $\tilde{\theta} = \theta$, the compatibility conditions in Eq.\;(\ref{eq:CompatNeighbor}) become
\begin{equation}
\begin{aligned}\label{eq:finalNeighborCompat}
&(\text{$\mathbf{u}_0$-neighbors:}) &&  \mathbf{R} \mathbf{t}_2^d(\gamma(\theta)) =  \mathbf{t}^d_2(\gamma(\theta)) \quad \text{ and } \quad  \mathbf{R} \mathbf{t}_4^r =  \mathbf{t}_4^r,  \\
&(\text{$\mathbf{v}_0$-neighbors:})&& \mathbf{R} \mathbf{t}_3^d(\theta) =  \mathbf{t}_3^d(\theta) \quad \text{ and } \quad \mathbf{R} \mathbf{t}_1^r = \mathbf{t}_1^r. 
\end{aligned}
\end{equation}
Next, we show that the rotation in these conditions is the identity. This means that the motions of the two unit cells are one and the same, and can be parameterized by the actuation of a single crease. 
\begin{lemma}\label{ImplicitLemma3}
Eq.\;(\ref{eq:finalNeighborCompat}) holds  for $\theta \in (\theta^{-}, \theta^+)$, $\mathbf{R} \in SO(3)$ if and only if $\mathbf{R} = \mathbf{I}$. 
\end{lemma}
\begin{proof}
We focus on the $\mathbf{u}_0$-neighbor case; the other case has an identical proof. Since $\theta \in (\theta^{-}, \theta^+)$, the deformed creases satisfy $\mathbf{t}_2^d(\gamma(\theta)) \cdot (\mathbf{t}_3^d(\theta) \times \mathbf{t}_4^r ) \neq 0$, i.e., they correspond to partly folded mountain-valley assignments. In particular, $\mathbf{t}_4^r$ and $\mathbf{t}_2^d(\gamma(\theta))$ are linearly independent. The only rotation in $SO(3)$ that maps two linearly independent vectors to themselves is $\mathbf{R} = \mathbf{I}$. 
\end{proof}

\subsection{Mechanism kinematics of the full pattern}\label{ssec:MechProofs3}

We now prove Proposition \ref{MechProp}. Specifically, we  construct the mechanism motion based on the crease vectors $\mathbf{t}_1(\theta) = \mathbf{R}_1(\theta) \mathbf{t}_i^r, \ldots, \mathbf{t}_4(\theta) = \mathbf{R}_4(\theta) \mathbf{t}_4^r$,  $\Omega_{\text{cell}}^{\theta}$ and  $\mathcal{T}_{\text{ori}}^{\theta}$ introduced in Eqs.\;(\ref{eq:MechConstruct3}-\ref{eq:MechConstructFinal}). We  then use the uniqueness arguments in Lemma \ref{ImplicitLemma2} and \ref{ImplicitLemma3} to conclude that this is the only mechanism motion that preserves the mountain-valley assignment of the reference configuration. 
\begin{proof}[Proof of Proposition \ref{MechProp}]
We first verify the existence of a continuous and rigid deformation. Note that, for any $\theta \in (\theta^{-}, \theta^+)$, the deformation $\mathbf{y}_{\theta}(\mathbf{x}) = \mathbf{R}_i(\theta) \mathbf{x}$ for $\mathbf{x} \in  \mathcal{P}_i$, $i = 1,\ldots 4$ is a well-defined, rigid and continuous deformation of $\Omega_{\text{cell}}$ with $\mathbf{y}_{\theta}(\Omega_{\text{cell}}) = \Omega_{\text{cell}}^{\theta}$  by Lemma \ref{ImplicitLemma} and the definitions above.   Let $\mathbf{y}^{R}_{\theta}(\mathbf{x}) := \mathbf{y}_{\theta}(\mathbf{x} - \mathbf{u}_0) + \mathbf{u}(\theta)$ for all $\mathbf{x} \in \Omega_{\text{cell}} + \mathbf{u}_0$.  Observe that $\mathbf{t}_1^r$ is on the boundary of  both $\Omega_{\text{cell}}$ and  $\Omega_{\text{cell}} + \mathbf{u}_0$. We  have $\mathbf{y}_{\theta}^R(\mathbf{t}_1^r) = \mathbf{y}_{\theta}( \mathbf{t}_3^r) + \mathbf{u}(\theta) = \mathbf{t}_3(\theta) + \mathbf{u}(\theta) = \mathbf{t}_1(\theta)$ and $\mathbf{y}_{\theta}(\mathbf{t}_1^r) = \mathbf{t}_1(\theta)$, verifying continuity at this boundary point.  We can verify continuity at the entire boundary $ (\Omega_{\text{cell}} + \mathbf{u}_0) \cap \Omega_{\text{cell}}$ in a similar fashion. We also clearly have $\mathbf{y}^R_{\theta}(\Omega_{\text{cell}}+ \mathbf{u}_0) = \Omega_{\text{cell}}^{\theta} + \mathbf{u}(\theta)$. Using the same ideas, we can verify that $\mathbf{y}_{\theta}^{U}(\mathbf{x}) := \mathbf{y}_{\theta}(\mathbf{x} -  \mathbf{v}_0) + \mathbf{v}(\theta)$, $\mathbf{x} \in \Omega_{\text{cell}} + \mathbf{v}_0$ connects perfectly with $\mathbf{y}_{\theta}(\mathbf{x})$ at the boundary $\Omega_{\text{cell}} \cap ( \Omega_{\text{cell}} + \mathbf{v}_0)$ and that $\mathbf{y}_{\theta}^{U}(\Omega_{\text{cell}} + \mathbf{v}_0) = \Omega_{\text{cell}}^{\theta} + \mathbf{v}(\theta)$. By these results, there is a continuous function $\mathbf{y}_{\theta} \colon \mathcal{T}_{\text{cell}} \rightarrow \mathbb{R}^3$ that satisfies 
\begin{equation}
\begin{aligned}\label{eq:MechConstructFinalFinal}
\mathbf{y}_{\theta}(\mathbf{x}) = \mathbf{y}_{\theta}(\mathbf{x} - i \mathbf{u}_0 - j \mathbf{v}_0) + i \mathbf{u}(\theta) + j \mathbf{v}(\theta), \quad  \mathbf{x} \in \Omega_{\text{cell}} + i \mathbf{u}_0 + j \mathbf{v}_0
\end{aligned}
\end{equation}
for all $i, j \in \mathbb{Z}$ with $\mathbf{y}_{\theta}(\mathcal{T}_{\text{cell}}) = \mathcal{T}_{\text{cell}}^{\theta}$, as desired.  

Next, we show that the family of deformations $\mathbf{y}_{\theta}(\mathcal{T}_{\text{cell}}), \theta \in (\theta^{-}, \theta^+)$ is a mechanism motion. Indeed, each $\theta$-dependent deformation  here is a rigid deformation since it is built by appropriately repeating a rigid deformation of $\Omega_{\text{cell}}$. Furthermore, the parameterization is analytic in $\theta$ since it is analytic in $\mathbf{R}_0(\theta)$ and $(\mathbf{t}_1^r, \mathbf{t}^d_2(\gamma(\theta)), \mathbf{t}_3^d(\theta), \mathbf{t}_4^r)$, which are themselves analytic in $\theta$ by Lemma \ref{ImplicitLemma} and the definitions in Eq.\;(\ref{eq:MechConstruct1}-\ref{eq:MechConstructFinalFinal}). 

Finally, it follows from by iteratively applying Lemmas \ref{ImplicitLemma2} and \ref{ImplicitLemma3} that this mechanism motion  is the unique one that contains the reference configuration and does not change the mountain-valley assignment. 
\end{proof}

\section{Solving the auxiliary PDE}\label{sec:ExistencePDE}

We prove the following result:
\begin{proposition}\label{AuxPDEProp}
There are vector and scalar fields $\boldsymbol{\omega}(\mathbf{x})$ and $ \xi(\mathbf{x})$ smoothly defined on a neighborhood of $\overline{\Omega}$ and solving the PDE system 
 \begin{equation}
\begin{aligned}\label{eq:theAuxPDEFinal}
&\partial_{\mathbf{u}_0} \Big(\mathbf{R}_{\emph{eff}}(\mathbf{x}) \big[ \boldsymbol{\omega}(\mathbf{x}) \times \mathbf{v}(\theta(\mathbf{x})) + \xi(\mathbf{x}) \mathbf{v}'(\theta(\mathbf{x})) \big] \Big)  + \frac{1}{2} \partial_{\mathbf{u}_0} \partial_{\mathbf{v}_0} \Big( \mathbf{R}_{\emph{eff}}(\mathbf{x}) \big[ \mathbf{t}_1(\theta(\mathbf{x})) + \mathbf{t}_3(\theta(\mathbf{x})) \big]\Big)  \\
&\quad = \partial_{\mathbf{v}_0} \Big(\mathbf{R}_{\emph{eff}}(\mathbf{x}) \big[ \boldsymbol{\omega}(\mathbf{x}) \times \mathbf{u}(\theta(\mathbf{x})) + \xi(\mathbf{x}) \mathbf{u}'(\theta(\mathbf{x})) \big] \Big)  + \frac{1}{2} \partial_{\mathbf{u}_0} \partial_{\mathbf{v}_0} \Big( \mathbf{R}_{\emph{eff}}(\mathbf{x}) \big[ \mathbf{t}_2(\theta(\mathbf{x})) + \mathbf{t}_4(\theta(\mathbf{x})) \big]\Big)
\end{aligned}
\end{equation}
on $\Omega$, under the assumptions of Theorem \ref{MainTheorem}.
\end{proposition}
As a reminder, there are two cases to consider based on the Poisson's ratio of the design $\nu(\theta)$ in Eq.\;(\ref{eq:PoissonsRatioDesign}). If  $\nu(\theta)$ is positive, the fields $\theta(\mathbf{x})$, $\boldsymbol{\omega}_{\mathbf{u}_0}(\mathbf{x})$ and  $\boldsymbol{\omega}_{\mathbf{v}_0}(\mathbf{x})$ are assumed to be smooth on a neighborhood of  $\overline{\Omega}$. If $\nu(\theta)$ is negative, then $\theta(\mathbf{x})$, $\boldsymbol{\omega}_{\mathbf{u}_0}(\mathbf{x})$ and $\boldsymbol{\omega}_{\mathbf{v}_0}(\mathbf{x})$ are assumed to be analytic instead. In either case, these fields solve the PDEs in Eqs.\;(\ref{eq:effPDE1}-\ref{eq:effPDE2}) on $\Omega$ and $\mathbf{R}_{\text{eff}}(\mathbf{x})$ is the unique rotation field solving the Pfaff system in Proposition \ref{firstProp}. We break the proof of Proposition \ref{AuxPDEProp} into several steps.

\subsection{Helmholtz--Hodge reformulation}

First, we make some definitions to rewrite the PDE in Eq.\;(\ref{eq:theAuxPDEFinal}) in a helpful way. Label the inhomogeneous terms from the PDE as 
\begin{equation}
\begin{aligned}
\mathbf{f}(\mathbf{x}) :=  \frac{1}{2} \partial_{\mathbf{u}_0} \partial_{\mathbf{v}_0} \Big( \mathbf{R}_{\text{eff}}(\mathbf{x}) \big[ \mathbf{t}_1(\theta(\mathbf{x})) + \mathbf{t}_3(\theta(\mathbf{x})) \big]\Big)  - \frac{1}{2} \partial_{\mathbf{u}_0} \partial_{\mathbf{v}_0} \Big( \mathbf{R}_{\text{eff}}(\mathbf{x}) \big[ \mathbf{t}_2(\theta(\mathbf{x})) + \mathbf{t}_4(\theta(\mathbf{x})) \big]\Big)
\end{aligned}
\end{equation}
and note $\mathbf{f}(\mathbf{x})$ is smooth by assumption. Next, for any tensor field $\mathbf{A}(\mathbf{x})\in \mathbb{R}^{3\times2}$, define the divergence and curl-like differential operators  
\begin{equation}
\begin{aligned}
\nabla_0 \cdot \mathbf{A}(\mathbf{x}) := \partial_{\mathbf{u}_0} \mathbf{A}(\mathbf{x}) \tilde{\mathbf{e}}_1 +  \partial_{\mathbf{v}_0} \mathbf{A}(\mathbf{x}) \tilde{\mathbf{e}}_2, \quad \nabla_0^{\perp} \cdot \mathbf{A}(\mathbf{x}) :=  \partial_{\mathbf{v}_0} \mathbf{A}(\mathbf{x}) \tilde{\mathbf{e}}_1 -  \partial_{\mathbf{u}_0} \mathbf{A}(\mathbf{x}) \tilde{\mathbf{e}}_2,
\end{aligned}
\end{equation}
where $\tilde{\mathbf{e}}_{1,2}$ are the standard basis of $\mathbb{R}^2$.
Likewise, for any vector and scalar fields $\mathbf{w}(\mathbf{x}) \in \mathbb{R}^{3}$ and $\zeta(\mathbf{x}) \in \mathbb{R}$, define the analogous gradient and rotated gradient-like differential operators 
\begin{equation}
\begin{aligned}\label{eq:scalarVect}
&\nabla_0 \mathbf{w}(\mathbf{x}) : = \partial_{\mathbf{u}_0} \mathbf{w}(\mathbf{x}) \otimes \tilde{\mathbf{e}}_1+ \partial_{\mathbf{v}_0}  \mathbf{w}(\mathbf{x}) \otimes  \tilde{\mathbf{e}}_2,  &&\nabla^{\perp}_0 \mathbf{w}(\mathbf{x}) : = \partial_{\mathbf{v}_0} \mathbf{w}(\mathbf{x}) \otimes \tilde{\mathbf{e}}_1- \partial_{\mathbf{u}_0}  \mathbf{w}(\mathbf{x}) \otimes  \tilde{\mathbf{e}}_2, \\
&\nabla_0 \zeta(\mathbf{x}) : =\partial_{\mathbf{u}_0} \zeta(\mathbf{x}) \tilde{\mathbf{e}}_1 + \partial_{\mathbf{v}_0} \zeta(\mathbf{x}) \tilde{\mathbf{e}}_2, && \nabla_0^{\perp} \zeta(\mathbf{x}) : =\partial_{\mathbf{u}_0} \zeta(\mathbf{x}) \tilde{\mathbf{e}}_1 + \partial_{\mathbf{v}_0} \zeta(\mathbf{x}) \tilde{\mathbf{e}}_2.
\end{aligned}
\end{equation}
With these definitions, solving the PDE in Eq.\;(\ref{eq:theAuxPDEFinal}) for $\boldsymbol{\omega}(\mathbf{x})$ and $\xi(\mathbf{x})$ is equivalent to solving the system
\begin{equation}
\begin{aligned}\label{eq:theAuxPDE1}
\begin{cases}
\nabla_0^{\perp} \cdot \mathbf{A}(\mathbf{x}) = \mathbf{f}(\mathbf{x}) \\
\mathbf{A}(\mathbf{x})\tilde{\mathbf{e}}_1= \mathbf{R}_{\text{eff}}(\mathbf{x}) \big[ \boldsymbol{\omega}(\mathbf{x}) \times \mathbf{u}(\theta(\mathbf{x})) + \xi(\mathbf{x}) \mathbf{u}'(\theta(\mathbf{x})) \big] \\
\mathbf{A}(\mathbf{x}) \tilde{\mathbf{e}}_2 =  \mathbf{R}_{\text{eff}}(\mathbf{x})  \big[ \boldsymbol{\omega}(\mathbf{x}) \times \mathbf{v}(\theta(\mathbf{x})) + \xi(\mathbf{x}) \mathbf{v}'(\theta(\mathbf{x})) \big] 
\end{cases}
\end{aligned}
\end{equation}
for $\mathbf{A}(\mathbf{x})$, $\boldsymbol{\omega}(\mathbf{x})$ and $\xi(\mathbf{x})$. This allows us to view Eq.\;(\ref{eq:theAuxPDEFinal})  as a PDE with linear algebraic constraints. We  deal with the PDE part of Eq.\;(\ref{eq:theAuxPDE1}) first, and then discuss the constraints.

Motivated by the  Helmholtz--Hodge decomposition for vector fields, we seek a solution to $\nabla_0^{\perp} \cdot \mathbf{A}(\mathbf{x}) = \mathbf{f}(\mathbf{x})$ of the form  
\begin{equation}
\begin{aligned}
\mathbf{A}(\mathbf{x}) = \nabla_0 \mathbf{p}(\mathbf{x}) + \nabla_0^{\perp} \mathbf{q}(\mathbf{x}),
\end{aligned}
\end{equation}
where $\mathbf{p}, \mathbf{q} \colon \Omega \rightarrow \mathbb{R}^3$ are the unknowns. Eventually, we shall choose the vector fields $\mathbf{p}(\mathbf{x})$, $\mathbf{q}(\mathbf{x})$  to solve Eq.\;(\ref{eq:theAuxPDE1}). For now, observe the following fact. 
\begin{lemma}\label{LemmaPDE1} 
For any smooth $\mathbf{p}(\mathbf{x})$, there is a smooth $\mathbf{q}(\mathbf{x})$ such that $\mathbf{A}(\mathbf{x}) = \nabla_0 \mathbf{p}(\mathbf{x}) + \nabla_0^{\perp} \mathbf{q}(\mathbf{x}) $ satisfies 
\begin{equation}
\begin{aligned}\label{eq:firstStepPDE}
\nabla_0^{\perp} \cdot \mathbf{A}(\mathbf{x}) = \mathbf{f}(\mathbf{x}) \quad \text{ on $\Omega$}.
\end{aligned}
\end{equation}
\end{lemma} 
\begin{proof}
The definitions above give that
\begin{equation}
\begin{aligned}
&\nabla_0^{\perp} \cdot \nabla_0  \mathbf{p}(\mathbf{x}) = \nabla_0^{\perp} \cdot \Big(\partial_{\mathbf{u}_0} \mathbf{p}(\mathbf{x}) \otimes \tilde{\mathbf{e}}_1+ \partial_{\mathbf{v}_0}  \mathbf{p}(\mathbf{x}) \otimes  \tilde{\mathbf{e}}_2 \Big)  = \partial_{\mathbf{v}_0} \partial_{\mathbf{u}_0}  \mathbf{p}(\mathbf{x}) - \partial_{\mathbf{u}_0} \partial_{\mathbf{v}_0} \mathbf{p}(\mathbf{x}) =  \mathbf{0}, \\
&\nabla_0^{\perp} \cdot \nabla_0^{\perp} \mathbf{q}(\mathbf{x})  =   \nabla_0^{\perp} \cdot \Big(\partial_{\mathbf{v}_0} \mathbf{q}(\mathbf{x}) \otimes \tilde{\mathbf{e}}_1 -\partial_{\mathbf{u}_0}  \mathbf{q}(\mathbf{x}) \otimes  \tilde{\mathbf{e}}_2 \Big) = \partial_{\mathbf{u}_0} \partial_{\mathbf{u}_0} \mathbf{q}(\mathbf{x}) + \partial_{\mathbf{v}_0} \partial_{\mathbf{v}_0} \mathbf{q}(\mathbf{x}) = \Delta_0 \mathbf{q}(\mathbf{x}),
\end{aligned}
\end{equation}
where $\Delta_0 = \partial^2_{\mathbf{u}_0} + \partial_{\mathbf{v}_0}^2$.
It follows that $\nabla_0^{\perp} \cdot \mathbf{A}(\mathbf{x}) = \Delta_0 \mathbf{q}(\mathbf{x})$ regardless of the choice of $\mathbf{p}(\mathbf{x})$. To solve Eq.\;(\ref{eq:firstStepPDE}), we seek $\mathbf{q}(\mathbf{x})$ such that   $\Delta_0 \mathbf{q}(\mathbf{x}) = \mathbf{f}(\mathbf{x})$. Since $\mathbf{u}_0$ and $\mathbf{v}_0$ are not parallel, this equation becomes the classical Poisson's equation after a linear change of coordinates. Therefore, by standard PDE theory (see, e.g., \cite{evans2022partial}), we can  assert the existence of a unique and smooth solution $\mathbf{q}(\mathbf{x})$ to $\Delta_0 \mathbf{q}(\mathbf{x}) = \mathbf{f}(\mathbf{x})$ subject to $\mathbf{q}(\mathbf{x}) = \mathbf{0}$ at $\partial \Omega$, since $\mathbf{f}(\mathbf{x})$ and $\partial\Omega$ are smooth.  
\end{proof}

The problem now is to find a $\mathbf{p}(\mathbf{x})$ enforcing the linear constraints in Eq.\;(\ref{eq:theAuxPDE1}). First, we characterize these constraints using linear algebra. Note $\mathbf{A} \colon \mathbf{B} : = \text{Tr}( \mathbf{A}^T \mathbf{B})$ for  $\mathbf{A}, \mathbf{B} \in \mathbb{R}^{n\times m}$. 

\begin{lemma}\label{LemmaPDE2} 
The following two statements are equivalent:
\begin{enumerate}
\item  $\mathbf{A}(\mathbf{x})$ satisfies $\mathbf{L}_i(\mathbf{x})  \colon \mathbf{A}(\mathbf{x}) = 0$, $i = 1,2$, for 
\begin{equation}
\begin{aligned}\label{eq:L1L2Def}
&\mathbf{L}_{1}(\mathbf{x}) := \mathbf{R}_{\emph{eff}}(\mathbf{x}) \big[ \mathbf{v}(\theta(\mathbf{x}) )\otimes \tilde{\mathbf{e}}_1 + \mathbf{u}(\theta(\mathbf{x}) )\otimes \tilde{\mathbf{e}}_2\big] , \\
&\mathbf{L}_{2}(\mathbf{x}) := \mathbf{R}_{\emph{eff}}(\mathbf{x}) \big[ \{\mathbf{v}(\theta(\mathbf{x})) \cdot \mathbf{v}'(\theta(\mathbf{x}))\}  \mathbf{u}(\theta(\mathbf{x}) )\otimes \tilde{\mathbf{e}}_1 -  \{\mathbf{u}(\theta(\mathbf{x})) \cdot \mathbf{u}'(\theta(\mathbf{x}))\}  \mathbf{v}(\theta(\mathbf{x}) )\otimes \tilde{\mathbf{e}}_2\big] .
\end{aligned}
\end{equation}
\item $\mathbf{A}(\mathbf{x})$ is given by 
\begin{equation}
\begin{aligned}
\mathbf{A}(\mathbf{x})\tilde{\mathbf{e}}_1= \mathbf{R}_{\emph{eff}}(\mathbf{x}) \big[ \boldsymbol{\omega}(\mathbf{x}) \times \mathbf{u}(\theta(\mathbf{x})) + \xi(\mathbf{x}) \mathbf{u}'(\theta(\mathbf{x})) \big], \\
\mathbf{A}(\mathbf{x}) \tilde{\mathbf{e}}_2 =  \mathbf{R}_{\emph{eff}}(\mathbf{x})  \big[ \boldsymbol{\omega}(\mathbf{x}) \times \mathbf{v}(\theta(\mathbf{x})) + \xi(\mathbf{x}) \mathbf{v}'(\theta(\mathbf{x})) \big],
\end{aligned}
\end{equation}
for some $\boldsymbol{\omega}(\mathbf{x})$ and $\xi(\mathbf{x})$ on $\Omega$. 
\end{enumerate}
\end{lemma}
\begin{proof}
 Suppressing the $\mathbf{x}$ dependence,  observe that a general parameterization for $\mathbf{A} \in \mathbb{R}^{3 \times2}$ is $\mathbf{A} \tilde{\mathbf{e}}_1 = \mathbf{R}_{\text{eff}} \mathbf{a}$ and $\mathbf{A} \tilde{\mathbf{e}}_2 := \mathbf{R}_{\text{eff}} \mathbf{b}$ for some $\mathbf{a}, \mathbf{b} \in \mathbb{R}^3$, since  $\mathbf{R}_{\text{eff}}$ is a rotation.  Next observe that general parameterizations for $\mathbf{a}$ and $\mathbf{b}$ are   
 \begin{equation}
     \begin{aligned}\label{eq:getABParam1}
        \mathbf{a} =  \boldsymbol{\omega}_a \times \mathbf{u}(\theta) + \xi_a \mathbf{u}'(\theta), \quad \mathbf{b} = \boldsymbol{\omega}_b \times \mathbf{v}(\theta) + \xi_b \mathbf{v}'(\theta) 
     \end{aligned}
 \end{equation}
 for  some $\boldsymbol{\omega}_{a,b} \in \mathbb{R}^3$ and $\xi_{a,b} \in \mathbb{R}$, since by Lemma \ref{SignsLemma} both $\mathbf{u}'(\theta)$ and $\mathbf{v}'(\theta)$ have non-zero components in the $\mathbf{u}(\theta)$ and $\mathbf{v}(\theta)$ directions, respectively. In fact, $\boldsymbol{\omega}_{a,b}$ can be further constrained, as there are redundancies in the parameterization.  Write 
 \begin{equation}
     \begin{aligned}\label{eq:getABParam2}
       \boldsymbol{\omega}_a := \alpha_a \mathbf{u}(\theta) + \beta_a \mathbf{v}(\theta) + \gamma_a \mathbf{e}_3, \quad  \boldsymbol{\omega}_b := \alpha_b \mathbf{u}(\theta) + \beta_b \mathbf{v}(\theta) + \gamma_b \mathbf{e}_3  
     \end{aligned}
 \end{equation}
and observe from Eqs.\;(\ref{eq:getABParam1}-\ref{eq:getABParam2}) that  
\begin{equation}
\begin{aligned}
\mathbf{a} = \beta_a \mathbf{v}(\theta) \times \mathbf{u}(\theta) + \gamma_a \mathbf{e}_3 \times \mathbf{u}(\theta) + \xi_a \mathbf{u}'(\theta), \quad 
\mathbf{b} &=  \alpha_b \mathbf{u}(\theta) \times \mathbf{v}(\theta) + \gamma_b \mathbf{e}_3 \times \mathbf{v}(\theta) + \xi_b \mathbf{v}'(\theta).
\end{aligned}
\end{equation}
Since $\alpha_a$ is not present in the first equation and  $\beta_b$ is not present in the second, we can without loss of generality set 
\begin{equation}
\begin{aligned}\label{eq:getDegeneracies}
\alpha_a = \alpha_b = \alpha, \quad \beta_a =  \beta_b = \beta 
\end{aligned}
\end{equation}
for some $\alpha$ and $\beta$. Taking $\mathbf{L}_1 := \mathbf{R}_{\text{eff}}[\mathbf{v}(\theta) \otimes \tilde{\mathbf{e}}_1 + \mathbf{u}(\theta) \otimes \tilde{\mathbf{e}}_2]$ and $\mathbf{L}_2 := \mathbf{R}_{\text{eff}}[\{ \mathbf{v}(\theta) \cdot \mathbf{v}'(\theta)\} \mathbf{u}(\theta) \otimes \tilde{\mathbf{e}}_1 - \{\mathbf{u}(\theta) \cdot \mathbf{u}'(\theta)\} \mathbf{v}(\theta) \otimes \tilde{\mathbf{e}}_2]$ as in Eq.\;(\ref{eq:L1L2Def}), we find that the above general parameterization for $\mathbf{A}$ satisfies  
\begin{equation}
\begin{aligned}
\mathbf{L}_1 \colon \mathbf{A} &= \mathbf{v}(\theta) \cdot  \big[ \boldsymbol{\omega}_a \times \mathbf{u}(\theta)   + \xi_a \mathbf{u}'(\theta)\big] + \mathbf{u}(\theta) \cdot  \big[ \boldsymbol{\omega}_b \times \mathbf{v}(\theta)   + \xi_b \mathbf{v}'(\theta)\big]    \\
& =  [\boldsymbol{\omega}_a - \boldsymbol{\omega}_b \big]  \cdot (\mathbf{u}(\theta) \times \mathbf{v}(\theta)) + ( \xi_a - \xi_b)  \mathbf{v}(\theta) \cdot \mathbf{u}'(\theta) \\
& = (\gamma_a - \gamma_b) \big[ \mathbf{e}_3 \cdot (  \mathbf{u}(\theta) \times \mathbf{v}(\theta))  \big] +  ( \xi_a - \xi_b)  \mathbf{v}(\theta) \cdot \mathbf{u}'(\theta)  , \\
\mathbf{L}_2 \colon \mathbf{A} &= \{ \mathbf{v}(\theta) \cdot \mathbf{v}'(\theta)\}  \mathbf{u}(\theta)  \cdot  \big[ \boldsymbol{\omega}_a \times \mathbf{u}(\theta)   + \xi_a \mathbf{u}'(\theta)\big] -  \{ \mathbf{u}(\theta) \cdot \mathbf{u}'(\theta)\}  \mathbf{v}(\theta) \cdot  \big[ \boldsymbol{\omega}_b \times \mathbf{v}(\theta)   + \xi_b \mathbf{v}'(\theta)\big]   \\
& = (\xi_a - \xi_b) \{ \mathbf{v}(\theta) \cdot \mathbf{v}'(\theta)\}  \{ \mathbf{u}(\theta) \cdot \mathbf{u}'(\theta)\}.
\end{aligned}
\end{equation}
Setting both of these to zero gives that $\xi_a = \xi_b = \xi$  for some $\xi$ and $\gamma_a = \gamma_b = \gamma$ for some $\gamma$. By Eq.\;(\ref{eq:getDegeneracies}),  $\boldsymbol{\omega}_a = \boldsymbol{\omega}_b = \alpha \mathbf{u}(\theta)  + \beta \mathbf{v}(\theta) + \gamma \mathbf{e}_3 := \boldsymbol{\omega}$ for a general $\boldsymbol{\omega}$. It follows that $\mathbf{A} \tilde{\mathbf{e}}_1 = \mathbf{R} \big[ \boldsymbol{\omega} \times \mathbf{u}(\theta) + \xi \mathbf{u}'(\theta)  \big] $ and $\mathbf{A} \tilde{\mathbf{e}}_2 = \mathbf{R} \big[ \boldsymbol{\omega} \times \mathbf{v}(\theta) + \xi \mathbf{v}'(\theta)  \big] $.
\end{proof}

Lemmas \ref{LemmaPDE1} and \ref{LemmaPDE2} combine to simplify Eq.\;(\ref{eq:theAuxPDE1}) as follows. Write $\mathbf{A}(\mathbf{x}) =\nabla_0 \mathbf{p}(\mathbf{x}) + \nabla_0^{\perp} \mathbf{q}(\mathbf{x})$ for a smooth $\mathbf{q}(\mathbf{x})$ as in Lemma \ref{LemmaPDE1}, so that $\nabla_0^{\perp} \cdot \mathbf{A}(\mathbf{x}) = \mathbf{f}(\mathbf{x})$. Then, use Lemma \ref{LemmaPDE2} to rewrite the algebraic constraints in Eq.\;(\ref{eq:theAuxPDE1}) equivalently via the inner product constraints involving  $\mathbf{L}_i(\mathbf{x})$. This manipulation results in a PDE in the unknown vector field $\mathbf{p}(\mathbf{x})$ of the form
\begin{equation}
\begin{aligned}\label{eq:needToSolveP}
\mathbf{L}_i(\mathbf{x}) \colon \nabla_0 \mathbf{p}(\mathbf{x}) = q_i(\mathbf{x}), \quad i = 1,2,
\end{aligned}
\end{equation}
where $q_i(\mathbf{x}) := - \mathbf{L}_i(\mathbf{x}) \colon \nabla_0^{\perp} \mathbf{q}(\mathbf{x})$. Solving Eq.\;(\ref{eq:needToSolveP}) is equivalent to solving (\ref{eq:theAuxPDE1}), which itself is a rewrite of  Eq.\;(\ref{eq:theAuxPDEFinal}).  
We turn to the question of finding $\mathbf{p}(\mathbf{x})$.

\subsection{A linear second-order PDE in two variables}
Here we reduce the PDE in Eq.\;(\ref{eq:needToSolveP}) for the unknown vector field $\mathbf{p}(\mathbf{x})$ to a second order system in two scalar fields that can be solved using standard PDE tools. 

Start with an ansatz of the form
\begin{equation}
\begin{aligned}\label{eq:parameterizeP}
\mathbf{p}(\mathbf{x}) := \mathbf{R}_{\text{eff}}(\mathbf{x}) \Big(  \lambda_u(\mathbf{x})  \mathbf{u}^r(\theta(\mathbf{x})) + \lambda_v(\mathbf{x})   \mathbf{v}^r(\theta(\mathbf{x}))\Big) 
\end{aligned}
\end{equation}
where $\{  \mathbf{u}^r(\theta), \mathbf{v}^r(\theta), \mathbf{e}_3\}$ is the reciprocal basis to $\{ \mathbf{u}(\theta), \mathbf{v}(\theta), \mathbf{e}_3\}$. Substituting Eq.\;(\ref{eq:parameterizeP}) into Eq.\;(\ref{eq:needToSolveP}) and using the definitions of $\mathbf{L}_i(\mathbf{x})$, $i = 1,2$, in Eq.\;(\ref{eq:L1L2Def}) furnishes  PDEs for the scalar fields $\lambda_{u}(\mathbf{x})$ and $\lambda_v(\mathbf{x})$:
\begin{equation}
\begin{aligned}\label{eq:firstOrderToSolve}
\begin{cases}
\partial_{\mathbf{v}_0} \lambda_{u}(\mathbf{x}) + \partial_{\mathbf{u}_0} \lambda_v(\mathbf{x})  + a_u(\mathbf{x}) \lambda_u(\mathbf{x}) + a_v(\mathbf{x}) \lambda_v(\mathbf{x})   = q_1(\mathbf{x}) \\
\big[\mathbf{v}(\theta(\mathbf{x})) \cdot \mathbf{v}'(\theta(\mathbf{x})) \big]  \partial_{\mathbf{u}_0}\lambda_u(\mathbf{x})  - \big[\mathbf{u}(\theta(\mathbf{x})) \cdot \mathbf{u}'(\theta(\mathbf{x})) \big]  \partial_{\mathbf{v}_0}\lambda_v(\mathbf{x}) + b_u(\mathbf{x}) \lambda_u(\mathbf{x}) + b_v(\mathbf{x}) \lambda_v(\mathbf{x})  = q_2(\mathbf{x}).
\end{cases}&
\end{aligned}
\end{equation}
The coefficients in these PDEs satisfy 
\begin{equation}
\begin{aligned}\label{eq:theCoeff0}
a_u(\mathbf{x}) := \mathbf{L}_1(\mathbf{x}) \colon \nabla_0 \big[\mathbf{R}_{\text{eff}}(\mathbf{x}) \mathbf{u}^r(\theta(\mathbf{x}))  \big], \quad a_v(\mathbf{x}) := \mathbf{L}_1(\mathbf{x}) \colon \nabla_0 \big[\mathbf{R}_{\text{eff}}(\mathbf{x}) \mathbf{v}^r(\theta(\mathbf{x}))  \big], \\
b_u(\mathbf{x}) := \mathbf{L}_2(\mathbf{x}) \colon \nabla_0 \big[\mathbf{R}_{\text{eff}}(\mathbf{x}) \mathbf{u}^r(\theta(\mathbf{x}))  \big], \quad b_v(\mathbf{x}) := \mathbf{L}_2(\mathbf{x}) \colon \nabla_0 \big[\mathbf{R}_{\text{eff}}(\mathbf{x}) \mathbf{v}^r(\theta(\mathbf{x}))  \big].
\end{aligned} 
\end{equation}
In fact, the explicit dependence on $\mathbf{R}_{\text{eff}}(\mathbf{x})$  cancels out in Eq.\;(\ref{eq:theCoeff0}) by expanding out the $\nabla_0$ terms and using that  $\mathbf{L}_i(\mathbf{x})$, $i=1,2$, and $\mathbf{R}_{\text{eff}}(\mathbf{x})$ satisfy Eq.\;(\ref{eq:L1L2Def}) and Eq.\;(\ref{eq:firstSurface1}), respectively.  Thus, the smoothness/analyticity of $a_u(\mathbf{x}), \ldots, b_v(\mathbf{x})$ on $\overline{\Omega}$ follows from the same for $\theta(\mathbf{x})$, $\boldsymbol{\omega}_{\mathbf{u}_0}(\mathbf{x})$ and $\boldsymbol{\omega}_{\mathbf{v}_0}(\mathbf{x})$. 

Next, consider $\lambda_{u}(\mathbf{x})$ and $\lambda_{v}(\mathbf{x})$ of the form 
\begin{equation}
\begin{aligned}\label{eq:theLambdasAnsatz}
&\lambda_u(\mathbf{x}) = c_{uu}(\mathbf{x}) \partial_{\mathbf{u}_0} \psi(\mathbf{x}) + c_{uv}(\mathbf{x}) \partial_{\mathbf{v}_0} \psi(\mathbf{x})  +  d_{uu}(\mathbf{x}) \partial_{\mathbf{u}_0} \varphi(\mathbf{x}) + d_{uv}(\mathbf{x}) \partial_{\mathbf{v}_0} \varphi(\mathbf{x})  , \\
&\lambda_v(\mathbf{x}) = c_{vu}(\mathbf{x}) \partial_{\mathbf{u}_0} \psi(\mathbf{x}) + c_{vv}(\mathbf{x}) \partial_{\mathbf{v}_0} \psi(\mathbf{x})  +  d_{vu}(\mathbf{x}) \partial_{\mathbf{u}_0} \varphi(\mathbf{x}) + d_{vv}(\mathbf{x}) \partial_{\mathbf{v}_0} \varphi(\mathbf{x})  
\end{aligned}
\end{equation}
for $\psi, \varphi \colon \Omega \rightarrow \mathbb{R}$. We seek a convenient choice of coefficients $c_{uu}(\mathbf{x}), \ldots, d_{vv}(\mathbf{x})$ to diagonalize the PDE in Eq.\;(\ref{eq:firstOrderToSolve}) at highest order. We succeed as follows. 
\begin{lemma}
Making the choices
\begin{equation}
\begin{aligned}\label{eq:chooseCoeff}
&c_{uu}(\mathbf{x}) = -1, \quad c_{vv}(\mathbf{x}) = 1, \quad c_{vu}(\mathbf{x}) = c_{uv}(\mathbf{x}) =0, \\ 
& d_{uu}(\mathbf{x}) = d_{vv}(\mathbf{x}) =0, \quad d_{vu}(\mathbf{x}) = -1, \quad d_{uv}(\mathbf{x}) = -\tfrac{\mathbf{u}(\theta(\mathbf{x})) \cdot \mathbf{u}'(\theta(\mathbf{x}))}{\mathbf{v}(\theta(\mathbf{x})) \cdot \mathbf{v}'(\theta(\mathbf{x}))} 
\end{aligned}
\end{equation}
in Eq.\;(\ref{eq:theLambdasAnsatz}) brings the first order system in Eq.\;(\ref{eq:firstOrderToSolve}) into the  second order system
\begin{equation}
\begin{aligned}\label{eq:L1L2L}
\begin{cases}
\mathcal{L}(\mathbf{x}) \psi(\mathbf{x}) + \tilde{\mathbf{m}}_1(\mathbf{x}) \cdot \nabla_0 \psi(\mathbf{x}) + \tilde{\mathbf{n}}_1(\mathbf{x}) \cdot \nabla_0 \varphi(\mathbf{x}) =[\mathbf{v}(\theta(\mathbf{x})) \cdot \mathbf{v}'(\theta(\mathbf{x}))] q_1(\mathbf{x}) \\
\mathcal{L}(\mathbf{x}) \varphi(\mathbf{x}) + \tilde{\mathbf{m}}_2(\mathbf{x}) \cdot \nabla_0 \psi(\mathbf{x}) + \tilde{\mathbf{n}}_2(\mathbf{x}) \cdot \nabla_0 \varphi(\mathbf{x}) = q_2(\mathbf{x}) 
\end{cases}
\end{aligned}
\end{equation}
where $\mathcal{L}(\mathbf{x}) := -[\mathbf{v}(\theta(\mathbf{x})) \cdot \mathbf{v}'(\theta(\mathbf{x}))] \partial^2_{\mathbf{u}_0}  -[\mathbf{u}(\theta(\mathbf{x})) \cdot \mathbf{u}'(\theta(\mathbf{x}))] \partial^2_{\mathbf{v}_0}$. The coefficients $\tilde{\mathbf{m}}_{1,2}(\mathbf{x})$ and $\tilde{\mathbf{n}}_{1,2}(\mathbf{x})$ have the same smoothness/analyticity properties as $\theta(\mathbf{x}), \boldsymbol{\omega}_{\mathbf{u}_0}(\mathbf{x})$ and $\boldsymbol{\omega}_{\mathbf{v}_0}(\mathbf{x})$.
\end{lemma}
\begin{proof}
Substituting Eqs.\;(\ref{eq:theLambdasAnsatz}) and (\ref{eq:chooseCoeff}) into the first PDE in Eq.\;(\ref{eq:firstOrderToSolve}) gives that 
\begin{equation}
\begin{aligned}\label{eq:almostDoneRewrite}
q_1(\mathbf{x}) &= \partial_{\mathbf{v}_0} \big[  -\partial_{\mathbf{u}_0} \psi(\mathbf{x})  - \tfrac{\mathbf{u}(\theta(\mathbf{x})) \cdot \mathbf{u}'(\theta(\mathbf{x}))}{\mathbf{v}(\theta(\mathbf{x})) \cdot \mathbf{v}'(\theta(\mathbf{x}))} \partial_{\mathbf{v}_0} \varphi(\mathbf{x}) \big] + \partial_{\mathbf{u}_0} \big[ \partial_{\mathbf{v}_0} \psi(\mathbf{x})  - \partial_{\mathbf{u}_0} \varphi(\mathbf{x})  \big]  + \text{l.o.t.}   \\
&= -\Big(  \partial^2_{\mathbf{u}_0}   + \tfrac{\mathbf{u}(\theta(\mathbf{x})) \cdot \mathbf{u}'(\theta(\mathbf{x})) }{\mathbf{v}(\theta(\mathbf{x})) \cdot \mathbf{v}'(\theta(\mathbf{x}))} \partial^2_{\mathbf{v}_0} \Big)  \varphi(\mathbf{x})  + \text{l.o.t.}
\end{aligned}
\end{equation}
where l.o.t.\;denotes the ``lower order terms" that are linear in $\nabla_0 \psi(\mathbf{x})$ and $\nabla_0 \varphi(\mathbf{x})$. 
Making the same substitutions into the second PDE gives that 
\begin{equation}
\begin{aligned}
q_2(\mathbf{x}) &= \big[\mathbf{v}(\theta(\mathbf{x})) \cdot \mathbf{v}'(\theta(\mathbf{x})) \big]  \partial_{\mathbf{u}_0} \big[-  \partial_{\mathbf{u}_0} \psi(\mathbf{x})  - \tfrac{\mathbf{u}(\theta(\mathbf{x})) \cdot \mathbf{u}'(\theta(\mathbf{x}))}{\mathbf{v}(\theta(\mathbf{x})) \cdot \mathbf{v}'(\theta(\mathbf{x}))} \partial_{\mathbf{v}_0} \varphi(\mathbf{x})\big]  \\
&\qquad - \big[\mathbf{u}(\theta(\mathbf{x})) \cdot \mathbf{u}'(\theta(\mathbf{x})) \big] \partial_{\mathbf{v}_0} \big[  \partial_{\mathbf{v}_0} \psi(\mathbf{x})  -  \partial_{\mathbf{u}_0} \varphi(\mathbf{x})   \big] + \text{l.o.t.} \\
&= -\Big(  [\mathbf{v}(\theta(\mathbf{x})) \cdot \mathbf{v}'(\theta(\mathbf{x}))] \partial^2_{\mathbf{u}_0}   + [\mathbf{u}(\theta(\mathbf{x})) \cdot \mathbf{u}'(\theta(\mathbf{x}))] \partial^2_{\mathbf{v}_0} \Big)  \psi(\mathbf{x})  + \text{l.o.t.}
\end{aligned}
\end{equation}
with analogous lower order terms. 
We have derived Eq.\;(\ref{eq:L1L2L}). The smoothness/analyticity of $\tilde{\mathbf{m}}_{1,2}(\mathbf{x})$ and $\tilde{\mathbf{n}}_{1,2}(\mathbf{x})$ follow since the l.o.t.\;only involve $\theta(\mathbf{x})$, its derivatives and the  coefficients $a_u(\mathbf{x}), \ldots, b_v(\mathbf{x})$. 
\end{proof}

 As a final step in this section, we rewrite Eq.\;(\ref{eq:L1L2L}) to involve derivatives along Cartesian axes through a simple change of variables. Let $\mathbf{x}(\eta_1, \eta_2) := \eta_1 \tilde{\mathbf{u}}_0 + \eta_2 \tilde{\mathbf{v}}_0$, $U := \{ (\eta_1, \eta_2) \colon \mathbf{x}(\eta_1, \eta_2)  \in \Omega \}$, and write $\boldsymbol{\eta} = (\eta_1, \eta_2)$ for short. Define $\boldsymbol{\Phi} \colon U \rightarrow \mathbb{R}^2$ and $\sigma \colon U \rightarrow \mathbb{R}$ such that 
\begin{equation}
\begin{aligned}
\boldsymbol{\Phi}(\boldsymbol{\eta})  = \begin{pmatrix} \varphi(\mathbf{x}(\boldsymbol{\eta})) \\  \psi(\mathbf{x}(\boldsymbol{\eta})) \end{pmatrix},  \quad  \sigma(\boldsymbol{\eta}) =  \frac{\mathbf{u}(\theta(\mathbf{x}(\boldsymbol{\eta}))) \cdot \mathbf{u}'(\theta(\mathbf{x}(\boldsymbol{\eta})))}{\mathbf{v}(\theta(\mathbf{x}(\boldsymbol{\eta}))) \cdot \mathbf{v}'(\theta(\mathbf{x}(\boldsymbol{\eta})))}.
\end{aligned}
\end{equation}
By the chain rule, the PDE system in Eq.\;(\ref{eq:L1L2L}) becomes in these new coordinates
\begin{equation}
\begin{aligned}\label{eq:finalManipPDE}
-\Big( \partial_1^2 + \sigma(\boldsymbol{\eta}) \partial_2^2 \Big) \boldsymbol{\Phi}(\boldsymbol{\eta})  + \mathcal{M}(\boldsymbol{\eta}) \colon \nabla \boldsymbol{\Phi}(\boldsymbol{\eta})   = \tilde{\mathbf{q}}(\boldsymbol{\eta})
\end{aligned}
\end{equation} 
where $\partial_i$ is the $\eta_i$-derivative, and $\mathcal{M}(\boldsymbol{\eta}) \in \mathbb{R}^{2\times2 \times2}$ satisfies $[\mathcal{M}(\boldsymbol{\eta}) \colon \nabla \boldsymbol{\Phi}(\boldsymbol{\eta})]_{\alpha} =[ \mathcal{M}(\boldsymbol{\eta})]_{\alpha \beta \gamma} [ \nabla \boldsymbol{\Phi}(\boldsymbol{\eta})]_{\beta \gamma}$ with summation implied. Again, the coefficients  $\sigma(\boldsymbol{\eta})$ and $[ \mathcal{M}(\boldsymbol{\eta})]_{\alpha \beta \gamma}$ are smooth/analytic on $\overline{U}$ and $\tilde{\mathbf{q}}(\boldsymbol{\eta})$ is a smooth two-dimensional vector field. 

Solving Eq.\;(\ref{eq:finalManipPDE}) on $U$ produces a solution to Eq.\;(\ref{eq:L1L2L}) on $\Omega$, and with it a proof of Proposition \ref{AuxPDEProp}. There are  two cases to consider, depending on the sign of $\sigma(\boldsymbol{\eta})$. The PDE is elliptic when $\sigma(\boldsymbol{\eta}) >0$ on $U$. Alternatively, it is hyperbolic when $\sigma(\boldsymbol{\eta}) <0$ on $U$.  Our assumption in Eq.\;(\ref{eq:MVPreserve}) that the deformations of the origami be restricted to a fixed mountain-valley assignment implies that $\sigma(\boldsymbol{\eta})$ is either always positive or always negative, per Lemma \ref{SignsLemma}. So,  Eq.\;(\ref{eq:finalManipPDE}) is either  elliptic or  hyperbolic and cannot change its type. Indeed,  there exists a $\sigma_{\theta} >0$ depending only on $\theta(\mathbf{x})$ such that
\begin{equation}
\begin{aligned}\label{eq:sigmaThetaProp}
|\sigma(\boldsymbol{\eta})| \geq \sigma_{\theta} > 0 \quad \text{ on $\overline{U}$},
\end{aligned}
\end{equation}
since $\theta(\mathbf{x})$ is continuous and belongs to the interval $(\theta^{-}, \theta^+)$ on $\overline{\Omega}$  per Remark \ref{IntervalRemark}. We handle the elliptic case in Section \ref{ssec:TheEllipCase} and the hyperbolic case in Section \ref{ssec:TheHyperCase}. Note the sign of $\sigma(\boldsymbol{\eta})$ is opposite that of the Poisson's ratio $\nu(\theta(\boldsymbol{\eta}))$ in Eq.\;(\ref{eq:PoissonsRatioDesign}). Hence, the elliptic case addresses Miura origami and all other auxetic parallelogram origami patterns, while the hyperbolic case addresses Eggbox origami  all other non-auxetic patterns.

\subsection{Existence in the elliptic case}\label{ssec:TheEllipCase}
This section proves the existence of a solution to the PDE in Eq.\;(\ref{eq:finalManipPDE}) in the elliptic case where $\sigma(\boldsymbol{\eta}) \geq \sigma_\theta > 0$. The proof  invokes the Fredholm alternative for linear elliptic PDEs, an important technical result from PDE theory \cite{evans2022partial,mclean2000strongly}. The idea is to  view the linear operator $\mathcal{L}_{\text{e}}(\boldsymbol{\eta}) := -\big[ \partial_1^2 + \sigma(\boldsymbol{\eta}) \partial_2^2 +  \mathcal{M}(\boldsymbol{\eta}) \colon \nabla  \big]$ as a matrix (of infinite dimension) and the problem of solving the PDE in Eq.\;(\ref{eq:finalManipPDE}) as a question of linear algebra. In particular, by making a good choice of boundary data, we will show that   $\tilde{\mathbf{q}}(\boldsymbol{\eta})$ is in the range space  of  $\mathcal{L}_{\text{e}}(\boldsymbol{\eta})$. We make crucial use of the analyticity hypothesis here.

Start by rewriting the operator in divergence form:
\begin{equation}
\begin{aligned}
\mathcal{L}_{\text{e}}(\boldsymbol{\eta}) := -\big[ \partial_1^2 + \sigma(\boldsymbol{\eta}) \partial_2^2 +  \mathcal{M}(\boldsymbol{\eta}) \colon \nabla  \big] =- \nabla\cdot \big[\mathbb{C}(\boldsymbol{\eta}\big) \colon \nabla ] + \mathcal{M}_{\text{e}}(\boldsymbol{\eta})  \colon \nabla 
\end{aligned}
\end{equation}
for appropriate higher-order tensors $\mathbb{C}(\boldsymbol{\eta}),\mathcal{M}_{\text{e}}(\boldsymbol{\eta}) \in \mathbb{R}^{2\times2\times2\times2}$. To make this formulation explicit, observe that  $\big[\partial_1^2 + \sigma(\boldsymbol{\eta}) \partial_2^2 \big]\boldsymbol{\Phi}(\boldsymbol{\eta}) = \nabla\cdot \big[ \partial_1 \boldsymbol{\Phi}(\boldsymbol{\eta})  \otimes \tilde{\mathbf{e}}_1 + \sigma(\boldsymbol{\eta})  \partial_2 \boldsymbol{\Phi}(\boldsymbol{\eta})  \otimes \tilde{\mathbf{e}}_2\big] + \text{l.o.t.}$ Thus, writing the bracketed term in index notation as $[ \mathbb{C}(\boldsymbol{\eta}) \colon \nabla \boldsymbol{\Phi}(\boldsymbol{\eta})]_{\alpha \beta} = [ \mathbb{C}(\boldsymbol{\eta})]_{\alpha \beta \gamma \delta} [ \nabla \boldsymbol{\Phi}(\boldsymbol{\eta})]_{\gamma \delta}$ gives that
\begin{equation}
\begin{aligned}\label{eq:getCtensor}
&[ \mathbb{C}(\boldsymbol{\eta})]_{1111} = [ \mathbb{C}(\boldsymbol{\eta})]_{2121}   = 1, \quad [ \mathbb{C}(\boldsymbol{\eta})]_{1212} = [ \mathbb{C}(\boldsymbol{\eta})]_{2222}   =  \sigma(\boldsymbol{\eta}), \\
&[ \mathbb{C}(\boldsymbol{\eta})]_{\alpha \beta \gamma \delta} = 0 \quad \text{ for } \alpha \beta \gamma \delta \text{ not as above.} 
\end{aligned}
\end{equation}
The tensor $\mathcal{M}_{\text{e}}(\boldsymbol{\eta})$ is then $\mathcal{M}(\boldsymbol{\eta})$ plus lower order terms. Since these lower order terms depend on $\boldsymbol{\eta}$ only through derivatives of $\sigma(\boldsymbol{\eta})$, $\mathcal{M}_{\text{e}}(\boldsymbol{\eta})$ is analytic on $\overline{U}$. Finally, we note using Eqs.\;(\ref{eq:sigmaThetaProp}) and (\ref{eq:getCtensor}) that $\mathbb{C}(\boldsymbol{\eta})$ is strongly elliptic, and in fact satisfies
\begin{equation}
\begin{aligned}\label{eq:StrongEllipticity}
\mathbf{F} \colon \boldsymbol{\mathbb{C}}(\boldsymbol{\eta}) \colon \mathbf{F} &:= [ \boldsymbol{\mathbb{C}}(\boldsymbol{\eta})]_{\alpha \beta \gamma \delta}  [\mathbf{F}]_{\alpha \beta}  [\mathbf{F}]_{\gamma \delta} \\
&  \;= ([\mathbf{F}]_{11})^2 + ([\mathbf{F}]_{12})^2 + \boldsymbol{\sigma}(\boldsymbol{\eta}) \big\{  ([\mathbf{F}]_{21})^2 + ([\mathbf{F}]_{22})^2\big\}  \geq \min\{ 1, \sigma_\theta \} |\mathbf{F}|^2
\end{aligned}
\end{equation}
for all $\mathbf{F}\in\mathbb{R}^{2\times2}$. It also has the  major symmetry   $\mathbf{A} \colon \mathbb{C}(\boldsymbol{\eta}) \colon \mathbf{B} = \mathbf{B} \colon \mathbb{C}(\boldsymbol{\eta})\colon \mathbf{A}$ for all $\mathbf{A},\mathbf{B} \in \mathbb{R}^{2\times2}$.

We will deduce the existence of a solution $\boldsymbol{\Phi}(\boldsymbol{\eta})$ to Eq.\;(\ref{eq:finalManipPDE}) by reviewing the well-posedness theory of the 
class of Dirichlet problems
\begin{equation}
\begin{aligned}\label{eq:TheDirichletProblem}
\begin{cases}
\mathcal{L}_{e}(\boldsymbol{\eta}) \boldsymbol{\Phi}(\boldsymbol{\eta}) = \tilde{\mathbf{q}}(\boldsymbol{\eta})  & \text{ in } U \\
 \boldsymbol{\Phi}(\boldsymbol{\eta})  = \boldsymbol{\Psi}(\boldsymbol{\eta}) & \text{ at } \partial U,
\end{cases}
\end{aligned}
\end{equation}
parameterized by the choice of Dirichlet boundary data $\boldsymbol{\Psi}(\boldsymbol{\eta})$.  The question is: for which boundary data does there exist a solution $\boldsymbol{\Phi}(\boldsymbol{\eta})$?  To answer this question, we pose the weak formulation of Eq.\;(\ref{eq:TheDirichletProblem}), which requires a brief digression on Sobolev spaces (see, e.g.,  \cite{evans2022partial}, Chapter 5,  for details). Recall
\begin{equation}
\begin{aligned}
H^1(U; \mathbb{R}^2)  :=\big\{ \boldsymbol{\Phi} \in L^2(U, \mathbb{R}^2) \colon \nabla \boldsymbol{\Phi} \in L^2(U;\mathbb{R}^2) \big\} 
\end{aligned}
\end{equation}
is the Sobolev space of square integrable maps with square integrable first derivatives, and $H_0^{1}(U;\mathbb{R}^2)$ is the subspace of maps $\boldsymbol{\Phi}\in H^1(U;\mathbb{R}^2)$ with $\boldsymbol{\Phi}|_{\partial U}=\boldsymbol{0}$ ``in the trace sense”. 
Correspondingly, we enforce the Dirichlet boundary condition from Eq.\;(\ref{eq:TheDirichletProblem}) by working with maps in
\begin{equation}
\begin{aligned}
H_{\boldsymbol{\Psi}}^1(U; \mathbb{R}^2) :=  \big\{ \boldsymbol{\Phi} \in H^1(U, \mathbb{R}^2) \colon \boldsymbol{\Phi}|_{\partial U} = \boldsymbol{\Psi} \big\}
\end{aligned}.
\end{equation}
With these definitions, a weak solution of Eq.\;(\ref{eq:TheDirichletProblem}) is any $\boldsymbol{\Phi} \in H_{\boldsymbol{\Psi}}^1(U; \mathbb{R}^2) $ such that 
\begin{equation}
\begin{aligned}\label{eq:WeakForm}
\int_{U} \Big\{    \big[ \mathbb{C}(\boldsymbol{\eta}) \colon \nabla \boldsymbol{\Phi}(\boldsymbol{\eta}) \big] \colon \nabla \boldsymbol{\Theta}(\boldsymbol{\eta})   + \big[ \mathcal{M}_{\text{e}}(\boldsymbol{\eta})  \colon \nabla \boldsymbol{\Phi}(\boldsymbol{\eta}) \big] \cdot \boldsymbol{\Theta}(\boldsymbol{\eta}) \Big\}   dA = \int_{U} \tilde{\mathbf{q}}(\boldsymbol{\eta}) \cdot   \boldsymbol{\Theta}(\boldsymbol{\eta}) dA \quad \forall \; \boldsymbol{\Theta} \in H^1_0(U;\mathbb{R}^2).
\end{aligned}
\end{equation}
Although we analyze weak solutions, by standard elliptic regularity results any weak solution of Eq.\;(\ref{eq:TheDirichletProblem}) will automatically be smooth on $U$, due to the strong ellipticity guaranteed by Eq.\;(\ref{eq:StrongEllipticity}), and because the coefficients of $\mathcal{L}_{e}(\boldsymbol{\eta})$ and the map $\tilde{\mathbf{q}}(\boldsymbol{\eta})$ are smooth. Actually, for our purposes in the main text and per the statement of Proposition \ref{AuxPDEProp}, we require that the solution is smooth on a neighborhood of $\overline{U}$. This can be achieved by replacing $U$ with a slightly larger domain, which again is possible due to the regularity properties of the coefficients of $\mathcal{L}_{e}(\boldsymbol{\eta})$ and $\tilde{\mathbf{q}}(\boldsymbol{\eta})$.
In the rest of this section we tacitly assume (with a slight abuse of notation) that $U$ is this slightly larger domain.

We desire a choice of $\boldsymbol{\Psi}(\boldsymbol{\eta})$ such that Eq.\;(\ref{eq:TheDirichletProblem}) has a weak solution. This question is completely resolved by the Fredholm alternative; Theorem 4.10 in  \cite{mclean2000strongly} gives a convenient statement, which we summarize now. There are two cases: either (i) the only weak solution of the homogeneous problem
\begin{equation}
\begin{aligned}\label{eq:TheHomogeneousProblem}
\begin{cases}
\mathcal{L}_{e}(\boldsymbol{\eta}) \boldsymbol{\Phi}(\boldsymbol{\eta}) = \mathbf{0}   & \text{ in } U \\
 \boldsymbol{\Phi}(\boldsymbol{\eta})  = \mathbf{0} & \text{ at } \partial U
\end{cases}
\end{aligned}
\end{equation}
is the trivial one $\boldsymbol{\Phi}(\boldsymbol{\eta}) = \mathbf{0}$, in which case the Fredholm alternative states that the inhomogeneous problem Eq.\;(\ref{eq:TheDirichletProblem}) is uniquely solvable for any $\boldsymbol{\Psi}(\boldsymbol{\eta})$; or (ii) the homogeneous problem in Eq.\;(\ref{eq:TheHomogeneousProblem}) admits non-zero solutions, in which case the alternative allows for Eq.\;(\ref{eq:TheDirichletProblem})  to have multiple solutions or no solutions at all. 

Now if $\mathcal{L}_{e}(\boldsymbol{\eta})$ belongs to  case (i), we are done (take $\boldsymbol{\Psi}=\boldsymbol{0}$). So, we focus on case (ii). For this purpose, we introduce the homogeneous adjoint problem to Eq.\;(\ref{eq:TheHomogeneousProblem}):
\begin{equation}
\begin{aligned}
\begin{cases}\label{eq:TheHomoAdjointProblem}
\mathcal{L}^{\star}_{e}(\boldsymbol{\eta}) \boldsymbol{\Theta}(\boldsymbol{\eta}) = \mathbf{0}   & \text{ in } U \\
 \boldsymbol{\Theta}(\boldsymbol{\eta})  = \mathbf{0} & \text{ at } \partial U
\end{cases}
\end{aligned}
\end{equation}
where $\mathcal{L}^{\star}_{e}(\boldsymbol{\eta}) \boldsymbol{\Theta}(\boldsymbol{\eta}) := -\nabla\cdot \big[ \mathbb{C}(\boldsymbol{\eta}) \colon  \nabla \boldsymbol{\Theta}(\boldsymbol{\eta}) + \boldsymbol{\Theta}(\boldsymbol{\eta}) \cdot \mathcal{M}_e(\boldsymbol{\eta})\big]$.
 Solving case (ii) turns out to be  a delicate matter of choosing the boundary data $\boldsymbol{\Psi}(\boldsymbol{\eta})$ to interact consistently with solutions of the adjoint problem. Specifically, in case (ii), the Fredholm alternative asserts that the solutions $\boldsymbol{\Theta}(\boldsymbol{\eta})$ of Eq.\;(\ref{eq:TheHomoAdjointProblem}) form a finite dimensional family, and that the original problem in Eq.\;(\ref{eq:TheDirichletProblem}) is solvable for $\boldsymbol{\Phi} \in H^1_{\boldsymbol{\Psi}}(U; \mathbb{R}^2)$  if and only if
\begin{equation}
\begin{aligned}\label{eq:keyIdentityFredholm}
\int_{U} \tilde{\mathbf{q}}(\boldsymbol{\eta}) \cdot \boldsymbol{\Theta}(\boldsymbol{\eta}) dA &= \int_{\partial U} \boldsymbol{\Psi}(\boldsymbol{\eta}) \cdot \big[ \mathbb{C}(\boldsymbol{\eta})  \colon \nabla \boldsymbol{\Theta}(\boldsymbol{\eta}) \big] \mathbf{n}_U( \boldsymbol{\eta})  d\Gamma
\end{aligned}
\end{equation} 
for all such $\boldsymbol{\Theta}(\boldsymbol{\eta})$. We now choose $\boldsymbol{\Psi}(\boldsymbol{\eta})$ to guarantee this identity and thus produce a solution to Eq.\;(\ref{eq:TheDirichletProblem}). 

By the Fredholm alternative, we are free to assume that  $\{ \boldsymbol{\Theta}_1, \ldots, \boldsymbol{\Theta}_N\} \subset H_0^1(U; \mathbb{R}^2)$ is a basis for the solutions of the homogeneous adjoint problem, for some $N \geq 1$. We construct $\boldsymbol{\Psi}(\boldsymbol{\eta})$ such that 
\begin{equation}
\begin{aligned}\label{eq:toProvekeyIdentityFredholm}
\int_{U} \tilde{\mathbf{q}}(\boldsymbol{\eta}) \cdot \boldsymbol{\Theta}_i(\boldsymbol{\eta}) dA &= \int_{\partial U} \boldsymbol{\Psi}(\boldsymbol{\eta}) \cdot \big[ \mathbb{C}(\boldsymbol{\eta})  \colon \nabla \boldsymbol{\Theta}_i(\boldsymbol{\eta}) \big] \mathbf{n}_U( \boldsymbol{\eta})  d\Gamma\quad  \text{for all } i =1,\ldots,N
\end{aligned}
\end{equation}
 by proving  that the collection of vector fields $\big\{ \big[ \mathbb{C}(\boldsymbol{\eta})  \colon \nabla \boldsymbol{\Theta}_i(\boldsymbol{\eta}) \big] \mathbf{n}_U( \boldsymbol{\eta}) \big\}_{i=1}^N$  is linearly independent on $\partial U$.  
 Indeed, this linear independence  allows us to introduce a dual basis $\{ \boldsymbol{\Psi}_i(\boldsymbol{\eta}) \}_{i=1}^{N} $ satisfying
\begin{equation}
\begin{aligned}
\int_{\partial U} \boldsymbol{\Psi}_j(\boldsymbol{\eta}) \cdot \big[ \mathbb{C}(\boldsymbol{\eta})  \colon \nabla \boldsymbol{\Theta}_i(\boldsymbol{\eta}) \big] \mathbf{n}_U( \boldsymbol{\eta}) d\Gamma = \delta_{ij} \quad \text{for all }  i, j =1,\ldots, N.
\end{aligned}
\end{equation}
Then, defining 
\begin{equation}
\begin{aligned}
\boldsymbol{\Psi}(\boldsymbol{\eta}):= \sum_{j =1}^N\Big( \int_{U} \tilde{\mathbf{q}}(\boldsymbol{\eta}) \cdot \boldsymbol{\Theta}_j(\boldsymbol{\eta}) dA  \Big)  \boldsymbol{\Psi}_j(\boldsymbol{\eta})
\end{aligned}
\end{equation}
yields the desired identity in Eq.\;(\ref{eq:toProvekeyIdentityFredholm}). 

That $\big\{ \big[ \mathbb{C}(\boldsymbol{\eta})  \colon \nabla \boldsymbol{\Theta}_i(\boldsymbol{\eta}) \big] \mathbf{n}_U( \boldsymbol{\eta})\big\}_{i=1}^{N}$ is linearly independent on $\partial U$ follows from the definition of $\{ \boldsymbol{\Theta}_i(\boldsymbol{\eta})\}_{i=1}^{N}$  as a basis for the space of solutions of the homogeneous adjoint problem, combined with a “unique continuation” result stating that the only solution of the overdetermined elliptic system
\begin{equation}\label{eq:overdetermined-system}
\begin{aligned}
\begin{cases}
\mathcal{L}_e^{\star}(\boldsymbol{\eta})  \boldsymbol{\Theta}(\boldsymbol{\eta}) = \mathbf{0} & \text{ in } U \\
\boldsymbol{\Theta}(\boldsymbol{\eta}) = \mathbf{0}  & \text{ at } \partial U  \\
\big[\mathbb{C}(\boldsymbol{\eta})  \colon \nabla \boldsymbol{\Theta}(\boldsymbol{\eta}) \big] \mathbf{n}_U( \boldsymbol{\eta}) = \mathbf{0} & \text{ at } \partial U
\end{cases}
\end{aligned}
\end{equation}
is the trivial one $\boldsymbol{\Theta}(\boldsymbol{\eta}) = \mathbf{0}$. To see this, note from the formula for $\mathcal{L}_{e}^{\star}(\boldsymbol{\eta})$ after Eq.\;(\ref{eq:TheHomoAdjointProblem}) that its coefficients are analytic on $\overline{U}$.
Thus, Holmgren’s uniqueness theorem gives the required unique continuation result: the version of Holmgren's theorem we use states that any $C^2$ solution of Eq.\;(\ref{eq:overdetermined-system}) must vanish identically (see \cite{john1991partial}, Chapter 3.5). 
To apply it, extend the domain $U$  to $U^{\text{ext}} := \{ \boldsymbol{\eta} \in \mathbb{R}^2 \colon \text{dist}(\boldsymbol{\eta}, U) < \epsilon\}$ for $\epsilon >0$ such that the coefficients of $\mathcal{L}^{\star}_{\text{e}}(\boldsymbol{\eta})$ remain analytic on $U^{\text{ext}}$,  take arbitrary constants $\alpha_1, \ldots, \alpha_N$, and define
\begin{equation}
\begin{aligned}\label{eq:SumAlphaTheta}
\boldsymbol{\Theta}(\boldsymbol{\eta}) = \begin{cases}
\sum_{i=1}^N \alpha_i \boldsymbol{\Theta}_i(\boldsymbol{\eta})  &\text{ if } \boldsymbol{\eta} \in U\\
\mathbf{0} & \text{ if } \boldsymbol{\eta} \in U^{\text{ext}} \setminus U
\end{cases}
\end{aligned}
\end{equation}
where $\{ \boldsymbol{\Theta}_1, \ldots,   \boldsymbol{\Theta}_N\} \subset H_0^1(U; \mathbb{R}^2)$ is the aforementioned basis. If
\begin{equation}
\begin{aligned}\label{eq:testLI}
 \big[ \mathbb{C}(\boldsymbol{\eta})  \colon \nabla \boldsymbol{\Theta}(\boldsymbol{\eta})  \big] \mathbf{n}_U( \boldsymbol{\eta}) = \sum_{i=1}^N \alpha_i  \big[ \mathbb{C}(\boldsymbol{\eta})  \colon \nabla \boldsymbol{\Theta}_i(\boldsymbol{\eta})  \big] \mathbf{n}_U( \boldsymbol{\eta})  = \mathbf{0} \quad \text{ on } \partial U,
\end{aligned}
\end{equation}
we must show that $\alpha_1=\cdots=\alpha_N = 0$. 
Given any $\boldsymbol{\Phi} \in H^1(U^{\text{ext}};\mathbb{R}^2)$, we can write that
\begin{equation}
\begin{aligned}\label{eq:intByParts}
&\int_{U^{\text{ext}}} \Big\{    \big[ \mathbb{C}(\boldsymbol{\eta}) \colon \nabla \boldsymbol{\Theta}(\boldsymbol{\eta}) \big] \colon \nabla \boldsymbol{\Phi}(\boldsymbol{\eta}) - \nabla\cdot \big[ \boldsymbol{\Theta}(\boldsymbol{\eta}) \cdot \mathcal{M}_e(\boldsymbol{\eta})\big]\cdot \boldsymbol{\Phi}(\boldsymbol{\eta}) \Big\}   dA \\
&\quad=\int_{U} \Big\{    \big[ \mathbb{C}(\boldsymbol{\eta}) \colon \nabla \boldsymbol{\Theta}(\boldsymbol{\eta}) \big] \colon \nabla \boldsymbol{\Phi}(\boldsymbol{\eta}) - \nabla\cdot \big[ \boldsymbol{\Theta}(\boldsymbol{\eta}) \cdot \mathcal{M}_e(\boldsymbol{\eta})\big]\cdot \boldsymbol{\Phi}(\boldsymbol{\eta}) \Big\}   dA \\
&\quad= \int_{\partial U}  \big\{ \big[ \mathbb{C}(\boldsymbol{\eta})  \colon \nabla \boldsymbol{\Theta}(\boldsymbol{\eta}) \big] \mathbf{n}_U( \boldsymbol{\eta}) \big\} \cdot \boldsymbol{\Phi}(\boldsymbol{\eta}) d\Gamma
\end{aligned}
\end{equation}
using integration by parts, since $\boldsymbol{\Theta}(\boldsymbol{\eta})$ solves Eq.\;(\ref{eq:TheHomoAdjointProblem}) and vanishes on $U^{\text{ext}}\setminus U$. Thus, if Eq.\;(\ref{eq:testLI}) holds, it follows from Eq.\;(\ref{eq:intByParts}) that $\mathcal{L}_{e}^{\star}(\boldsymbol{\eta})\boldsymbol{\Theta}(\boldsymbol{\eta})=\mathbf{0}$ weakly on  $U^\text{ext}$ (c.f.\ Eq.\;(\ref{eq:WeakForm})). Then, $\boldsymbol{\Theta}(\boldsymbol{\eta})$ must be smooth by elliptic regularity, and hence $\boldsymbol{\Theta}(\boldsymbol{\eta}) = \mathbf{0}$ by Holmgren's uniqueness theorem.  Looking back at  Eq.\;(\ref{eq:SumAlphaTheta}), we conclude that  $\alpha_1=\cdots=\alpha_N=0$ since $\{ \boldsymbol{\Theta}_1(\boldsymbol{\eta}), \ldots,  \boldsymbol{\Theta}_N(\boldsymbol{\eta})\}$ were chosen to be linearly independent. This completes the proof of existence of solutions to the PDE in Eq.\;(\ref{eq:finalManipPDE}) in the elliptic case.

The above argument is the only part of the proof of Theorem~\ref{MainTheorem} where we need $\theta(\mathbf{x})$, $\boldsymbol{\omega}_{\mathbf{u}_0}(\mathbf{x})$ and $\boldsymbol{\omega}_{\mathbf{v}_0}(\mathbf{x})$ to be analytic, rather than just smooth. Presumably, our assumption of analyticity is suboptimal and can be gotten rid of by  a stronger unique continuation result.

\subsection{Existence in the hyperbolic case}\label{ssec:TheHyperCase}

This section proves the existence of a solution to the PDE in Eq.\;(\ref{eq:finalManipPDE}) in the hyperbolic case where $\sigma(\boldsymbol{\eta}) \leq -\sigma_{\theta} < 0$. Start by writing the PDE as 
\begin{equation}
\begin{aligned}
\square(\boldsymbol{\eta}) \boldsymbol{\Phi}(\boldsymbol{\eta}) + \mathcal{M}_{h}(\boldsymbol{\eta}) \colon \nabla \boldsymbol{\Phi}(\boldsymbol{\eta})  = \tilde{\mathbf{q}}(\boldsymbol{\eta}), 
\end{aligned}
\end{equation}
using the second-order hyperbolic operator $\square(\boldsymbol{\eta}) := -\partial_1^2 + \partial_2 ( |\sigma(\boldsymbol{\eta})| \partial_2)$. The tensor $\mathcal{M}_h(\boldsymbol{\eta}) \in \mathbb{R}^{2\times2\times2}$ collects the sum of $\mathcal{M}(\boldsymbol{\eta})$ and the lower order terms from the calculation $-\partial_1^2 \boldsymbol{\Phi}(\boldsymbol{\eta})  - \sigma(\boldsymbol{\eta})   \partial^2_2 \boldsymbol{\Phi}(\boldsymbol{\eta}) = -\partial_1^2 \boldsymbol{\Phi}(\boldsymbol{\eta})  +\partial_2 ( |\sigma(\boldsymbol{\eta})| \partial_2 \boldsymbol{\Phi}(\boldsymbol{\eta})) + \text{l.o.t.}$ Since $\sigma(\boldsymbol{\eta})$, $\mathcal{M}_{h}(\boldsymbol{\eta})$ and  $ \tilde{\mathbf{q}}(\boldsymbol{\eta})$ are smooth on a neighborhood of $\overline{U}$, they can be extended smoothly to the closure of a square domain $(-r, r)^2$ containing $\overline{U}$, while preserving the strict negativity of $\sigma(\boldsymbol{\eta})$.  
We seek a solution to the  initial/boundary value problem 
\begin{equation}
\begin{aligned}
\begin{cases}\label{eq:solveForHyperCase}
\square(\boldsymbol{\eta}) \boldsymbol{\Phi}(\boldsymbol{\eta}) + \mathcal{M}_{h}(\boldsymbol{\eta}) \colon \nabla \boldsymbol{\Phi}(\boldsymbol{\eta})  = \tilde{\mathbf{q}}(\boldsymbol{\eta})  &  -r < \eta_1, \eta_2 < r \\
\boldsymbol{\Phi}(\boldsymbol{\eta}) = \mathbf{0} & -r < \eta_1< r,\eta_2 = \pm r\\
\boldsymbol{\Phi}(\boldsymbol{\eta}) = \mathbf{0}, \partial_1\boldsymbol{\Phi}(\boldsymbol{\eta})= \mathbf{0}  &  \eta_1 = -r,  -r  < \eta_2 < r
\end{cases}
\end{aligned}
\end{equation}
on the extended domain to obtain the desired solution of Eq.\;(\ref{eq:finalManipPDE}). 

There are several techniques available to prove that Eq.\;(\ref{eq:solveForHyperCase}) admits a unique smooth solution, such as the Galerkin truncation-based argument presented in \cite{evans2022partial}, Section 7.2, which we briefly summarize. The idea is to view  Eq.\;(\ref{eq:solveForHyperCase}) as an ODE (in $\eta_1$) through a suitable function space (in $\eta_2$), and to find its solution $\boldsymbol{\Phi}(\boldsymbol{\eta})$ by solving a sequence of successively larger yet finite dimensional ODE problems built to approximate Eq.\;(\ref{eq:solveForHyperCase}). Smoothness of the resulting solution follows from the regularity theory of hyperbolic PDEs.

In a bit more detail, suppose we expand  $\boldsymbol{\Phi}(\boldsymbol{\eta})$ in a Fourier basis adapted to the zero Dirichlet boundary conditions by writing
\begin{equation}
\begin{aligned}
\boldsymbol{\Phi}(\boldsymbol{\eta}) = \sum_{k = 1}^{\infty} \boldsymbol{\Phi}_k(\eta_1) s_k(\eta_2)
\end{aligned}
\end{equation}
where $\{ s_k(\eta_2)\}_{k = 1}^{\infty}$ is the intended basis of $H_0^{1}((-r, r))$. Projecting the PDE part of Eq.\;(\ref{eq:solveForHyperCase}) to the span of the first $N$ basis functions gives a way of approximating the unknown $\{\boldsymbol{\Phi}_k(\eta_1)\}$. For each fixed $N$, we seek a vector field $\boldsymbol{\Phi}^{(N)}(\boldsymbol{\eta})$ satisfying the  integral constraints 
\begin{equation}
\begin{aligned}\label{eq:ApproxHyperbolic}
&\int_{-r}^r \Big( \big\{  \partial_1^2 \boldsymbol{\Phi}^{(N)}(\boldsymbol{\eta}) + \mathcal{M}_h(\boldsymbol{\eta}) \colon \nabla \boldsymbol{\Phi}^{(N)}(\boldsymbol{\eta})  - \tilde{\mathbf{q}}(\boldsymbol{\eta}) \big\}   s_k(\eta_2)  + |\sigma(\boldsymbol{\eta})| \partial_2 \boldsymbol{\Phi}^{(N)}(\boldsymbol{\eta})  s_k'(\eta_2) \Big) d\eta_2  = \mathbf{0}
\end{aligned}
\end{equation} 
for all $\eta_1 \in (-r,r)$ and $k = 1,\ldots, N$, along with the initial and boundary conditions. Writing
\begin{equation}
\begin{aligned}
\boldsymbol{\Phi}^{(N)}(\boldsymbol{\eta}) := \sum_{k = 1}^N \boldsymbol{\Phi}^{(N)}_k(\eta_1) s_k(\eta_2)
\end{aligned}
\end{equation}
yields the following system of ODEs on substitution into Eq.\;(\ref{eq:ApproxHyperbolic}): 
\begin{equation}
\begin{aligned}\label{eq:theODESystemProof}
\frac{d^2}{d \eta_1^2} \begin{pmatrix} \boldsymbol{\Phi}^{(N)}_1(\eta_1)  \\ \vdots \\  \boldsymbol{\Phi}^{(N)}_N(\eta_1)  \end{pmatrix}  + \mathbf{M}_1(\eta_1)  \frac{d}{d \eta_1} \begin{pmatrix} \boldsymbol{\Phi}^{(N)}_1(\eta_1)  \\ \vdots \\  \boldsymbol{\Phi}^{(N)}_N(\eta_1)  \end{pmatrix}  + \mathbf{M}_0(\eta_1)  \begin{pmatrix} \boldsymbol{\Phi}^{(N)}_1(\eta_1)  \\ \vdots \\  \boldsymbol{\Phi}^{(N)}_N(\eta_1)  \end{pmatrix}  = \begin{pmatrix} \tilde{\mathbf{q}}_1(\eta_1) \\ \vdots \\  \tilde{\mathbf{q}}_N(\eta_1) \end{pmatrix}.
\end{aligned} 
\end{equation}
Note $\tilde{\mathbf{q}}_k(\eta_1) := \int_{-r}^{r} \tilde{\mathbf{q}}(\boldsymbol{\eta}) s_k(\eta_2) d\eta_2$, and $\mathbf{M}_{0,1} (\eta_1)\in \mathbb{R}^{2N \times 2N}$ are the matrices that arise from integrating out the $\eta_2$-dependence in the $N$ integral equations above and organizing the calculation around their linearity in $(\boldsymbol{\Phi}^{(N)}_1(\eta_1), \ldots,\boldsymbol{\Phi}^{(N)}_N(\eta_1))$. The boundary conditions are taken care of by the basis functions $\{s_k(\eta_2)\}$. For the initial conditions, we impose 
\begin{equation}\label{eq:theODESystem-initconds}
\begin{aligned}
\boldsymbol{\Phi}^{(N)}_k(-r) = \mathbf{0},  \quad \frac{d}{d\eta_1} \boldsymbol{\Phi}^{(N)}_k(-r) = \mathbf{0}, \quad \text{for all } k =1,\ldots,N.
\end{aligned}
\end{equation}
A unique solution to the initial value problem in Eqs.\;(\ref{eq:theODESystemProof}-\ref{eq:theODESystem-initconds}) exists. Thus, we produce an $N$-dependent sequence of approximate solutions to Eq.\;(\ref{eq:solveForHyperCase}). The desired solution $\boldsymbol{\Phi}(\boldsymbol{\eta})$ is obtained by taking $N \rightarrow \infty$.

Justifying this last step is standard fare in the theory of PDEs, but for completeness we give an overall picture of the argument here. The goal is to apply a compactness theorem to extract a convergent sub-sequence of approximate solutions; the key is to check that the maps $\{ \boldsymbol{\Phi}^{(N)}(\boldsymbol{\eta}) \}$ remain bounded in the sense that their “energies”
\begin{equation}
\begin{aligned}
E^{(N)}(\eta_1) := \frac{1}{2} \int_{-r}^{r} \big\{ | \partial_1 \boldsymbol{\Phi}^{(N)}(\boldsymbol{\eta})|^2 + |\sigma(\boldsymbol{\eta})|  | \partial_2 \boldsymbol{\Phi}^{(N)}(\boldsymbol{\eta})|^2 \big\} d\eta_2 
\end{aligned}
\end{equation}
do not go to infinity as $N \rightarrow \infty$. This boundedness is possible due to the estimate 
\begin{equation}\label{eq:bound-to-prove-N}
\begin{aligned}
\max_{-r \leq \eta_1 \leq r} E^{(N)}( \eta_1) \leq C \sum_{k=1}^N \int_{(-r,r)^2}  |\tilde{\mathbf{q}}_k(\boldsymbol{\eta})|^2 dA,
\end{aligned}
\end{equation}
which follows by the same reasoning leading to the analogous estimate for the original problem in Eq.\;(\ref{eq:solveForHyperCase}). To prove the latter, let  $\boldsymbol{\Phi}(\boldsymbol{\eta})$ be any solution of Eq.\;(\ref{eq:solveForHyperCase}) and introduce its energy 
\begin{equation}
\begin{aligned}
E(\eta_1) := \frac{1}{2} \int_{-r}^{r} \big\{ | \partial_1 \boldsymbol{\Phi}(\boldsymbol{\eta})|^2 + |\sigma(\boldsymbol{\eta})|  | \partial_2 \boldsymbol{\Phi}(\boldsymbol{\eta})|^2 \big\} d\eta_1.
\end{aligned}
\end{equation}
This energy remains bounded, per the following observations:
\begin{equation}
\begin{aligned}\label{eq:chainInequalities}
\Big|\frac{d E(\eta_1)}{d\eta_1} \Big|  &= \Big| \int_{-r}^r \Big\{  \partial_1 \boldsymbol{\Phi}(\boldsymbol{\eta}) \cdot \partial_1^2 \boldsymbol{\Phi}(\boldsymbol{\eta})  +  |\sigma(\boldsymbol{\eta})|   \partial_2 \boldsymbol{\Phi}(\boldsymbol{\eta}) \cdot \partial_1 \partial_2  \boldsymbol{\Phi}(\boldsymbol{\eta}) + \partial_1( |\sigma(\boldsymbol{\eta})|) |\partial_2 \boldsymbol{\Phi}(\boldsymbol{\eta})|^2  \Big\}  d\eta_2   \Big| \\
& = \Big| \int_{-r}^r \Big\{  -\partial_1 \boldsymbol{\Phi}(\boldsymbol{\eta}) \cdot  \square(\boldsymbol{\eta}) \boldsymbol{\Phi}(\boldsymbol{\eta})  + \partial_1( |\sigma(\boldsymbol{\eta})|) |\partial_2 \boldsymbol{\Phi}(\boldsymbol{\eta})|^2  \Big\}  d\eta_2   \Big| \\
&\leq \int_{-r}^r \Big\{  |\partial_1 \boldsymbol{\Phi}(\boldsymbol{\eta})|^2 + \frac{1}{2} |\tilde{\mathbf{q}}(\boldsymbol{\eta})|^2 + \frac{1}{2} | \mathcal{M}_h(\boldsymbol{\eta}) \colon \nabla \boldsymbol{\Phi}(\boldsymbol{\eta})|^2  + |\partial_1 \sigma(\boldsymbol{\eta})||\partial_2 \boldsymbol{\Phi}(\boldsymbol{\eta})|^2 \Big\} d\eta_2.
\end{aligned}
\end{equation}
In the second step, we integrated by parts with the boundary conditions in Eq.\;(\ref{eq:solveForHyperCase}) and the definition of the  operator $\square(\boldsymbol{\eta})$ (hyperbolicity is used here). In the third step, we applied the arithmetic–geometric inequality twice in the form $2 \tilde{\mathbf{a}} \cdot \tilde{\mathbf{b}} \leq  |\tilde{\mathbf{a}}|^2 + |\tilde{\mathbf{b}}|^2$. Since $\mathcal{M}_h(\boldsymbol{\eta})$ and $\sigma(\boldsymbol{\eta})$ are smooth on $[-r,r]^2$, the terms in the inequality are bounded as
\begin{equation}
\begin{aligned}
| \mathcal{M}_h(\boldsymbol{\eta}) \colon \nabla \boldsymbol{\Phi}(\boldsymbol{\eta})|^2 &\leq |\mathcal{M}_h(\boldsymbol{\eta})|^2 |\nabla  \boldsymbol{\Phi}(\boldsymbol{\eta})|^2 \\
 & \leq \frac{C_{\theta,\boldsymbol{\omega}_{\mathbf{u}_0},\boldsymbol{\omega}_{\mathbf{v}_0}}}{\min\{ \sigma_\theta ,1\}}   \Big( |\partial_1 \boldsymbol{\Phi}(\boldsymbol{\eta})|^2 + |\sigma(\boldsymbol{\eta})| |\partial_2 \boldsymbol{\Phi}(\boldsymbol{\eta})|^2\Big),  \\
 |\partial_1 \sigma(\boldsymbol{\eta})|&
 \leq \frac{C_{\theta}}{\sigma_{\theta}} |\sigma(\boldsymbol{\eta})| 
\end{aligned}
\end{equation}
for constants $C_{\theta,\boldsymbol{\omega}_{\mathbf{u}_0},\boldsymbol{\omega}_{\mathbf{v}_0}}, C_{\theta} >0$ depending only on $\theta(\mathbf{x})$, 
$\boldsymbol{\omega}_{\mathbf{u}_0}(\mathbf{x})$ and $\boldsymbol{\omega}_{\mathbf{v}_0}(\mathbf{x})$, and for $\sigma_{\theta}$ from Eq.\;(\ref{eq:sigmaThetaProp}). Applying these estimates to Eq.\;(\ref{eq:chainInequalities}) yields 
\begin{equation}
\begin{aligned}
\frac{d E(\eta_1)}{d\eta_1} \leq   \frac{C^{\text{Tot}}_{\theta,\boldsymbol{\omega}_{\mathbf{u}_0},\boldsymbol{\omega}_{\mathbf{v}_0}}}{\min\{ \sigma_\theta ,1\}} E(\eta_1) + \frac{1}{2}\int_{-r}^{r} |\tilde{\mathbf{q}}(\boldsymbol{\eta})|^2  d \eta_2
\end{aligned}
\end{equation}
for a constant $C^{\text{Tot}}_{\theta,\boldsymbol{\omega}_{\mathbf{u}_0},\boldsymbol{\omega}_{\mathbf{v}_0}} > 0$. Since $E(0) = 0$ by the initial data in Eq.\;(\ref{eq:solveForHyperCase}), Gr\"onwall’s inequality gives the desired bound:
\begin{equation}
\begin{aligned}\label{eq:theKeyHypEst}
\max_{-r \leq  \eta_1 \leq r} E(\eta_1) \leq C\int_{(-r,r)^2} |\tilde{\mathbf{q}}(\boldsymbol{\eta})|^2  dA.
\end{aligned}
\end{equation}

The above inequality bounds the energy of a solution $\boldsymbol{\Phi}(\boldsymbol{\eta})$ to Eq.\;(\ref{eq:solveForHyperCase}). The analogous bound in Eq.\;(\ref{eq:bound-to-prove-N}) for the approximate solutions $\{\boldsymbol{\Phi}^{(N)}(\boldsymbol{\eta})\}$ follows from a similar argument using the ODE system in Eq.\;(\ref{eq:theODESystemProof}) in place of the PDE.  
Having eliminated the possibility that the sequence  $\{\boldsymbol{\Phi}^{(N)}(\boldsymbol{\eta})\}$ blows up, one finally defines $\boldsymbol{\Phi}(\boldsymbol{\eta}) := \lim_{N \rightarrow \infty} \boldsymbol{\Phi}^{(N)}(\boldsymbol{\eta})$ and checks that it solves the desired PDE. See \cite{evans2022partial} for details. This completes our construction of solutions to the PDE in Eq.\;(\ref{eq:finalManipPDE}) in the hyperbolic case, and finishes the proof of Proposition \ref{AuxPDEProp}.

\section{One-dimensional solutions of the surface theory}\label{sec:1DSolve}

This appendix shows how to solve the surface theory in Eq.\;(\ref{eq:GenPDESurf}) under the assumption that the actuation, bend and twist fields vary along a single direction. 
Consider a one-dimensional ansatz  for $\theta(\mathbf{x})$, $\kappa(\mathbf{x})$ and $\tau(\mathbf{x})$ of the form 
\begin{equation}\label{eq:1d-ansatz-def}
\begin{aligned}
\theta(\mathbf{x}) \equiv \theta(\mathbf{x} \cdot \tilde{\mathbf{s}}), \quad \kappa (\mathbf{x})  \equiv \kappa ( \mathbf{x} \cdot \tilde{\mathbf{s}}), \quad  \tau(\mathbf{x})  \equiv  \tau(\mathbf{x} \cdot \tilde{\mathbf{s}} ) 
\end{aligned}
\end{equation}
for a non-zero $\tilde{\mathbf{s}} \in \mathbb{R}^3$. 
Substituting the ansatz into Eq.\;(\ref{eq:GenPDESurf}) gives the system of first order ODEs
\begin{equation}
\begin{aligned}
\begin{cases}\label{eq:theODEForExamples}
\theta(s) \in (\theta^{-}, \theta^+)  \\ 
\frac{d}{ds} \Bigg\{ \theta'(s) \begin{pmatrix} \tilde{\mathbf{v}}_0 \cdot \tilde{\mathbf{s}} \\ -\tilde{\mathbf{u}}_0 \cdot \tilde{\mathbf{s}}  \end{pmatrix} \cdot  \boldsymbol{\Gamma}(\theta(s)) \begin{pmatrix} \tilde{\mathbf{u}}_0 \cdot \tilde{\mathbf{s}} \\ \tilde{\mathbf{v}}_0 \cdot \tilde{\mathbf{s}}  \end{pmatrix} \Bigg\}   =  \begin{pmatrix} \kappa(s)  \\ \tau(s) \end{pmatrix} \cdot \boldsymbol{\Lambda}(\theta(s))   \begin{pmatrix} \kappa(s)  \\ \tau(s) \end{pmatrix}  \\
\Big[(\tilde{\mathbf{s}} \cdot \tilde{\mathbf{v}}_0) \mathbf{M}_{\mathbf{v}_0}(\theta(s))  - (\tilde{\mathbf{s}} \cdot \tilde{\mathbf{u}}_0 ) \mathbf{M}_{\mathbf{u}_0}(\theta(s))    \Big]\begin{pmatrix} \kappa'(s)  \\ \tau'(s) \end{pmatrix}  =   \theta'(s) \mathbf{M}\Big(\theta(s),  \begin{pmatrix} \tilde{\mathbf{u}}_0 \cdot \tilde{\mathbf{s}} \\  \tilde{\mathbf{v}}_0 \cdot \tilde{\mathbf{s}} \end{pmatrix}  \Big)  \begin{pmatrix} \kappa(s)  \\ \tau(s) \end{pmatrix}  
\end{cases}
\end{aligned}
\end{equation}
where $s := \mathbf{x} \cdot \tilde{\mathbf{s}}$. This ODE can be solved using the initial conditions $(\theta'(0), \theta(0), \kappa(0), \tau(0)) = (\bar{\zeta}, \bar{\theta}, \bar{\kappa}, \bar{\tau})$ under the assumption that $\bar{\theta}$ satisfies 
\begin{equation}
\begin{aligned}
\bar{\theta} \in (\theta^{-}, \theta^+), \quad \det \Big[ (\tilde{\mathbf{s}} \cdot \tilde{\mathbf{v}}_0) \mathbf{M}_{\mathbf{v}_0}(\bar{\theta})  - (\tilde{\mathbf{s}} \cdot \tilde{\mathbf{u}}_0 ) \mathbf{M}_{\mathbf{u}_0}(\bar{\theta})  \Big] \neq 0, \quad \begin{pmatrix} \tilde{\mathbf{v}}_0 \cdot \tilde{\mathbf{s}} \\ -\tilde{\mathbf{u}}_0 \cdot \tilde{\mathbf{s}}  \end{pmatrix} \cdot   \boldsymbol{\Gamma}(\bar{\theta}) \begin{pmatrix} \tilde{\mathbf{u}}_0 \cdot \tilde{\mathbf{s}} \\ \tilde{\mathbf{v}}_0 \cdot \tilde{\mathbf{s}}  \end{pmatrix} \neq 0.
\end{aligned}
\end{equation} 
When these conditions hold,  Eq.\;(\ref{eq:theODEForExamples}) can be written in the standard  first order form $\dot{\boldsymbol{\Phi}}(s) = \mathbf{f}(\boldsymbol{\Phi}(s))$ through a change of variables. A  solution to the initial value problem therefore exists and is unique  on some interval $s \in (0, \bar{s})$,  $\bar{s} >0$, whose extent may depend on  $(\bar{\zeta}, \bar{\theta}, \bar{\kappa}, \bar{\tau})$.  Such solutions can be constructed numerically using a standard ODE solver. 

Fig.\;\ref{Fig:MiuraEggboxExample} from the main text shows examples of one-dimensional solutions for the Miura and Eggbox patterns. In these cases, the lattice vectors $\mathbf{u}(\theta)$ and $\mathbf{v}(\theta)$ are orthogonal and can be written as
\begin{equation}
\begin{aligned}\label{eq:orthogonalForSolving}
\mathbf{u}(\theta) = \lambda_u(\theta) \mathbf{e}_1, \quad \mathbf{v}(\theta) = \lambda_v(\theta) \mathbf{e}_2
\end{aligned}
\end{equation}
for scalar functions $\lambda_u(\theta), \lambda_v(\theta) >0$. These scalars are strictly monotonic in $\theta$, per Lemma \ref{SignsLemma}. Noting that $\mathbf{u}_0 = \mathbf{u}(0) = \lambda_u(0) \mathbf{e}_1$ and $\mathbf{v}_0 = \mathbf{v}(0) = \lambda_v(0) \mathbf{e}_2$,  we apply the one-dimensional ansatz in Eq.\;(\ref{eq:1d-ansatz-def}) with
\begin{equation}
\begin{aligned}
\tilde{\mathbf{s}} = \tilde{\mathbf{u}}^r_0 =(\lambda_u(0))^{-1}\tilde{\mathbf{e}}_1 \quad \text{ or } \quad \tilde{\mathbf{s}} = \tilde{\mathbf{v}}^r_0 = (\lambda_v(0))^{-1} \tilde{\mathbf{e}}_2.
\end{aligned}
\end{equation}
This produces the following ODEs for $(\theta(s), \kappa(s), \tau(s))$, where again  $s := \mathbf{x} \cdot \tilde{\mathbf{s}}$. Focusing on the  $\tilde{\mathbf{s}} = \tilde{\mathbf{v}}^r_0$ case --- the $\tilde{\mathbf{s}} = \tilde{\mathbf{u}}_0^r$ case is similar and omitted --- we substitute Eq.\;(\ref{eq:orthogonalForSolving}) into Eq.\;(\ref{eq:theODEForExamples}) to obtain 
\begin{equation}
\begin{aligned}
&\frac{d}{ds} \Big[ \frac{\lambda_u'(\theta(s))}{\lambda_v(\theta(s))}  \theta'(s) \Big] = \lambda_u(\theta(s)) \lambda_v(\theta(s)) \Big[ \tau^2(s) + \kappa^2(s) \lambda_u(\theta(s)) \lambda_v(\theta(s))  \lambda'_u(\theta(s)) \lambda'_v(\theta(s))\Big],\\
&\kappa'(s) = \theta'(s)\Big[-2   \frac{ \lambda_v'(\theta(s))}{\lambda_v(\theta(s))  } - \frac{\tfrac{d}{d\theta}\{(\lambda_u(\theta(s)) \lambda'_u(\theta(s))\}}{[\lambda_u(\theta(s))  \lambda_u'(\theta(s)) ] } \Big] \kappa(s),  \\
& \tau'(s)   = -2 \theta'(s)  \frac{\lambda_u'(\theta(s))}{\lambda_u(\theta(s))} \tau(s).
\end{aligned}
\end{equation}
The general solutions to the last two equations are  
\begin{equation}
\begin{aligned}\label{eq:tauKappaAnalytical}
\kappa(s) = \frac{c_{\kappa}}{\lambda_u(\theta(s)) \lambda^{2}_{v}(\theta(s))\lambda_{u}'(\theta(s))},\quad \tau(s) = \frac{c_{\tau}}{\lambda^2_u(\theta(s))}
\end{aligned}
\end{equation}
for arbitrary constants $c_{\kappa}, c_{\tau}$. The first equation becomes 
\begin{equation}
\begin{aligned}\label{eq:finalODE}
\frac{d}{ds}\Big[ \frac{\lambda_u'(\theta(s))}{\lambda_v(\theta(s))}  \theta'(s) \Big] = c_{\tau}^2 \frac{\lambda_v(\theta(s))}{\lambda_u^3(\theta(s))} + c_{\kappa}^2  \frac{ \lambda_v'(\theta(s))}{\lambda_v^2(\theta(s)) \lambda_u'(\theta(s))} . 
\end{aligned}
\end{equation} 
This ODE can be solved given choices of $c_{\kappa}, c_{\tau}$ and initial conditions $(\theta(0), \theta'(0)) = (\bar{\theta}, \bar{\zeta})$. 

We compute the soft modes in Fig.\;\ref{Fig:MiuraEggboxExample} using this ansatz as follows. First, we parameterize $\lambda_{u}(\theta)$ and $\lambda_v(\theta)$ according the mechanism of a given design.  Then, we substitute this parameterization  into the ODE in Eq.\;(\ref{eq:finalODE}) and solve it numerically for suitable choices of $c_{\kappa}, c_{\tau}$ and $(\theta(0), \theta'(0)) = (\bar{\theta}, \bar{\zeta})$. Next, we substitute the bend $\kappa(\mathbf{x}) \equiv \kappa(\mathbf{x} \cdot \tilde{\mathbf{v}}_0^r)$, twist $\tau(\mathbf{x}) \equiv \tau(\mathbf{x} \cdot \tilde{\mathbf{v}}_0^r$) and actuation  $\theta(\mathbf{x}) \equiv \theta(\mathbf{x} \cdot \tilde{\mathbf{v}}_0^r)$ fields back into Eq.\;(\ref{eq:omegaParamsSolveAlg}) to produce compatible $\boldsymbol{\omega}_{\mathbf{u}_0}(\mathbf{x})$ and $\boldsymbol{\omega}_{\mathbf{v}_0}(\mathbf{x})$. From these two vector fields,  we build the rotation field $\mathbf{R}_{\text{eff}}(\mathbf{x})$ numerically by  first solving the ODE 
 \begin{equation}
 \begin{aligned}\label{eq:firstODERotSolve}
 \frac{d}{d\eta_1} \mathbf{R}_{\text{eff}}(\eta_1 \tilde{\mathbf{u}}_0)  =  \mathbf{R}_{\text{eff}}(\eta_1 \tilde{\mathbf{u}}_0)\big[ \boldsymbol{\omega}_{\mathbf{u}_0}( \eta_1 \tilde{\mathbf{u}}_0) \times \big]  \quad \text{subject to }  \mathbf{R}_{\text{eff}}(\mathbf{0}) = \mathbf{I},
 \end{aligned}
 \end{equation}
and then by solving the complementary set of ODEs
 \begin{equation}
 \begin{aligned}\label{eq:secondODERotSolve}
& \frac{d}{d\eta_2} \mathbf{R}_{\text{eff}}(\bar{\eta}_1 \tilde{\mathbf{u}}_0 + \eta_2 \tilde{\mathbf{v}}_0)  =  \mathbf{R}_{\text{eff}}( \bar{\eta}_1 \tilde{\mathbf{u}}_0 + \eta_2 \tilde{\mathbf{v}}_0)\big[ \boldsymbol{\omega}_{\mathbf{v}_0}( \bar{\eta}_1 \tilde{\mathbf{u}}_0 + \eta_2 \tilde{\mathbf{v}}_0) \times \big]  
 \end{aligned}
 \end{equation}
using the initial condition  $\mathbf{R}_{\text{eff}}(\bar{\eta}_1 \tilde{\mathbf{u}}_0) $ supplied by the solution to Eq.\;(\ref{eq:firstODERotSolve}). This two-step approach produces the desired rotation fields because $\boldsymbol{\omega}_{\mathbf{u}_0}(\mathbf{x}), \boldsymbol{\omega}_{\mathbf{v}_0}(\mathbf{x})$ and $\theta(\mathbf{x})$ are compatible per Eq.\;(\ref{eq:SurfaceTheory}).
We construct the effective deformation $\mathbf{y}_{\text{eff}}(\mathbf{x})$ numerically in the same fashion, using the solved for $\theta(\mathbf{x})$ and $\mathbf{R}_{\text{eff}}(\mathbf{x})$. Finally, we substitute these effective fields into the construction in Section \ref{ssec:GlobalOrigami} to obtain the origami deformations. 


\bibliographystyle{unsrt}

\end{document}